\documentclass[ twoside]{article}

\usepackage{amsmath,amsthm,amssymb}
\usepackage{times}
\usepackage{enumerate}
\usepackage{hyperref}

\def\titlerunning#1{\gdef\titrun{#1}}
\makeatletter
\def\author#1{\gdef\autrun{\def\and{\unskip, }#1}\gdef\@author{#1}}
\def\address#1{{\def\and{\\\hspace*{18pt}}\renewcommand{\thefootnote}{}%
\footnote {#1}}%
\markboth{\autrun}{\titrun}}
\makeatother
\def\email#1{e-mail: #1}
\def\subjclass#1{{\renewcommand{\thefootnote}{}%
\footnote{\emph{Mathematics Subject Classification (2010):} #1}}}
\def\keywords#1{\par\medskip
\noindent\textbf{Keywords.} #1}

\newcommand{\eq}[1]{\eqref{#1}}

\newtheorem{theorem}{Theorem}[section]
\newtheorem{proposition}[theorem]{Proposition}
\newtheorem{lemma}[theorem]{Lemma}
\newtheorem{corollary}[theorem]{Corollary}
\newtheorem{sublemma}[theorem]{Sublemma}

\theoremstyle{definition}
\newtheorem{definition}[theorem]{Definition}
\newtheorem{remark}[theorem]{Remark}

\newtheorem*{quotethm}{}

\numberwithin{equation}{section}

\DeclareMathOperator{\supp}{supp}
\DeclareMathOperator{\tr}{tr}
\DeclareMathOperator{\Ran}{Ran}
\DeclareMathOperator{\dist}{dist}

\DeclareMathOperator{\essinf}{ess\,inf}

\DeclareMathOperator{\diam}{diam}

\newcommand{\pr}{\prime}

\newcommand\R{\mathbb R}
\newcommand\N{\mathbb N}
\newcommand\C{\mathbb C}
\newcommand\Z{\mathbb Z}
\newcommand\G{\mathbb{G}} 

\renewcommand\P{\mathbb P}
\newcommand\E{\mathbb E}

\renewcommand\H{\mathcal{H}}
\renewcommand\L{\mathrm{L}}
 
\newcommand\D{\mathcal{D}}

\newcommand\I{\mathcal{I}}

\newcommand{\cX}{\mathcal{X}}
\newcommand{\cP}{\mathcal{P}}
\newcommand{\cQ}{\mathcal{Q}}
\newcommand{\cC}{\mathcal{C}}
\newcommand{\cJ}{\mathcal{J}}
\newcommand{\cD}{\mathcal{D}}
\newcommand{\cR}{\mathcal{R}}
\newcommand{\cG}{\mathcal{G}}
\newcommand{\cE}{\mathcal{E}}
\newcommand{\cA}{\mathcal{A}}
\newcommand{\cS}{\mathcal{S}}
\newcommand{\cL}{\mathcal{L}}
\newcommand{\cZ}{\mathcal{Z}}
\newcommand{\cF}{\mathcal{F}}

\newcommand{\cN}{\mathcal{N}}
\newcommand{\cB}{\mathcal{B}}
\newcommand{\cW}{\mathcal{W}}
\newcommand{\cY}{\mathcal{Y}}
\newcommand{\cT}{\mathcal{T}}
\newcommand{\cM}{\mathcal{M}}
\newcommand{\cU}{\mathcal{U}}

\newcommand{\bE}{\boldsymbol{E}}
\newcommand{\Pb}{\boldsymbol{P}}
\newcommand{\W}{\boldsymbol{W}}

\newcommand\bt{\boldsymbol{t}}

\newcommand\bfe{\boldsymbol{e}}

\newcommand\di{\mathrm{d}}
\newcommand\e{\mathrm{e}}

\newcommand\beps{\boldsymbol{\varepsilon}}
\newcommand{\bom}{\boldsymbol{\omega}}
\newcommand{\bzt}{\boldsymbol{\zeta}}

\newcommand{\bta}{\boldsymbol{\eta}}
\newcommand{\bSig}{\boldsymbol{\Sigma}}

\newcommand{\om}{\omega}
\newcommand\eps{\varepsilon}
\newcommand{\vrho}{\varrho}
\newcommand\La{\Lambda}
\newcommand{\vphi}{\varphi}
\newcommand{\zt}{\zeta}
\newcommand{\vs}{\varsigma}

\newcommand\Chi{\raisebox{.2ex}{$\chi$}}

\newcommand{\la}{\langle}
\newcommand{\ra}{\rangle}

\newcommand{\abs}[1]{\left\lvert #1 \right\rvert}
\newcommand{\norm}[1]{\left\lVert #1 \right\rVert}
\newcommand{\scal}[1]{\left\langle #1 \right\rangle}
\newcommand{\set}[1]{\left\{ #1 \right\}}
\newcommand{\pa}[1]{\left( #1 \right)}
\newcommand{\br}[1]{\left[ #1 \right]}
\newcommand{\rb}[1]{\left] #1 \right[}
\newcommand{\hnorm}[1]{\left\{ \!\left\{ #1\right\}\! \right\}}

\newcommand{\up}[1]{^{(#1)}}

\newcommand\beq{\begin{equation}}
\newcommand\eeq{\end{equation}}

\newcommand{\Vper}{V_{\mathrm{per}}}

\newcommand{\qtx}[1]{\quad\text{#1}\quad}

\begin{document}

\titlerunning{Localization  for continuous Anderson models}
\title{A comprehensive proof of localization for continuous Anderson models with singular random potentials}

\author{Fran\c cois Germinet
\and 
Abel Klein}

\date{}

\maketitle

\address{F. Germinet: Universit\'e de Cergy-Pontoise, IUF,
UMR 8088 CNRS, D\'epartement de Math\'ematiques,
F-95000 Cergy-Pontoise, France;  \email{francois.germinet@u-cergy.fr}
\and \!\!A. Klein:
University of California, Irvine,
Department of Mathematics,
Irvine, CA 92697-3875,  USA;
 \email{aklein@uci.edu}}
 

\begin{abstract}
We study continuous Anderson Hamiltonians with  non-degenerate 
single site probability distribution of bounded support, without  any regularity  condition on the single site probability distribution. We prove the existence
of  a strong form of localization   at the bottom of the spectrum, which includes Anderson localization (pure point spectrum with exponentially decaying eigenfunctions) with finite multiplicity of eigenvalues,  dynamical localization (no spreading of wave packets under the time evolution),  decay of eigenfunctions correlations, and decay of the Fermi projections.
   We also prove  log-H\" older continuity of the integrated density of states at the bottom of the spectrum.

 \end{abstract}


 \setcounter{tocdepth}{2}
 \tableofcontents


 \normalsize
\section*{Introduction}
\addcontentsline{toc}{section}{Introduction}

Anderson Hamiltonians are 
 alloy-type   random Schr\"odinger 
operators on $\mathrm{L}^2(\mathbb{R}^d)$   that model the motion of an electron moving in a randomly disordered crystal.   They  are the continuous analogue of the Anderson  model,  a random Schr\"odinger 
operator on $\ell^{2}(\Z^{d})$.

In this paper we prove
 a strong form of localization  at  the bottom of the spectrum for  Anderson Hamiltonians with a non-degenerate single site probability distribution with compact support, without  any regularity  condition on the single site probability distribution.  This strong form of localization includes   Anderson localization (pure point spectrum with exponentially decaying eigenfunctions) with finite multiplicity of eigenvalues,  dynamical localization (no spreading of wave packets under the time evolution),  decay of eigenfunctions correlations, and decay of the Fermi projections.  We also prove  log-H\" older continuity of the integrated density of states at the bottom of the spectrum.

Localization for random Schr\"odinger operators was first established in the celebrated
paper by  Gol'dsheid,  Molchanov and Pastur \cite{GMP} for a certain
one dimensional continuous random Schr\"odinger operator.  Localization is by now well established for one and quasi-one  random Schr\"odinger operators
\cite{KS,Lac,KlMP,CKM,KLS,Stolz,DSS1}.

In the multi-dimensional case there is a wealth of results concerning localization    for the (discrete) Anderson model and the (continuous) Anderson Hamiltonian as long as the single site probability distribution has enough regularity (absolutely continuous with a bounded density,    H\"older continuous,   log-H\"older continuous).  In this case  Anderson and dynamical localization are well established,   e.g., \cite{FS,MS,FMSS,DLS,SW,SVW,vD,vDK,S,vDK2,AM,Klong,FKlat,A,ASFH,W2,Klweak,HM,CH1,Klop95,GDB,FKac,KSS,KSS2,DS,GKboot,GKfinvol,GKhighdis,AENSS,Kle}.   Localization is also known in a random displacement model where the displacement probability distribution has a bounded density \cite{Klop93, GhK,KLNS},
    for a class of Gaussian random potentials \cite{FLM,U,LMW}, and    for Poisson models where the single-site potentials are multiplied by random variables with bounded densities   \cite{MS2,CH1}. What all these results have in common is the availability of random variables with  sufficiently regular probability distributions,  which can be exploited, in an averaging procedure,  to produce an \emph{a priori} Wegner estimate at all scales (an estimate on the probability of energy resonances in finite volumes), e.g.,  \cite{Weg,FS,HM,CKM,CH1,Klop95,CHM,Kir,FLM,St,CHN,CHKN,CHK1,CHK2}.

In contrast, 
for the  most natural  random Schr\"odinger operators on the continuum (cf.\  \cite[Subsection~1.1]{LGP}),  the Bernoulli-Anderson Hamiltonian (simplest disordered substitutional alloy) and the Poisson Hamiltonian (simplest disordered amorphous medium),  
 localization results in two or more dimensions were much harder to obtain.   The  Bernoulli-Anderson Hamiltonian is an Anderson Hamiltonian where  the single site probability distribution is the distribution of a Bernoulli random variable, and the Poisson Hamiltonian
 is a random Schr\"odinger operator corresponding to identical impurities  placed at locations given by a homogeneous Poisson point process on $\R^{d}$.
 In both cases the  random variables with regular probability distributions are not available, so there is no a priori Wegner estimate.

Bourgain and Kenig \cite{BK} proved Anderson localization at the bottom of the spectrum for   the Bernoulli-Anderson Hamiltonian.  In their remarkable paper the Wegner estimate is established   by a multiscale analysis using  ``free sites" and a new quantitative version of the unique  continuation principle which gives  a lower bound on eigenfunctions. 
Since this Wegner estimate  has weak probability estimates and  the underlying random variables are discrete, they also introduced a new method to prove Anderson localization from estimates on the finite-volume resolvents given by a single energy multiscale analysis.  The new
method does not use spectral averaging as in \cite{DLS,SW,CH1}, which requires random variables with bounded densities.  It is also not an energy-interval multiscale analysis as in \cite{FMSS,vDK,FKac,GKboot,Kle}, which requires better probability estimates.

Germinet, Hislop and Klein \cite{GHK1, GHK2, GHK3} established Anderson localization at the bottom of the spectrum for the Poisson Hamiltonian,  
using   a multiscale analysis that 
exploits  the   probabilistic properties of Poisson point processes to
 control  the randomness of the configurations, and at the same time allows  the  use of the  new ideas introduced by  Bourgain and Kenig.

 Aizenman, Germinet, Klein, and Warzel \cite{AGKW} 
used a Bernoulli decomposition for random variables to show that  spectral localization (pure point spectrum with probability one)  for  Anderson Hamiltonians  follows from an extension of the Bourgain-Kenig results to   nonhomogeneous Bernoulli-Anderson Hamiltonians,   which incorporate an additional 
background potential  and allow
the variances of the Bernoulli terms  not to be identical but   
only uniformly positive.  Such random Schr\"odinger operators are generalized Anderson Hamiltonians as in Definition~\ref{defgenAndH}, for which we prove  Anderson and dynamical localization in this paper, thus providing a proof of  the required extension  stated  in \cite[Theorem~1.4]{AGKW}.

In this article we provide a comprehensive proof of localization for Anderson Hamiltonians, drawing on the methods of   \cite{FS,FMSS,vDK,CH1,FKac,GKboot,GKsudec,Kle} and incorporating the new  ideas of \cite{BK}.  We  make no assumptions on the  single site probability distribution except for  compact support. (The proof can be extended to distributions of unbounded support with appropriate assumptions on the tails of the distribution.)  We perform a  multiscale analysis to obtain  probabilistic statements about restrictions of the Anderson Hamiltonian to finite volumes.    From the conclusions of the multiscale analysis we extract  an infinite volume characterization of localization:  a  probabilistic statement concerning the generalized eigenfunctions of  the (infinite volume) Anderson Hamiltonian, from which we derive both  Anderson and dynamical localization, as well as other consequences of localization, such as decay of eigenfunctions correlations (e.g., SULE, SUDEC) and decay of the Fermi projections.

  This new infinite volume description of localization (given in Theorem~\ref{thmthemainev}(\ref{probloc})) yields all the manifestations of localization that have been previously derived  from  the energy interval multiscale analysis for sufficiently regular single site probability  distribution    \cite{FMSS,vDK,GDB,DS,GKboot,GKsudec,Kle}.  This description may also be derived  from the energy interval multiscale analysis (see Remark~\ref{remenergyint});  it is  implicit in \cite{GKsudec}.  One of the main achievements of this paper is the extraction of  such a clean and simple statement of localization for Bernoulli and other singular single site probability distributions.

We give a detailed account of this single energy multiscale analysis, which uses  `free sites' and the  quantitative unique continuation principle as in \cite{BK} to obtain control of the finite volume resonances.  We also explain in  detail how  all forms of localization can be extracted from this  single energy multiscale analysis.  
To put this extraction in perspective,   Fr\" ohlich and Spencer,  in their seminal paper \cite{FS}, obtained  a single energy multiscale analysis
 for the discrete  Anderson model  with good probability estimates, but  were not able to derive Anderson localization from their result. The desired localization was later obtained from a multiscale analysis by   two different methods.   Spectral averaging gets  Anderson localization from a single energy multiscale analysis as in \cite{FS}, but requires  absolutely
continuous single site probability  distributions with a bounded density  \cite{DLS,SW,CH1}.  Anderson localization, and later dynamical localization, can be proven from an  energy interval multiscale analysis using generalized eigenfunctions \cite{FMSS,vDK,DS,GKboot,Kle}. None of these methods were available in Bourgain and Kenig's setting.  Spectral averaging is
 not feasible for
Bernoulli random variables, and the  energy interval multiscale analysis  requires better probability estimates than possible using the quantitative unique continuation principle. In response, Bourgain and Kenig developed a new method for obtaining Anderson localization from a single energy multiscale analysis,  using Peierl's argument, generalized eigenfunctions, and two energy reductions [BouK, Section 7] . (Their method is simpler in the setting of [FrS], where the second energy reduction is not needed--see Remarks~\ref{remp>1} and \ref{remp>3}.)  In this paper we combine the ideas of [BouK, Section 7] with methods we developed in [GK1, GK6] to extract all forms of localization from a  single energy multiscale analysis, giving a detailed account of all steps.

We also derive  log-H\" older continuity of the integrated density of states from the conclusions of the multiscale analysis. The multiscale analysis requires  the probabilistic control of finite volume resonances subexponentially close to the given energy (and no more, as  noted in \cite{vDK}). In \cite{BK} and in this article, this control is obtained as part of the multiscale analysis. We show that, in the presence of a  multiscale analysis,  log-H\" older continuity of the integrated density of states is the infinite volume trace of this probabilistic control  (the `Wegner estimate').

The integrated density of states of  
 the discrete Anderson model   is always  log-H\" older continuous \cite{CS}. 
If  the single site probability distribution  is continuous  (i.e., it has no atoms), then   the integrated density of states for  both  discrete  Anderson models and continuous Anderson Hamiltonians has at least as much regularity as the concentration function of this probability distribution \cite{CHK2}. 
  Although for the discrete Anderson model there is an easy proof of continuity of the integrated density of states for arbitrary single site probability distribution  \cite{DeS}, for the continuous Anderson Hamiltonian  it is not even known if the integrated density of states is always a continuous function  if this probability distribution has an atom.

Neither Anderson localization nor dynamical localization carry information about the regularity of the integrated density of states.  
Roughly speaking, dynamical localization and regularity of the integrated density of states carry complementary types of information. This  is made more precise in \cite{GK5}, where we showed that for Anderson Hamiltonians    with an a priori  Wegner estimate, dynamical localization is necessary and sufficient to perform a multiscale analysis. The multiscale analysis contains more information than just localization properties: it also encodes  regularity of the integrated density of states. This fact has been overlooked,  since,
previous to the multiscale analysis in \cite{BK}, 
all multiscale analyses for Anderson models were performed with an a priori  Wegner estimate which readily implied regularity of the integrated density of states, even without localization.  In view of our results in \cite{GK5}, we may argue that,  by proving both localization and log-H\" older continuity of the integrated density of states,  we have extracted  from the multiscale analysis all the encoded information.   This `philosophical' remark would become a mathematical statement if we could prove  that localization combined with the log-H\" older continuity of the  integrated density of states is enough to start a multiscale analysis, extending the results of \cite{GK5} to the setting of this article.

 The strong  localization results, including  Anderson localization and  dynamical localization, and the  log-H\" older continuity of the integrated density of states,
 presented in this paper  for Anderson Hamiltonians, are also valid for  Poisson Hamiltonians using the probabilistic properties of Poisson point processes to
 control  the randomness of the configurations as in \cite{GHK2}.
 
It remains a challenge to prove localization for other random Schr\"odinger operators with  no assumptions on the  single site probability distribution except for  compact support (e.g., for  a Bernoulli distribution).  In particular, there is no proof of localization for the multidimensional discrete Bernoulli-Anderson model, for which     everything in \cite{BK} and this paper is valid  except for the quantitative unique continuation principle;   there is no unique continuation principle for discrete Schr\" odinger operators, where non-zero eigenfunctions may vanish on arbitrarily large sets \cite[Theorem~2]{J}.  The same applies to random Landau Hamiltonians \cite{CH2,W1,GKS,GKS2,GKM}, where, although the unique continuation principle holds,  an appropriate quantitative unique continuation principle is missing. (There is a quantitative unique continuation principle for Landau Hamiltonians, but it comes with the exponent $2$ instead of $\frac 4 3$ \cite{Davey}.  The multiscale analysis requires an exponent $<\frac {1+\sqrt{3}} 2$, as discused in Remark~\ref{remQUCP}. Note that $\frac 4 3<\frac {1+\sqrt{3}} 2< 2$.) The same is also true for a continuous alloy-type   random Schr\"odinger 
operators with single site potentials of indefinite sign \cite{Klop95,KloN,HK}, where, although we have the  quantitative unique continuation principle, it cannot be used to  control  the finite volume resonances.

This article is organized  as follows:

\begin{enumerate}[1]
\item  \textbf{Main results:}  In Section~\ref{secmainr} we define Anderson Hamiltonians and state our main results, Theorem~\ref{thmthemainev} and  Corollary~\ref{corallloc}.

\item    \textbf{Anderson Hamiltonians:}  In Section~\ref{secgenAnd}  we introduce (normalized) generalized Anderson Hamiltonians, finite volume operators,  and prove some  basic deterministic  properties.  We always work with generalized Anderson Hamiltonians in the following sections. 

\item    \textbf{Preamble to the multiscale analysis:}   In Section~\ref{secpreambleMSA}  we introduce the machinery  for the multiscale analysis.  We define `good boxes', `free sites', `suitable coverings'  of boxes and annuli, and prove some basic lemmas.  

\item   \textbf{The multiscale analysis with a Wegner estimate:}  Section~\ref{sectMSAWegner} is devoted to  the multiscale analysis;  Theorem~\ref{thmMSA} states the full result at the bottom of the spectrum.  Proposition~\ref{propInitial} gives  a priori finite volume estimates at the bottom of the spectrum  that yield the starting condition for the multiscale analysis.  The single energy multiscale analysis with a Wegner estimate is performed in Proposition~\ref{propA} on any energy  interval where we have a priori finite volume estimates. 

\item   \textbf{Preamble to localization:} In Section~\ref{secpreambleloc}  we introduce tools for extracting localization from the multiscale analysis.  We discuss generalized eigenfunctions and the the generalized eigenfunction expansion, and show that generalized eigenfunctions are small in good boxes (eg., Lemma~\ref{lemgoodW}).

\item    \textbf{From the multiscale analysis to localization:}  In Section~\ref{secMSAloc} we extract localization from the multiscale analysis.   We assume that the conclusions of the multiscale analysis  (i.e., of Proposition~\ref{propA}) hold  for all energies in   a bounded open interval (not necessarily at the bottom of the spectrum), and derive localization in that interval.  Theorem~\ref{thmmainev}  encapsulates all forms of  localization. 

\item    \textbf{Localization:} In Section~\ref{secloc} we extract the usual forms of localization from Theorem~\ref{thmmainev}.   Anderson localization and finite multiplicity of eigenvalues is proven in 
Theorem~\ref{thmAndloc}.  Eigenfunctions correlations  (e.g., SUDEC, SULE) are obtained with probability one  in Theorem~\ref{thmSUDEC} and in expectation in Theorem~\ref{thmsgenSUDEC}. Dynamical localization and decay of Fermi projections are proved with probability one  in Corollary~\ref{cordynloc}  and in expectation   in   Corollary~\ref{coroldynlocexp}.

\item   \textbf{Log-H\" older continuity of the integrated density of states:} In Section~\ref{seclogH} we derive log-H\" older continuity of the integrated density of states from the multiscale analysis with a Wegner estimate; see Theorem~\ref{logHolder}.  

\item[A]  \textbf{A quantitative unique continuation principle for Schr\"odinger operators:} In  Appendix~\ref{appendixQUP}  we rewrite   Bourgain and Kenig's quantitative unique continuation principle
for Schr\"odinger operators, i.e.,    \cite[Lemma~3.10]{BK},  in a convenient form   for our purposes; see Theorem~\ref{thmucp}  and Corollary~\ref{corQUCPD}.
We also give an application of this quantitative unique continuation principle to periodic Schr\"odinger operators, providing an alternative proof to Combes, Hislop and Klopp's  lower bound estimate concerning periodic potentials and  spectral projections \cite[Theorem~4.1]{CHK1}.

\end{enumerate}

 \section{Main results}\label{secmainr}
 We start by defining Anderson Hamiltonians. 

\begin{definition} \label{defAndHam} An \emph{Anderson Hamiltonian} is a  random Schr\"odinger 
operator on 
$\mathrm{L}^2(\mathbb{R}^d)$ of the form
\beq\label{AndH}
H_{\bom}: =  -\Delta + V_{\mathrm{per}} +
V_{\bom} ,
\eeq
where
\begin{enumerate}
\item 
$\Delta$ is the $d$-dimensional Laplacian operator, 
\item $V_{\mathrm{per}}$ is a bounded periodic potential with period $q \in \N$,

\item$V_{\bom}$ is an alloy-type random potential,
\beq
V_{\bom} (x):= 
\sum_{\zeta \in \Z^d} \omega_\zeta \,  u(x - 
\zeta),\label{AndV}
\eeq
where  
\begin{enumerate}
\item
the single site potential $u$ is a  nonnegative bounded 
measurable function
on $\R^{d}$ with compact support, uniformly 
bounded away from zero in
a neighborhood of the origin,
\item 

$\bom=\{ \omega_\zeta \}_{\zeta\in
\Z^d}$ is a family of independent 
identically distributed random
variables  whose  common probability 
distribution $\mu$ is non\--dege\-nerate with bounded support.

\end{enumerate}
\end{enumerate}
\end{definition} 

Given an Anderson Hamiltonian $H_{\bom}$,  we set  $P_{\bom}(B):=\Chi_{B}(H_{\bom})$ for a Borel set $B \subset \R^{d}$,  $P_{\bom}({E}):= P_{\bom}(\{{E}\})$ and    $P_{\bom}^{(E)}:=P_{\bom}(]-\infty,E])$ for ${E} \in \R$.

An Anderson Hamiltonian  $H_{\bom}$  is a $q\Z^d$-ergodic family of
random self-adjoint operators ($q=1$ if $V_{\mathrm{per}}=0$).
It follows   (see \cite{KiM,CL,PF})
that there exists fixed subsets $\Sigma$,  $\Sigma_{\mathrm{pp}}$, $\Sigma_{\mathrm{ac}}$  and $\Sigma_{\mathrm{sc}}$ of $\R$ so that the spectrum $\sigma(H_{\bom})$
of $H_{\bom}$,  as well as its pure point, 
absolutely continuous, and singular continuous  components,
are equal to these fixed sets with probability one.  We let $E_{\inf} = \inf \Sigma> - \infty$,  the bottom of the non-random spectrum; note that there exists $E_{1}> E_{\inf}$ such that $[E_{\inf}, E_{1}] \subset \Sigma$ \cite{KiM2}.

We will use the following notation:
\begin{itemize}
\item 
 Given $x=\pa{x_1,x_2,\ldots, x_d} \in \R^d$, we set
\beq \label{supnorm}
\norm{x}:= \max \set{\abs{x_1},\abs{x_2},  \ldots,\abs{x_d}}\quad \text{and}\quad \scal{x} := \pa{1 + \norm{x}^{2}}^{\frac 12}.
\eeq
\item 
 Given $\nu>0$  and $y \in \R^{d}$, we let 
 $T_{\nu,y} $  be  the operator on   $\mathrm{L}^2(\mathbb{R}^d)$ given by multiplication by the function
$T_{\nu,y}(x):= \la x-y\ra^{\nu}$. We set  $\scal{X-y}:= T_{1,y}$ and  $T_{\nu}:=T_{\nu,0}=\scal{X}^\nu$. 

\item
  We let  
\beq
\Lambda_{L}(x):= \set{y \in \R^d; \; \norm{y-x} < \tfrac L 2} =x + \left]-\tfrac L 2,\tfrac L 2\right[^d 
\eeq
 denote  the  (open)
box of side $L$ centered at $x \in \R^{d}$.  By a box $\Lambda_L$ we will mean a box $\Lambda_{L}(x)$ for some $x \in \R^d$. We write $\overline{\Lambda}_L=\overline{\Lambda_L} $ for the closed box. Given scales $L_{1} < L_{2}$, we consider the (open)  annulus 
\beq
\La_{L_{2},L_{1}}(x):=\La_{L_{2}}(x)\setminus {\overline{\La}}_{L_{1}}(x) =
\set{y \in \R^d; \; \tfrac {L_{1}} 2 <  \norm{y-x} < \tfrac {L_{2}} 2},
\eeq
and  let  $\overline{\La}_{L_{2},L_{1}}(x):=\overline{\La_{L_{2},L_{1}}(x)}$ be the closed annulus.

\item 
Given a set $B$, we write  $\Chi_B$ for its characteristic function. 

\item   $\Chi_x$ will  denote the characteristic function of the unit box centered at $x \in \R^d$, i.e., $\Chi_x:=\Chi_{\Lambda_1(x)} $.

\item The cardinality of a set $A$ will be denoted by $\# A$.

\item Given a Borel set  ${\Xi} \subset \R^d$, we will denote its Lebesgue measure by $\abs{{\Xi}}$.

\item  We will use the notation  $\sqcup$ for disjoint unions:  given sets $A$ and $ B$,  then  $C=A\sqcup B$ means that  $C=A \cup B$ and  $A \cap B =\emptyset$.

\item  We let $\cB_{b}$ denote the collection of bounded complex-valued Borel functions on $\R$,  and set
 $\cB_{b,1}:=\set{f \in \cB_{b};  \ \sup_{t \in \R} \abs{f(t)}\le 1}$.

\item Given an open set ${\Xi} \subset \R^d$ and $n \in \N \cup \set{\infty}$, $C^n(\Xi)$ will denote the collection of  $n$-times continuously differentiable complex-valued functions on $\Xi$, with $C_c^n(\Xi)$ denoting the subset of functions with compact support.

\item 
  By a  constant  we will always mean a finite constant.  We will use  $C_{a,b, \ldots}$, $C^{\pr}_{a,b, \ldots}$,  $C(a,b, \ldots)$, etc., to  denote a constant depending only on the parameters
$a,b, \ldots$.
 
\end{itemize}

We prove  a probabilistic statement about the generalized eigenfunctions of an Anderson Hamiltonian, from which we will derive all the usual statements about localization.   Generalized eigenfunctions, originally used by Martinelli and Scoppola \cite{MS}  to extract absence of absolutely continuous from the multiscale analysis, have been an indispensable tool in all proofs of localization that do not use spectral averaging \cite{FMSS,vDK,GKboot,Kle,BK}.

 Let $H_{\bom}$ be an Anderson Hamiltonian on $\mathrm{L}^2(\mathbb{R}^d)$ and fix 
 $\nu>0 $.    A generalized eigenfunction for a realization $H_{{\bom}}$ (i.e., we fix the values of the random variables  $\bom$) with generalized eigenvalue ${E} \in \R$  is a  measurable function $\psi$ on $\R^{d}$, with $ 0< \norm{T_{\nu}^{-1}\psi} < \infty$,
satisfying the eigenvalue equation for ${E}$ in the weak sense, i.e.,
\beq  \label{eigb100}
\scal{H_{{\bom}}\vphi,\psi}= {E} \scal{\vphi,\psi} \quad \text{for all}\quad \vphi \in C_{c}^{\infty}(\R^{d}).
\eeq   
We will denote by
$\Theta\up{\nu}_{\bom}({E})$ the collection of  generalized eigenfunctions for $H_{{\bom}}$ with generalized eigenvalue ${E}$.

To detect localization  for a realization $H_{{\bom}}$, we introduce quantities that measure  the concentration of the  generalized eigenfunctions with generalized eigenvalue ${E}$ in certain  subsets of $\R^{d}$.  Given $x\in \R^{d}$,  we will measure this concentration  at $x$ by
 \begin{align} \label{defGWxint}
W\up{\nu}_{\bom,x}({E}):=\begin{cases} 
{\sup_{\psi \in \Theta\up{\nu}_{\bom}({E})} }
\ \frac {\| \Chi_{x}\psi \|}
{\|T_{\nu,x}^{-1}\psi \|}&
 \text{if $\Theta\up{\nu}_{\bom}({E})\not=\emptyset$}\\0 & \text{otherwise}\end{cases},
\end{align}
 and at  an annulus around $x$ at scale  
 $L\ge 1$ by 
  \begin{align} \label{defGWxLint}
W\up{\nu}_{\bom,x,L}({E}):=\begin{cases} {\sup_{\psi \in \Theta\up{\nu}_{\bom}({E})} }
\frac {\norm{ \Chi_{x,L}\psi}}
{\norm{T_{\nu,x}^{-1}\psi }}&
 \text{if $\Theta\up{\nu}_{\bom}({E})\not=\emptyset$}\\0 & \text{otherwise}
\end{cases},
\end{align}
where $ \Chi_{x,L}:= \Chi_{\La_{2L +1,L-1}(x)}$.  (For technical reasons we  will need an annulus slightly bigger than $ \Chi_{\La_{2L,L}(x)}$.) We always have $0\le W\up{\nu}_{\bom,x}({E})\le \left( \tfrac 5 4\right)^{\frac {\nu} 2}<2^{\frac {\nu} 2}  $ and  $0\le W\up{\nu}_{\bom,x,L}({E})\le 2^{\frac {\nu} 2} L^{\nu}$.
We will work with a fixed $\nu> \frac d 2$, but note that $W\up{\nu}_{\bom,x}({E})$ and $W\up{\nu}_{\bom,x,L}({E})$ are increasing in $\nu$.

We also prove  log-H\"older continuity of the integrated density of states. The integrated density of states $N(E)$ for an Anderson Hamiltonian $H_{\bom}$, usually defined  through the infinite volume limit of the normalized eigenvalue counting functions
of appropriate restrictions to finite volumes (e.g., \cite{CL,PF}), equals (e.g.,  \cite{DIM}) 
\beq
N(E) =\frac 1 {q^{d}}\,  \E\set{ \tr \pa{\Chi_{\Lambda_q(0)}P_{\bom}^{(E)}\Chi_{\Lambda_q(0)}} }\qtx{for} E \in \R.
\eeq

The following theorem contains our main results; item (\ref{probloc}) encapsulates localization for Anderson Hamiltonians.

 \begin{theorem}\label{thmthemainev}
 Let $H_{\bom}$ be an Anderson Hamiltonian on $\mathrm{L}^2(\mathbb{R}^d)$. For each $p \in \rb{\frac 1 3, \frac 3 8}$ there exists an energy  $E_0> E_{\inf}$ such that the following holds for all  $\widetilde{p}\in ]0,p[$:
  \begin{enumerate}[\upshape (A)]

 \item  \label{logHoldc} The integrated density of states $N(E)$ is locally log-H\"older continuous of order $\widetilde{p}d$ in the interval $ [E_{\inf},E_{0}[$, i.e.,  for all $\widetilde{p} \in ]0,p[$ and compact intervals $I \subset  [E_{\inf},E_{0}[$ with length $\abs{I}\le \frac 1 2$ we have
  \beq\label{Nxest11}
\abs{N(E_{2})-N(E_{1})} \le \frac {C_{\widetilde{p},I}} {\abs {\log \abs{E_{2}-E_{1}} }^{\widetilde{p}d}} \qtx{for all} E_{1},E_{2} \in I.
\eeq

\item \label{probloc}  Let $\vartheta=\frac 1 2 \rho^{n_1}$ for some $\rho \in ]\frac 1 {1+p} ,1[$  and $n_1 \in\N$  with $(n_1+1) \rho^{n_1} < p-\widetilde{p}$. There exists a constant  $M>0$ so, fixing $\nu> \frac d 2$, there is a finite scale $L_{0}$ such that for all
 $L \ge L_{0}$ and $x_{0}\in \R^{d}$    there exists an event $\cU_{L,x_{0}}$ with the following properties:
 \begin{enumerate}[(i)]
\item $\cU_{L,x_{0}}$  depends only on the random variables $\set{\omega_{\zeta}}_{\zeta \in \Lambda_{\frac {1001L}{500}}(x_{0})}$,  and  
\beq\label{cUdesiredint}
  \P\set{\cU_{L,x_{0}} }\ge  1 - L^{- \widetilde{p}d}.
\eeq

\item If  $\bom \in\cU_{L,x_{0}}$,   for all  ${E} \in [E_{\inf},E_{0}[ $ we have that 
 \beq \label{distEIred-int}
\text{either}\quad  W_{\bom,x_{0}}\up{\nu}({E}) \le   \e^{-  M L^{\vartheta} } \qtx{or}  W_{
\bom,x_{0},L}\up{\nu}({E}) \le \e^{- M L}.
 \eeq
In particular, for all   $\bom \in\cU_{L,x_{0}}$   we have 
\beq \label{EDI9L99int}
W_{\bom,x_{0}}\up{\nu}({E})W_{\bom,x_{0},L}\up{\nu}({E}) \le \e^{-  \frac 1 2 {M}  L^{\vartheta}}\qtx{for}  {E} \in  [E_{\inf},E_{0}  [  .
\eeq

\end{enumerate}

\end{enumerate}
  \end{theorem}

\begin{remark}\label{remmainev}  The conclusions of Theorem~\ref{thmthemainev} hold on 
any bounded open interval $\I$ in which we  verify the starting condition (i.e., hypotheses) for the multiscale analysis of Proposition~\ref{propA}.    Theorem~\ref{thmthemainev} is stated for an interval at the bottom of the spectrum,  where the starting condition for the multiscale analysis is derived from Lifshitz tails estimates in Proposition~\ref{propInitial}.   This starting condition, and hence the analogue of Theorem~\ref{thmthemainev},  can also be proved  in intervals at the edge of spectral gaps,  similarly to Proposition~\ref{propInitial},   using the  internal Lifshitz tails estimates given in \cite{Klop99}.   This starting condition is also derived in  Proposition~\ref{propInitialHD} for a fixed interval at the bottom of the spectrum at high disorder, provided $\mu\pa{\set{\inf \supp \mu}}=0$, and the conclusions of Theorem~\ref{thmthemainev} hold in this fixed interval at high disorder if $\mu([{\inf \supp \mu},{\inf \supp \mu} +t])\le C t^\gamma$, with $\gamma >0$ appropriately  large. Note  that Theorem~\ref{thmthemainev} holds also if the single site potential $u$ in Definition~\ref{defAndHam} is assumed to be nonpositive instead of nonnegative, since in this case replacing $u$ by $-u$ and $\mu$ by $\widetilde{\mu}$, where $\widetilde{\mu}(B)=\mu(-B)$, rewrites the random Schr\" odinger operator as an Anderson Hamiltonian as in Definition~\ref{defAndHam}.
\end{remark}

Theorem~\ref{thmthemainev}(\ref{logHoldc}) says that in the interval $ [E_{\inf},E_{0} [$ (more generally, in the interval where we have a multiscale analysis) the  integrated density of states  $N(E)$ is log-H\" older continuous regardless of the (lack of)  regularity of $\mu$.  If  the single site probability distribution $\mu$ is continuous  (i.e., $\mu$ has no atoms), then it is known that   the integrated density of states  has at least as much regularity as the concentration function $S_\mu$ of $\mu$ \cite{CHK2}: for all compact intervals $I\subset \R$ we have
 \beq\label{Nxest1199}
\abs{N(E_{2})-N(E_{1})} \le  {C_{I}} S_\mu( \abs{E_{2}-E_{1}} ) \qtx{for all} E_{1},E_{2} \in I,
\eeq
where  $S_\mu(s):= \sup_{t \in \R}\mu([t,t+s])$ for $s \ge0$. If $\mu$ has an atom, \eq{Nxest1199} is still true but useless, since $\inf_{s>0} S_\mu(s)>0$.  For the continuous Anderson Hamiltonian  it is not even known if $N(E)$ is  a continuous function on $\R$ if $\mu$ has an atom.

Theorem~\ref{thmthemainev}(\ref{probloc}) is a probabilistic statement about the infinite volume Anderson Hamiltonian; there is no mention of finite volume operators. It captures
 all the usual forms of localization. Anderson localization with finite multiplicity of eigenvalues will  follow from \eq{cUdesiredint} and \eq{distEIred-int} by a simple application of the  Borel-Cantelli Lemma. Dynamical localization,  decay of eigenfunctions correlations (e.g., SULE, SUDEC), and decay of the Fermi projections will be consequences of \eq{cUdesiredint} and \eq{EDI9L99int}. These and other familiar localization properties are stated in Corollary~\ref{corallloc}. (Theorem~\ref{thmthemainev}(\ref{logHoldc}) is not needed for Corollary~\ref{corallloc}.)

\begin{corollary}\label{corallloc} Let $H_{\bom}$ be an Anderson Hamiltonian on $\mathrm{L}^2(\mathbb{R}^d)$.  Fix  $p \in \rb{\frac 1 3, \frac 3 8}$, and  let  $E_0> E_{\inf}$, $\widetilde{p}\in ]0,p[$,  $\vartheta > 0$ and $M>0$  be as in  Theorem~\ref{thmthemainev}.  Then  $H_{\bom}$ exhibits strong localization  in the energy interval $[E_{\inf},E_0[$ in the following sense:
\begin{enumerate}
\item \label{Andlocitem} The following holds with probability one:
\begin{enumerate}
\item   {$H_{\bom}$}  has  pure point spectrum  in the interval $[E_{\inf},E_0[$. 
\item  
For all $E\in [E_{\inf},E_0[$, 
 $\psi \in  \Ran P_{\bom}(E)$, and  $\nu>\frac d 2$,  we have 
\begin{equation}\label{expdecayIntro}
 \|\Chi_x \psi\| \le C_{\bom,E,\nu}\norm{T_{\nu}^{-1} \psi}\, e^{- M \norm{x}} \qquad \text{for all}\quad  x \in \R^{d}.
\end{equation}
In particular,  each   {eigenfunction} {$\psi$}   of  {$H_{\bom}$} 
with {eigenvalue}  {$E \in [E_{\inf},E_0[$} is exponentially localized  with the non-random rate of decay $ M>0$.

\item  The {eigenvalues} of  {$H_{\bom}$} in $[E_{\inf},E_0[$  have {finite multiplicity}: 
\beq\label{finmultIntro}
\tr P_{\bom}(E) < \infty  \quad\text{for all} \quad E\in [E_{\inf},E_0[.
\eeq
\end{enumerate}

\item \label{dynlocitem} The following  holds with probability one for all $\eps >0$ on all compact intervals $I \subset [E_{\inf},E_0[$ :

\begin{enumerate}

\item For all $E\in  I$,   $ x, y \in \R^{d}$, and  $\nu>\frac d 2$, we have
\beq\label{SUDECasIntro}
\norm{\Chi_{x} \phi}\norm{\Chi_{y} \psi}\le C_{\bom,I,\nu,\eps}\norm{T_{\nu}^{-1}\phi}\norm{
T_{\nu}^{-1} \psi}\e^{\norm{x}^{(1+\eps)\frac {\vartheta}{\widetilde{p}} }}  \e^{-\frac 1 4 M \norm{x-y}^{\vartheta}}
\eeq
for all  $\phi,\psi \in  \Ran P_{\bom}(E)$, and 
\beq \label{SUDECasPIntro}
 \|\Chi_x P_{\bom}(E)\|_{2} \norm{\Chi_{y} P_{\bom}(E)}_{2}\le  C_{\bom,I,\nu,\eps}\norm{T_{\nu}^{-1}P_{\bom}(E)}_{2}^{2}\e^{\norm{x}^{(1+\eps)\frac {\vartheta}{\widetilde{p}} }}  \e^{-\frac 1 4 M \norm{x-y}^{\vartheta}}.
 \eeq

\item    For all $E\in  I$, there exists a ``center of localization'' $y_{\bom,E}\in \R^{d}$
for all eigenfunctions with eigenvalue $E$, in the sense that  for all $x \in \R^{d}$ and $\nu>\frac d 2$    we have 
\beq\label{SULEasIntro}
\norm{\Chi_{x} \phi}\le  C_{\bom,I,\nu,\eps} \norm{T_{\nu}^{-1} \phi}\e^{\norm{y_{\bom,E}}^{(1+\eps)\frac {\vartheta}{\widetilde{p}} }} \e^{- \frac 1 4 M \norm{x-y_{\bom,E}}^{\vartheta}}
\eeq
for  all $\phi \in  \Ran P_{\bom}(E)$,
and 
\beq\label{SULEasPIntro}
\norm{\Chi_{x}  P_{\bom}(E)}_{2}\le   C_{\bom,I,\nu,\eps} \norm{T_{\nu}^{-1} P_{\bom}(E)}_{2}\e^{\norm{y_{\bom,E}}^{(1+\eps)\frac {\vartheta}{\widetilde{p}} }}\e^{-\frac 1 4 M \norm{x-y_{\bom,E}}^{\vartheta}} .
\eeq
Moreover,\beq \label{NSULEIntro}
N_{\bom,I}(L):= \sum_{\substack{E \in I\\\norm{y_{\bom,E}}\le L } }\tr P_{\bom}(E) \le  C_{\bom,I,\eps}\, L^{(1 +\eps)\frac {d}{\widetilde{p}}} \quad \text{for}\quad L \ge 1.
\eeq

\item For all   $ x, y \in \R^{d}$  we  have
\beq\label{decaykernelasIntro}
\sup_{f \in \cB_{b,1}}    \norm{\Chi_{y} f(H_{\bom})P_{\bom}(I)
\Chi_{x} }_{1} \le  C_{\bom,I,\eps}  \e^{\norm{x}^{(1+\eps)\frac {\vartheta}{\widetilde{p}} }} \e^{- \frac 1 4 M  \norm{x-y}^{\vartheta}}.
\eeq

\item For all $E\in  I$ and $ x, y \in \R^{d}$ we have
  \beq\label{decayFermipIntro}
\norm{\Chi_{y}P_{\bom}\up{E} \Chi_{x}}_{1} \le  C_{\bom,I,\eps}\e^{\norm{x}^{(1+\eps)\frac {\vartheta}{\widetilde{p}} }}\e^{-\frac 1 4 M \norm{x-y}^{\vartheta}}.
  \eeq
\end{enumerate}

\item\label{strdynlocitem} Given $b >0$, for all 
 $s \in \left]0,\frac {p}{ b + \frac 1 2  }\right[$,  $x_{0}\in \R^{d}$, and  compact intervals $I \subset [E_{\inf},E_0[$,   we have 
\beq \label{HSdynlocIntro}
\E\set{\sup_{f \in \cB_{b,1}}    \norm{\scal{X}^{bd} f(H_{\bom})P_{\bom}( I)\Chi_{x_{0}} }_{1}^{s}      } < \infty,
\eeq
\beq \label{HSdynloc2Intro}
\E\set{\sup_{t \in \R}    \norm{\scal{X}^{bd} \e^{-itH_{\bom}} P_{\bom}(I)
\Chi_{x_{0}} }_{1}^{s}      }< \infty,
\eeq
and
 \beq\label{decayFermipexpIntro}
\E\set{\sup_{E \in   I}\norm{\scal{X}^{bd}P_{\bom}\up{E} \Chi_{x_{0}}}_{1}^{s}} < \infty.
  \eeq

\end{enumerate}
\end{corollary}

\begin{remark} 
If Theorem~\ref{thmthemainev}(\ref{probloc}) holds on a given  bounded open interval $\I$ (instead of the interval $ [E_{\inf},E_0[$ at the bottom of the spectrum, as discussed in Remark~\ref{remmainev}), then Corollary~\ref{corallloc} also holds as stated in the interval $\I$.
\end{remark}

\begin{remark} Theorem~\ref{thmthemainev} and Corollary~\ref{corallloc}  hold also for Poisson Hamiltonians,  with minor modifications.  Their proofs can be modified for Poisson Hamiltonians  using the methods of \cite{GHK2,GHK3}, both for positive and attractive Poisson potentials. 
\end{remark}

\begin{remark}\label{remenergyint} It is instructive to compare  Theorem~\ref{thmthemainev} and Corollary~\ref{corallloc}  to the known results for the case when the single site probability distribution $\mu$ is   absolutely continuous with a bounded density (or   H\"older  continuous),  for which slightly stronger versions of these results have been be derived from an energy-interval multiscale analysis as in \cite{FMSS,vDK,FKac,GKboot,GKsudec,Kle}.  In this case the probability estimate \eq{cUdesiredint} is much stronger, one gets sub-exponential decay  $\e^{-L^{\xi}}$ for any $\xi \in ]0,1[$ for  the bad probabilities \cite{GKboot}, and even exponential decay
when the fractional moment method applies \cite{AENSS}.  The `either or' statement in \eq{distEIred-int} is stronger: \emph{either $W_{\bom,x_{0}}\up{\nu}({E}) \le   \e^{-  M L } $ or $ W_{\bom,x_{0},L}\up{\nu}({E}) \le \e^{- M L}$}.  We also have exponential decay in \eq{EDI9L99int} and in Corollary~\ref{corallloc}~\ref{dynlocitem}, that is, they hold with $\vartheta =1$.  Corollary~\ref{corallloc}~\ref{strdynlocitem}  holds for all $b>0$ with $s=1$. The SUDEC estimate \eq{SUDECasIntro} and the SULE estimate \eq{SULEasIntro}  hold with exponential decay and milder than exponential growth in $x$ or $y$; moreover they are equivalent, one can be derived from the other  (see \cite{GKboot,GKsudec}).
But in the general case \eq{SUDECasIntro} and \eq{SULEasIntro}   are not equivalent; \eq{SUDECasIntro} implies \eq{SULEasIntro} but the converse is not true.
\end{remark}

 Theorem~\ref{thmthemainev} and Corollary~\ref{corallloc}  will be proved in the context of generalized Anderson Hamiltonians.   Theorem~\ref{thmthemainev}(\ref{logHoldc}) is proven in Theorem~\ref{logHolder}, and  Theorem~\ref{thmthemainev}(\ref{probloc}) is contained in  Theorem~\ref{thmmainev}.  Corollary~\ref{corallloc}\ref{Andlocitem}   is proven in Theorem~\ref{thmAndloc}, Corollary~\ref{corallloc}\ref{dynlocitem} in Theorem~\ref{thmSUDEC} and Corollary~\ref{cordynloc}, and Corollary~\ref{corallloc}\ref{strdynlocitem}  follows from  Corollary~\ref{coroldynlocexp}.

 \section{Anderson Hamiltonians}\label{secgenAnd}

 \subsection{Normalized Anderson Hamiltonians}
 
  Given an Anderson Hamiltonian  $H_{\bom}$, it follows from Definition~\ref{defAndHam} that the common probability
distribution $\mu$ of the random variables  $\bom=\{ \omega_\zeta \}_{\zeta\in
\Z^d}$   satisfies
 \beq\set{M_{-} ,M_{+}}\subset \supp \mu \subset   [M_{-} ,M_{+}] \quad \text{for some} \quad -\infty <M_{-} <M_{+}< \infty.
 \eeq
Letting
\begin{gather}\notag
\widehat{ V}_{\mathrm{per}}=  \widetilde{ V}_{\mathrm{per}} -  \inf \sigma(-\Delta + \widetilde{ V}_{\mathrm{per}})
, \; \text{with}\; \widetilde{ V}_{\mathrm{per}}(x)=  V_{\mathrm{per}}(x) + M_{-}\sum_{\zeta \in \Z^d} u(x-\zeta), \\
\widehat{ V}_{\widehat{\bom}}(x)=\sum_{\zeta \in \Z^d} \widehat{\omega}_\zeta \,   \widehat{u}(x - 
\zeta),  \; \text{with}\;\; \widehat{u}= (M_{+}-M_{-})u  \;\; \text{and}\; \; \widehat{\omega}_\zeta=\frac  {\omega_{\zeta}- M_{-}}{M_{+}-M_{-}}, \\
\widehat{H}_{\widehat{\bom}}= -\Delta + \widehat{ V}_{\mathrm{per}} + \widehat{ V}_{\widehat{\bom}},\notag
\end{gather}
we have
\beq
H_{\bom}= \widehat{H}_{\widehat{\bom}} + \inf \sigma(-\Delta + \widetilde{ V}_{\mathrm{per}}).
\eeq
Since $\widehat{H}_{\widehat{\bom}}$ is a normalized  Anderson Hamiltonian as in Definition~\ref{defAndH} below, we conclude that  every Anderson Hamiltonian  equals a normalized Anderson Hamiltonian plus a constant.
Thus, without loss of generality,  it suffices to  study  normalized Anderson Hamiltonians as in Definition~\ref{defAndH},  which makes the relevant parameters explicit.

 \begin{definition} \label{defAndH}
 A \emph{normalized  Anderson Hamiltonian} is  an   Anderson Hamiltonian $H_{\bom}$ such that:
 \begin{enumerate} 
\item  the periodic potential $ {V}_{\mathrm{per}}$ satisfies 
\beq\label{AndH287}
\inf \sigma(-\Delta + {V}_{\mathrm{per}})= 0 ,
\eeq

\item the single site potential $u$ is a measurable function on $\R^d$
 with 
 \begin{equation} \label{u}
u_{-}\Chi_{\Lambda_{\delta_{-}}(0)}\le u \le u_{+}\Chi_{\Lambda_{\delta_{+}}(0)}\quad \text{for some constants $u_{\pm}, \delta_{\pm}\in ]0,\infty[ $},
\end{equation}

\item  $\bom=\{ \om_\zeta\}_{\zeta \in \Z^d}$ is a family of independent, identically distributed random variables with a  common probability distribution $\mu$  satisfying
\beq \label{mu}
\set{0  ,1}\subset\supp \mu \subset   [0,1] .
\eeq
\end{enumerate}

 \end{definition}

The condition \eq{AndH287} implies that $[0, E_1] \subset\sigma\pa{H_{0}}$  for some $E_1>0$. 
It follows that
  the non-random spectrum $\Sigma$ of    a normalized   Anderson Hamitonian  $H_{\bom}$
 satisfies (see \cite{KiM2})
\begin{gather}
\sigma\pa{ H_0} \subset \Sigma \subset [0, \infty[,
\intertext{so}
 \inf \Sigma=0 \quad \text{and} \quad     [0, E_1] \subset\Sigma \quad \text{for some $E_1=E_1(V_{\mathrm{per}})>0$}.
\intertext{In particular, we have}
 \Sigma= \sigma\pa{ -\Delta } =[0,\infty[  \quad \text{if} \quad  V_{\mathrm{per}}=0 .
\end{gather}

 \subsection{Generalized Anderson Hamiltonians}
  We will conduct our analysis of normalized Anderson Hamiltonians  in a more general context  which incorporates an additional background potential, bounded and nonnegative, but otherwise arbitrary, and allows variability  in the single site potentials as long as they satisfy uniform bounds.

  \begin{definition} \label{defgenAndH}
 A \emph{generalized (normalized)   Anderson Hamiltonian} is  a random Schr\"odinger 
operator    on  $\mathrm{L}^2(\mathbb{R}^d)$ of the form
\beq \label{genAndH}
H_{\bom}= H_0 + V_{\bom}, \quad \text{with} \quad  H_0 = -\Delta + V_{\mathrm{per}} +U,  \eeq 
where $V_{\mathrm{per}}$   is a bounded periodic potential with period $q \in \N$ such that
\beq
\inf \sigma( -\Delta + V_{\mathrm{per}})= 0 ,
\eeq
 $U$ is a measurable function  on $\R^d$ satisfying
\begin{equation}
 0 \le U(x) \le U_+ \quad \text{for all} \quad x \in \R^d \quad \text{for some constant}\quad  U_+\in [0,\infty[, \label{potU}
\end{equation}
and  $V_{\bom}$ is the random potential 
\beq
V_{\bom} (x):= 
\sum_{\zeta \in \Z^d} \omega_\zeta \,  u_\zeta(x),\label{AndVzeta}
\eeq
 where the family of  random variables  $\bom=\{ \om_\zeta\}_{\zeta \in \Z^d}$ is as in Definition~\ref{defAndH}, and  $\mathbf{u}=\set{u_\zeta}_{\zeta \in \Z^d}$ is a family of measurable functions on $\R^d$
such that there are constants $u_{\pm}, \delta_{\pm}\in ]0,\infty[ $ for which
 \begin{equation} \label{uzeta}
u_{-}\Chi_{\Lambda_{\delta_{-}}(\zt)}\le u_\zeta \le u_{+}\Chi_{\Lambda_{\delta_{+}}(\zt)}\quad \text{for all} \quad  \zeta \in \Z^d.
\end{equation}
 \end{definition}

 Without loss of generality, we realize the random variables $\{ \om_\zeta\}_{\zeta \in \Z^d}$ as  the coordinate functions on the probability space  $\pa{\Omega,\cF,\P}$, where  $\Omega=  [0,1]^{\Z^d}$,   $\cF$ denotes the $\sigma$-algebra generated by  the coordinate functions,  and   $\P=   \mu^{\Z^{d}}$,  the  product measure of $\Z^{d}$ copies of  the common probability distribution $\mu$ of the random variables $\{ \om_\zeta\}_{\zeta \in \Z^d}$.  In other words, $\pa{\Omega,\cF,\P}=  \pa{[0,1], \cB_{[0,1]},\mu}^{\Z^{d}}$, the product measure space of $\Z^{d}$ copies of the measure space  $\pa{[0,1], \cB_{[0,1]},\mu}$, 
 where $\cB_{[0,1]}$ is the Borel $\sigma$-algebra on $[0,1]$. The expectation with respect to $\P$ will be denoted by $\E$.
  Note that $\Omega$ is a compact Hausdorff space with the product topology
  and $\cF$ is the corresponding Borel $\sigma$-algebra.  
  A set $\cU\in \cF$ will be called an event.

 A generalized  Anderson Hamiltonian $H_{\bom}$ is a
 measurable map
from the probability space $(\Omega,\mathcal{F},\mathbb{P})$
 to  
self-adjoint operators on the Hilbert space  $\mathrm{L}^2(\mathbb{R}^d)$.  Measurability of $H_\omega$ 
means that
 the maps $\omega \to f(H_\omega)$ are weakly
(and hence strongly)
measurable for all bounded Borel  measurable functions $f$ on $\mathbb{R}$.
 
A generalized Anderson Hamiltonian  $H_{\bom} $ is   not, in general,  a  $q\Z^d$-ergodic family of
random self-adjoint operators for any $q \in \N$, so the spectrum of $H_{\bom}$,  as well as its pure point, 
absolutely continuous, and singular continuous  components,
 need not be non-random (i.e., equal to some  fixed set with probability one).  But we always have    $\sigma(H_{\bom}) \subset [0,\infty)$ for all $\bom \in \Omega$.

 \subsection{Finite volume  Anderson Hamiltonians}
    
  Given a  set ${\Xi} \subset \R^d$, we set $\widetilde{{\Xi}} :={\Xi}   \cap \Z^d$ 
  and consider  the product measure space  $\pa{\Omega_{\Xi},\cF_{\Xi},\P_{\Xi}}=   \pa{[0,1], \cB_{[0,1]},\mu}^{ \widetilde{{\Xi}}}$;  in particular,
$\Omega_{\Xi}=[0,1]^{ \widetilde{{\Xi}}}$.  We  identify $\cF_{\Xi}$ with the sub-$\sigma$-algebra of subsets of $\Omega$ generated by the coordinate functions $\bom_{\Xi}=\{ \om_\zeta\}_{\zeta \in  \widetilde{{\Xi}}}$, in which case
$\P_{\Xi}$ is the restriction of  $\P$ to $\cF_{\Xi}$.

Given a generalized  Anderson Hamiltonian $H_{\bom}$, we set 
\beq \label{Vlocal}
V_{{\bom}_{{\Xi}}}(x):=  \sum_{\zeta \in \widetilde{{\Xi}}} \omega_\zeta \,u_\zeta(x )
\quad \text{for}\quad  \bom \in \Omega \quad \text{and}\quad  {\Xi} \subset \R^d,
 \eeq
and define the corresponding  finite volume (generalized) Anderson Hamiltonian on a  box $\Lambda={\Lambda}_{L}(x)$ in $\R^d$  as follows:  
\begin{align}\label{finvolH}
H_{{\bom},\Lambda} :={H_{0,\Lambda}}+ V_{{\bom},\Lambda} \quad \text{on}   \quad \L^{2}(\Lambda),
\end{align}
with
\beq\label{H0La}
H_{0,\Lambda}:= - \Delta_{\Lambda} +  V_{\mathrm{per},\Lambda} + U_\Lambda,
\eeq
where 
$\Delta_{\Lambda}$ is the  Laplacian on $\Lambda$ with Dirichlet boundary condition,  and 
$ V_{\mathrm{per},\Lambda}$,  $ U_\Lambda$ and $V_{{\bom},\Lambda}$  are the restrictions of $ V_{\mathrm{per}}$,  $ U$ and $V_{{\bom}_\Lambda}$  to $\Lambda$.
Since we are using Dirichlet boundary condition, we always have   $\inf \sigma(H_{0,\Lambda}) \ge  0$ (easy to see using quadratic forms), and hence    $\inf \sigma(H_{{\bom},\Lambda}) \ge  0$. 
The finite volume resolvent, defined for $z \notin  \sigma(H_{{\bom},\Lambda})$ by  
\begin{align}\label{finvolResolv}
R_{\bom,\Lambda} (z):=(H_{\bom,\Lambda} - z)^{-1} \quad \text{on}   \quad \L^{2}(\Lambda),
\end{align}
  is a compact operator.
Note that
$\Delta_{\Lambda}= \nabla_{\Lambda}\cdot\nabla_{\Lambda}$, where 
$\nabla_{\Lambda}$ is the gradient with Dirichlet boundary condition.

We will identify $\L^{2}(\Lambda)$ with $\Chi_{\Lambda }\L^{2}(\R^{d})$ when convenient,
and,  if necessary, we  will use subscripts $\Lambda$ and $\R^{d}$ to distinguish between the norms and inner products of $\L^{2}(\Lambda)$ and $\L^{2}(\R^{d})$.
In particular, we use the identification $ V_{\mathrm{per},\Lambda}=\Chi_\Lambda  V_{\mathrm{per}}$, 
$U_\Lambda = \Chi_\Lambda U$, and $V_{{\bom},\Lambda} =\Chi_\Lambda V_{{\bom}_\Lambda}$. If $\Lambda \subset \Lambda^\pr$, we  will also extend operators on $\L^{2}(\Lambda)$,  such as  $R_{\bom,\Lambda} (z)$, to operators   on  $\L^{2}(\Lambda^\pr)$ by making them the zero operator on $\L^{2}(\Lambda^\pr \setminus \Lambda)$.
If $\eta \in  \L^{\infty}(\Lambda)$, we will also use $\eta$ to denote the operator given by multiplication by $\eta$  on  $\L^{2}(\Lambda)$.

 If $\Xi\subset \R^{d}$, $\overline{\Xi}$ will denote its closure, $\Xi^0$ its interior, and $\partial \Xi:=\overline{\Xi}\setminus \Xi^0$ its boundary. If $\Xi \subset \Xi^{\pr}\subset \R^{d}$, $\partial^{ \Xi^{\pr}} \Xi:=\partial \Xi \setminus \partial \Xi^\pr $ will denote the boundary of $\Xi$ in $\Xi^\pr$. ($\partial^{ \Xi^{\pr}} \Xi$ is the boundary of $\Xi$ with respect to the relative topology on $\Xi^{\pr}$.) 

 Given a box $\La\subset \La^{\pr}$, where $\La^{\pr}$ is either a box or $\R^{d}$, and $\delta>0$, we set (the distance is given by the norm in \eq{supnorm})
\beq\begin{split}\label{LambdaupL}
 \Lambda^{\Lambda^\pr,\delta}& := \set{x \in \Lambda; \;  \Lambda_{2\delta}(x) \cap \Lambda^\pr \subset \Lambda}= \set{x \in \Lambda; \; \dist \pa{x, \partial^{\Lambda^\pr} \Lambda} \ge \delta },\\
  \partial^{\Lambda^\pr,\delta} \Lambda & := \Lambda \setminus  \Lambda^{\Lambda^\pr,\delta}.
\end{split}\eeq
If $\Lambda^\pr=\R^d$ we generally omit it from the notation.

In general  $ V_{\bom,\Lambda}\not=  \Chi_{\Lambda} V_{\bom,\Lambda^\pr}$ for   $\Lambda \subset \Lambda^\pr$, but we always have
\begin{equation}\label{compat}
 \Chi_{ \Lambda^{\Lambda^\pr,\frac {\delta_+} 2}} V_{\bom,\Lambda}=   \Chi_{ \Lambda^{\Lambda^\pr,\frac {\delta_+} 2}} V_{\bom,\Lambda^\pr}.
\end{equation}

In this paper we will always assume that the  finite volumes $\Lambda=\Lambda_{L}$ where we define $H_{{\bom},\Lambda}$  have $L \ge 100\pa{\delta_{+}+1}$.

\subsection{Generalized eigenfunctions} Let $H_{\bom}
$ be a generalized Anderson Hamiltonian, fix $\bom \in \Omega$, and let $\La$ be either  $\R^{d}$ or a box  $\La_{L}$. Recall that $\cD(H_{{\bom},\Lambda})=\cD(\Delta_{\La})$.

\begin{definition}  \label{defgeneiginitial}
 A \emph{generalized eigenfunction} for $H_{{\bom},\Lambda}$ with  \emph{generalized eigenvalue} $E \in \R$  is a  measurable function $\psi$ on $\Lambda$ with
  \beq
 0< \norm{T_{\nu}^{-1}\psi}_{\Lambda} < \infty \quad \text{for some}\quad \nu >0,
 \eeq
such that
\beq  \label{eigb1}
\scal{H_{{\bom},\Lambda}\vphi,\psi}= {E} \scal{\vphi,\psi} \quad \text{for all}\quad \vphi \in C_{c}^{\infty}(\La).
\eeq
\end{definition}

It follows (e.g., \cite{KKS}) that if $\psi$ is a generalized eigenfunction for $H_{{\bom},\Lambda}$ with generalized eigenvalue ${E}\in \R$, then for all 
$\phi \in C_{c}^{2}(\La)$ we have  $\phi \psi \in \cD(\Delta_{\La})\subset \cD(\nabla_{\La})$  and 
\beq\label{localeig}
\pa{H_{\bom,\Lambda}-{E}}\phi \psi=  W_\Lambda(\phi)\psi,
\eeq
where  $W_\Lambda({\phi})$ is the closed densely defined operator on  $\L^{2}(\Lambda)$ given by \beq \label{Wphi}
W_\Lambda({\phi}) = -2( \nabla {\phi}) \cdot \nabla_\Lambda - (\Delta {\phi}).
\eeq
(More precisely, $ W_\Lambda(\phi)\psi : =  W_\Lambda(\phi)\widetilde{\phi}\psi$ for all $\widetilde{\phi} \in  C_{c}^{2}(\La)$ such that $ \widetilde{\phi}\equiv 1$ on $\supp \phi$.)

Eigenfunctions are always  generalized eigenfunctions.  

\subsection{Properties of finite volume operators} We will now derive some deterministic properties of the finite volume operators corresponding to a  generalized Anderson Hamiltonian $H_{\bom}$.

Given  $\La$,  either a finite box or $ \R^{d}$, and   $x,y \in \La$,  $\norm{\Chi_y R_{\bom,\Lambda} (z)  \Chi_x}\in [0,\infty[$  is well defined  for $z \notin  \sigma\pa{H_{\bom,\Lambda}}$.  We will abuse the notation and make the extension to $z\in  \sigma\pa{H_{\bom,\Lambda}}$ by 
\beq\label{exttoE}
\norm{\Chi_y R_{\bom,\Lambda}(z) \Chi_x} := \limsup_{\eps \to 0}\norm{\Chi_y R_{\bom,\Lambda}(z+i\eps) \Chi_x} \in [0,\infty].
\eeq

We will consider boxes  $\Lambda \subset  \Lambda^\pr$ without requiring the interior box $\Lambda$ to be at a certain distance from the boundary of $ \Lambda^\pr$. For this reason we work with $\partial^{\Lambda^\pr} \Lambda$ (the boundary of $ \Lambda$ in $\Lambda^{\pr}$) instead of  $\partial \Lambda$.

\begin{lemma} \label{lempreSLI}  Consider a box  $\Lambda=\Lambda_\ell \subset  \Lambda^\pr$, where $\La^{\pr}$ is either a finite box or $\R^{d}$, and  let $z \notin  \sigma\pa{H_{\bom,\Lambda}} $. 
Then, given   $x \in \Lambda$ with  $\Lambda_{ \delta_+  + 3}(x) \cap \Lambda^\pr \subset \Lambda$ and $y\in \Lambda^\pr $, we can find $x^\pr \in \Upsilon_\Lambda^{\Lambda^\pr}$, where
\beq \label{Upsilon}
\Upsilon_\Lambda^{\Lambda^\pr} := \set{x \in \Lambda; \; \dist \pa{x, \partial^{\Lambda^\pr} \Lambda} = \tfrac {\delta_+ + 1} 2 } ,
\eeq
such that
\begin{align}\label{yRxall}
&\norm{\Chi_y R_{\bom,\Lambda^\pr} (z)  \Chi_x} \\
& \qquad \le 
\norm{\Chi_y \Chi_{\Lambda\up{\La^\pr,\frac 1 2}}  R_{\bom,\Lambda} (z) \Chi_x }+ \gamma_{z} \ell^{d-1}
\norm{\Chi_y  R_{\bom,\Lambda^\pr}(z) \Chi_{x^\pr}}\norm{
\Chi_{x^\pr}  R_{\bom,\Lambda} (z)\Chi_x },\notag
\end{align}
with
\beq\label{defgammaz}
\gamma_{z}= \gamma_{ z,d, V_{\mathrm{per}}}= C_d \pa{1 +\max \set{0,\Re z  -\essinf V_{\mathrm{per}} }}^{\frac 1 2} .
\eeq
In particular, 
   \begin{enumerate}
   \item if $y \in \Lambda^\pr \setminus \Lambda$, we have 
  \begin{align}\label{yRx239}
\norm{\Chi_y R_{\bom,\Lambda^\pr} (z)  \Chi_x}  \le 
 \gamma_{z}\ell^{d-1}
\norm{\Chi_y  R_{\bom,\Lambda^\pr}(z) \Chi_{x^\pr}}\norm{
\Chi_{x^\pr}  R_{\bom,\Lambda} (z)\Chi_x },
\end{align}

\item if $y \in \Lambda$, we have 
\begin{align}\label{yRx249}
&\norm{\Chi_y R_{\bom,\Lambda^\pr} (z)  \Chi_x} \\
& \qquad \le 
\norm{\Chi_y  R_{\bom,\Lambda} (z) \Chi_x }+ \gamma_{z}\ell^{d-1}
\norm{\Chi_y  R_{\bom,\Lambda^\pr}(z) \Chi_{x^\pr}}\norm{
\Chi_{x^\pr}  R_{\bom,\Lambda} (z)\Chi_x }.\notag
\end{align}
   \end{enumerate}
\end{lemma}

\begin{proof}
Given boxes  $\Lambda \subset  \Lambda^\pr$, we let $\Upsilon_\Lambda^{\Lambda^\pr}$ be as in    \eq{Upsilon}) and set
\beq
 \widehat{ \Upsilon}= \widehat{ \Upsilon}_\Lambda^{\Lambda^\pr}  := \set{x \in \Lambda; \; \dist \pa{x,  \Upsilon_\Lambda^{\Lambda^\pr}} <  \tfrac { 1} 4 }=\set{\bigcup_{y^\pr \in  \Upsilon_\Lambda^{\Lambda^\pr}} \Lambda_{\frac 1 2}(y^\pr)}\cap \Lambda. 
\eeq
There exists  a constant  $C_d$, independent of $\Lambda$ and $\Lambda^\pr$, for which we can find  a function ${\phi}={\phi}_\Lambda^{\Lambda^\pr} \in C^2(\Lambda^\pr)$, with  $0\le {\phi} \le 1$,  such that
\begin{gather}\label{defphi}
{\phi}\equiv 1 \quad \text{on} \quad   { \Lambda^{\Lambda^\pr,\frac {\delta_+ + 1} 2 +\frac 1 4 }} ,\\
{\phi}\equiv 0 \quad \text{on} \quad  \Lambda^\pr \setminus  { \Lambda^{\Lambda^\pr,\frac {\delta_+ + 1} 2 -\frac 1 4 }} ,\\
\abs{\nabla {\phi}}, \abs{\Delta {\phi}}\le C_d. \label{nablaphi}
\end{gather}
Note that  
\beq
\supp {\phi} \subset \Lambda\up{1}:= { \Lambda^{\Lambda^\pr,\frac {\delta_+ + 1} 2 -\frac 1 4 }}\quad \text{and}\quad \supp {\nabla {\phi}} \subset \widehat{ \Upsilon}=\widehat{ \Upsilon}_\Lambda^{\Lambda^\pr} .
\eeq
In particular, we have  ${\phi} \mathcal{D}{(\Delta_{\Lambda})} \subset  \mathcal{D}{(\Delta_{\Lambda})}$ and   ${\phi} \mathcal{D}{(\Delta_{\Lambda^\pr})} \subset  \mathcal{D}{(\Delta_{\Lambda^\pr})}$.

Suppose first that  $z \notin  \sigma\pa{H_{\bom,\Lambda}} \cup  \sigma\pa{H_{\bom,\Lambda^\pr}}$.
In this case we use the geometric resolvent identity (cf. \cite{CH1,FKac,BK}). In view of \eq{compat}, if $z \notin  \sigma\pa{H_{\bom,\Lambda}} \cup  \sigma\pa{H_{\bom,\Lambda^\pr}}$ we get
\beq \label{geometricresolvent}
R_{\bom,\Lambda^\pr} (z)  {\phi} = {\phi} R_{\bom,\Lambda} (z) + R_{\bom,\Lambda^\pr} (z) W_\Lambda({\phi}) R_{\bom,\Lambda} (z),
\eeq
as operators from $\L^{2}(\La)$ to $\L^{2}(\La^{\pr})$, where $W_\Lambda({\phi})$ is as in \eq{Wphi}.
Given $x \in \Lambda$ with  $\Lambda_{\delta_+ + 3}(x) \cap \Lambda^\pr \subset \Lambda$, i.e., $x \in \Lambda^{\Lambda^\pr,\frac {\delta_+ + 1} 2 +1 }$, we have 
  \beq \label{xdisjoint}
   \Chi_x={\phi} \Chi_x,  \quad \Chi_{y^{\pr} } \Chi_x=0 \quad\text{for}\quad  y^{\pr} \in  \Upsilon_\Lambda^{\Lambda^\pr}.
   \eeq
  It follows that for $y \in \Lambda^\pr$ we have
\begin{align}\label{yRx}
\Chi_y R_{\bom,\Lambda^\pr} (z)  \Chi_x & = \Chi_y R_{\bom,\Lambda^\pr} (z) {\phi} \Chi_x\\
&=\Chi_y  {\phi} R_{\bom,\Lambda} (z) \Chi_x +\Chi_y  R_{\bom,\Lambda^\pr} (z) W_\Lambda({\phi}) R_{\bom,\Lambda} (z)\Chi_x   \notag \\
&=\Chi_y  {\phi} R_{\bom,\Lambda} (z) \Chi_x +\Chi_y  R_{\bom,\Lambda^\pr}(z) \Chi_{ \widehat{ \Upsilon}} W_\Lambda({\phi}) R_{\bom,\Lambda} (z)\Chi_x . \notag
\end{align}

Let $\ell$ be the length of the side of the box $\Lambda$, i.e., $\Lambda=\Lambda_\ell$.  Then we can pick $y_1,y_2,\ldots,y_{J} \in  \Upsilon_\Lambda^{\Lambda^\pr}$, 
where $C^\pr_d \ell^{d-1}\le J \le C_d^{\pr\pr} \ell^{d-1}$, and $y_1^\pr,y_2^\pr,\ldots,y_{J^\pr}^\pr \in \overline{ \Upsilon_\Lambda^{\Lambda^\pr}}\setminus \Upsilon_\Lambda^{\Lambda^\pr}$, with $0\le J^\pr \le  C_d^{\pr\pr\pr} \ell^{d-2}$  (note $J^\pr=0$ if $\partial^{\Lambda^\pr} \Lambda=\partial \Lambda$, in which case  $\overline{ \Upsilon_\Lambda^{\Lambda^\pr}}= \Upsilon_\Lambda^{\Lambda^\pr}$), such that $\Lambda_{1}(y_j) \subset  \Lambda$ for $j=1,2,\ldots,J$,
 \beq
\overline{\widehat{ \Upsilon}_\Lambda^{\Lambda^\pr}} =\set{ \bigcup_{j=1}^J
\overline{\Lambda}_{\frac 1 2}(y_j)} \bigcup \set{ \bigcup_{j^\pr =1}^{J^\pr}
\overline{\Lambda}_{\frac 1 2}(y^\pr_{j^\pr})\cap \overline{\Lambda}}.
 \eeq
 and $y_1,y_2,\ldots,y_{J},y_1^\pr,y_2^\pr,\ldots,y_{J^\pr}^\pr$ form a minimal set with respect to this properties.
It follows that we can select disjoint open sets ${\mathcal{O}_j \subset {\Lambda}_{\frac 1 2}(y_j)}$ and  ${\mathcal{O}^\pr_{j^\pr} \subset {\Lambda}_{\frac 1 2}(y^\pr_{j^\pr})}\cap {\Lambda}$,  where $j=1,2,\ldots,J$ and $j^\pr=1,2,\ldots,J^\pr$, such that
 \beq
\overline{\widehat{ \Upsilon}_\Lambda^{\Lambda^\pr}} =\set{ \bigcup_{j=1}^J
\overline{\mathcal{O}_j}} \bigcup \set{ \bigcup_{j^\pr =1}^{J^\pr}
\overline{\mathcal{O}^\pr_{j^\pr}}}.
 \eeq
It follows that
\begin{align}\label{startSLI}
&\norm{\Chi_y  R_{\bom,\Lambda^\pr}(z) \Chi_{ \widehat{ \Upsilon}} W_\Lambda({\phi}) R_{\bom,\Lambda} (z)\Chi_x }\\ \notag
& \quad \quad  \le \sum_{j=1}^J \norm{\Chi_y  R_{\bom,\Lambda^\pr}(z) \Chi_{\mathcal{O}_j} W_\Lambda({\phi}) R_{\bom,\Lambda} (z)\Chi_x }\\ \notag
&  \qquad \qquad \qquad + \sum_{j^\pr=1}^{J^\pr} \norm{\Chi_y  R_{\bom,\Lambda^\pr}(z) \Chi_{\mathcal{O}^\pr_{j^\pr} } W_\Lambda({\phi}) R_{\bom,\Lambda} (z)\Chi_x }\\ \notag 
& \quad  \quad  \le \sum_{j=1}^J \set{\norm{\Chi_y  R_{\bom,\Lambda^\pr}(z)\Chi_{{\Lambda}_{\frac 1 2}(y_j)} }  \norm{\Chi_{{\Lambda}_{\frac 1 2}(y_j)} W_\Lambda({\phi}) R_{\bom,\Lambda} (z)\Chi_x }}\\ \notag
& \qquad \qquad \qquad + \sum_{j^\pr=1}^{J^\pr}\set{ \norm{\Chi_y  R_{\bom,\Lambda^\pr}(z)\Chi_{ {\Lambda}_{\frac 1 2}(y^\pr_{j^\pr})\cap {\Lambda}} }\norm{\Chi_{ {\Lambda}_{\frac 1 2}(y^\pr_{j^\pr})\cap {\Lambda}} W_\Lambda({\phi}) R_{\bom,\Lambda} (z)\Chi_x }}.
\end{align}

Let $\Lambda_\sharp$  be either   ${\Lambda}_{\frac 1 2}(y_j)={\Lambda}_{\frac 1 2}(y_j)\cap \Lambda$ or  ${\Lambda}_{\frac 1 2}(y^\pr_{j^\pr})\cap {\Lambda}$ for some $j$ or $j^\pr$. We write $\Lambda_\sharp^\pr$ for the corresponding ${\Lambda}_{1}(y_j)\cap \Lambda$ or ${\Lambda}_{1}(y^\pr_{j^\pr})\cap \Lambda$. Using \eq{Wphi} and \eq{nablaphi} we get
\begin{align}
 \norm{\Chi_{\Lambda_\sharp} W_\Lambda({\phi}) R_{\bom,\Lambda} (z)\Chi_x }
 \le 2C_d  \norm{\Chi_{\Lambda_\sharp}  \nabla_\Lambda   R_{\bom,\Lambda} (z)\Chi_x } + C_d \norm{\Chi_{\Lambda_\sharp}  R_{\bom,\Lambda} (z)\Chi_x }.
\end{align}

We now use the following  interior estimate
 (e.g., \cite[Lemma~A.2]{GK5}):   Let
$\eta \in C^1(\mathcal{O})$ with  $\|\eta\|_\infty \le 1$, where $\mathcal{O}\subset \R^d$ is an open set.  Given a finite  box $\Lambda$ such that $\overline{\Lambda} \subset \mathcal{O}$, 
we set $\eta_\Lambda=\eta \Chi_\Lambda$. Then, 
 for all $\bom \in [0,1]^
{\Z^d}$, $z \in \C $, and   $\psi \in \mathcal{D}{(\Delta_{\Lambda})}$, we have
 \begin{align} \label{interior}
\norm{{\eta_\Lambda} \nabla_{\Lambda} \psi}^2 &\le 
 \norm{\Chi_{\supp {\eta_\Lambda}} \pa{H_{\bom,\Lambda}-z} \psi}^2\\ 
 &\qquad  + 
\left(1 +\max \set{0,\Re z  -\essinf V_{\mathrm{per}} } +
4\|\nabla {\eta_\Lambda}\|_\infty^2\right) \norm{
\Chi_{\supp {\eta_\Lambda}} \psi}^2 . \notag
\end{align}
(Although   \cite[Lemma~A.2]{GK5} is stated with somewhat different conditions on $\eta$,  the proof applies with $\eta$ as above. The important observation is that with Dirichlet  boundary condition we have 
$\eta \psi=\eta_\Lambda \psi \in  \mathcal{D}{(\Delta_{\Lambda})}$ for all $\psi \in \mathcal{D}{(\Delta_{\Lambda})}$.)

Given a box ${\Lambda}_{\frac 1 2}(x^\pr)$, we fix a function $\eta \in  C^1(\R^d)$ with  $0\le \eta \le 1$, such that $\eta \equiv 1$ on ${\Lambda}_{\frac 1 2}(x^\pr)$, 
$\supp \eta \subset {\Lambda}_{1}(x^\pr)$, and $\norm{\nabla \eta}_\infty \le C^{\pr\pr\pr}_d$.    We have, using \eq{interior} and $\eta \Chi_x =0$ (see \eq{xdisjoint}),
\begin{align}\label{usinginterior}
  \norm{\Chi_{\Lambda_\sharp}  \nabla_\Lambda   R_{\bom,\Lambda} (z)\Chi_x } \le
    \norm{\eta_\Lambda  \nabla_\Lambda   R_{\bom,\Lambda} (z)\Chi_x } \le   \gamma^\pr_{\Re z, d,V_{\mathrm{per}}}
 \norm{
\Chi_{\Lambda_\sharp^\pr}  R_{\bom,\Lambda} (z)\Chi_x }  ,
 \end{align}
with
\beq
 \gamma^\pr_{\Re z, d,V_{\mathrm{per}}}
:=  C^{\pr\pr\pr}_d\pa{1 +\max \set{0,\Re z  -\essinf V_{\mathrm{per}} }}^{\frac 1 2} 
\eeq

If $\Lambda_\sharp={\Lambda}_{\frac 1 2}(y_j)$, we have $\Chi_{\Lambda_\sharp}\le \Chi_{\Lambda_\sharp^\pr}=\Chi_{y_j}$. If $\Lambda_\sharp= {\Lambda}_{\frac 1 2}(y^\pr_{j^\pr})\cap {\Lambda}$, we have $\Chi_{\Lambda_\sharp}\le \Chi_{\Lambda_\sharp^\pr}\le\Chi_{y^{\pr\pr}_{j^\pr}}$ for some $y^{\pr\pr}_{j^\pr} \in  \Upsilon_\Lambda^{\Lambda^\pr}$.  Thus, it  follows from 
\eq{startSLI} and \eq{usinginterior}  that
\beq\label{endSLI}
\norm{\Chi_y  R_{\bom,\Lambda^\pr}(z)\Chi_{ \widehat{ \Upsilon}} W_\Lambda({\phi}) R_{\bom,\Lambda} (z)\Chi_x }
  \le c_d (1 + \gamma^\pr_{\Re z, d,V_{\mathrm{per}}})\ell^{d-1}
\norm{\Chi_y  R_{\bom,\Lambda^\pr}(z)\Chi_{x^\pr}}\norm{
\Chi_{x^\pr}  R_{\bom,\Lambda} (z)\Chi_x }
\eeq
for some $x^\pr \in \Upsilon_\Lambda^{\Lambda^\pr}$.

Combining  \eq{yRx} and \eq{endSLI} we conclude that
\begin{align}\label{yRx2}
\norm{\Chi_y R_{\bom,\Lambda^\pr} (z)  \Chi_x}  \le 
\norm{\Chi_y  {\phi} R_{\bom,\Lambda} (z) \Chi_x }+ \gamma_{z}\ell^{d-1}
\norm{\Chi_y  R_{\bom,\Lambda^\pr}(z)\Chi_{x^\pr}}\norm{
\Chi_{x^\pr}  R_{\bom,\Lambda} (z)\Chi_x }\
\end{align}
for some $x^\pr \in \Upsilon_\Lambda^{\Lambda^\pr}$, where $\gamma_{z}$ is  as in \eq{defgammaz},
 which yields  \eq{yRxall}.
If $y \in \Lambda^\pr \setminus \Lambda$, $\Chi_y  {\phi}=0$, and we get \eq{yRx239}.
If $y \in   \Lambda$, using  $0\le \phi\le 1$   we get \eq{yRx249}.

If $z \in  \sigma\pa{H_{\bom,\Lambda^\pr}}\setminus  \sigma\pa{H_{\bom,\Lambda}}$, for all have $\eps \not=0$ we have
$z +i\eps \notin  \sigma\pa{H_{\bom,\Lambda}} \cup  \sigma\pa{H_{\bom,\Lambda^\pr}}$, and the lemma holds for $z +i\eps$.  The lemma then follows for $z$ in view of \eq{exttoE}.
\end{proof}

\begin{lemma}\label{lemEDI} Consider a box  $\Lambda=\Lambda_\ell \subset  \Lambda^\pr$, where $\La^{\pr}$ is either a finite box or $\R^{d}$.  Let $\psi$ be a generalized eigenfunction of $H_{\bom,\Lambda^\pr}$ with generalized eigenvalue $E \in \R\setminus \sigma(H_{\bom,\Lambda})$.
  Then for every  $x \in \Lambda$, with  $\Lambda_{ \delta_+ + 3 }(x) \cap \Lambda^\pr \subset \Lambda$, we can find $x^\pr \in \Upsilon_\Lambda^{\Lambda^\pr}$ such that
\beq\label{EDI}
\norm{\Chi_{x}\psi}\le  \gamma_{E}\ell^{d-1}
\norm{\Chi_{x^\pr}  R_{\bom,\Lambda} (E)\Chi_x } \norm{\Chi_{x^\pr}\psi}.
\eeq
\end{lemma}

\begin{proof} Let ${\phi}={\phi}_\Lambda^{\Lambda^\pr}$ be the function in the proof of the previous lemma (cf. \eq{defphi}-\eq{nablaphi}).  It follows from  \eq{localeig}
that
\beq
\phi \psi= R_{\bom,\Lambda}(E) W_\Lambda(\phi)\psi.
\eeq
Thus, given   $x \in \Lambda$ with  $\Lambda_{ \delta_+  + 3}(x) \cap \Lambda^\pr \subset \Lambda$, we have
\begin{align}
\norm{\Chi_{x}\psi}=\norm{\Chi_{x}\phi \psi}= \norm{\Chi_{x} R_{\bom,\Lambda}(E) W_\Lambda(\phi)\psi}.
\end{align}
Proceeding as in \eq{startSLI}-\eq{endSLI} we get \eq{EDI}.
\end{proof}

\section{Preamble to the multiscale analysis}\label{secpreambleMSA}

 We  fix a generalized  Anderson Hamiltonian $H_{\bom}$.

\subsection{Good boxes and free sites} A finite box will be called `good' at an energy  $E$ when the finite volume resolvent is not too big and exhibits exponential decay.
   As in \cite{B,BK,GHK2}, we will also require `free sites'. 
   
    Given a box $\Lambda$, a subset  $S \subset\widetilde{\Lambda}$, and $\bt_S=\{t_\zeta\}_{\zeta \in S} \in [0,1]^S$, we set
\begin{align}\label{finvolHS}
H_{{\bom},{\bt_S},\Lambda} :=H_{0,\Lambda}+  V_{{\bom},{\bt_S},\Lambda}\quad \text{on}   \quad \L^{2}(\Lambda),
\end{align}
where  $ V_{{\bom},{\bt_S},\Lambda}= \Chi_\La V_{{\bom}_\Lambda,{\bt_S}}$ with 
\begin{equation}
 V_{{\bom_\Lambda},{\bt_S}}(x):= V_{\bom_{\Lambda \setminus S}}(x)+ V_{\bt_S}(x)=
  \sum_{\zeta \in\widetilde{\Lambda}\setminus S} \omega_\zeta \,u_\zt (x)+ \sum_{\zeta \in  S} t_\zeta \,u_\zt(x ).  \label{finvolVS}
\end{equation}
 $R_{\bom,{\bt_S},\Lambda} (z)$ will denote the corresponding finite volume resolvent.

 \begin{definition}\label{defgood} Consider a configuration  $\bom \in{\Omega}$, an energy $E\in \C$,   a rate of decay $m>0$,  $0< \varsigma< 1$, and $S \subset\widetilde{\Lambda}_L$.  A   box  $\Lambda_{L}$ is said to be $(\bom,E,m,\varsigma,S)$- \emph{good}  if 
  the following holds for all $\bt_S \in [0,1]^{S}$:
  \begin{align}\label{weg}
\| R_{\bom, \bt_S, \Lambda_{L}}(E) \|& \le \e^{L^{1-\varsigma}}  
\intertext{and} 
\| \Chi_x R_{\bom,\bt_S, \Lambda_{L}}(E) \Chi_y \|& \le \e^{-m  \norm{x-y}} \qtx{for all}  x,y \in \Lambda_{L} \qtx{with}  \norm{x-y}\ge \tfrac L{100}. \label{good}
\end{align} 
 In this case $S$ consists of $(\bom,E,m,\vs)$- \emph{free sites} for the box  $\Lambda_{L}$. 
 If  no free sites are specified, i.e., $S=\emptyset$, $\Lambda_{L}$ is said to be $(\bom,E,m,\varsigma)$- \emph{good}.
 \end{definition}

  \begin{remark}\label{remgood} Condition \eq{good} is stronger than the usual condition in the definition of a good box  (cf. \cite{vDK,CH1,GKboot,Kle}), where decay is postulated only from the center of the box to its boundary.  We introduce  the exponential decay in $\|x-y\|$ for arbitrary $x,y$ in the box, not too close to each other,   in order to prove Lemma~\ref{lemkeyMSA}, where we will need to consider  locations $x$ and $y$  that may be anywhere in a box $\Lambda^{\pr}$. In particular, we will need to consider the case when both $x$ and $y$ are close to the boundary of $\Lambda^{\pr}$. Thus,  we will need to apply Lemma~\ref{lempreSLI}  for boxes $\Lambda \subset \Lambda^{\pr}$ that  touch the boundary of $\Lambda^{\pr}$ (i.e., $ \partial \Lambda \cap  \partial \Lambda^{\pr}\not= \emptyset$). For this reason  we defined $\Upsilon_\Lambda^{\Lambda^\pr}$  in \eq{Upsilon} in terms of $\partial^{\Lambda^\pr} \Lambda$, the boundary of $ \Lambda$ in $\Lambda^{\pr}$.
\end{remark}

\begin{remark} \label{remgoodmble} 
 It follows from \eq{Vlocal} and \eq{finvolH} that for all $E \in \C$ we have
  \beq
 \set{  \Lambda_{L} \text{\ is  $(E,m,\varsigma,S)$-good}}:=\set{\bom \in \Omega ; \; \Lambda_{L} \text{\ is  $(\bom,E,m,\varsigma,S)$-good}}\in \cF_{\Lambda_{L}}.
 \eeq
Moreover, the set
 \beq \label{goodmeas}
 \set{\pa{E,\bom_{\Lambda_{L}}} \in \R \times \Omega_{\Lambda_{L}} ; \; \Lambda_{L} \text{\ is  $(\bom,E,m,\varsigma,S)$-good}} 
 \eeq
is  closed in $\R \times\Omega_{\Lambda_{L}} $, and hence jointly measurable in $\pa{E,\bom_{\Lambda_{L}}}$.
 \end{remark}

  \begin{definition} Consider an energy $E\in \R$,   a rate of decay $m>0$, and numbers  $0< \varsigma < 1$ and $p>0$. A scale $L>0$ is called $(E,m,\varsigma, p)$- \emph{good} if for every  $x \in \R^d$ we have
  \beq\label{probgood}
  \P\set{  \Lambda_{L}(x) \text{\ is  $(E,m,\varsigma)$-good}} \ge 1 - L^{-pd}  .
  \eeq
  \end{definition}

 If a box  $\Lambda_{L}$ is  $(\bom,E,m,\varsigma)$-good, then it is just as good for energies $E^\pr$ such that $\abs{E^\pr - E} \le  \e^{-cL}$, the precise statement being given in the  following definition and lemma.
 
  \begin{definition}\label{defjgood} Consider a configuration  $\bom \in{\Omega}$, an  energy $E\in \C$,   a rate of decay $m>0$, and   $0< \varsigma< 1$.  A   box  $\Lambda_{L}$ is said to be $(\bom,E,m,\varsigma)$- \emph{jgood}  (just as good)  if 
   \begin{align}\label{jweg}
\| R_{\bom, \Lambda_{L}}(E) \|& \le 2\e^{L^{1-\varsigma}}  \\
\intertext{and} 
\| \Chi_x R_{\bom, \Lambda_{L}}(E) \Chi_y \|& \le 2\e^{-m  \norm{x-y}} \qtx{for all}  x,y \in \Lambda_{L} \qtx{with}  \norm{x-y}\ge \tfrac L{100}. \label{jgood}
\end{align} 
 \end{definition}

 \begin{lemma}\label{lemjgood}  Let $\bom \in{\Omega}$,  $E\in \C$, $0< \tau <\varsigma< 1$.  Suppose the   box  $\Lambda_{L}$ is  $(\bom,E,m,\varsigma)$-good with a rate of decay  $m\ge L^{-\tau}$.   Then, if  $L \ge \widetilde{L}_{\vs,\tau}$, the box  $\Lambda_{L}$ is  $(\bom,E^\pr,m,\varsigma)$-jgood  for all  energies  $E^\pr  \in \C$ such that $\abs{E^\pr -E}\le \e^{-2 mL}$.
 \end{lemma}
 
 \begin{proof} By the resolvent identity, 
 \beq R_{\bom, \Lambda_{L}}(E^\pr)=  R_{\bom,\Lambda_{L}}(E) - (E^\pr-E) R_{\bom,\Lambda_{L}}(E)  R_{\bom, \Lambda_{L}}(E^\pr).
 \eeq
Thus, for   $\abs{E^\pr - E} \le  \e^{-2mL}$, we get
\beq
\norm{R_{\bom, \Lambda_{L}}(E^\pr)} \le \e^{L^{1-\varsigma}} +  \e^{-2mL}\e^{L^{1-\varsigma}}\norm{R_{\bom, \Lambda_{L}}(E^\pr)}
\eeq
 Since $0< \tau <\varsigma< 1$, \eq{jweg} follows.
 
 Similarly, using also \eq{jweg},  given  $x,y \in \Lambda_{L}$ with $ \norm{x-y}\ge \tfrac L{100}$, we have
 \beq
 \norm{ \Chi_x R_{\bom, \Lambda_{L}}(E^\pr) \Chi_y} \le \e^{-m  \norm{x-y}} +2 \e^{-2mL}\e^{2L^{1-\varsigma}},
 \eeq
 and \eq{jgood} follows.
 \end{proof}

 We also need the following variant of Lemma~\ref{lemjgood}; the proof is almost identical.

\begin{lemma}\label{lemjgoodm1} Let $\bom \in{\Omega}$,  $E\in \C$, $0< \varsigma< 1$,    $0 <\widetilde{m}<m$.   Suppose the box  $\Lambda_{L}$ is $(\bom,E,m,\varsigma)$-good.  Then,  if  $L \ge \widetilde{L}_{\vs, \widetilde{m}}$, given  $E^{\pr}\in \C$ with  $\abs{E^{\pr}-E}\le \e^{-m_{1}L}$, where  $m_{1} \in [\widetilde{m},m]$,
  the box 
 $\Lambda_{L}$ is $(\bom,E^{\pr},m_{2},\varsigma)$-jgood with
 \beq\label{m2m1}
 m_{2}= m_{1}\pa{1 - C \widetilde{m}^{-1}L^{ -\vs }}.
 \eeq
\end{lemma}

 The following definition will be needed only for real energies.

  \begin{definition} Consider an energy $E\in \R$,   a rate of decay $m>0$, and numbers  $0< \varsigma,\varsigma^\pr < 1$ and $p>0$.
  \begin{enumerate}
  
  \item Given a   box  $\Lambda_{L}$, a subset $S \subset\widetilde{ \Lambda}_{L}$ is called $\vs^\pr$- \emph{abundant} if  
     \beq\label{abundant}
  \# \pa{S \cap \Lambda_{\frac L 5}}\ge L^{(1 - \vs^\pr)d}\quad \text{for all boxes $ \Lambda_{\frac L 5} \subset \Lambda_L$}.
  \eeq

  \item   Given a   box  $\Lambda_{L}$,    an event $\cC$ is said to be
  $(\Lambda_{L},E,m,\varsigma, \varsigma^\pr)$- \emph{adapted} if  there exists a $\vs^\pr$-abundant subset $S_\cC \subset\widetilde{ \Lambda}_{L}$ such that  $\cC \in \cF_{\Lambda_L\setminus S_\cC}$ and 
   $\Lambda_{L}$ is $(\bom,E,m,\varsigma,S_\cC)$-good for all $\bom \in \cC$. In this case $\cC$ will also be called $(\Lambda_{L},E,m,\varsigma, \varsigma^\pr,S_\cC)$- \emph{adapted}.

\item     Given a   box  $\Lambda_{L}$, an event $\cE$ is called  $(\Lambda_{L},E,m,\varsigma, \varsigma^\pr)$- \emph{extra good} if it is the disjoint union of a finite number of  $(\Lambda_{L},E,m,\varsigma, \varsigma^\pr)$-adapted events, i.e.,  there exist disjoint   $(\Lambda_{L},E,m,\varsigma, \varsigma^\pr)$-adapted events $\set{\cC_i}_{i=1,2,\ldots,I}$ such that
\beq  \label{Eloc}
\cE= \bigsqcup_{i=1}^I  \cC_i  .
\eeq

\item A scale $L>0$ is called $(E,m,\varsigma, \varsigma^\pr,p)$- \emph{extra good} if for every  $x \in \R^d$  there exists a $(\Lambda_{L}(x),E,m,\varsigma, \varsigma^\pr)$-{extra good} event  $\cE_{L,x}$ such that
\beq \label{PEL0}
\P\set{\cE_{L,x}} \ge 1 - L^{-pd}  .
\eeq
\end{enumerate}
 \end{definition}
 
If a scale $L$ is  $(E,m,\varsigma, \varsigma^\pr,p)$-{extra good}, it is clearly also  $(E,m,\varsigma, p)$-good.

 \subsection{Tools for the multiscale analysis}
 We now combine Lemmas~\ref{lempreSLI} and \ref{lemEDI} with good boxes to obtain crucial tools for the multiscale analysis. In Lemmas~\ref{lemSLI} and \ref{lemkeyMSA}  we 
will  not know a priori that $E \notin \sigma(H_{\bom,\Lambda})$, and we will apply  Lemma~\ref{lempreSLI} with the notation given in  \eq{exttoE}.

\begin{lemma}\label{lemSLI}  Fix a  configuration  $\bom \in{\Omega}$ and an energy $E\in \C$.   Let $\La$  be either $\R^{d}$ or a box $\La_{L}$. Consider a scale $\ell$, with  $\ell <\frac  L  6$ if $\La=\La_{L}$, numbers $0< \tau<\varsigma< 1$,   and 
 $m\ge \ell^{-\tau}$. Let $\Theta\subset  \Lambda$ be such that for all $x \in \La \setminus \Theta$ there exists a $(\bom,E,m,\varsigma)$-good box, denoted by $\La_\ell\up{x}$, such that $\La_\ell\up{x}\subset \La$ with  $\La_{\frac {\ell}5}(x) \cap \La \subset \La_\ell\up{x}$. Then there exists a constant  $C=C_{d, V_{\mathrm{per}},E}$, locally bounded in $E$, such that setting 
 \beq\label{mpr}
m^\pr= m\pa{1 - C (\log \ell)  \ell^{\tau -1}},
\eeq
the following holds:
\begin{enumerate}
\item \label{sublemSLI} For all  $x, y \in \Lambda$ with  $x \notin \Theta $ we have
\beq\label{goodcasem}
\norm{\Chi_y R_{\bom,\Lambda}(E) \Chi_x} \le \norm{\Chi_y \Chi_{\La_\ell\up{x,\La,\frac 1 2}} R_{\bom,\La_\ell\up{x}}(E)\Chi_{x}}+  \e^{-m^\pr \norm{x-x_1} }\norm{\Chi_y R_{\bom,\Lambda}(E) \Chi_{x_1}},
\eeq
for some $x_1 \in  \Upsilon_{\La_\ell\up{x}}^{\Lambda} $, so in particular
\beq \label{x1dist}
\tfrac \ell {11}   \le \norm{x-x_1} \le \ell.
\eeq

\item \label{sublemSLI2} Let $x, y \in \Lambda$ with  $x \notin \Theta $ and $\norm{x-y} \ge \ell $.
  Then
\beq\label{goodcasem2}
\norm{\Chi_y R_{\bom,\Lambda}(E) \Chi_x} \le   \e^{-m^{\pr} \norm{x-x^\pr} }\norm{\Chi_y R_{\bom,\Lambda}(E) \Chi_{x^\pr}}
\eeq
for some $x^\pr \in \Lambda$ such that {either}  $x^\pr \in \Theta$  or  $ \norm{x^\pr-y} <  \ell $, i.e.,
\beq \label{x1dist2}
 x^\pr \in \Theta \cup  \Lambda_{2\ell}(y).
\eeq

 \item \label{sublemEDII2}   Suppose $E\in \R$ and $\psi$ is a generalized   eigenfunction 
 of $H_{\bom,\Lambda}$  with generalized eigenvalue $E^\pr  \in \br{E - \e^{-2 m \ell},E + \e^{-2 m\ell}}$. 
 Then for all $x \in \Lambda \setminus \Theta$ we have
 \begin{align}\label{EDI22}
\norm{\Chi_{x}\psi}&\le   \e^{-m^{\pr} \norm{x-x^\pr} } \norm{\Chi_{x^\pr}\psi}\le   \e^{-\frac {m^{\pr}} {11} \ell }  \norm{\Chi_{x^\pr}\psi}\qtx{for some} x^\pr  \in  \Upsilon_{\La_\ell\up{x}}^{\Lambda}, \\
\intertext{and also}
\label{EDI2}
\norm{\Chi_{x}\psi}&\le   \e^{-m^{\pr} \norm{x-x^{\pr\pr}} } \norm{\Chi_{x^{\pr\pr}}\psi}\le  \e^{-{m^{\pr}} \dist\set{x,\Theta}} \norm{\Chi_{x^{\pr\pr}}\psi}\qtx{for some} x^{\pr\pr} \in \Theta.
\end{align}
If $E^\pr=E$, \eq{EDI22} and \eq{EDI2} hold with $m$ substituted for $m^\pr$.

\end{enumerate}
\end{lemma}

\begin{proof} 
\begin{enumerate}
\item  Since $x \notin \Theta$,   we use the existence of the good box $\La_\ell\up{x}$ and   apply \eq{yRxall} to get
\beq\label{goodcase222}
\norm{\Chi_y R_{\bom,\Lambda}(E) \Chi_x} \le  \norm{\Chi_y\Chi_{\La_\ell\up{x,\La,\frac 1 2}} R_{\bom,\La_\ell\up{x}}(E)\Chi_{x}}+ \gamma_{E} \ell^{d-1} \e^{-m  \norm{x-x_1} }\norm{\Chi_y R_{\bom,\Lambda}(E) \Chi_{x_1}}
\eeq
for some $x_1 \in  \Upsilon_{\La_\ell\up{x}}^{\Lambda} $, so $ \frac \ell {10} -\tfrac {\delta_{+}+1}{2}\le \norm{x-x_1} \le  \ell -\frac { \delta_{+}+1}{2} $,  hence \eq{x1dist},  and we have \eq{goodcasem} with \eq{mpr}.

\item  Since $x \notin \Theta $ and $\norm{x-y} \ge  \ell $, we  apply \eq{goodcasem} repeatedly to get  
\beq\label{goodcasem3}
\norm{\Chi_y R_{\bom,\Lambda}(E) \Chi_x} \le  \e^{-m^\pr \sum_{i=1}^n
\norm{x_{i-1} -x_i}}\norm{\Chi_y R_{\bom,\Lambda}(E) \Chi_{x_n}},
\eeq
with $x_0=x$ and  $x_i \in  \Upsilon_{\La_\ell\up{x_{i-1}}}^{\Lambda} $, $ i=1,2,\ldots,n$, where $n\in \N$ is such that $x_i \notin \Theta $ and $\norm{x_i-y} \ge  \ell $ for $ i=0,1,\ldots,n-1$, and either $x_n\in \Theta $ or $\norm{x_n-y} <  \ell $.   Since
$\norm{x_0 - x_n} \le \sum_{i=1}^n
\norm{x_{i-1} -x_i}$, \eq{goodcasem2} follows.

\item   It follows from Lemma~\ref{lemjgood} that for all $x \notin \Theta$ the box $\La_\ell\up{x}$ is $(\bom,E^\pr,m,\varsigma)$-jgood.  Thus, given $x \notin \Theta$,  we apply Lemma~\ref{lemEDI} with the  box $\La_\ell\up{x}$ to get  \eq{EDI22}.   To prove  \eq{EDI2}, we proceed similarly to the proof of \eq{goodcasem2},  applying  Lemma~\ref{lemEDI} repeatedly. 

\end{enumerate}

Note that in (iii) the constance $C$ in \eq{mpr} depends on $E^{\pr}$.  Since $\abs{E^{\pr}-E}\le 1$,  we can fix a constant $C=C_{E,d, V_{\mathrm{per}}}$, locally bounded in $E$, that works for all the conclusions of the lemma.
\end{proof}

 The following lemma will play an important role in the multiscale analysis. We use the notation given in \eq{LambdaupL}.

 \begin{lemma} \label{lemkeyMSA} Fix a  configuration  $\bom \in{\Omega}$ and an energy $E\in \C$.  Consider a box $\Lambda=\Lambda_L$ and   let $\varsigma,\rho, \kappa,\tau \in ]0,1[$,  $\ell= L^\rho$,  $m\ge \ell^{-\tau}$, $K,K^{\pr} \in \N$, where
 \beq\label{rhovstau}
 \kappa \vs >   {\tau}\rho.
 \eeq
Suppose   there exist   $ \Theta=\sqcup_{j=1}^K \Theta_j \subset \Lambda$  satisfying the following conditions:
\begin{enumerate}
\item There exist disjoint boxes  $\Lambda_j= \Lambda_{L_j}(y_j) \subset  \Lambda$ with  $L^{\kappa}\le L_{j}\le K^{\pr }L^{\kappa}$, $j=1,2,\ldots,K$, 
such that     
\beq\label{Thetajs}
\Theta_j \subset \Lambda_j\up{\Lambda, \frac {L^{\kappa}} {10}},
\eeq
and
\beq\label{wegL1}
\| R_{\bom,  \Lambda_{j}}(E) \| \le \e^{L^{\kappa(1-\varsigma)}} .
\eeq

\item  For all $x \in \La \setminus \Theta$ there exists a $(\bom,E,m,\varsigma)$-good box $\La_\ell\up{x}\subset \La$ such that $\La_{\frac {\ell}5}(x) \cap \La \subset \La_\ell\up{x}$.

\end{enumerate}
Then the box $\La$ is     $(\bom,E,M,\varsigma)$-good for $L\ge L_{d,E,\Vper,\vs,K,\tau,\kappa}$, where 
\beq \label{massM}
M \ge m(1 - C\pa{L^{\kappa-1} + L^{\rho \tau -\kappa \zeta}}) \ge L^{-\tau}
\eeq
and   $C=C_{E,d,  V_{\mathrm{per}},K,K^{\pr}}$ is locally bounded in $E$.
 \end{lemma}

 \begin{proof}

 We start by proving  \eq{weg} for $\La$. Since $H_{\bom,\Lambda}$ has discrete spectrum,
 there exists $\eps>0$ such that $E^{\pr}\notin \sigma(H_{\bom,\Lambda}) $ if  $0 <\abs{E^{\pr}-E}< \eps$.  We take $\eps \le \e^{-2 m\ell}$, so the boxes $\La_\ell\up{x}$ given in  condition (ii) are $(\bom,E^{\pr},m,\varsigma)$-jgood by Lemma~\ref{lemjgood}, and small enough such that it follows from  \eq{wegL1} that
 \beq\label{wegL15}
\| R_{\bom,  \Lambda_{j}}(E^{\pr}) \| \le 2\e^{L^{\kappa(1-\varsigma)}} \quad \text{for}\quad  j=1,2,\ldots,K.
\eeq

We will estimate $\norm{ R_{\bom,\Lambda}(E^{\pr}) }$ for  $0 <\abs{E^{\pr}-E}< \eps$.  
Suppose either $x$ or $y$ are not in $\Theta$, say $x \notin \Theta$.  In this case 
we apply Lemma~\ref{lemSLI}\ref{sublemSLI}. It  follows from \eq{goodcasem},  appropriately modified for jgood boxes, Definition~\ref{defjgood}, and \eq{x1dist}, that
\beq\begin{split}\label{goodcase}
\norm{\Chi_y R_{\bom,\Lambda}(E^{\pr}) \Chi_x} &\le 2\e^{\ell^{1 -\vs}} + 2\e^{-m^\pr \frac \ell {11}}\norm{R_{\bom,\Lambda}(E^{\pr})}\\
&\le 2\e^{\ell^{1 -\vs}} + 2 \e^{-\frac 1 {15} \ell^{1 -\tau} }\norm{R_{\bom,\Lambda}(E^{\pr})}\le 2 \e^{\ell^{1 -\vs}} + \tfrac 1 2 L^{-2d}\norm{R_{\bom,\Lambda}(E^{\pr})}
\end{split}\eeq
 for large $L$.   If $x \in \Theta$ and   $y \notin \Theta$ we use $\norm{\Chi_y R_{\bom,\Lambda}(E^{\pr}) \Chi_x}=\norm{\Chi_x R_{\bom,\Lambda}(E^{\pr}) \Chi_y}$ to get \eq{goodcase}.  Suppose now
 $x, y \in  \Theta$, say $x \in \Theta_s$.  Then we apply  \eq{yRxall} with the box $\Lambda_s$, and use \eq{wegL15}, getting
 \beq\label{bad11}
\norm{\Chi_y R_{\bom,\Lambda}(E^{\pr}) \Chi_x} \le 2\e^{L^{\kappa(1 -\vs)}} +2 \gamma \pa{K^{\pr}L^{\kappa}}^{d-1}  \e^{L^{\kappa(1 -\vs)}}\norm{\Chi_y R_{\bom,\Lambda}(E^{\pr}) \Chi_{x_0}},
\eeq
where $x_0 \in  \Upsilon_{\Lambda_s}^{\Lambda} $ and $\gamma=\gamma_{E+1}$.  Note that \eq{Thetajs} implies $\dist\set{x_0,\Theta}\ge \frac {L^{\kappa}} {11}$;  in particular, $\norm{x_{0} -y}\ge \frac {L^{\kappa}} {11}$ as $y\in \Theta$.  We can now use  
Lemma~\ref{lemSLI}\ref{sublemSLI2}, with $\frac{m^{\pr}} 2$ replacing $m^{\pr}$  in \eq{goodcasem2} to compensate for using jgood boxes instead of good boxes,  to conclude that
\beq \label{bad12}
\norm{\Chi_y R_{\bom,\Lambda}(E^{\pr}) \Chi_{x_0}}\le \e^{-\frac{m^\pr} 2 \norm{x_0 - x^\pr} }\norm{\Chi_y R_{\bom,\Lambda}(E^{\pr}) \Chi_{x^\pr}}\le \e^{-m^\pr  \frac {L^{\kappa}} {30} }\norm{\Chi_y R_{\bom,\Lambda}(E^{\pr}) \Chi_{x^\pr}} ,
\eeq
where $x^\pr$ satisfies \eq{x1dist2}, so $ \norm{x_0 - x^\pr} \ge   \frac {L^{\kappa}} {11} -  \ell >  \frac {L^{\kappa}} {15} $ for large $L$. From \eq{bad11},  \eq{bad12}, and \eq{rhovstau}, we conclude that,  for large $L$, we have 
\begin{align}\notag
\norm{\Chi_y R_{\bom,\Lambda}(E^{\pr}) \Chi_x}& \le2 \e^{L^{\kappa(1 -\vs)}} + 2\gamma \pa{K^{\pr}L^{\kappa}}^{d-1}  \e^{L^{\kappa(1 -\vs)}}\e^{-\frac 1 {60}  {\ell^{-\tau}} {L^{\kappa}} }\norm{ R_{\bom,\Lambda}(E^{\pr}) }\\
& \le 2 \e^{L^{\kappa(1 -\vs)}} +\tfrac 1 2 L^{-2d}\norm{R_{\bom,\Lambda}(E^{\pr}
)}
 .\label{badcase}
\end{align}
Combining \eq{goodcase} and \eq{badcase} we get
\beq\begin{split}\label{badcase2}
\norm{ R_{\bom,\Lambda}(E^{\pr}) }&\le  L^{2d}\set{2 \e^{L^{\kappa(1 -\vs)}} + \tfrac 1 2 L^{-2d}\norm{R_{\bom,\Lambda}(E^{\pr})}}\\
& \le 2 L^{2d} \e^{L^{\kappa(1 -\vs)}}+\tfrac 1 2 \norm{R_{\bom,\Lambda}(E^{\pr})},
\end{split}\eeq
and hence, for large $L$,
\beq\label{msweg345}
\norm{ R_{\bom,\Lambda}(E^{\pr}) }\le 2   L^{2d}  \e^{L^{\kappa(1 -\vs)}}\le \e^{L^{1 -\vs}}.
\eeq

We now conclude that for large $L$ we have
\beq\label{msweg}
\norm{ R_{\bom,\Lambda}(E) }=\lim_{E^{\pr}\to E}\norm{ R_{\bom,\Lambda}(E^{\pr}) }\le \e^{L^{1 -\vs}}
\eeq

To finish the proof, we need to prove \eq{good} for the box $\Lambda$.

\begin{sublemma} \label{sublemallcases} Given $s \in \set{1,2,\ldots,K}$, let $x,y \in \Lambda$ with $x\in \Theta_s $ and $\norm{x-y}\ge  {L_{s}}$.
 Then there exist  $x\up{0} \in  \Upsilon_{\Lambda_s}^{\Lambda} $ and   $x^\pr \in \Lambda$, with  $x^\pr$ satisfying \eq{x1dist2} and
\beq \label{x0xy}
\tfrac {1 } {11} L^{\kappa}\le \norm{x\up{0} - x } \le L_{s  }- \tfrac {1 } {10} L^{\kappa} \quad \text{and}\quad  \norm{x\up{0} -y} \ge  \tfrac{1 }{10} L^{\kappa},
\eeq 
such that
 \beq\label{bad1125}
\norm{\Chi_y R_{\bom,\Lambda}(E) \Chi_{x}} \le  \e^{-m^{\pr\pr} \norm{x\up{0} - x^{\pr} } }\norm{\Chi_y R_{\bom,\Lambda}(E) \Chi_{x^{\pr}}},
\eeq
where
\beq  \label{mprpr}
m^{\pr\pr}= m^\pr\pa{1- C \ell^\tau L^{-\kappa \vs}} \quad \text{with} \quad C=C_{E,d,V_{\mathrm{per}},K^{\pr}} \quad \text{locally bounded in $E$}.
\eeq

\end{sublemma}

\begin{proof}
 Let $x,y \in \Lambda$ with $x\in \Theta_s $ and $\norm{x-y}\ge  L_{s}$.  
We proceed as in \eq{bad11} and  \eq{bad12} (note that we are now working at energy $E$, so we have \eq{wegL1} and   condition (ii) holds), getting
 \beq\begin{split}\label{bad112}
\norm{\Chi_y R_{\bom,\Lambda}(E) \Chi_{x}}& \le \gamma_{E} \pa{K^\pr L^\kappa}^{d-1}  \e^{L^{\kappa(1 -\vs)}}\norm{\Chi_y R_{\bom,\Lambda}(E) \Chi_{x\up{0}}}\\
& \le \gamma_{E} \pa{K^\pr L^\kappa}^{d-1}\e^{L^{\kappa(1 -\vs)}} \e^{-m^\pr \norm{x\up{0} - x^{\pr} } }\norm{\Chi_y R_{\bom,\Lambda}(E) \Chi_{x^{\pr}}}\\
& \le  \e^{-m^{\pr\pr} \norm{x\up{0} - x^{\pr} } }\norm{\Chi_y R_{\bom,\Lambda}(E) \Chi_{x^{\pr}}},
\end{split}\eeq
where $x\up{0} \in  \Upsilon_{\Lambda_s}^{\Lambda} $, so   we have \eq{x0xy}, and  
 $x^\pr \in \Lambda $ satisfies \eq{x1dist2},
  so $ \norm{x\up{0} - x^\pr} \ge   \frac {L^{\kappa}} {11} - \ell >  \frac {L^{\kappa}} {15} $, 
and thus $m^{\pr\pr}$ is as in \eq{mprpr}.  
\end{proof}

Now let  $x,y \in \Lambda$ with $\norm{x-y}\ge\frac {L}{100} \ge K^{\pr}  {L^{\kappa}}$. If $x\notin \Theta$, we  apply Lemma~\ref{lemSLI}\ref{sublemSLI2}, obtaining $x^\pr$ satisfying \eq{x1dist2}.  If $\norm{x^\pr -y }<K^{\pr} L^{\kappa}$, we stop.   Otherwise we then start from $x^\pr$ and apply Sublemma~\ref{sublemallcases}   repeatedly, until we get\beq\label{goodcasem45}
\norm{\Chi_y R_{\bom,\Lambda}(E) \Chi_x} \le  \e^{-m^{\pr\pr} \sum_{i=1}^n
\norm{x_{i-1}\up{0} -x_i}}\norm{\Chi_y R_{\bom,\Lambda}(E) \Chi_{x_n}},
\eeq
with $x=x_0=x_0\up{0}$, $x_1=x^\pr$,  $x_{i-1}\up{0}$ and $x_i$ correspond to  $x\up{0} $ and $x^\pr$ in Sublemma~\ref{sublemallcases}  for $x_{i-1}, y$ for  $ i=2,\ldots,n$, 
and  $n\in \N$ is such that  $\norm{x_i-y}\ge K^{\pr} L^{\kappa}$  (and hence $x_i \in \Theta$) for $ i=1,2,\ldots,n-1$, and  $\norm{x_n-y} <K^{\pr} L^{\kappa}$. If $x \in \Theta$, we start directly with   Sublemma~\ref{sublemallcases} obtaining also \eq{goodcasem45} but with $x=x_0$, and 
$x_{0}\up{0}$ and $x_1$ corresponding to  $x\up{0} $ and $x^\pr$ in Sublemma~\ref{sublemallcases}  for $x_{0}$ and $ y$.

Now let us choose distinct  $j_0,j_1,\ldots,j_r \in \set{0,1,2,\ldots,K+1}$, where $0\le r\le  K+1$, as follows:
\begin{enumerate}
\item[(a)]   If $x\notin \Theta$, we set $j_0=0$ and  $\Lambda_0=\Theta_0=\set{x}$.  If $x\in \Theta$, $j_0$ is determined by $x \in \Theta_{j_0}$.  Set also $\Theta_{K+1}=\set{x_n}$.

\item[(b)]  Pick $j_1\not= j_0$  such that for some $i_1\in \set{1,2,\ldots,n}$ we have
$x_{i_1-1}\up{0} \in  \Lambda_{j_0}$ and $x_{i_1} \in  \Theta_{j_1}$

\item[(c)]  Given $j_0,j_1,\ldots, j_{s}$,  if $i_{s}=n$, $r=s$, so stop. If not,   pick $j_{s+1} \notin \set{j_0,j_1,\ldots, j_{s}}$ such that  that for some $i_{s+1} \in  \set{1,2,\ldots,n}$ we have $x_{i_{s+1} -1}\up{0} \in \Lambda_{j_{s} }$ and   $x_{i_{s+1}}\in \Theta_{j_{s+1}}$.
\end{enumerate}
It then follows from \eq{goodcasem45} that
\beq\label{goodcasem458}
\norm{\Chi_y R_{\bom,\Lambda}(E) \Chi_x} \le  \e^{-m^{\pr\pr} \sum_{s=1}^r
\norm{x_{i_s-1}\up{0} -x_{i_s}}}\norm{R_{\bom,\Lambda}(E) }.
\eeq
By our construction,
\beq
\sum_{s=1}^r \norm{x_{i_s-1}\up{0} -x_{i_s}}\ge \sum_{s=1}^r
 \dist\set{\Lambda_{s-1},\La_s}
 \ge \norm{x - x_n} - KK^{\pr}L^{\kappa}\ge \norm{x - y} - (KK^{\pr}+1) L^{\kappa}.
\eeq
It follows, using also \eq{msweg},  that
\beq\label{goodcasem459}
\norm{\Chi_y R_{\bom,\Lambda}(E) \Chi_x} \le  \e^{-m^{\pr\pr} \pa{ \norm{x - y} - (KK^{\pr}+1) L^{\kappa}}}
 \e^{L^{1 -\vs}}\le  \e^{-M \pa{ \norm{x - y}}},
\eeq
where
\beq \label{massM111}
M= m^{\pr\pr}\pa{1- C \pa{ L^{\kappa} L^{-1} + \ell^\tau L^{-\vs}}}, 
\eeq
with   a constant $ C=C_{E,d,V_{\mathrm{per}},K,K^{\pr}}$ locally bounded in $E$.

The lemma is proved.
\end{proof}

\subsection{Suitable coverings of  boxes and annuli}

\subsubsection{Suitable coverings of  boxes}
  
\begin{definition}\label{defsuitcov}
Given scales $\ell <  L$,  a  \emph{suitable $\ell$-covering} of a box $\Lambda_{L}(x)$ is
a collection of boxes $\Lambda_{\ell}$ of the form
\begin{align}\label{standardcover}
\cG_{\Lambda_{L}(x)}^{(\ell)}= \{ \Lambda_{\ell}(r)\}_{r \in \G_{\Lambda_{L}(x)}^{(\ell)}},
\end{align}
where
\beq  \label{bbG}
\G_{\Lambda_{L}(x)}^{(\ell)}:= \{ x + \alpha\ell  \Z^{d}\}\cap \Lambda_{L}(x)\quad 
\text{with}  \quad \alpha \in \br{\tfrac {3} {5},\tfrac {4} {5}}   \cap \set{\tfrac {L-\ell}{2 \ell n}; \, n \in \N }.
\eeq
\end{definition}

\begin{lemma}\label{lemcovering} Let $\ell \le \frac  L   6$. Then every box   $\Lambda_{L}(x)$  has a   suitable $\ell$-covering, and   for any  suitable $\ell$-covering
  $\cG_{\Lambda_{L}(x)}^{(\ell)}$  of $\Lambda_{L}(x)$ we have 
  \begin{align}\label{nestingproperty}
&\Lambda_{L}(x) =\bigcup_{r \in\G_{\Lambda_{L}(x)}^{(\ell)} } \Lambda_{\ell}(r),\\ \label{bdrycover}
&\text{for each $y \in \Lambda_{L}(x)$ there is $r \in \G_{\Lambda_{L}(x)}^{(\ell)}$ with $\Lambda_{\frac{\ell} 5}(y)\cap \Lambda_{L}(x) \subset \Lambda_{\ell}(r)$},\\
\label{freeguarantee}
&\Lambda_{\frac{\ell}{5}}(r)\cap \Lambda_{\ell}(r^{\prime})=\emptyset
\quad \text{for all} \quad r,r^{\pr}\in  x + \alpha\ell  \Z^{d}, \  r\ne r^{\prime},\\ \label{number}
&  (\tfrac{L} {\ell})^{d}\le \#\G_{\Lambda_{L}(x)}^{(\ell)}= \pa{ \tfrac{L-\ell} {\alpha \ell}+1}^d \le   (\tfrac{2L} {\ell})^{d}.
\end{align}
Moreover,  given $y \in  x + \alpha \ell  \Z^{d}$ and $n \in \N$, it follows that
\beq \label{nesting}
\Lambda_{(2  n \alpha  +1) \ell}(y)= \bigcup_{r \in  \{ x + \alpha \ell \Z^{d}\}\cap\Lambda_{(2n \alpha  +1) \ell}(y) } \Lambda_{\ell}(r),
\eeq
and  $ \{ \Lambda_{\ell}(r)\}_{r \in  \{ x + \alpha\ell  \Z^{d}\}\cap\Lambda_{(2n\alpha +1) \ell}(y)}$ is a suitable $\ell$-covering of the box $\Lambda_{(2 n\alpha +1) \ell}(y)$.  In particular,
\beq \label{ell5cover}
\text{for each $y \in \Z^{d}$ there is $r \in  x + \alpha\ell  \Z^{d}$ with $\Lambda_{\frac{\ell} 5}(y) \subset \Lambda_{\ell}(r)$}.
\eeq
\end{lemma}

\begin{proof}
It suffices to note that $\ell \le \frac  L   6$ ensures
$\br{\tfrac {3} {5},\tfrac {4} {5}}   \cap \set{\tfrac {L-\ell}{2 \ell n}; \, n \in \N }\not= \emptyset$,
$\alpha\le \frac 4 5$ gives \eqref{bdrycover} and \eq{ell5cover}, and  $\alpha\ge\frac 35$ yields  \eqref{freeguarantee}. 
\end{proof} 

To fixate ideas we make the following definition.

\begin{definition}\label{defstdcov}
The  \emph{standard $\ell$-covering} of a box $\Lambda_{L}(x)$ is the unique   suitable $\ell$-covering of  $\Lambda_{L}(x)$ with 
\beq\label{alphaL}
\alpha =\alpha_{L,\ell}:= \max \set{\br{\tfrac {3} {5},\tfrac {4} {5}}   \cap \set{\tfrac {L-\ell}{2 \ell n}; \, n \in \N }}.
\eeq
\end{definition} 

We now consider standard coverings by good boxes.

 \begin{definition}\label{defpgood} Consider a configuration  $\bom \in{\Omega}$,  an energy $E\in \R$,  a rate of decay $m>0$,  $0< \varsigma< 1$, and $\eta >0$.  A   box  $\Lambda_{L}$ is said to be $(\bom,E,m,\varsigma,\eta)$- \emph{pgood} (for predecessor of good)  if, letting $\ell=L^{\frac 1{1 + \eta}}$, every box $\Lambda_{\ell}$ in the  standard $\ell$-covering of $\La_{L}$ is 
 $(\bom,E,m,\varsigma)$-good.
  \end{definition}

  \begin{lemma} \label{lempggodtogood}  Suppose
 the box  $\Lambda_{L}$ is $(\bom,E,m,\varsigma,\eta)$-pgood for some  $\bom \in{\Omega}$,   $E\in \R$,   $m>0$,  $0< \varsigma< 1$, and $\eta >0$,   set $\ell=L^{\frac 1{1 + \eta}}$, and let $0<\widehat{m}\le m$.  Then, if  $L \ge \widehat{L}_{\vs, \widetilde{m},\eta}$, given  $m_{1} \in [\widehat{m},m]$,  the box $\Lambda_{L}$ is   $(\bom,E^\pr,{M_{1}},\varsigma)$-good for all   energies $E^\pr \in \C$ such that
  $\abs{E^\pr-E}\le \e^{-m_{1}\ell}$, where 
  \beq \label{massM1}
M_{1}= m_{1}\pa{1- C_{d,p,\widehat{m}}  L^{-\frac {\min\set{\vs,\eta}}{1 + \eta}} } .
\eeq
    \end{lemma}

  \begin{proof} Let $\Lambda_{\ell}$ be  $(\bom,E,m,\varsigma)$-good and  $E^\pr \in \C$ with $\abs{E^{\pr}-E}\le \e^{-m_{1}\ell}$.   It follows from Lemma \ref{lemjgoodm1} that $\Lambda_{\ell}$ is $(\bom,E^{\pr},m_{2},\varsigma)$-jgood if  $\ell \ge \widetilde{\ell}_{\vs, \widehat{m}}$, with $m_{2}= m_{1}\pa{1 - C \widehat{m}^{-1}\ell^{ -\vs }}.$

Now suppose $\La=\Lambda_{L}$ is $(\bom,E,m,\varsigma,\eta)$-pgood and $\ell \ge \widetilde{\ell}_{\vs, \widehat{m}}$ . We proceed as in Lemma~\ref{lemkeyMSA} (but note that we have $\Theta=\emptyset$).   Proceeding as in \eq{goodcase} and \eq{badcase2},  using the fact that  every box $\Lambda_{\ell}$ in the  standard $\ell$-covering of $\La_{L}$ is 
$ (\bom,E^{\pr},m_{2},\varsigma)$-jgood,
 we get, for $L$ sufficently large,
\beq\begin{split}\label{badcase25}
\norm{ R_{\bom,\Lambda}(E^{\pr}) }&\le  L^{2d}\pa{ 2\e^{\ell^{1 -\vs}} + \e^{-m_{3} \frac \ell {11}}\norm{R_{\bom,\Lambda}(E^{\pr})}}\\
&\le  2L^{2d} e^{\ell^{1 -\vs}}+\tfrac 1 2 \norm{R_{\bom,\Lambda}(E)},
\end{split}\eeq
where $m_{3}= m_{2}\pa{1 - C_{d,V_{\mathrm{per}},p,\widehat{m},\I} \frac {\log \ell} \ell}$,
and hence
\beq\label{2msweg}
\norm{ R_{\bom,\Lambda}(E^{\pr}) }\le 4   L^{2d} e^{\ell^{1 -\vs}}\le \e^{L^{1 -\vs}}.
\eeq
Given  $x,y \in\La= \Lambda_{L}$ with $\norm{x-y}\ge\frac {L}{100}$, we proceed as in the derivation of \eq{goodcasem459}  (with $\Theta=\emptyset$) to obtain, using \eq{2msweg},
\beq\label{goodcasem45988}
\norm{\Chi_y R_{\bom,\Lambda}(E^{\pr}) \Chi_x} \le  \e^{-m_{3} \pa{ \norm{x - y} - \ell) }}
 4   L^{2d} e^{\ell^{1 -\vs}}\le  \e^{-M_{1}  \norm{x - y}},
\eeq
where $M_{1}$ is as in \eq{massM1}.
  \end{proof}

 \begin{lemma}\label{lemprobpgood} Suppose the scale $\ell$ is  $(E,m,\varsigma,p)$-good, where   $E\in \R$,   $m>0$,  $0< \varsigma< 1$, and $p>0$.  Then, if $L=\ell^{1+\eta}$, where  $0<\eta <p$, we have
    \beq\label{probxgood}
  \P\set{  \Lambda_{L}(x) \text{\ is  $(\bom,E,m,\varsigma,\eta)$-pgood}} \ge 1 -2^d L^{-\frac {p-\eta}{1+\eta}d} \quad \text{for all} \quad  x \in \R^d.
  \eeq  
  \end{lemma}

  \begin{proof} It follows from \eq{probgood} and  \eq{number} that
 \begin{align}
  \P\set{  \Lambda_{L} \text{\ is not  $E$-pgood}} <  \set{2 L^{\frac \eta {1 +\eta}}}^{d}  L^{-\frac {pd} {1 +\eta}}=2^d L^{-\frac {p-\eta}{1+\eta}d}.
\end{align} 
  \end{proof}

\subsubsection{Suitable coverings of   annuli}\label{seccovannulus}
Given scales $L_{1} < L_{2}$, we consider the open \emph{annulus} 
\beq
\La_{L_{2},L_{1}}(x):=\La_{L_{2}}(x)\setminus {\overline{\La}}_{L_{1}}(x) =
\set{y \in \R^d; \; \tfrac {L_{1}} 2 <  \norm{y-x} < \tfrac {L_{2}} 2}.
\eeq
We let  $\bar{\La}_{L_{2},L_{1}}(x):=\overline{\La_{L_{2},L_{1}}(x)}$ be the closed annulus,
and set $\La_{\infty,L}(x):=\R^{d}\setminus {\overline{\La}}_{L}(x)$.  

\begin{definition}\label{defsuitcovann}
Given scales $\ell , L_{1}, L_{2}$ with $ L_{1}<L_{2}$ and   $\ell < \frac{ L_{1}- L_{2}} 2$,  a  \emph{suitable $\ell$-covering} of an annulus $\La_{L_{2},L_{1}}(x)$ is
a collection of boxes $\Lambda_{\ell}$ of the form
\begin{align}\label{standardcoverann}
\cG_{\Lambda_{L_{2},L_{1}}(x)}^{(\ell)}= \{ \Lambda_{\ell}(r)\}_{r \in \G_{\Lambda_{L_{2},L_{1}}(x)}^{(\ell)}},
\end{align}
where
\begin{align} \label{bbGLL}
\G_{\Lambda_{L_{2},L_{1}}(x)}^{(\ell)}&:= \set{r \in  x +\mathbb{U}_{L_{1},\ell}+ \alpha\ell  \Z^{d}; \; \Lambda_{\ell}(r) \subset \Lambda_{L_{2},L_{1}}(x)}, \quad
\text{with}  \\ \mathbb{U}_{L_{1},\ell}&:= \set{0,\tfrac {L_{1}}2,-\tfrac {L_{1}}2, \tfrac {L_{1}+\ell}2,-\tfrac {L_{1}+\ell}2}^{d}\setminus \set{0,\tfrac {L_{1}}2,-\tfrac {L_{1}}2}^{d}, \\
\alpha & \in \br{\tfrac {3} {5},\tfrac {4} {5}}   \cap \set{\tfrac {L_{2}-L_{1}-2\ell}{2 \ell n}; \, n \in \N }.
\end{align}
\end{definition}

\begin{lemma}\label{lemcoveringann} Consider  scales $\ell , L_{1}, L_{2}$ with $ L_{1}<L_{2}$ and   $\ell < \frac{ L_{2}- L_{1}} 7$. Then every   annulus $\La_{L_{2},L_{1}}(x)$  has a   suitable $\ell$-covering, and   for any  suitable $\ell$-covering
  $\cG_{\Lambda_{L_{2},L_{1}}(x)}^{(\ell)}$  of $\La_{L_{2},L_{1}}(x)$ we have 
  \begin{align}\label{nestingpropertyann}
&\Lambda_{L_{2},L_{1}}(x) =\bigcup_{r \in\G_{\Lambda_{L_{2},L_{1}}(x)}^{(\ell)} } \Lambda_{\ell}(r),\\ \label{bdrycoverann}
&\text{given $y \in \Lambda_{L_{2},L_{1}}(x)$ there is $r \in \G_{\Lambda_{L_{2},L_{1}}(x)}^{(\ell)}$ with $\Lambda_{\frac{\ell} 5}(y)\cap\Lambda_{L_{2},L_{1}}(x) \subset \Lambda_{\ell}(r)$},\\
&   \label{number22}
 \#\G_{\Lambda_{L_{2},L_{1}}(x)}^{(\ell)} \le  (\tfrac{2L_{2}} {\ell})^{d}\# \mathbb{U}_{L_{1},\ell} \le  (\tfrac{10 L_{2}} {\ell})^{d}.
\end{align}
\end{lemma}

Definition~\ref{defsuitcovann} is similar  to Definition~\ref{defsuitcov}, and  Lemma~\ref{lemcoveringann} is proven similarly to Lemma~\ref{lemcovering}, but there are some differences.  In particular, we do not have the analog of \eq{freeguarantee}.

As in Definition~\ref{defstdcov}, the standard $\ell$-covering of $\Lambda_{L_{2},L_{1}}(x)$
corresponds to
\beq\label{alphaLL}
\alpha =\alpha_{L_{2},L_{1},\ell}:= \max \set{\br{\tfrac {3} {5},\tfrac {4} {5}}   \cap \set{\tfrac {L_{2} - L_{1}-\ell}{2 \ell n}; \, n \in \N }}.
\eeq

\section{The multiscale analysis with  a Wegner estimate}\label{sectMSAWegner}

We will prove the following theorem.

\begin{theorem} \label{thmMSA}
 Let $H_{\bom}$ be a generalized Anderson Hamiltonian on  $\L^{2}(\R^{d})$.     Fix    $p\in \rb{ \frac 1 3, \frac 3 8}$ and   $ \vs,\vs^\pr  \in]0,1[ $.   
 Then there exist an energy $E_{0}>0$, a rate of decay $m>0$, and a scale $L_{0}$, all depending only  on $d, \Vper, \delta_{\pm},u_{\pm},U_{+},\mu,p, \vs,\vs^\pr$, such that
 all scales  $L\ge L_{0}$ are $(E,m,\varsigma, \varsigma^\pr,p)$-{extra good} for all energies $E \in [0,E_{0}]$.  In particular, all scales $L\ge L_{0}$ are $(E,m,\varsigma, p)$-good for all energies $E \in [0,E_{0}]$.
\end{theorem}

To prove the theorem we first obtain an {\it a priori}   estimate on the probability that a box $\Lambda_{L}$ is good    with an adequate supply of free sites for all energies in an interval   at the bottom of the spectrum (Proposition~\ref{propInitial}). Next, we   perform a multiscale analysis to  show that if such a probabilistic estimate  holds for a given energy at a sufficiently large scale, then it holds  all large scales  (Proposition~\ref{propA}).  
Theorem~\ref{thmMSA} is an immediate consequence of Propositions~\ref{propInitial} and \ref{propA}.

\begin{remark}
 If  $0$ is not an atom for the measure $\mu$ in \eq{mu}, Proposition~\ref{propInitialHD}  provides an alternative to Proposition~\ref{propInitial}, giving an {\it a priori}   estimate  in a fixed interval at the bottom of the spectrum for sufficiently high disorder.   If we also have $\mu([0,t])\le C t^\gamma$, with $\gamma >0$ appropriately  large, 
Proposition~\ref{propInitialHD} and Proposition~\ref{propA} (and their proofs) yield an alternative high disorder version of Theorem~\ref{thmMSA}.
\end{remark}

 \subsection{`A priori' finite volume estimates}  \label{sectinit} 

We set $\tilde{q}=\max \set{q,2}$, where   $q\in \N$ is the period of the background periodic operator  $V_{\mathrm{per}}$ in  \eqref{genAndH}.

\begin{proposition}\label{propInitial}  Let $H_{\bom}$ be a generalized Anderson Hamiltonian on  $\L^{2}(\R^{d})$, and  fix  $p>0$ and $0<\eps \le 1$.  
There exists  $\tilde{L}=\tilde{L}(d, V_{\mathrm{per}},u_-,\delta_-,\mu,p,\eps)$, such that for all 
 scales $L\ge\tilde{L}$ and all $x \in \R^d$ we have
  \beq \label{Hlowerbd}
\P \set{H_{{\bom},{\bt_S},\Lambda_L(x)}\ge  \pa{(p+1)d    \log (L +\delta_{+} + \tilde{q})}^{- \frac {2 + \eps} d} \, \text{for all} \; \,
\bt_S \in [0,1]^{S}} \ge 1 - L^{-pd} ,
\eeq
where  $S=S_{x,L,q} =\widetilde{\Lambda}_L(x) \setminus \tilde{q} \Z^d  $.   In particular,
 setting  
\beq  \label{EmL}
 E_{L}= \tfrac 1 2\pa{(p+1)d    \log (L +\delta_{+} + \tilde{q})}^{- \frac {2 + \eps} d} \quad \text{and} \quad m_{L}= \tfrac 1 2 \sqrt{E_{L}},
\eeq
it follows that  for all 
 scales $L\ge\tilde{L}$,    $x \in \R^d$,  $t_{S} \in [0,1]^{{S}}$, and  energies $E \in [0, E_{L}]$, we have,   with probability  $ \ge  1 - L^{- pd}$, that
\begin{align}\label{wegL0}
\| R_{\bom,t_{S},\Lambda_L(x)}(E) \|& \le \tfrac 1 {E_{L}},
\intertext{and, for all  $y,y^{\prime} \in \Lambda_L$ with   $\norm{y-y^{\pr}} \ge 20 \sqrt{d}$,}   \label{goodL01}
\| \Chi_y R_{\bom,t_{S}, \Lambda_L(x)}(E)  \Chi_{y^{\prime}} \|& \le 
\tfrac 2 {E_{L}} \e^{-\frac23 \sqrt{E_{L}}\norm{y-y^{\prime}}}.
 \end{align}
In particular, given $ \varsigma,\vs^\pr \in ]0,1[$, there is  $\tilde{\tilde{L}}=\tilde{\tilde{L}}(d,V_{\mathrm{per}},u_-,\delta_-,\mu,p,\vs,\vs^\pr,\eps)$, such that all scales  $L\ge\tilde{\tilde{L}}$ are $(E,m_L,\varsigma, \varsigma^\pr,p)$-{extra good} for all energies $E \in [0,E_{L}]$.
 \end{proposition}

\begin{proof} It suffices to prove \eq{Hlowerbd}, since given $H_{{\bom},{\bt_S},\Lambda_L(x)}\ge   2 E_L $,   for all   $E \in [0, E_{L}]$ we get immediately \eqref{wegL0},  and  \eqref{goodL01}   follows  by the Combes-Thomas estimate. (We use the precise estimate given in \cite[Eq. (19)]{GK2}, which  is also valid for finite volume operators with Dirichlet boundary condition.)   Moreover, in view of \eq{potU} and \eq{uzeta}, it suffices to prove  \eq{Hlowerbd} for the case when $U=0$,  and $u_\zeta= u_{-}\Chi_{\Lambda_{\delta_{-}}(\zt)}$ for all $\zeta \in \tilde{q}\Z^d$,  $u_\zt =0$ otherwise.

So let  
\beq \label{genAndHq}
H_{\bom}^{(q)}= H_0 + V_{\bom}^{(q)}, \quad \text{with} \quad V_{\bom}^{(q)} (x):= 
\sum_{\zeta \in \tilde{q} \Z^d} \omega_\zeta \,  u(x-\zeta ),  \eeq 
where $u= u_{-}\Chi_{\Lambda_{\delta_{-}}(0)}$.  Note that  $H_{\bom}^{(q)}$ is an Anderson Hamiltonian as in Definition~\ref{defAndH}, except that $\Z^d$ was replaced by $ \tilde{q} \Z^d$ and the periodic potential has period $\tilde{q}$, and hence its integrated density of states  $N^{(q)}(E)$ is well defined with the usual properties (cf. \cite{CL,PF}).   Given a box $\Lambda$, we define the corresponding finite volume  operator $H_{\bom,\Lambda}^{(q)}$ as in \eq{finvolH}.  For scales $L \in  \tilde{q}\N$ we set
\beq \label{defNq}
N_{\bom,\Lambda_{L}}^{(q)}(E):= \tr \Chi_{]- \infty,E]}\pa{\widetilde{H}_{\bom,\Lambda_{L}}^{(q)}},
\eeq
where 
\begin{align}\label{finvolHcorrected}
\widetilde{H}_{\bom,\Lambda_{L}}^{(q)} :={H_{0,\Lambda}}+ \widetilde{V}_{\bom,\La}^{(q)} \quad \text{on}   \quad \L^{2}(\Lambda),
\end{align}
where ${H_{0,\Lambda}}$ is as in \eq{H0La} and $ \widetilde{V}_{\bom,\La}^{(q)} $ is the restriction of ${V}_{\bom}^{(q)}$ to $\La$. In general $ \widetilde{V}_{\bom,\La}^{(q)}\not= {V}_{\bom,\La}^{(q)}$,  but we have \eq{compat}.

We recall (e.g., \cite[Eq.~(VI.15) on page 311]{CL}) that
\beq  \label{Nineq}
\E \pa{N_{\bom,\Lambda_{L}}^{(q)}(E)}\le N^{(q)}(E) \abs{\Lambda_{L}}\quad \text{for all} \quad  L \in  \tilde{q}\N.
\eeq

We now use the  Lifshitz tails estimate as in  \cite[Remark~7.1]{Klop99} (note that it applies  with $\mu$ as in \eq{mu}):
\beq
\lim_{E \downarrow 0}     \frac{\log \abs{\log    N^{(q)}(E)}}{\log E}   
\le  - \frac d 2  .
\eeq
It follows that there is an energy $E_1=E_1(d, V_{\mathrm{per}},u_-,\delta_-,\mu,\eps) >0$, such that
\beq \label{lifest}
 N^{(q)}(E) \le \e^{- E^{- \frac d {2 + \eps}} } \quad \text{for all energies} \quad E \le E_1.
\eeq
Combining \eq{defNq}, \eq{Nineq} and \eq{lifest}, and using Chebyschev's inequality, we get
that for all scales $ L \in  \tilde{q}\N$, $x \in \R^d$, and energies $ E \le E_1$,
\beq
\P\set{\sigma\pa{\widetilde{H}_{\bom,\Lambda_L(x)}^{(q)}} \cap [0, E] \not= \emptyset}\le 
\E \pa{N_{\bom,\Lambda_L(x)}^{(q)}(E)} \le \e^{- E^{- \frac d {2 + \eps}} } L^d ,
\eeq
and hence
\beq\label{Hlowerbdqq}
\P \set{\widetilde{H}_{\bom,\Lambda_L(x)}^{(q)}  \ge \min \set{\pa{(p+1)d  \log L}^{- \frac {2 + \eps} d}, E_1}   }  \ge 1 - L^{-pd}.
\eeq

To get  \eq{Hlowerbd} from \eq{Hlowerbdqq}, given
a scale $L\ge 1$ we set
\beq
L_{q}:= \min \set{L^{\pr}\in   \tilde{q}\N, L +\delta_{+} \le L^{\pr}}.
\eeq
It follows from \eq{compat} that  
\beq
\Chi_{\La_{L}(x)} \widetilde{V}_{\bom,\La_{L_{q}}(x)}^{(q)}= {V}_{\bom,\La_{L}(x)}^{(q)}.
\eeq
Since we are using Dirichlet boundary condition for the Laplacian, we conclude that
$\inf \sigma(H_{\bom,\La_{L}(x)}^{(q)})\ge  \inf \sigma(\widetilde{H}_{\bom,\La_{L_{q}}(x)}^{(q)})$.  Since  $L +\delta_{+}\le L_{q} < L +\delta_{+}+ \tilde{q}$, we conclude that
\beq\label{Hlowerbdqq2}
\P \set{H_{\bom,\Lambda_L(x)}^{(q)}  \ge \min \set{\pa{(p+1)d  \log (L +\delta_{+}+ \tilde{q})}^{- \frac {2 + \eps} d}, E_1}   }  \ge 1 - L^{-pd}
\eeq
for all $L\ge 1$.  The desired estimate \eq{Hlowerbd} follows for all scales $ L \ge \tilde{L}$, where 
$\tilde{L}=\tilde{L}(d, V_{\mathrm{per}},u_-,\delta_-,\mu,p,\eps)$.
\end{proof}

\begin{remark} \label{remnoVper} In the absence of a periodic background potential, i.e.,  $V_{\mathrm{per}}= 0$,  one can prove a slightly modified form of Proposition~\ref{propInitial} using ideas from \cite{BK} instead of Lifshitz tails. 
 As in the proof of Proposition~\ref{propInitial}, it suffices to consider the operator  $H_{\bom}= -\Delta + V_{\bom}$, where $V_{\bom}$ is as in \eq{genAndHq}. Setting $K > 10 \delta_{-}$, $\Lambda=\Lambda_L$,  It follows from   the lower bound in \eqref{uzeta} that there exists a constant  $c_{u_{-}, \delta_{-},d} >0$ such that
\begin{equation}\label{average}
\overline{V}_{{\bom}_\Lambda}(x) := \frac1{K^d} \int_{\Lambda_{K}(0)} V_{{\bom}_\Lambda} (x-a)\, \mathrm{d} a   \ge c_{u_{-}, \delta_{-},d}\, Y_{{\bom},\Lambda} \Chi_{\Lambda}(x) , 
\end{equation}
where 
\beq    Y_{{\bom},\Lambda}:= \min_{\xi \in \widetilde{\Lambda}}  \frac1{K^d}  \sum_{\zeta \in 
 \widetilde{\Lambda_{\frac K 3}(\xi)}} {\om}_\zeta\,  .
\eeq
It follows from standard estimates (e.g., \cite[Proposition~3.3.1]{Y}) that, with $\bar{\mu}$ and $\sigma$ the mean and standard deviation of the probability measure $\mu$, we have 
\beq
\P\set{\frac1{K^d}  \sum_{\zeta \in 
 \widetilde{\Lambda_{\frac K 3}(\xi)}} \om_\zeta \le \frac {\bar{\mu}} 2 } \le \e^{- A K^d},
\eeq
where
\beq
A = A_{\mu,d} = \frac {\bar{\mu}}{3^d 8 \sigma^2 (1 + \frac {\bar{\mu}} 2)} >0,
\eeq
and hence
\beq \label{Yprob} 
\P\set{ Y_{\bom,\Lambda} \le \frac {\bar{\mu}} 2 }\le L^d \e^{- A K^d}.
\eeq
It follows from \eqref{average} and \eqref{Yprob} that. with $ c_{u_{-}, \delta_{-},d}^{\pr}=\frac 1 2 { c_{u_{-}, \delta_{-},d}}$,
\beq \label{proboverV}
\P\set{\overline{V}_{\bom_\Lambda} >   c_{u_{-}, \delta_{-},d}^{\pr} \,  {\bar{\mu}}  \Chi_{\Lambda}}\ge 1-  L^d \e^{- A K^d},
\eeq
so, if $\overline{V}_{{\bom}_\Lambda} >  c_{u_{-}, \delta_{-},d}^{\pr} \,  {\bar{\mu}}  \Chi_{\Lambda}$, we have
\begin{equation}\label{lb2}
\overline{ H}_{ {\bom},\Lambda}:= -\Delta_{\Lambda } +\Chi_{\Lambda}\overline{V}_{{\bom}_\Lambda} \ge   c_{u_{-}, \delta_{-},d}^{\pr} \,  {\bar{\mu}}\quad \text{on} \quad \L^{2}(\Lambda). 
\end{equation}
Thus, if $\vphi \in C_{c}^{\infty}({\Lambda})$ with $\norm{\vphi}=1$, we have
\begin{align}\notag
\scal{\vphi,H_{ {\bom},\Lambda}\vphi }_{\Lambda}& = \scal{\vphi,\overline{H}_{ {\bom},\Lambda}\vphi }_{\Lambda}
+ \scal{\vphi,\left({V}_{{\bom}_\Lambda}- \overline{V}_{{\bom}_\Lambda}\right) \vphi }_{\Lambda}\\
& \ge  c_{u_{-}, \delta_{-},d}^{\pr} \,  {\bar{\mu}} +  \scal{\vphi,\left({V}_{{\bom}_\Lambda}- \overline{V}_{{\bom}_\Lambda}\right) \vphi }_{\R^{d}}\\
&  \ge  c_{u_{-}, \delta_{-},d}^{\pr} \,  {\bar{\mu}} + \scal{\vphi,{V}_{{\bom}_\Lambda}\vphi }_{\R^{d}}-  \frac1{K^d} \int_{\Lambda_{K}(0)} \scal{\vphi(\cdot +a),{V}_{{\bom}_\Lambda}\vphi(\cdot +a) } \mathrm{d} a    \notag\\
&  \ge c_{u_{-}, \delta_{-},d}^{\pr}\,  {\bar{\mu}}  -\frac1{K^d} \int_{\Lambda_{K}(0)}  \left|\scal{\vphi,{V}_{{\bom}_\Lambda}\vphi }- \scal{\vphi(\cdot +a),{V}_{{\bom}_\Lambda}\vphi(\cdot +a) }\right|  \mathrm{d} a    \notag\\
&  \ge  c_{u_{-}, \delta_{-},d}^{\pr} \,  {\bar{\mu}}  - c^{\prime}_{u} K \norm{\nabla_{\Lambda}\vphi}_{\Lambda}
\ge c_{u_{-}, \delta_{-},d}^{\pr} \,  {\bar{\mu}}  - c^{\prime}_{u}K  \scal{\vphi,H_{ {\bom},\Lambda}\vphi }_{\Lambda}^{\frac 1 2},\notag
\end{align}
where we used
\begin{equation}
\norm{\vphi(\cdot +a)-\vphi}_{\R^{d}}= \norm{(\e^{ a\cdot \nabla} -1)\vphi}_{\R^{d}}\le \abs{a} \norm{\nabla\vphi}_{\R^{d}}= \abs{a} \norm{\nabla_{\Lambda}\vphi}_{\Lambda}.
\end{equation}
It follows that there is $ \tilde{K}_{u,d}>0$, such that  for $K> \tilde{K}_{u,d}$ we have 
\begin{equation}
 \scal{\vphi,H_{ {\bom},\Lambda}\vphi }_{\Lambda} \ge  c_{u_{-}, \delta_{-},d}^{\pr\pr}\,\frac{ {\bar{\mu}}^2}{K^2}.
\end{equation}
Since this holds for all $\vphi \in C_{c}^{\infty}({\Lambda})$ with $\norm{\vphi}=1$, we have
\beq \label{lowerbd}
H_{ {\bom},\Lambda} \ge c_{u_{-}, \delta_{-},d}^{\pr\pr}\,\frac{ {\bar{\mu}}^2}{K^2} \quad  \text{on} \  \L^{2}(\Lambda).
\eeq
From \eqref{proboverV} and \eqref{lowerbd} we get
\beq
\P\set{H_{ \bom,\Lambda} \ge c_{u_{-}, \delta_{-},d}^{\pr\pr}\,\frac{ {\bar{\mu}}^2}{K^2}} 
 >   1-  L^d \e^{- A K^d}.
\eeq
Given $ p>0$, we take $K= \pa{\frac {(p+1)d} {A_{\mu,d}} \log L}^{\frac 1 d}$ and get
\begin{equation} 
\P\set{H_{ \bom,\Lambda_L} \ge  2 C_{u_{-}, \delta_{-},\mu,d,p}\,\pa{\log L}^{-\frac 2 d}}
 >   1-  L^{-pd} ,
\end{equation} 
for $L\ge \tilde{L}_{u_{-}, \delta_{-},\mu,d,p}$, where $C_{ u_{-}, \delta_{-},\mu,d,p}>0$ is an appropriate constant.

We then take   $n \in \N$ and let $S=S_\Lambda = n \Z^d \cap \Lambda$.   If $n \ll K$, 
we get, as in \eqref{average}, that for all   $t_S \in [0,1]^S$ we have
\begin{equation}\label{averageS}
\overline{V}_{{\bom}_\Lambda,t_S}(x):= \frac1{K^d} \int_{\Lambda_{K}(0)} V_{{\bom}_\Lambda, t_S} (x-a)\, \mathrm{d} a  \ge  c_{u_{-}, \delta_{-},d} \, Y_{{\bom},S,\Lambda}\, \Chi_{\Lambda}(x), 
\end{equation}
where 
\beq    Y_{{\bom},S,\Lambda}:= \min_{\xi \in \widetilde{\Lambda}}  \frac1{K^d}  \sum_{\zeta \in 
 \widetilde{\Lambda_{\frac K 3}(\xi)}\setminus S} {\om}_\zeta\,  .
\eeq
Proceeding as above, we conclude that
\beq \label{lowerSboundevent}
\P\set{H_{ \bom, t_S,\Lambda_L} \ge  2 C_{u_{-}, \delta_{-},\mu,d,p,q}\,\pa{\log L}^{-\frac 2d} \quad \text{for all  $t_S \in [0,1]^S$}}
 >   1-  L^{-pd} ,
\eeq
for $L\ge \tilde{L}_{u_{-}, \delta_{-},\mu,d,p,q}$, where $C_{u_{-}, \delta_{-},\mu,d,p,q}>0$ is an appropriate constant.
\end{remark}

If $0$ is not an atom for the measure $\mu$ in \eq{mu}, i.e., if  $\mu(\{0\})=0$, we can also obtain a high disorder  `a priori' finite volume estimate.

\begin{proposition}\label{propInitialHD}  Consider the generalized Anderson Hamiltonian $H_{\bom,\lambda}= H_{0 }+ \lambda V_{\bom}$ on  $\L^{2}(\R^{d})$, where $H_{0}$ and  $V_{\bom}$ are as in \eq{genAndH} and $\lambda >0$. Suppose  $0$ is not an atom for the measure $\mu$ in \eq{mu}. There exists an energy  $\widetilde{E}=\widetilde{E}(d, V_{\mathrm{per}},u_-,\delta_-) >0$, such that, fixing $E_{0}\in ]0,\widetilde{E}[$ and  $p>0$, given  $L \ge 100(\delta_{+}+1)$ there exists a constant  $\widetilde{\lambda}(L)=\widetilde{\lambda}(d, V_{\mathrm{per}},u_-,\delta_-,\mu,p,E_{0},L)$,  nondecreasing as a function of $E_0$, so for all $\lambda\ge \widetilde{\lambda}(L)$  we have
\beq \label{HlowerbdHD}
\P \set{H_{\bom,{\bt_S},\lambda,\Lambda_{L}(x)} \ge E_{0} \; \text{for all} \; \;
\bt_S \in [0,1]^{S}}\ge  1 - L^{-pd} \quad\text{for all}\quad  x\in \R^{d},
\eeq
where  $S=S_{x,L,q} =\widetilde{\Lambda}_L(x) \setminus \tilde{q} \Z^d  $. 
Thus,  for  all $E\in [0, E_{0}[$, $x \in \R^d$,  $\bt_{S} \in [0,1]^{{S}}$, and 
 $\lambda\ge\widetilde{\lambda}(L)$,   it follows, with probability  $ \ge  1 - L^{- pd}$, that
\begin{align}\label{wegL02HD}
\| R_{\bom,\bt_{S},\lambda,\Lambda_L(x)}(E) \|& \le  (E_{0} - E)^{-1},
\intertext{and,  {for $y,y^{\prime} \in \Lambda_L$,  $\norm{y-y^{\pr}} \ge 20 \sqrt{d}$},}   \label{goodL012HD}
\| \Chi_y R_{\bom,\bt_{S},\lambda, \Lambda_L(x)}(E)  \Chi_{y^{\prime}} \|& \le 
2  (E_{0} - E)^{-1} \e^{-\frac23 \sqrt{E_{0} - E}\norm{y-y^{\prime}}}.
 \end{align}
In particular, given $ \varsigma,\vs^\pr \in ]0,1[$, and  $0<E_{1}<E_{0}<\widetilde{E}$, there is  $\tilde{\tilde{L}}=\tilde{\tilde{L}}(d,\vs,\vs^\pr,E_{0}-E_{1})$, such that for all energies $E \in [0,E_{1}]$ a scale  $L\ge\tilde{\tilde{L}}$ is  $(E,\frac 1 2 \sqrt{E_{0}-E_{1}},\varsigma, \varsigma^\pr,p)$-{extra good}  if  $\lambda\ge \widetilde{\lambda}(L)$.
  \end{proposition}

\begin{proof}  Similarly to the proof of Theorem~\ref{propInitial}, in view of \eq{uzeta} it suffices to consider the  case when $u_{\zeta}= u_{-}\Chi_{\La_{\delta_{-}}(\zeta)}$ for all $\zeta \in \Z^{d}$.   Given $t \ge 0$, we set  $H(t)= H_{0} + V(t)$, where
$V(t)=t \sum_{\zeta \in \Z^d}   u_{\zeta} $ is a periodic potential with period one.  Then
$E(t)=\inf \sigma(H(t))$  is a strictly increasing continuous function of $t\ge 0$ with 
$E(0)=0$ (see \cite[Lemma~3.1 and its proof]{Klop95}); we set  $E(\infty)=\lim_{t \to \infty} E(t) >0$.  Given a box  $\La=\La_{L}(x)$ we let
$H_{ \bom,\lambda,\La}$ and $H_{ \La} (t)$ be the corresponding finite volume operators  to $H_{ \bom,\lambda}$ and $H (t)$.  ($H_{ \La} (t)=H_{ \bom,t,\La} $ with $\omega_{j}=1$ for all $j\in\Z^{d}$.)  Since we are using Dirichlet boundary condition, we have 
$H_{ \La} (t)\ge E(t)$, and hence
\beq
H_{\bom,\lambda,\La} \ge  E(t) \quad \text{if}\quad \lambda \omega_{j}\ge t \quad \text{for all}\quad j \in \widetilde{\La}.
\eeq
Given $E_{0} \in ]0,E(\infty)[$, let $t_{0}>0$ be defined by  $E_{0}=E(t_{0})$.   We conclude that
\beq
\P \set{H_{\bom,\lambda,\Lambda} \ge E_{0}}\ge 1 - L^{d} \mu([0, \tfrac {t_{0}}{\lambda}[)).
\eeq
Since $\mu(\set{0})=0$ by hypothesis, we have $\lim_{\lambda \to \infty}\mu([0, \tfrac {t_{0}}{\lambda}[))=0$, 
and hence there exists  $\widetilde{\lambda}(L)=\widetilde{\lambda}(d, V_{\mathrm{per}},u_-,\delta_-,\mu,p,E_{0},L)<\infty$,  such that 
\beq
\P \set{H_{\bom,\lambda,\Lambda} \ge E_{0}} \ge 1 - L^{-pd} \quad\text{for}\quad \lambda\ge \widetilde{\lambda}(L).
\eeq

To prove   a similar estimate with free sites, we set 
$H_{\bom,\lambda}^{(q)}= H_0 +\lambda  V_{\bom}^{(q)}$ with $V_{\bom}^{(q)}$ as in \eq{genAndHq}.  Proceeding as above, let $H^{(q)}(t)= H_{0} + V^{(q)}(t)$, where
$V^{(q)}(t)=t \sum_{\zeta \in \tilde{q}\Z^d}   u(x - \zeta) $, set $E^{(q)}(t)=\inf \sigma(H^{(q)}(t))$, and let   $E^{(q)}(\infty)=\lim_{t \to \infty} E^{(q)}(t) >0$. Given $E_{0} \in ]0,E^{(q)}(\infty)[$, let $t^{(q)}_{0}>0$ be defined by  $E_{0}=E^{(q)}(t^{(q)}_{0})$. Civen  a box $\Lambda= \Lambda_{L}$, we set  $S =\widetilde{\Lambda}_L \setminus \tilde{q} \Z^d  $.
 We conclude that
there is $\widetilde{\lambda}(L)< \infty$  such that for all $\lambda \ge \widetilde{\lambda}(L)$ we have 
\beq
\P \set{H_{\bom,\bt_{S},\lambda,\Lambda} \ge E_{0}  \; \text{for all} \; \;
\bt_S \in [0,1]^{S}}\ge 1 - (\tfrac L q)^{d} \mu([0, \tfrac {t^{(q)}_{0}}{\lambda}))\ge 1 - L^{-pd},
\eeq
which is \eq{HlowerbdHD}.
 As in Proposition~\ref{propInitial}, if  $\lambda \ge \widetilde{\lambda}(L)$, then for  all $E\in [0, E_{0}[$ and 
 $\bt_{S} \in [0,1]^{{S}}$,   it follows, with probability  $ \ge  1 - L^{- pd}$, that we have \eq{wegL02HD} and \eq{goodL012HD}.
 \end{proof}

\subsection{The multiscale analysis}
\label{sectMSA}

We  now state our single energy  multiscale analysis    for generalized Anderson Hamiltonians.

\begin{proposition}  \label{propA} Let $H_{\bom}$ be a generalized Anderson Hamiltonian on  $\L^{2}(\R^{d})$.    Fix  $E_{0} >0$,  $p\in\rb{\frac 1 3, \frac 3 8}$,  $ \vs,\vs^\pr , \tau,\rho_1,\rho_2  \in]0,1[ $,  
 with $\tau < \vs$ and   $ \rho_{2}=\rho_{1}^{n_{1}}$,   $n_{1}\in  \N$, such that 
\beq  \label{rhos}
 \tfrac 1 {1 +p} <  \rho_{1}< \tfrac 3 4  (1 -\vs)\quad
 \text{and} \quad   p < \tfrac 1 2 \rho_{1}(1 - \vs^\pr) -\rho_{2} .
\eeq
There exists a finite scale  $\widetilde{L}_{0}=\widetilde{L}_{0}\pa{d, \Vper, \delta_{\pm},u_{\pm},U_{+},\mu,E_{0},p, \rho_{1},\rho_{2},\vs,\vs^\pr,\tau}$ with the following property: 
given an energy  $ E \in [0,E_{0}]$,  a scale $L_{0}\ge \widetilde{L}_{0}$, and a number
\beq \label{m0}
 m_{0}\ge  L_{0}^{-\tau} ,
   \eeq 
if  all scales
  $L \in \br{L_{0}, L_{0}^{ {\frac 1 {\rho_{1}\rho_{2}}}}}$ are   $(E,m_0,\varsigma, \varsigma^\pr,p)$-{extra good}, it follows that every scale $L \ge L_{0}$
  is $\pa{E,\frac  {m_{0}} 2,\vs,\vs^\pr,p}$-{extra good}.  
\end{proposition}

\begin{remark}
To satisfy \eq{rhos} and \eq{m0}, we may
  pick      $p=
\frac 3 8  -$, and appropriate  $  \rho_{1}=\frac 3 4 -$,  $\vs=0+$, $\vs^\pr=0+$,  $\tau=0+$,
$ \rho_{2}=0+$. 
\end{remark}

\begin{remark}\label{remQUCP} The restriction $p\in\rb{\frac 1 3, \frac 3 8}$ comes from the use of the quantitative unique continuation principle, stated in Theorem~\ref{thmucp} and used in the form given in Corollary~\ref{corQUCPD}\ref{corQUCPDi}, which gives  a lower bound of the form $R^{-CR^{\frac 4 3}}$ in \eq{UCPbound1}.  It is instructive to see what happens if this lower bound  was of the form $R^{-CR^{\gamma}}$ for some $\gamma >0$.  In the  multiscale analysis,   Lemma~\ref{lemast} requires $ \tfrac 1 {1 +p} <  \rho_{1}$ in \eq{probast}.   The lower bound of  \eq{UCPbound1} is used to prove Lemma~\ref{lemwegner}; the important estimate \eq{wegner2} is useful only if   $\gamma   \rho_{1} <1$.   Lemma~\ref{lemwegnerplus2} uses $p< \frac 1 2\rho_{1}$ to get the probability estimate \eq{wegnerplusprobest}.  We conclude that the multiscale analysis requires
 \begin{gather}
\gamma < \tfrac {1 +\sqrt{3}} 2  \qtx{and} \gamma -1 <p< \tfrac 1 {2\gamma} , \\
\tfrac 1 {1 +p} <  \rho_{1}< \tfrac 1\gamma (1 -\vs)\quad
 \text{and} \quad   p < \tfrac 1 2 \rho_{1}(1 - \vs^\pr) -\rho_{2} .
\end{gather}
Since the quantitative unique continuation principle gives $\gamma=\frac 4 3<  \frac {1 +\sqrt{3}} 2$, we can perform the multiscale analysis with $p\in\rb{\frac 1 3, \frac 3 8}$ and \eq {rhos}.
\end{remark}

The proof of proposition~\ref{propA} will require several lemmas and definitions. 
We fix an energy $ E \in [0,E_{0}]$, and let  $p,\vs,\vs^\pr,\rho_1,\rho_2,n_1,  \tau$ be as in Proposition~\ref{propA}, satisfying \eqref{rhos}. 

\begin{definition} A collection $\cL$ of scales  is called  $(E ,\varsigma, \varsigma^\pr,p,\tau)$- \emph{extra good}  if for each  ${\ell}\in \cL$ there is a rate of decay $m_\ell$, with 
\beq \label{mL}
 m_{\ell}\ge  \ell^{-\tau} ,
  \eeq
such that for each 
 box $\Lambda_{{\ell}}$  
there is  a $(\Lambda_{\ell},E,m_\ell ,\varsigma, \varsigma^\pr)$-extra good event  $\cE_{\Lambda_{{\ell}}}$ satisfying
 \eqref{PEL0}.  
\end{definition}

   In the following definitions and lemmas, given a scale $L$, we set   $\ell_{1}=L^{\rho_{1}}$ and  $\ell_{2}=\ell_1^{\rho_2}= L^{\rho_1\rho_{2}}$.   We also set  $L_{n}=\ell_{1}^{\rho_{1}^{n}}$ for   $n=0,1,\ldots,n_{1}$;  note $L_{0}=\ell_{1} $  and $L_{n_{1}}=\ell_{2}$.

We start by defining an event that incorporates  \cite[property $(\ast)$]{BK}.  Note that  by writing ``$\cR=\{\Lambda_{\ell}(r)\}_{r \in R}$  is the standard  $\ell$-covering  of  $\Lambda_{L}$" (cf. Definitions~\ref{defsuitcov} and \ref{defstdcov}), we will  mean that $\cR=\cG_{\Lambda_{L}}^{(\ell)}$ as in \eq{standardcover} with $\alpha$ as in \eq{alphaL}; in particular,  $R= \G_{\Lambda_{L}}^{(\ell)}$ as in \eq{bbG}.

\begin{definition}\label{defnotsobad}
 Given a box $\Lambda_{\ell_{1}}$,    let $\cR_{n}=\{\Lambda_{L_{n}}(r)\}_{r \in R_{n}}$ be the standard  $L_{n}$-covering  of  $\Lambda_{\ell_{1}}$.   Fix a number $K_{2}\in \N$. Then:
\begin{enumerate}
\item A box $\Lambda_{\ell_{1}}$   is said to be $(\bom,E,K_2)$- \emph{notsobad}   if   there is $\Theta=\cup_{r \in R^{\pr}_{n_{1}}} \Lambda_{3\ell_{2}}(r)$, where $R^{\pr}_{n_{1}} \subset  R_{n_{1}}$ with $\# R^{\pr}_{n_{1}} \le K_{2}$,  
 such that for all $ x \in   {\Lambda_{\ell_{1}}}\setminus {\Theta}$ there is a
$(\bom,E,m_{L_n},\varsigma)$-good  box $\Lambda_{L_{n}}(r)$, with   $r\in R_{n}$ for some $n\in \{1,\ldots,n_{1}\}$, and $ \Lambda_{\frac { L_{n}} 5}(x )  \cap \Lambda_{\ell_{1}}\subset \Lambda_{L_{n}}(r)$.

\item   An event $\cN$ is  $(\Lambda_{\ell_{1}},E,K_2)$- \emph{notsobad} if
$\cN \in \cF_{\Lambda_{\ell_{1}}}$ and 
the box $\Lambda_{\ell_{1}}$   is  $(\bom,E,K_2)$-notsobad for all  $\bom \in \cN$.
\end{enumerate}
\end{definition}

\begin{lemma} \label{lemast}   Suppose $\set{L_n; \, n=1,2\ldots,n_1}$ is $(E ,\varsigma, \varsigma^\pr,p,\tau)$-extra good. There exists a constant $\widehat{K}_{2}= \widehat{K}_{2}(d,p,\rho_{1}, \rho_{2})$, and for  $K_{2} \in \N$ with  $K_{2}\ge \widehat{K}_{2}$  a constant $\widehat{\ell}_{1}=\widehat{\ell}_{1}(d,p,\rho_{1}, \rho_{2}, K_2)$, such that
 for any box $\Lambda_{\ell_{1}}$ with   $\ell_{1}\ge \widehat{\ell}_{1}$ there exists  a $(\Lambda_{\ell_{1}},E,K_2)$-notsobad
 event $\cN_{\Lambda_{\ell_{1}}}$
 with
\beq \label{notsobad}
\P\{\cN_{\Lambda_{\ell_{1}}}\} > 1- \ell_{1}^{-5d}.
\eeq
\end{lemma}

\begin{proof}  Given $\Lambda_{L_{n-1}}(r) \in \cR_{n-1}$, we set
\beq \begin{split}
\cR_{n}(r)& := \{\Lambda_{L_{n}}(s) \in \cR_{n}; \,\Lambda_{L_{n}}(s)\cap \Lambda_{L_{n-1}}(r)\not= \emptyset \} \quad \text{and} \quad\\
 R_{n}(r)&:=\{s\in R_{n}; \, \Lambda_{L_{n}}(s)\in \cR_{n}(r)\}.
\end{split}
\eeq
We  have  $\Lambda_{L_{n-1}}(r) \subset \bigcup_{s \in R_{n}(r)} \Lambda_{L_{n}}(s) $
and, similarly to \eqref{number}, $\#  R_{n}(r) \le  (\frac{3L_{n-1}} {L_{n}})^{d}$.  Fix  a number $K^{\pr}$, and define the event $\cN_{\Lambda_{\ell_{1}}}$ as consisting of $\bom \in \Omega$ such that,
for all $n=1,\dots,n_{1}$ and all $r\in R_{n-1}$, we have
$\bom \in \cE_{\Lambda_{L_{n}}(s)}$ for all $s \in R_{n}(r)$, with the possible exception of at most $K^{\pr}$ disjoint boxes $\Lambda_{L_{n}}(s)$ with  $s \in R_{n}(r)$. We clearly have $\cN_{\Lambda_{\ell_{1}}}\in \cF_{\Lambda_{\ell_{1}}}$.  Since $\set{L_n; \, n=1,2\ldots,n_1}$ is $(E ,\varsigma, \varsigma^\pr,p,\tau)$-extra good, the probability of its complementary event to  $\cN_{\Lambda_{\ell_{1}}}$ can be estimated from \eqref{PEL0}:
\begin{align} \notag
\P\left\{\Omega \setminus \cN_{\Lambda_{\ell_{1}}}\right \}
& \le \sum_{n=1}^{n_{1}}    (\tfrac{2\ell_{1}} {L_{n-1}})^{d}   (\tfrac{3L_{n-1}} {L_{n}})^{K^{\pr} d}   L_{n}^{-K^{\pr}p d}\\\label{probast}
&  \le 2^{d} 3^{K^{\pr}d} n_{1 }
 \ell_{1}^{-\rho_{1}^{n_1 -1}(K^{\pr}(\rho_{1}(pd+d)-d)+d)+ d}\\
 & =  2^{d} 3^{K^{\pr}d} n_{1 }
  \ell_{1}^{-d \pa{\rho_{1}^{n_1 -1}(K^{\pr}(\rho_{1}(p+1)-1)+1)- 1}}
 \le \ell_{1}^{-5d},\notag
\end{align}
where the last inequality  holds for all large $\ell_{1}$ after choosing $K^{\pr}$  sufficiently large using \eqref{rhos}.

Given  $\bom \in\cN_{\Lambda_{\ell_{1}}}$, then for each  $n=1,\dots,n_{1}$ and  $r\in R_{n-1}$ we can find $s_{1},s_{2},\ldots,s_{K^{\pr\pr}}\in R_{n}(r)$, with $K^{\pr\pr}\le K^{\pr}-1$, such that $\bom \in   \cE_{\Lambda_{L_{n}}(s)}$  if $s \in R_{n(r)}$
and $s \notin \bigcup_{j=1}^{K^{\pr\pr}} {\Lambda_{3L_{n}}(s_{j})}$.  
(Here we need boxes of side $3L_{n}$ because we only ruled out the existence of $K^{\pr}$ \emph{disjoint} boxes of side $L_{n}$.)  Since each box $ {\Lambda_{3L_{n}}(s_{j})}$ is contained in the union of  at most $C^{\pr\pr} $ boxes in 
$\cR_{n}$, we conclude that for each   $\bom \in\cN_{\Lambda_{\ell_{1}}}$ there are  $t_{1},t_{2},\ldots,t_{K^{\pr\pr\pr}} \in R_{n_{1}}$, with   $K^{\pr\pr\pr}\le K_{2}=(C^{\pr\pr} (K^{\pr}-1))^{n_{1}} $, such that, setting ${\Theta}= \bigcup_{t_{j}=1}^{K^{\pr\pr\pr}} \Lambda_{3\ell_{2} }(t_{j})$,  for all $ x \in   {\Lambda_{\ell_{1}}}\setminus {\Theta}$ we have $\bom \in   \cE_{\Lambda_{L_{n}}(s)}$  for some  $n=1, 2,\ldots,n_{1}$ and  $s \in R_{n}$, with   $ \Lambda_{\frac { L_{n}} 5}(x )   \cap \Lambda_{\ell_{1}}\subset \Lambda_{L_{n}}(s)$.
\end{proof}

\begin{definition}\label{defprepared}
Fix   $K_{1}, K_2 \in \N$.   Then:
 \begin{enumerate}
\item 
An event 
 $\cP$  is called     $(\Lambda,E,K_1,K_2)$-\emph{prepared}  if, with
 $\cR=\{\Lambda_{\ell_{1}}(r)\}_{r \in R}$ being the standard $\ell_{1}$-covering of  $\Lambda=\Lambda_{L}$, there exists a disjoint decomposition  $R=R^{\prime} \sqcup R^{\prime\prime} $ with $ \# R^{\prime\prime}\le K_{1} $, such that
 \beq \label{prepared}
 \cP=   \left\{ \bigcap_{r \in R^{\prime}} \cC_{\Lambda_{\ell_{1}}(r)} \right\} \bigcap \left\{ \bigcap_{r \in  R^{\prime\prime}} \cN_{\Lambda_{\ell_{1}}(r)}  \right\},
 \eeq
 where $\cC_{\Lambda_{\ell_{1}}(r)}$ is a  $(\Lambda_{\ell_{1}}(r),E,m_{\ell_{1}},\varsigma, \varsigma^\pr,S_{\cC_{\Lambda_{\ell_{1}}(r)}})$-adapted event for each $r \in R^{\prime}$,
 and $ \cN_{\Lambda_{\ell_{1}}(r)}$ is a $(\Lambda_{\ell_{1}}(r),E,K_2)$-notsobad event for each $r \in  R^{\prime\prime}$.  In this case we set
\begin{align} \label{ScP}
S_{\cP}&:=   \set{s \in \widetilde{\Lambda}; \quad  s \in \Lambda_{\ell_{1}}(r) \; \Rightarrow  \; r\in R^{\prime} \; \;\text{and} \;\;  s \in S_{\cC_{\Lambda_{\ell_{1}}(r)}}}\\ 
&= \bigcup _{r \in R^{\prime}}\pa{ S_{\cC_{\Lambda_{\ell_{1}}(r)}} \setminus \bigcup_{r^{\pr} \in  R^{\prime}\setminus \{r\}}\pa{ \Lambda_{\ell_{1}}(r^{\pr})\setminus S_{\cC_{\Lambda_{\ell_{1}}(r^{\pr})} } }} \setminus  \bigcup_{r \in  R^{\prime\prime}} \Lambda_{\ell_{1}}(r) .\notag
\end{align}

\item An event $\cQ$ is called $(\Lambda,E,K_1,K_2)$-\emph{ready} if it is the disjoint union of a finite number of $(\Lambda,E,K_1,K_2)$-prepared events, i.e.,   there exist  
  disjoint $(\Lambda,E,K_1,K_2)$-prepared events
$\{\cP_j\}_{j=1,2,\ldots,J}$  such that 
\beq \label{ready}
\cQ= \bigsqcup_{j=1}^{J} \cP_j.
\eeq
\end{enumerate}
\end{definition}

The set $S_{\cP}$ in \eq{ScP} is the maximal set with the required properties. 
It follows from
 \eqref{freeguarantee} that
\beq\label{Pfreesites}
\bigsqcup_{r \in R^{\prime}}\set{ S_{\cC_{\Lambda_{\ell_{1}}(r)}} \cap \Lambda_{\frac{\ell_{1}}{5}}(r)}\subset S_{\cP},
\eeq
and nothing would be lost if we had defined $S_{\cP}$ by making \eq{Pfreesites} an equality.

\begin{lemma} \label{lemstructure}   Suppose $\set{L_n; \, n=0,1,\ldots,n_1}$ is $(E ,\varsigma, \varsigma^\pr,p,\tau)$-extra good. For sufficiently large $K_{1},K_2 \in \N$,
depending only on $d,p,\rho_{1}, \rho_2$, if  $L$ is taken large enough,  depending only on $d,p,\rho_{1}, \rho_2,\vs^{\pr}, K_1, K_2$, the following holds:
\begin{enumerate}

\item\label{lemstructure1} If $\cP$ is  a $(\Lambda,E,K_1,K_2)$-prepared event, then $S_{\cP}$ is a $\vs^\pr$-abundant subset of  $ \widetilde{ \Lambda}$ and $\cP \in \cF_{\Lambda \setminus S_\cP}$.

\item \label{lemstructure2}  There exists  a  $(\Lambda,E,K_1,K_2)$-ready event $\cQ$ such that
\beq \label{readyprob}
\P\{ \cQ\} > 1-2L^{-2d}.
\eeq
\end{enumerate}
\end{lemma}

\begin{proof} 
Let $\cP$ be  a $(\Lambda,E,K_1,K_2)$-prepared event, as in \eq{prepared}, and let $S_{\cP}$ be as in \eq{ScP}. In particular,  $\cP \in \cF_{\Lambda \setminus S_\cP}$.
Since  $\# R^{\prime\prime} \le K_{1}$, it follows from \eqref{Pfreesites},  using \eq{abundant},  that  for all boxes  $\Lambda_{\frac{L}{5}}\subset \Lambda$ we have, with $L$ sufficiently large, 
\beq
\#\pa{S_\cP \cap \Lambda_{\frac{L}{5}}}\ge \ell_1^{(1 - \vs^\pr)d}\pa{ \pa{\tfrac 5 4 \tfrac{L} {5\ell_1}-2}^{d} - K_1}\ge  L^{(1 -\vs^\pr)d},
\eeq
and hence $S_{\cP}$ is a $\vs^\pr$-abundant subset of  $ \widetilde{ \Lambda}$.

We now use the hypothesis that $\set{L_n; \, n=0,1,\ldots,n_1}$ is $(E ,\varsigma, \varsigma^\pr,p,\tau)$-extra good.  
For each $r \in R$ we pick a   $(\Lambda_{\ell_{1}}(r),E,m_{\ell_1} ,\varsigma, \varsigma^\pr)$-{extra good} event  $\cE_{\Lambda_{\ell_{1}}(r)}$ as in \eq{Eloc} with
 \eqref{PEL0}.  Taking $K_2$ and $L$ sufficiently large so we can use Lemma~\ref{lemast},
 for each  $r \in R$ we also  pick a $(\Lambda_{\ell_{1}}(r),E,K_2)$-notsobad
 event $\cN_{\Lambda_{\ell_{1}}(r)}$ with \eq{notsobad}, and set
   $\cN_{\Lambda_{\ell_{1}}(r)}^\ast= \cN_{\Lambda_{\ell_{1}}(r)} \setminus \cE_{\Lambda_{\ell_{1}}(r)}$, clearly also a $(\Lambda_{\ell_{1}}(r),E,K_2)$-notsobad event.    Given $K_1 \in \N$, 
 define the  event $\cQ$ by the disjoint union
 \beq  \begin{split}\label{cQ}
  \cQ& := \bigsqcup_{\substack{ R^{\prime}\subset R\\
  \#(R \setminus R^{\prime})\le K_{1}}}  \cQ(R^{\prime}), \quad\text{where} \\
 \cQ(R^{\prime})&=    \left\{ \bigcap_{r \in R^{\prime}} \cE_{\Lambda_{\ell_{1}}(r)} \right\} \bigcap \left\{ \bigcap_{r \in R\setminus R^{\prime}} \cN_{\Lambda_{\ell_{1}}(r)}^\ast \right\} .
 \end{split} \eeq 
Using the probability estimates in  \eqref{PEL0}  and \eqref{notsobad},  and taking $K_{1}$ sufficiently large (independently of the scale), we get
\eq{readyprob}.  
This can be seen as follows.  First, using  \eqref{notsobad},  we have 
\beq
\P\left\{  \cE_{\Lambda_{\ell_{1}}(r)} \cup  \cN_{\Lambda_{\ell_{1}}(r)}^\ast \right \}\ge \P\left\{  \cN_{\Lambda_{\ell_{1}}(r)} \right \} > 1- L^{-5 \rho_{1 }d},
\eeq
and hence 
\beq  \begin{split} \label{Ld1}
&\P\left\{ \bigcap_{r \in R}\left\{   \cE_{\Lambda_{\ell_{1}}(r)} \cup  \cN_{\Lambda_{\ell_{1}}(r)}^\ast \right\}\right\}> 1 -  \left( \frac {2L}{\ell_{1} }\right)^{d} L^{-5 \rho_{1 }d}\\
& \qquad  \qquad  \qquad \qquad  \ge 1 - 2^{d} L^{-(6 \rho_{1 }-1)d}>
1 - L^{-2d},
\end{split}
\eeq
for large $L$, where we used \eq{number} and  \eqref{rhos}.  On the other hand, letting $K_{1}=C^{\pr}(K^{\pr}-1)$,   it follows from \eqref{PEL0} and \eqref{rhos} that
\beq \begin{split} \label{Ld2}
&\P\left\{\text{there are $K^{\pr}$  disjoint boxes $\Lambda_{\ell_{1}}(r) \in \cR$
with $\bom \notin \cE_{\Lambda_{\ell_{1}}(r)}$}\right\}\\
& \qquad \qquad   \le   (\tfrac{2L} {\ell_{1}})^{d K^{\pr}}\ell_{1}^{-pdK^{\pr}}  \le  2 ^{dK^{\pr}}
L^{-dK^{\pr}(\rho_{1}(p+1) -1)}\le L^{-2d},
\end{split} \eeq
if $K_{1} > \tfrac {2C^{\pr}} {\rho_{1}(p+1) -1}$ and $L$ is large enough. We now take   $C^{\pr}=3^{d}- 1$, ensuring that the complementary event has at most $K_{1}$ (not necessarily disjoint)   boxes $\Lambda_{\ell_{1}}(r) \in \cR$
with $\bom \notin\cE_{\Lambda_{\ell_{1}}(r)}$.
The estimate \eqref{readyprob} follows from  \eqref{Ld1} and  \eqref{Ld2}.

Moreover, it follows from  \eqref{Eloc}  and \eq{cQ}  that each   $  \cQ(R^{\prime})$ is a disjoint union of (non-empty) events of the form
\beq \label{cDR}
\cP_{R^{\prime}}= \left\{ \bigcap_{r \in R^{\prime}}\cC_{\Lambda_{\ell_{1}}(r)} \right\} \bigcap \left\{ \bigcap_{r \in R\setminus R^{\prime}} \cN_{\Lambda_{\ell_{1}}(r)}^\ast \right\},
\eeq
where 
$\cC_{\Lambda_{\ell_{1}}(r)}$ is a  $(\Lambda_{\ell_{1}}(r),E,m_{\ell_{1}},\varsigma, \varsigma^\pr, S_{\cC_{\Lambda_{\ell_{1}}(r)}})$-adapted event for each $r \in R^{\prime}$.
  Thus  $\cQ$ is a $(\Lambda,E,K_1,K_2)$-ready event.
 \end{proof}

 Given a box $\Lambda$ and a number $Y>0$, 
\beq
\cW_{\Lambda,Y}:=\set{\bom \in \Omega; \; \norm{R_{\bom, \Lambda}(E)}\ge Y}
\eeq
is a  measurable subset of $\Omega$, i.e., an event, and  moreover $\cW_{\Lambda,Y} \in \cF_\Lambda$.

\begin{lemma}\label{lemwegner}  Given a
 box $\Lambda= \Lambda_L$,
let  $\cP$ be   a $(\Lambda,E,K_1,K_2)$-prepared event, and consider a box $\Lambda_{L_{1}}\subset \Lambda $
with   $L_{1}=(2 k_1 \alpha +1)\ell_{1}$,    constructed as in \eqref{nesting} from  the standard $\ell_{1}$-covering   $\cR=\{\Lambda_{\ell_{1}}(r)\}_{r \in R}$ of  $\Lambda$,  where  $k_1 \in \N$,  $k_1 \ge  100 K_1$. 
  Then, there exist constants $C_{1}=C_{d,\Vper,\mu,\delta_{\pm}, u_{\pm},U_{+},\rho_{1},E_{0},K_{1,}K_{2}}$,   $C_{2}=C_{d,\Vper,\mu,K_{1},K_{2},E_{0}}$, and
  $\widehat{L}= \widehat{L}_{d,\mu, \delta_\pm, \Vper,U_+,E_0,\rho_1,\rho_2,\vs,\vs^\pr,K_1,K_2}$  (the  constants are  all independent of $k_1$),  such that    for  all scales $L\ge \widehat{L}$ we have  the conditional probability estimate
\begin{equation} \label{wegner2}
\P\set{\left .   \norm{ R_{\bom, \Lambda_{L_{1}}}(E)}\ge \e^{C_{1}L^{\frac 4 3 \rho_{1}}\log L } \right |\cP }\le C_{2}L^{-\frac d 2  \pa{  \rho_{1}(1 - \vs^\pr) -2\rho_{2}}}.
\end{equation}
\end{lemma}

\begin{proof}
 Let $\cP$ be a $(\Lambda,E,K_1,K_2)$-prepared event as in \eq{prepared}, and let $\set{\Lambda_b}_{b=1,2\ldots,B}$ be an enumeration of the  notsobad boxes    $\set{\Lambda_{\ell_1}(r)}_{r \in R^{\pr\pr} \cap \La_{L_1}}$; note $B \le K_1$.  For each $b=1,2\ldots,B$  we let $\Theta_b\subset \La_b$ be as in Definition~\ref{defnotsobad}, so $\abs{\Theta_b} \le 3^d K_2 \ell_2^d$.  We set $\Theta=\cup_{b=1}^B \Theta_b$, and note  $\abs{\Theta} \le 3^d K_1 K_2 \ell_2^d$.
 
 It follows from \eq{number} and $k_1 \ge  100 K_1$ that $ \# \pa{R\cap \Lambda_{L_1}} \ge \pa{200 K_1}^d$, so  we can pick distinct  $\set{r_b}_{b=1,2\ldots,B}\subset R^\pr\cap \Lambda_{L_1}$ such that for all $b=1,2\ldots,B$ we have 
 \beq\label{ellrB}
4 \ell_1 \le \dist \set{r_b, \Lambda_b} \le 12 K_1 \ell_1
\quad \text{and}\quad   \dist \set{r_b,\cup_{b^\pr=1}^B \Lambda_{b^\pr}}   \ge4 \ell_1 .
 \eeq
 Thus, the boxes $\set{\Lambda_{\frac {\ell_1} 5}(r_b)}_{b=1,2\ldots,B}$  are disjoint, and  it follows from \eq{abundant} that for each $b$ we have 
 \beq
  \# \pa{S_{\cC_{\Lambda_{\ell_{1}}(r_b)}} \cap \Lambda_{\frac {\ell_1} 5}(r_b)}\ge N_1:=\br{\ell_1^{(1 - \vs^\pr)d}}.
 \eeq

We now pick distinct  free sites $\set{\zt_{b,j}}_{j=1}^{N_1}\subset S_{\cC_{\Lambda_{\ell_{1}}(r_b)}}$,   $b=1,2\ldots,B$, and let $S=\cup_{b=1}^B \set{\zt_{b,j}}_{j=1}^{N_1}$, so $S \subset S_{\cP}$ by    \eq{Pfreesites} and we have 
\beq\label{Snumber}
\# {S}=   B N_1\le  K_1\ell_1^{(1 - \vs^\pr)d}.
\eeq
 Given  $\bt_{S}=\set{t_{\zt}}_{\zt \in S}\in [0,1]^S$, we consider $  H_{\bom,\bt_{S},\Lambda_{L_1}}$ as in \eq{finvolHS}. We fix $\bom \in \cP  \in \cF_{\Lambda \setminus S_\cP}\subset \cF_{\Lambda \setminus S}$ and set
\beq
\widetilde{H}_{\bt_{S}}= \widetilde{H}_{\bom,\bt_{S}}= H_{\bom,\bt_{S},\Lambda_{L_1}} \quad \text{on}   \quad \L^{2}(\Lambda_{L_1}).
\eeq

Since $\widetilde{H}_{\bt_{S}}\ge 0$ has compact resolvent, it has nonnegative discrete spectrum.  Using the min-max principle as in \cite[Theorem~A.1]{FKmid}, these eigenvalues (repeated according to the finite multiplicity) are given by
\begin{equation}
E _{n}(\bt_{S})=\inf_{\cL \subset \cD(\Delta_{\La_{L_1}})  ; \,\dim \cL =n}  \left[
\sup_{\psi \in \cL:\,\left\| \psi \right\| =1}{\left\langle \psi ,\widetilde{H}_{\bt_{S}}\psi
\right\rangle }\right] \quad \text{for} \quad n\in \N.  \label{mima}
\end{equation}
Thus, $0 \le E _{1}(\bt_{S})\le E _{2}(\bt_{S})\le \ldots E _{n}(\bt_{S})\le E _{n+1}(\bt_{S})\le\ldots$, and each   $E _{n}(\bt_{S})$ is a continuous function of $\bt_S$, monotone increasing in $t_\zt$ for each $\zt \in S$. In fact, we have
\beq
\abs{E _{n}(\bt_{S})-E _{n}(\bt_{S}^\pr)}\le \norm{V_{\bt_S}- V_{\bt_S^\pr}}\le  \abs{\bt_S-\bt_S^\pr}_1 u_+  ,
\eeq
a general bound does that does not take advantage of  our construction.  To do so, we note
that for  $\zt \in S$, each   $E _{n}(\bt_{S})$  is piecewise differentiable in $t_\zt$ for fixed $\bt_{S\setminus \set{\zt} }$ (cf. \cite[Section~VII.3.5]{Ka}), with
\beq\label{partialE}
\frac {\partial} {\partial t_\zt} E _{n}(\bt_{S})= \scal{\psi_{n}(\bt_{S}), u_\zt \psi_{n}(\bt_{S})},
\eeq
where by $\psi_{n}(\bt_{S})$ we denote a corresponding  normalized eigenfunction:  
\beq
\widetilde{H}_{\bt_{S}} \psi_{n}(\bt_{S})=  E _{n}(\bt_{S})\psi_{n}(\bt_{S}), \quad \psi_{n}(\bt_{S}) \in \cD(\Delta_{\La_{L_1}}) \quad \text{with}\quad  \norm{\psi_{n}(\bt_{S})}=1.
\eeq    Combining with \eq{uzeta}, we get
\beq\label{eigbounds}
u_- \norm{\Chi_{\La_{\delta_-}(\zt)}\psi_{n}(\bt_{S})}^2 \le \frac {\partial} {\partial t_\zt} E _{n}(\bt_{S})
\le u_+ \norm{\Chi_{\La_{\delta_+}(\zt)}\psi_{n}(\bt_{S})}^2.
\eeq

We set $m_1=m_{\ell_1}$, an consider the intervals 
\beq
I_1=[E-\e^{-2 m_1\ell_1},E+ \e^{-2 m_1\ell_1}] \quad\text{and}\quad I_2=[E-\e^{-4 m_1\ell_1},E+ \e^{-4 m_1\ell_1}]
\eeq
If $E _{n}(\bt_{S}) \in I_{2}$  for some $\bt_{S}\in [0,1]^S$, we can use Lemma~\ref{lemSLI}\ref{sublemEDII2}, namely  \eq{EDI2}, to conclude from the upper bound in \eq{eigbounds}, using   \eq{Snumber},  that for all $\bt_{S}^\pr \in [0,1]^S$ we have
\beq\label{eigcount1}
\abs{ E _{n}(\bt_{S}^\pr) - E}\le \e^{-4 m_1\ell_1} + u_+\delta_{+}^{d} K_1\ell_1^{(1 - \vs^\pr)d}\e^{- 3{m_1^{\pr}}  \ell_1 }\le \e^{-2 m_1\ell_1},
\eeq
and hence  $ E _{n}(\bt_{S}^\pr)\in I_{1}$.   In particular,  if  $\bt_S=0_S$ means  $t_\zt=0$ for all $\zt \in S$, we have 
\begin{align}\label{eigcount}
& \# \set{n\in \N; \; E _{n}(\bt_{S})\in I_{2}\; \text{for some} \; \bt_{S}\in [0,1]^S} \qquad \\
& \qquad \qquad  \qquad   \le N_2:=  \# \set{n\in \N; \; E _{n}(0_{S})\in I_{1} }
 =\tr \set{\Chi_{I_1}(\widetilde{H}_{0_{S}})} .\notag
\end{align}

General estimates yield (cf. \cite[Eq.~(A.7)]{GK5})
\beq \label{N2general}
N_2 \le C_{d,\Vper}  (E+ \e^{-2 m_1\ell_1})^{\frac d 2}L_1^d \le  C_{d,\Vper,E_{0}} L_{1}^d,
\eeq
  which is not good enough for our purposes.  To improve the estimate, we  apply Lemma~\ref{lemSLI}\ref{sublemEDII2}. 
    If  $ E _{n}(\bt_{S}) \in I_{1}$,  it follows from  \eq{EDI22} and our construction that
\beq
x \notin \Theta \quad  \Longrightarrow\quad  \norm{\Chi_{x}\psi_{n}(\bt_{S})} \le  \e^{-\frac m {11} \ell_2}\le  \e^{-\frac 1 {11} \ell_2^{1-\tau} },
\eeq
and hence, for large $L$,
\beq \label{LThetapsi}
 \norm{\Chi_{\Lambda_{L_1}\setminus \Theta}\psi_{n}(\bt_{S})} \le L_1^d\e^{-\frac 1 {11} \ell_2^{1-\tau} } \le \e^{-\frac 1 {14} \ell_2^{1-\tau} }.
\eeq
It follows that
\beq
\tr \set{\Chi_{\Lambda_{L_1}\setminus \Theta}\Chi_{I_1}(\widetilde{H}_{0_{S}})}\le  \e^{-\frac 1 7 \ell_2^{1-\tau} }N_2.
\eeq
Recalling that $\Theta$ is a union of at most $K_1K_2$ boxes  of side ${3\ell_2}$, and using  the trace estimate given  in \cite[Lemma~A.4]{GK5}, we obtain
\begin{align}\label{N2bound}
N_2&\le \pa{1- \e^{-\frac 1 7  \ell_2^{1-\tau} }}^{-1}  \tr \set{\Chi_{ \Theta}\Chi_{I_1}(\widetilde{H}_{0_{S}})}\\
& \le 2 \sum_{x \in \Z^d}  \tr \set{\Chi_x \Chi_{ \Theta}\Chi_{I_1}(\widetilde{H}_{0_{S}})}\le C_{d,\Vper,E_{0}}  \pa{K_1K_2}^d \ell_2^d,\notag
\end{align}
a huge improvement over \eq{N2general}.

 In addition, if $ E _{n}(\bt_{S}) \in I_{1}$  we conclude  from \eq{LThetapsi} that 
  there exists  $b^{\pr } \in \set{1,2,\ldots,B}$ such that
\beq\label{someb}
\norm{\Chi_{ \Theta_{b^{\pr}}}\psi_{n}(\bt_{S})} \ge B^{-\frac 1 2} \pa{1- \e^{-\frac 1 7 \ell_2^{1-\tau} }}^{\frac 12}
\ge \pa
{2 K_1}^{-\frac 12} .
\eeq
In view of \eq{ellrB}, it now follows from the quantitative unique continuation principle (\cite[Lemma~3.10]{BK}, see Theorem~\ref{thmucp}), which we use in the form given in Corollary~\ref{corQUCPD}\ref{corQUCPDi}, that 
\begin{align}\label{someb2}
\norm{\Chi_{\La_{\delta_-}( \zt_{b^{\pr},j})}\psi_{n}(\bt_{S})} \ge \e^{-  C_{3}  \ell_{1}^{\frac 4 3} (\log \ell_{1})}
\quad \text{for all  $j=1,2,\ldots,N_{1}$},
\end{align}
with a constant 
\beq
C_{3}=  C_{d,K_{1},\delta_{-}}\pa{1 + \norm{\Vper}+\delta_{+}^{d}u_{+}+ U_{+}+ E_{0}}^{\frac 2 3}, \quad \text{where} \quad  C_{d,K_{1},\delta_{-}}>0.  
\eeq

To exploit \eq{someb2}, we set
 $\bzt_{j}=\set{\zt_{b,j}}_{b=1,2,\ldots,B} $ for $j=1,2,\ldots,N_{1}$, and let
 $u_{\bzt_{j}}: =\sum_{b=1}^{B} u_{\zt_{b,j}} $,   $ \Chi_{\La_{\delta_\pm}( \bzt_{j})}
:= \sum_{b=1}^{B} \Chi_{\La_{\delta_\pm}( \zt_{b^{\pr},j})}
 $.  It follows from  \eq{someb2} that  
\begin{align}\label{someb3}
\norm{ \Chi_{\La_{\delta_-}( \bzt_{j})}
\psi_{n}(\bt_{S})} \ge \e^{-  C_{3}  \ell_{1}^{\frac 4 3} (\log \ell_{1})}
\quad \text{for all  $j=1,2,\ldots,N_{1}$}.
\end{align}
Given $J \subset \set{1,2,\ldots,N_{1}}$ we let
$S_{J}= \cup_{j \in J}\bzt_{j}$.

We now set   $\bt_{j}
=\set{t_{\zt_{b,j}}}_{b=1,2,\ldots,B} $ for $j=1,2,\ldots,N_{1}$, and write $\bt_{S}=\set{\bt_{j}}_{j=1}^{N_{1}}$.   
Given $j^\pr =1,2,\ldots,N_{1}$, we also  define   $\bfe_{j}\up{j^\pr}=\set{e\up{j^\pr}_{\zt_{b,j}}}_{b=1,2,\ldots,B} $  by    
$e\up{j^\pr}_{\zt_{b,j}} =\delta_{j^{\pr},j}$ for ${b=1,2,\ldots,B}$,  $j=1,2,\ldots,N_{1}$, and let   $\bfe_{S}\up{j^\pr}=\set{\bfe_{j}\up{j^\pr}}_{j=1}^{N_{1}}$.  
  It follows, as in \eq{partialE} and \eq{eigbounds}, that for
 $ E _{n}(\bt_{S}) \in I_{1}$   we have
\beq\label{eigbounds2}
u_- \norm{ \Chi_{\La_{\delta_-}( \bzt_{j})}
\psi_{n}(\bt_{S})}^2 \le  {\partial_{j}} E _{n}(\bt_{S})
\le u_+ \norm{ \Chi_{\La_{\delta_+}( \bzt_{j})}
\psi_{n}(\bt_{S})}^2,
\eeq
where
\beq
 {\partial_{j}} E _{n}(\bt_{S})=  \lim_{s \to 0} \tfrac 1 s \pa{E _{n}(\bt_{S}+ s\, \bfe_{S}\up{j})- E _{n}(\bt_{S})},
\eeq
so  \eq{someb3} yields
\beq \label{eigbounds4}
 {\partial_{j}} E _{n}(\bt_{S})\ge u_{-}  \e^{- 2 C_{3}  \ell_{1}^{\frac 4 3} (\log \ell_{1})}.
\eeq

We   pick  $0\le \theta_{-}<\theta_{+}\le 1$  such that, letting  
 \beq\label{mucond}
p_{-} = \mu\pa{\set{\omega \le \theta_{-}}}  \quad \text{and} \quad p_{+}=\mu\pa{\set{\omega \ge \theta_{+}}},
\eeq
 we have $p_{\pm}\in ]0,1[$.  ($\mu$ is the  probability distribution in \eq{mu}.) Such $\theta_{\pm}$ always exist  since $\mu$ is non-degenerate, and we have  $p_{-} + p_{+} \le 1$.
We  set ${\theta_{\mu}}= \theta_{+} - \theta_{-}\in ]0,1]$.

We now define random variables
\beq
\omega^{+}_{j}:= \max_{b=1,2,\ldots,B} \omega_{\zt_{b,j}}\quad \text{and} \quad \omega^{-}_{j}:= \min_{b=1,2,\ldots,B} \omega_{\zt_{b,j}}, \quad j=1,2,\ldots,N_{1},
\eeq
and consider the  events
\beq
\cY_{j}\up{1}=\set{\omega_{j}^{+}\le \theta_{-}}, \quad  \cY_{j}\up{2}=\set{\omega_{j}^{-}  \ge \theta_{+}}, \quad \text{and} \quad \cY_{j}\up{0}=\cY_{j}\up{1}\sqcup \cY_{j}\up{1}.
 \eeq
It follows from \eq{mucond} that
\beq
p\up{1}:=\P\pa{\cY_{j}\up{1}} = p_{-}^{B }, \; p\up{2}:=\P\pa{\cY_{j}\up{2}} = p_{+}^{B }, \; \text{and} \;
p\up{0}:=\P\pa{\cY_{j}\up{0}}=p\up{1} + p\up{2}.
\eeq
We now introduce Bernoulli  random variables
$\eta_{j}\up{a}=\Chi_{\cY_{j}\up{a}}$,  $a=0,1,2$.  Then $\bta\up{a}=\set{\eta_{j}\up{a}}_{j=1}^{N_{1}}$ are independent, identically distributed Bernoulli random variables with $\P\set{\eta_{j}\up{a}=1}= p\up{a}$. Note that $\eta_{j}\up{0}=\eta_{j}\up{1}+\eta_{j}\up{2}$, and
\beq\label{condBer}
\eta_{j}\up{0}=\eta_{j}\up{1}+\eta_{j}\up{2}  \quad \text{and}\quad \P\set{\eta_{j}\up{a}=1\, |\,  \eta_{j}\up{0}=1}=\frac { p\up{a}}{p\up{0}}, \;\; a=1,2.
\eeq
 We consider the random index  set  given by 
$J_{\bta\up{0}}=\set{j \in \set{1,2,\ldots,N_{1}}; \eta_{j}\up{0}=1}$. Then $\# J_{\bta\up{0}}=\sum_{j=1}^{N_{1}} \eta_{j}\up{0}$, and standard large deviation estimates \cite[Theorem~1]{Hoeff} give
\beq\label{largedev}
\P\set{\# J_{\bta\up{0}}\le  \tfrac 1 2 N_{1}p\up{0}} 
\le \e^{-\frac 1 2  N_{1}\pa{p\up{0}}^{2} }=  \e^{-\frac 1 2 \pa{p_{-}^{B} +p_{+}^{B}}^{2}N_{1}}  .
\eeq

 Suppose  $ E _{n}(\bom_{S}), E _{n}(\bom^{\pr}_{S})  \in I_{1}$, such that for some
 $j$ we have $\omega_{\zt}=\omega^{\pr}_{\zt}$ for $\zt \in S \setminus \bzt_{j}$, and we have
  $\eta_{j}\up{1}(\bom_{S})= \eta_{j}\up{2}(\bom^{\pr}_{S})=1$.   It then follows from \eq{eigbounds4} that
  \beq\label{incomp}
  E _{n}(\bom^{\pr}_{S})- E _{n}(\bom_{S}) \ge u_{-} \theta_{\mu} \e^{- 2 C_{3}  \ell_{1}^{\frac 4 3} (\log \ell_{1})}.
  \eeq  
We  set  
\beq \label{incomp2}
I=\br{E - \tfrac 1 2 u_{-} \theta_{\mu} \e^{- 2 C_{3}  \ell_{1}^{\frac 4 3} (\log \ell_{1})}, E + \tfrac 1 2u_{-} \theta_{\mu} \e^{- 2 C_{3}  \ell_{1}^{\frac 4 3} (\log \ell_{1})}},
\eeq
and we will estimate (we write $ \bta= \bta\up{0}$)
\beq \label{Peta}
\P_{S} \set{E _{n}(\bom_{S}) \in I\, |\, \bta} =\widehat{\widehat{\E}}_{J_{\bta}}{\set{\widehat{\P}_{J_{\bta}}\set{E _{n}(\bom_{S}) \in I}}},
\eeq
where, given  $J \subset \set{1,2,\ldots,N_{1}}$, we write
\beq\begin{split}
\widehat{\P}_{{J}}\set{\, \cdot \, }&: = \P_{S_{J}}\set{\, \cdot \, | \, \eta_{j}=1, \; j\in J },\\
\widehat{\widehat{\E}}_{{J}}\set{\, \cdot \, }&: = \E_{S\setminus S_{J}}\set{\, \cdot \, | \, \eta_{j}=0, \; j\notin  J }.
\end{split}\eeq
It follows from \eq{condBer} that, with respect to $\widehat{\P}_{{J}}$, $\bta_{J}^{(2)} =\set{\eta_{j}^{(2)}}_{j \in J} $ is a family of independent identically distributed  Bernoulli random variables with
\beq\label{Berprob}
\widehat{\P}_{{J}}\set{\eta_{j}^{(2)}=1}= \frac { p\up{2}}{p\up{0}}= \frac { p_{+}^{B} }{ p_{-}^{B}  +  p_{+}^{B} }, \;\; j \in J.
\eeq
The configuration space of  $\bta_{J}^{(2)}$, $\set{0,1}^{J}$, is partially ordered by the relation defined by  $\beps  \prec \beps' \; \Longleftrightarrow  \; \eps_j \le \eps_j' \ ;\mbox{for all}\; j \in J$.
Let us write $\bom_{S}= \pa{ \bom_{S \setminus S_{J}},\bom_{S_{J}} }$.  For a fixed  $ \bom_{S \setminus S_{J}}$ we set 
\beq
\cA_{\bom_{S \setminus S_{J}}}=  \set{\bta_{J}^{(2)} \pa{ \bom_{S \setminus S_{J}},\bom_{S_{J}} }; \;  E _{n}\pa{ \bom_{S \setminus S_{J}},\bom_{S_{J}} } \in I }\subset \set{0,1}^{J}.
\eeq
It follows from \eq{incomp}  and \eq{incomp2} that $\cA_{\bom_{S \setminus S_{J}}}$ is an anti-chain in  $\set{0,1}^{J}$, i.e., if  $\beps,\beps^{\pr}\in \cA_{\bom_{S \setminus S_{J}}}$ and $\beps  \prec \beps'$, then  $\beps= \beps^{\pr}$.  Using the probabilistic Sperner  Lemma given in \cite[Lemma~3.1]{AGKW} with  \eq{Berprob},  we get
\begin{align}\label{probsperner}
 \widehat{\P}_{{J}}\set{E _{n}\pa{ \bom_{S \setminus S_{J}},\bom_{S_{J}} } \in I }
= \widehat{\P}_{{J}}\set{\bta_{J}^{(2)} \in \cA_{\bom_{S \setminus S_{J}}}}
 \le \frac {2\sqrt{2}\pa{ p_{-}^{B}  +  p_{+}^{B} }}{\pa{p_{-}    p_{+}}^{\frac B 2}\sqrt{ \#J}}.
\end{align}

It follows from \eq{Peta} and \eq{probsperner} that
\beq \label{Peta2}
\P_{S} \set{E _{n}(\bom_{S}) \in I\, |\, \bta} \le   \frac {2\sqrt{2}}{\pa{p_{-}    p_{+}}^{\frac {K_{1}} 2}\sqrt{ \#J_{\bta}}}
\eeq
Combining \eq{largedev} and \eq{Peta2} we obtain
\begin{align}
 \label{Peta23}
\P_{S} \set{E _{n}(\bom_{S}) \in I}& \le  \e^{-\frac 1 2 \pa{p_{-}^{B} +p_{+
}^{B}}^{2}N_{1}} + 
  4  \pa{p_{-}    p_{+}}^{- \frac B 2}   \pa{p_{-}^{B} +  p_{+}^{B}}^{ \frac 1 2} N_{1}^{-\frac 1 2}\\
& \le  \e^{-\frac 1 2 \pa{p_{-}^{K_{1}} +p_{+
}^{K_{1}}}^{2}N_{1}} + 
  4  \pa{p_{-}    p_{+}}^{- \frac {K_{1}} 2}   N_{1}^{-\frac 1 2}\le C_{\mu,K_{1}}\ell_1^{- \frac 1 2(1 - \vs^\pr)d}. \notag
\end{align}

We now conclude from \eq{Peta23}, \eq{eigcount}, and \eq{N2bound} that
\begin{align}\label{almostweg}
\P_{S} \set{  \norm{ \pa{\widetilde{H}_{\bom_{S}}-E}^{-1}}  \ge   2 \pa{u_{-} \theta_{\mu} }^{-1}\e^{2 C_{3}  \ell_{1}^{\frac 4 3} (\log \ell_{1})}   } 
& =  \P_{S}\set{ \sigma\pa{\widetilde{H}_{\bom_{S}}}\cap I \not
 = \emptyset} \\ &
 \le C_{4}\ell_1^{- \frac 1 2(1 - \vs^\pr)d}\ell_{2}^{d},\notag
\end{align}
with a constant $C_{4}=C_{d,\mu,\Vper,K_{1},K_{2},E_{0}}$.

Recalling  $ \cP  \in \cF_{\Lambda \setminus S_\cP}\subset \cF_{\Lambda \setminus S}$, 
it follows from \eq{almostweg}  that
\begin{align}
&\P\set{ \set{ \norm{ R_{\bom, \Lambda_{L_{1}}}(E)}\ge 2 \pa{u_{-} \theta_{\mu} }^{-1}\e^{2 C_{3}  \ell_{1}^{\frac 4 3} (\log \ell_{1})}}  \cap\cP }\\\notag 
& \quad = \P\set{ \Chi_{\cP}(\bom) \P_{S }\set{ \norm{ R_{\bom, \Lambda_{L_{1}}}(E)}\le 2 \pa{u_{-} \theta_{\mu} }^{-1}\e^{2 C_{3}  \ell_{1}^{\frac 4 3} (\log \ell_{1})}} }  \le C_{4}\ell_1^{- \frac 1 2(1 - \vs^\pr)d}\ell_{2}^{d}\,  \P\set{\cP},\notag
\end{align}
which yields   \eq{wegner2}.
\end{proof}

\begin{lemma}\label{lemwegnerplus}  Given a
 box $\Lambda= \Lambda_L$,
let  $\cP$ be   a $(\Lambda,E,K_1,K_2)$-prepared event.  Then, if   $L$ is large enough,  depending only on $d,\mu, \delta_\pm, \Vper,U_+,p, E_0,\rho_1,\rho_2,\vs,\vs^\pr,\tau,K_1,K_2$, there exists an event
$\cW_{\cP} \subset \cP$, with
\beq\label{probW}
\P\set{\cW_{\cP}} \le C_{2} K_{1} L^{-\frac d 2  \pa{  \rho_{1}(1 - \vs^\pr) -2\rho_{2}}}\P\set{\cP},
\eeq
where the constant $C_{2}$ is as in \eq{wegner2}, such that  the event $\cP\setminus \cW_{\cP}$ is $(\Lambda,E,m_{L},\varsigma, \varsigma^\pr)$-adapted with
\beq\begin{split}
\label{mL34}
m_{L} &= m_{\ell_{1}}\pa{1 - C_{d, V_{\mathrm{per}},E_{0},K_{1}}L^{-\beta
}}\ge L^{-\tau},\quad \text{where}\\
\beta & =  \min \set{\pa{\tfrac 4 3\vs (1-\vs)^{-1}-\tau}\rho_{1} ,\tfrac 12 \pa{1- \tfrac 4 3 (1-\vs)^{-1} \rho_{1}} }>0.
\end{split}\eeq
\end{lemma}

\begin{proof}
Let $\cP$ be a $(\Lambda,E,K_1,K_2)$-prepared event as in \eq{prepared}.  
We take 
\beq\label{defkappa}
\kappa= \tfrac 1 2\pa{1 + \tfrac 4 3 \rho_{1}(1-\vs)^{-1}}, \quad \text{so} \quad \tfrac 4 3 \rho_{1}(1-\vs)^{-1}<\kappa <1, 
\eeq 
where we used   \eq{rhos}, and we have
\beq
 \kappa \vs -\tau \rho_{1} > \pa{\tfrac 4 3\vs (1-\vs)^{-1}-\tau}\rho_{1} .
\eeq
By geometrical considerations, we can find  disjoint boxes $ \set{\Lambda_{j}}_{j=1}^{J}$, $J\le \#R^{\pr\pr} \le K_1$, where  each $\Lambda_{{j}}=\Lambda_{L_{j}}\subset \Lambda$  is
    constructed as in \eqref{nesting} from  the standard $\ell_{1}$-covering   $\cR=\{\Lambda_{\ell_{1}}(r)\}_{r \in R}$ of  $\Lambda$ with   $ L^{\kappa}\le L_{j}\le K_{1} L^{\kappa}$, and for every $r \in  R^{\pr\pr}$
    there exists a (unique) $j_{r }\in \set{1,2\ldots,J}$ with
$ \Lambda_{\ell_{1}(r)} \subset \Lambda_{{j_{r}}}\up{\Lambda, \frac {L^{\kappa}} {10}}$.  
   Since it follows from \eq{defkappa} that ($L$ large enough),
\beq
 \e^{C_{1}L^{\frac 4 3 \rho_{1}}\log L } \le  \e^{ L^{{\kappa}(1-\vs) }},
\eeq
we conclude from  Lemma~ \ref{lemwegner}  that for all $ j=1,2,\ldots,J$, letting
\beq
\cW_{j}= \set{\norm{ R_{\bom, \Lambda_{{j}}}(E)}\ge\e^{ L^{{\kappa}(1-\vs) }}}\cap
   \cP ,\eeq 
we have
\beq
\P\set{ \cW_{j}
  }\le C_{2}L^{-\frac d 2  \pa{  \rho_{1}(1 - \vs^\pr) -2\rho_{2}}} \P\set{\cP}.
  \eeq
We  set $\cW_{\cP}= \cup_{j=1}^{J} \cW_{j}\subset \cP$, so \eq{probW} holds.

Since $\cP$ is a $(\Lambda,E,K_1,K_2)$-prepared event, the hypotheses of Lemma ~\ref{lemkeyMSA}  are satisfied for $\bom \in \cP\setminus \cW_{\cP}$, so we conclude that
the box $\La$ is  $(\bom,E,m_{L},\varsigma)$-good  for all  $\bom \in \cP\setminus \cW_{\cP}$ with
$m_{L}$ as in \eq{mL34}.

Moreover, for all $j$ we have  $ \set{\norm{ R_{\bom, \Lambda_{{j}}}(E)}\ge\e^{ L^{{\kappa}(1-\vs) }}} \in \cF_{\Lambda_{{j}}}$, so it follows from \eq{prepared} that
$\cW_{j} \in \cF_{{\La}_{j}^{\pr}}$, where ${\La}_{j}^{\pr}= \set{x \in \Lambda; \; \dist \pa{x, \La_{j}}  < \ell_{1}} $. Let $\La^{\pr}= \cup_{j=1}^{J}\La_{j}^{\pr}$.    It follows that $S_{\cP\setminus \cW_{\cP}}:= S_{\cP}\setminus  \La^{\pr}$ consists of free sites for  $\cP\setminus \cW_{\cP}$, i.e., the box $\La$ is  $(\bom,E,m_{L},\varsigma,S_{\cP\setminus \cW_{\cP}})$-good  for all  $\bom \in \cP\setminus \cW_{\cP}$.  

To conclude that 
$\cP\setminus \cW_{\cP}$ is $(\Lambda,E,m_{L},\varsigma, \varsigma^\pr)$-adapted  we need only to show that $S_{\cP\setminus \cW_{\cP}}$ is   $\vs^\pr$-abundant.  This can be shown is as in Lemma~\ref{lemSLI}\ref{lemstructure1}.  Since 
\beq\label{Pfreesites22}
\bigsqcup_{r \in R^{\prime}\setminus \La^{\pr}} S_{\cC_{\Lambda_{\ell_{1}}(r)}} \cap \Lambda_{\frac{\ell_{1}}{5}}(r)\subset S_{\cP\setminus \cW_{\cP}},
\eeq
it follows, using \eq{abundant},  that  for all boxes  $\Lambda_{\frac{L}{5}}\subset \Lambda$ we have ($L$ sufficiently large)
\beq
\#\pa {S_{\cP\setminus \cW_{\cP}} \cap \Lambda_{\frac{L}{5}}}\ge \ell_1^{(1 - \vs^\pr)d}\pa{ \pa{\tfrac 5 4 \tfrac{L} {5\ell_1}-2}^{d} - K_1\pa{\tfrac 5 3 \tfrac {K_{1}L^{\kappa} }{\ell_{1}}}^{d}}\ge  L^{(1 -\vs^\pr)d},
\eeq
and hence $S_{\cP\setminus \cW_{\cP}}$ is a $\vs^\pr$-abundant subset of  $ { \Lambda}$.
\end{proof}

\begin{lemma}\label{lemwegnerplus2}  Suppose $\set{L_n; \, n=0,1,\ldots,n_1}$ is $(E ,\varsigma, \varsigma^\pr,p,\tau)$-extra good. Then,  if   $L$ is sufficiently large,  depending only on $d,\mu, \delta_\pm, \Vper,U_+,p, E_0,\rho_1,\rho_2,\vs,\vs^\pr,\tau$, the scale $L$ is  $\pa{E, m_{L},\vs,\vs^\pr,p}$-{extra good}, and 
\beq \label{mL345}
m_{L}= m_{\ell_{1}}\pa{1 - C_{d, V_{\mathrm{per}},E_{0},\rho_1,\rho_2}L^{-\beta
}}\ge L^{-\tau},
\eeq
where $\beta$ is given in \eq{mL34}.
\end{lemma}

\begin{proof} Since by hypothesis  $\set{L_n; \, n=0,1,\ldots,n_1}$ is $(E ,\varsigma, \varsigma^\pr,p,\tau)$-extra good,  it follows from Lemma~\ref{lemstructure} that there exist  $K_{1},K_2 \in \N$ such that,  given a box $\Lambda=\Lambda_{L}$, if $L$ is sufficiently large  there exists 
 a $(\Lambda,E,K_1,K_2)$-ready event $\cQ$ satisfying \eq{readyprob}.   We write $\cQ$ as in \eq{ready}, and apply Lemma~\ref{lemwegnerplus}  to each  $(\Lambda,E,K_1,K_2)$-prepared events $\cP_{j}$, letting $\cW_{\cP_{j}}$ denote the corresponding event.  In particular, $\cW_{\cP_{j}}$ satisfies \eq{probW} and $\cP_{j}\setminus \cW_{\cP_{j}}$
 is a $(\Lambda,E,m_{L},\varsigma, \varsigma^\pr)$-adapted event  with $m_{L}$ is as in \eq{mL34}, which yields \eq{mL345} since $K_1,K_2$  depend only on $d,p,\rho_{1}, \rho_2$.  It follows  then that
 \beq
 \cE =\bigsqcup_{j=1}^{J} \pa{ \cP_{j}\setminus \cW_{\cP_{j}}}= \cQ \setminus \pa{\bigcup_{j=1}^{J} \pa{ \cW_{\cP_{j}}}}
 \eeq
is   a  $(\Lambda_{L},E,m_{L},\varsigma, \varsigma^\pr)$-{extra good} event.
Since it follows from \eq{ready} and  \eq{probW} that 
\beq
\P\set{\bigcup_{j=1}^{J} \pa{ \cW_{\cP_{j}}}} \le  C_{2} K_{1} L^{-\frac d 2  \pa{  \rho_{1}(1 - \vs^\pr) -2\rho_{2}}}\P\set{\cQ},
\eeq
we get, using \eq{readyprob} and \eq{rhos}, that 
\beq\label{wegnerplusprobest}
\P \set{\cE} \ge \pa{1 - 2L^{-2d}}\pa{1 -   C_{2} K_{1} L^{-\frac d 2  \pa{  \rho_{1}(1 - \vs^\pr) -2\rho_{2}}}}\ge 1 - L^{-pd}.
\eeq
\end{proof}

We can now finish the proof of Proposition~\ref{propA}.

\begin{proof}[Proof of Proposition~\ref{propA}]  Let $ E \in [0,E_{0}]$ and suppose  that for some scale $L_{0} $  we know that $L$ is $(E,m_0,\varsigma, \varsigma^\pr,p)$-{extra good} for all 
  $L \in \br{L_{0}, L_{0}^{ \rho_{1}^{-1}\rho_{2}^{-1}}}$, with $m_{0}$ satisfying \eq{m0}. In other words, the interval $\br{L_{0}, L_{0}^{ \rho_{1}^{-1}\rho_{2}^{-1}}}$ is $(E ,\varsigma, \varsigma^\pr,p,\tau)$-extra good with   $m_{L}=m_{0}$ for $L \in \br{L_{0}, L_{0}^{ \rho_{1}^{-1}\rho_{2}^{-1}}}$.  We also assume that $L_0$ is large enough so we can use  Lemma~\ref{lemwegnerplus2} for all $L\ge L_0$
  
  Let $\cL_0=   \br{L_{0}, L_{0}^{ \rho_{1}^{-1}\rho_{2}^{-1}}}$ and  $\cL_k= \br{L_{0}^{ \rho_{1}^{-k}\rho_{2}^{-1}}, L_{0}^{ \rho_{1}^{-(k+1)}\rho_{2}^{-1}}}$ for $k=1,2,\ldots$. We set
  \beq
  m_{k}= m_{0} \prod_{k^\pr=1}^{k}\pa{1 -C_{E_0} L_{0}^{ -\beta \rho_{1}^{-k^\pr}\rho_{2}^{-1}}}\ge  L_{0}^{ -\tau \beta \rho_{1}^{-k}\rho_{2}^{-1}},
  \eeq   
where $C_{E_0}=C_{d, V_{\mathrm{per}},\rho_1,\rho_2,E_{0}}$ and $\beta$ are as in \eq{mL345}, the inequality holding for all $k$ by taking $L_0$ sufficiently large.
 We consider  statements $(\cS_k)$,   given for  $k=0,1,2,\ldots$ by:
  
  \begin{quotethm}[$\cS_k$] The scale interval $ \cL_k$
  is   $(E ,\varsigma, \varsigma^\pr,p,\tau)$-extra good with $m_L \ge m_k$ for all  $L \in \cL_k$.
 \end{quotethm}

  We will prove that $(\cS_k)$ is valid for all $k=0,1,2,\ldots$ by induction. Note  that the validity of $(\cS_0)$ is our hypothesis, and   $(\cS_1)$ follows immediately from $(\cS_0)$ by Lemma~\ref{lemwegnerplus2}.  If  $k=1,2,\ldots$,  and  $(\cS_{k-1})$ and $(\cS_k)$ are valid, we can apply Lemma~\ref{lemwegnerplus2} for all  $L \in  \cL_{k+1}$, and  conclude that  $(\cS_{k+1})$holds   with 
\beq
 m_{L} \ge   m_k\pa{1 -C_{E_0}L^{-\beta}}\ge m_{k+1} \ge  L_{0}^{ -\tau \beta \rho_{1}^{-(k+1)}\rho_{2}^{-1}}\ge L^{-\tau}.
 \eeq

  Since we have  $(\cS_k)$  for all $k=0,1,2,\ldots$, we conclude that the scale interval $ [ L_{0}, \infty[=\bigcup_{k=0}^\infty  \cL_{k}$
  is   $(E ,\varsigma, \varsigma^\pr,p,\tau)$-extra good, and for all $L \in  [ L_{0}, \infty[$ we have
   \beq
  m_{L}\ge m_{0} \prod_{k=1}^{\infty}\pa{1 - C_{E_{0}}L_{0}^{ -\beta \rho_{1}^{-k}\rho_{2}^{-1}}}\ge\frac {m_{0}} 2
  \eeq 
  for sufficiently large $L_0$. In particular,  every scale $L \ge L_{0}$
  is $\pa{E,\frac  {m_{0}} 2,\vs,\vs^\pr,p}$-{extra good}, so
  the theorem is proved.
  \end{proof}

\section{Preamble to localization}\label{secpreambleloc}

In this section we introduce tools for extracting localization from the multiscale analysis.

Let $ {\nu} > \frac d 2$.  (We will work with a fixed ${\nu}$ that will be generally omitted from the notation.) Given $y \in \R^{d}$, we  recall that
 $T_y=T_{\nu,y}$ denotes the operator on  the Hilbert space  $\H=\mathrm{L}^2(\mathbb{R}^d)$ given by multiplication by the function
$T_y(x)=T_{\nu,y}(x):= \la x-y\ra^{\nu}$ for  $x \in \R^d$, with $T:=T_0$.
Since $\langle y_{1} +y_{2} \rangle \le \sqrt{2}\langle y_{1} \rangle
\langle y_{2}\rangle$, we have
\begin{equation}\label{Tab}
\| T_{y_{1}} T_{y_{2}}^{-1} \| \le 2^{\frac  {\nu}  2} \la y_{1} -y_{2} \ra^{\nu}.\end{equation}

The  domain of $T$,  $\mathcal{D}(T)$, equipped with the norm
$\|\phi\|_+ = \|T\phi\|$, is a Hilbert space,  denoted by
$\mathcal{H}_+=\mathcal{H}_{\nu,+}\,$.  The Hilbert space $\mathcal{H}_-=\mathcal{H}_{\nu,-}$ is defined as
 the completion of  $\mathcal{H}$ in the norm $\|\psi\|_- = \|T^{-1}\psi\|$.
 By construction,
$\mathcal{H}_+\subset\mathcal{H}\subset\mathcal{H}_-\,$, and the
natural injections  
 $\imath_+: \mathcal{H}_+\rightarrow\mathcal{H}$  and 
$\imath_-:\mathcal{H}\rightarrow\mathcal{H}_-$ are
continuous with dense range.  The operators $T_+: \mathcal{H}_+ \to \mathcal{H}$
and $T_-: \mathcal{H} \to \mathcal{H}_-$, defined by
$T_+= T \imath_+\,$, and 
 $T_- = \imath_- T $ on $ \mathcal{D}(T)$, are unitary. Note that it follows from \eq{Tab}
 that
 \beq
  \|T_{y}^{-1}\psi\|\le  2^{\frac  {\nu}  2} \la y \ra^{\nu} \|T^{-1}\psi\|\quad \text{for all} \quad y\in \R^{d}\quad \text{and}\quad \psi \in\mathcal{H}_-.
 \eeq

\subsection{$\nu$-generalized eigenfunctions} Let $H_{\bom}
$ be a generalized Anderson Hamiltonian.  For a fixed  $\bom \in \Omega$ we now consider only generalized eigenfunctions $\psi \in \H_{-}= \mathcal{H}_{\nu,-}$, so we rewrite Definition~\ref{defgeneiginitial} as follows.

\begin{definition}  \label{defgeneiginitialnu}
 A $\nu$- \emph{generalized eigenfunction} for $H_{{\bom}}$ with  \emph{generalized eigenvalue} ${E}$  is a function  $\psi\in \mathcal{H}_{\nu,-}$  such that $\psi\not=0$ and
\beq  \label{eigb3}
\scal{H_{{\bom}}\vphi,\psi}= {E} \scal{\vphi,\psi} \quad \text{for all}\quad \vphi \in C_{c}^{\infty}(\R^{d}).
\eeq
\end{definition}
Given ${E} \in \R$ we let 
$\Theta_{\bom}({E})=\Theta_{\nu,\bom}({E})$ denote the collection of  $\nu$-generalized eigenfunctions for $H_{{\bom}}$ with generalized eigenvalue ${E}$, and set
$\widetilde{\Theta}_{\bom}({E})= \Theta_{\bom}({E})\cup \{0\}$.   We will drop  $\nu$ from the notation:  $\psi$ will be called a generalized eigenfunction for $H_{{\bom}}$ with generalized eigenvalue ${E}$ if and only if $\psi \in {\Theta}_{\bom}({E})$.
We will also call ${E}\in \R$  a generalized eigenvalue for $H_{{\bom}}$ if and only if $\Theta_{\bom}({E})\not= \emptyset$.

 The generalized eigenvalues and eigenfunctions of $H_{\bom}$  are  the same as the eigenvalues and eigenfunctions of the operator   $H_{\bom,-}$:  a function 
$\psi \in \mathcal{H}_-$, $\psi \not= 0$, is a generalized eigenfunction of $H_{\bom}$ with generalized 
eigenvalue ${E}$ if and only if $\psi \in \mathcal{D}(H_{\bom,-})$ and
$H_{\bom,-} \psi={E} \psi$, i.e.,
\begin{equation} \label{eigb}
\langle H_{\bom} \phi,\psi \rangle = {E} \langle \phi,\psi \rangle
\quad  \mbox{for all} \quad\phi \in \mathcal{D}(H_{\bom}) \cap \H_{+}\,.
\end{equation}
This follows from the fact that  \eq{eigb} is equivalent to \eq{eigb3} since $C_{c}^{\infty}(\R^{d})$ is a core for  the $H_{\bom}$.

Eigenvalues and eigenfunctions of $H_{\bom}$ are always generalized eigenvalues and eigenfunctions. Conversely, if $\psi \in  \Theta_{\bom}({E})\cap \H$, i.e.,  $\psi \in \mathcal{H}$ is a generalized eigenfunction of $H_{\bom}$ with generalized eigenvalue ${E}$, then $\psi $ is an  eigenfunction of $H_{\bom}$ with  eigenvalue ${E}$.

\subsection{Generalized eigenfunctions and good boxes}

Given $\bom\in \Omega$, $x \in \R^{d}$ and ${E} \in \R$, we set
 \begin{align} \label{defGWx}
W_{\bom,x}({E})=W\up{\nu}_{\bom,x}({E}):=\begin{cases} 
{\sup_{\psi \in \Theta_{\bom}({E})} }
\ \frac {\| \Chi_{x}\psi \|}
{\|T_{x}^{-1}\psi \|}&
 \text{if $\Theta_{\bom}({E})\not=\emptyset$}\\0 & \text{otherwise}\end{cases}.
\end{align}
Note that
\begin{equation}\label{boundGW}
W_{\bom,x}({E})\le 
 \left( \tfrac 5 4\right)^{\frac {\nu} 2}<2^{\frac {\nu} 2} .
\end{equation}

\begin{remark}  By the unique continuation principle  $\Theta_{\bom}({E})\not=\emptyset$
if and only if  $W_{\bom,x}({E})\not=0$ for all $x \in \R$.
\end{remark}

\begin{lemma}\label{lemgoodW} Let $\bom \in{\Omega}$, $I\subset \R$ a bounded interval,  ${E}\in I$, $0< \varsigma< 1$,    $m>0$. Suppose the box  $\Lambda_{L}(x)$ is $(\bom,{E},m,\varsigma)$-jgood.
  Then,  if  $m \ge  C_{d,\nu,V_{\mathrm{per}},I}\,  \frac  {\log L} L$,   we have
\beq \label{WxEgood}
W_{\bom,y}({E})  \le  \e^{-{\frac m {15} L}} \quad \text{for all $y \in \Lambda_{L}(x)$ with $\Lambda_{\frac L 5}(y)\subset \Lambda_{L}(x)$}.
\eeq
 \end{lemma}
 
 \begin{proof}  We can assume $\Theta_{\bom}({E})\not=\emptyset$.  Given  $\psi \in {\Theta}_{\bom}({E})$, it follows from  Lemma~\ref{lemEDI} that  for all $y \in \Lambda_{L}(x)$ with $\Lambda_{\frac L 5}(y)\subset \Lambda_{L}(x)$ we have
 \begin{align}
 \norm{\Chi_{y}\psi}\le 2 \gamma_{{E}}L^{d-1}  \e^{-{\frac m {11} L}}\max_{y^{\pr}\in \Upsilon_{\Lambda_{L}(x)}} \norm{\Chi_{y^{\pr}}\psi}
 \le  2 \gamma_{{E}}L^{d-1}\pa{1 +L^{2}}^{\frac \nu 2 } \e^{-{\frac m {11} L}} \norm{T_{y}^{-1}\psi},
 \end{align}
so
\beq
\frac { \norm{\Chi_{y}\psi}}{\norm{T_{y}^{-1}\psi}}\le  \e^{-{\frac m {15} L}}\quad \text{for}\quad m \ge  C_{d,\nu,V_{\mathrm{per}},I}\,  \frac   {\log L} L. 
\eeq 
\end{proof}

\subsection{Generalized eigenfunctions and annuli of good boxes} 
Given $\bom \in \Omega$, $x \in \R^{d}$, ${E} \in \R$, and a scale $L$, we set (cf. \eq{defGWxLint})
  \begin{align} \label{defGWxL}
W_{\bom,x,L}({E})=W\up{\nu}_{\bom,x,L}({E}):=\begin{cases} {\sup_{\psi \in \Theta_{\bom}({E})} }
\frac {\norm{ \Chi_{x,L}\psi}}
{\norm{T_{x}^{-1}\psi }}&
 \text{if $\Theta_{\bom}({E})\not=\emptyset$}\\0 & \text{otherwise}
.\end{cases},
\end{align}
where $ \Chi_{x,L}:= \Chi_{\La_{2L+1,L-1}(x)}$, and
\begin{equation}
 L_{-}:=L - \tfrac 1 {5}\tfrac L {100}= \tfrac  {499}{500}L,  \quad  L_{+}:=2L +\tfrac 1 {5}\tfrac L {100}= \tfrac  {1001}{500}L. 
\end{equation}
In particular, we have ($L \ge 2$)
\begin{equation}\label{boundGWL}
W_{\bom,x,L}({E})\le 
 \left(1 + (L + \tfrac 1 2)^2\right)^{\frac {\nu} 2} \le 2^{\frac {\nu} 2} L^{{\nu} }.
\end{equation}
Note also that, using \eq{Tab},
\beq\label{Wconv}
W_{\bom,y}({E}) \le  2^{\frac  {\nu}  2} \la {
y-x} \ra^{\nu}W_{\bom,x,L}(E)\le 2^{  {\nu}  } L^{\nu }W_{\bom,x,L}(E)\;\; \text{for}\;\; y \in  \overline{\La}_{2L,L}(x).
\eeq

\begin{lemma}\label{lemgoodWL} Let $\bom \in{\Omega}$, $I\subset \R$ a bounded interval,  ${E}\in I$, $0< \varsigma< 1$,    $0 <\widetilde{m}<m$. Suppose every box $\La_{\frac L {100}}$ in  the standard  $\frac L {100}$-covering  of  the  annulus $  \Lambda_{L_{+},L_{-}}(x)$ is  $(\bom,{E}, m ,\varsigma)$-jgood.
  Then,  if  $ m \ge  C^\pr_{d,\nu,V_{\mathrm{per}},I}\,  \frac  {\log L} L$,   we have
\beq \label{WxLEgood}
W_{\bom,x,L}({E})\le    \e^{-{\frac m {2000} L}}.
\eeq
 \end{lemma}

\begin{proof}  We can assume $\Theta_{\bom}({E})\not=\emptyset$. 
Given   $y \in \overline{\La}_{2L,L}(x)$ there exists a box $\La_{\frac {L} {100}}\up{y}$ in the 
standard  $\frac {L} {100}$-covering  of  the  annulus $  \La_{L_{+},L_{-}}(x)$ with
$\La_{ \frac {L} {500}}(y)\subset \La_{\frac {L} {100}}\up{y}$.  Since the box $\La_{\frac {L} {100}}\up{y}$ is $(\bom,{E}, m ,\varsigma)$-jgood by hypothesis, it follows from Lemma~\ref{lemgoodW}  that for all  $\psi \in {\Theta}_{\bom}({E})$ we have, with $\ell =\frac {L} {100}$, that if  $ m \ge  C_{d,\nu,V_{\mathrm{per}},I}\,  \frac  {\log \ell} \ell$ we have
\begin{align}
 \norm{\Chi_{y}\psi} \le   \norm{T_{y}^{-1}\psi}  \e^{-{\frac m {15} \ell}}\le 2^{\frac\nu 2}\scal{y-x}^{\nu} \norm{T_{x}^{-1}\psi} \e^{-{\frac m {15} \ell}}
  \le 2 ^\nu  L^{ \nu }  \norm{T_{x}^{-1}\psi} \e^{-{\frac m {15} \ell}}.
\end{align}
It follows that
\begin{align}
 \norm{ \Chi_{x,L}\psi}& \le  C_{d,\nu} L^{\nu + d }  \norm{T_{x}^{-1}\psi} \e^{-{\frac m {15} \ell}}\le \norm{T_{x}^{-1}\psi} \e^{-{\frac m {20} \ell}},
 \end{align}
which yields \eq{WxLEgood}.
\end{proof}

  \subsection{Generalized eigenfunction expansion} 

A generalized  Anderson Hamiltonian $H_{\bom}$ has a generalized eigenfunction expansion, which we will now review.  We follow  \cite[Section 3]{KKS}, to which we refer for  all the details. (Although the results in \cite{KKS} are stated for classical wave operators, which include $-\Delta$, they clearly hold  for $-\Delta +V$ with $V$ a bounded potential; in particular they hold for  generalized  Anderson Hamiltonians as in Definition~\ref{defgenAndH}.)

 Let $H_{\bom}$ be a generalized  Anderson Hamiltonian.  For all $\bom \in \Omega$ we have the estimate (e.g., \cite[Lemma~A.4]{GK5})
\beq\label{traceestimate}
\tr \set{T^{-1}\pa{H_{\bom} + 1 + \norm{V^{-}_{\mathrm{per}}}  }^{- 2[[\frac d 4]]} T^{-1}}\le
C_{d,{\nu},\norm{V^{-}_{\mathrm{per}}}}< \infty,
\eeq
where $[[\frac d 4]]= \min \set{n \in \N; n > \frac d 4}$ and $V^{-}_{\mathrm{per}}$ is the negative part of $V_{\mathrm{per}}$.  
We define the spectral measure
\begin{equation}\label{trestmeas}
\mu_{\bom}({B}):= \tr \{T^{-1}P_{\bom}(B)T^{-1}\}=\|T^{-1}P_{\bom}(B))\|_2^2,
\quad\text{ ${B}\subset \R$  a Borel set}.
\end{equation}
As a consequence of \eq{traceestimate}, for all Borel sets $B$ with $\sup B<\infty$ we have
\begin{equation}\label{trestmeas234}
\mu_{\bom}({B}) \le C_{d,{\nu},\norm{V^{-}_{\mathrm{per}}},\sup B}<\infty\quad\text{for all}\quad \bom \in \Omega.
\end{equation}
Moreover, since   the constants  in \eq{traceestimate} and \eq{trestmeas234} depend on the potential only through $\norm{V^{-}_{\mathrm{per}}}$ (they are independent of the background potential $U\ge 0$ and the random potential $V_{\bom}\ge 0$), we have, 
similarly to \cite[Eq.~(2.5)]{GKsudec}, that for all $\bom \in \Omega$  and Borel sets $B $ with $\sup B<\infty$  we have
\begin{equation}\label{trestmeas99}
\mu_{\bom,y}({B}):= \tr \{T_{y}^{-1}P_{\bom}(B)T_{y}^{-1}\}\le
C_{d,{\nu},\norm{V^{-}_{\mathrm{per}}},\sup B}<\infty \quad\text{for all}\quad y\in \R^{d},
\end{equation}
and hence 
\begin{equation}\label{trestmeas9988}
\norm{\Chi_{y} P_{\bom}(B)}_{2}\le
C_{d,\norm{V^{-}_{\mathrm{per}}},\sup B}<\infty \quad\text{for all}\quad y\in \R^{d}.
\end{equation}
Note also that $\mu_{\bom}$ and $\mu_{\bom,y}$ are absolutely continuous with respect to each other.

Let ${\mathcal{ T}_1(\mathcal{H}_+,\mathcal{H}_-)}$ be the Banach space of bounded linear operators  $A\colon \mathcal{H}_+ \to
\mathcal{H}_-$  with  $T_-^{-1}AT_+^{-1}$ trace class.  Then
 for all $\bom \in \Omega$ there exists a $\mu_{\bom}$-locally integrable function 
${\Pb_{\bom}}\colon \R \to {\mathcal{ T}_1(\mathcal{H}_+,\mathcal{H}_-)}$,
such that
\begin{equation}\label{a16a}
\tr \left\{T_-^{-1}{\Pb_{\bom}}({{E}})T_+^{-1}\right\}=1
\quad  \text{for          
$\mu_{\bom}$-a.e.\  ${{E}}$},
\end{equation}
and, for all   Borel sets $B $ with $\sup B<\infty$,
\begin{equation}\label{geneigexp}
\imath_- P_{\bom}(B)\imath_+=
\int_{B}\Pb_{\bom}({{E}})\,d\mu_{\bom}({{E}}),
\end{equation} 
where the integral is the Bochner integral of 
${\mathcal{ T}_1(\mathcal{H}_+,\mathcal{H}_-)}$-valued 
functions.  Note that ${\Pb_{\bom}}({{E}})$ is jointly measurable in $(\bom,{E})$.  (This can be see from \cite[Eq.~(46)]{KKS}.)  Moreover,  we have  (e.g.,  \cite[Corollary~3.1]{KKS}) 
\beq \label{geneigPeq}
H_{\bom,-}\Pb_{\bom}({{E}})={E}  \Pb_{\bom}({{E}})\quad \text{for} \quad \mu_{\bom}\text{-a.e.} \quad {E}\in \R, 
\eeq
where $H_{\bom,-}$ is   the  closure of the operator $H_{\bom}$
in the Hilbert space  $\H_-$. It 
 follows  that
 \beq\label{geneigPgeneig}
 \Pb_{\bom}({{E}})\H_{+}\subset\widetilde{\Theta}_{\bom}({E})  \quad \text{for} \quad \mu_{\bom}\text{-a.e.} \quad {E}\in \R. 
 \eeq 

If for a given Borel set $B$ we have ($\H=\imath_- \H$ as sets)
\beq\label{allgeneig}
\Pb_{\bom}({{E}})\H_{+}\subset  \H \quad \text{for}\quad \mu_{\bom}\text{-a.e}  \;\; {E}\in B,
\eeq
 it follows from \eq{geneigPgeneig} that $H_{\bom}$ has pure point spectrum in $B$.

 Given $\bom\in \Omega$, $x \in \R^{d}$, ${E} \in \R$, and a scale $L$, we set (cf.\cite{GKsudec})
 \begin{align} \label{defGWx2}
\W_{\bom,x}({E}):=\begin{cases} 
\displaystyle{\sup_{\substack{\phi \in \H_+\\
\Pb_{\bom}({E})\phi  \not= 0}}} 
\ \frac {\| \Chi_{x} \Pb_{\bom}({E})\phi \|}
{\|T_{x}^{-1}\Pb_{\bom}({E})\phi \|}&
 \text{if $\Pb_{\bom}({E})\not=0$}\\0 & \text{otherwise}\end{cases},
\end{align}
  \begin{align} \label{defGWxL2}
\W_{\bom,x,L}({E}):=\begin{cases} 
\displaystyle{\sup_{\substack{\phi \in \H_+\\
\Pb_{\bom}({E})\phi  \not= 0}}} 
\ \frac {\| \Chi_{x,L}\Pb_{\bom}({E})\phi \|}
{\|T_{x}^{-1}\Pb_{\bom}({E})\phi \|}&
 \text{if $\Pb_{\bom}({E})\not=0$}\\0 & \text{otherwise}\end{cases}.
\end{align}

$\W_{\bom,x}({E})$ and $\W_{\bom,x,L}({E})$ are
measurable functions of  $(\bom,{E})$ for each $x\in \R^d$ with
\begin{align}\label{boundGW2}
\W_{\bom,x}({E})&\le 
 \left( \tfrac 5 4\right)^{\frac {\nu} 2}<2^{\frac {\nu} 2}  ,\\ \label{boundGWL2}
\W_{\bom,x,L}({E})&\le 
 \left(1 + (L +\tfrac 12 )^2\right)^{\frac {\nu} 2} \le 2^{\frac {\nu} 2}L^{\nu},\\
\W_{\bom,y}({E})& \le 2^{\nu} L^{\nu }\W_{\bom,x,L}\quad\text{for}\quad y \in  \overline{\La}_{2L,L}(x). \label {Wconvbold}
\end{align}
Moreover, it follows from \eq{geneigPgeneig} that
 \beq\label{WW<W}
\W_{\bom,x}({E}) \le W_{\bom,x}({E}) \;\; \text{and}\;\;  \W_{\bom,x,L}({E}) \le W_{\bom,x,L}({E}) \quad \text{for}\;\mu_{\bom}\text{-a.e. }   {E}\in \R.
\eeq

\begin{remark} There is a difference between  $\W_{\bom,x}({E})$ and $\W_{\bom,x,L}({E})$,   defined in  \eq{defGWx2}  and \eq{defGWxL2}, and     $W_{\bom,x}({E})$ and $W_{\bom,x,L}({E})$, previouusly defined in \eq{defGWx}  and \eq{defGWxL}.       The conclusions of the multiscale analysis of  Proposition~\ref{propA}  will yield  bounds on $W_{\bom,x}({E})$ and $W_{\bom,x,L}({E})$ in an energy interval $I$.  In view of \eq{WW<W}, these bounds will hold for    $\W_{\bom,x}({E})$ and $\W_{\bom,x,L}({E})$  for $ \mu_{\bom}$-a.e.\    ${E}\in I$,  yielding \eq{allgeneig}  for $ \mu_{\bom}$-a.e.\    ${E}\in I$, and hence establishing pure point spectrum in the interval $I$.  Note  that  $\W_{\bom,x}({E})$ and $\W_{\bom,x,L}({E})$ are measurable functions of  $(\bom,{E})$ for each $x\in \R^d$, but we do not make such a claim for  $W_{\bom,x}({E})$ and $W_{\bom,x,L}({E})$.
\end{remark}

\subsection{Connection with point spectrum}

 Given ${{E}} \in \R$,  we set   
 \beq
 P_{\bom}({E}):=  \Chi_{\{{E}\}}(H_{\bom})\quad \text{and}\quad
\mu_{\bom}({{E}}) :=\mu_{\bom}(\{{{E}}\})=  \norm{T^{-1}P_{\bom}({E})}_{2}^{2}.
\eeq
In particular, 
$P_{\bom}({E})\not=0$ if and only if  $\mu_{\bom}({{E}})\not=0$.

 It follows from 
\eqref{geneigexp} that 
\beq\label{Ppurepoint}
\imath_- P_{\bom}({E}) \imath_+=
\Pb_{\bom}({E})\mu_{\bom}({{E}}).
\eeq  
Thus, given $x \in \R^{d}$ and a scale $L$, we have
\begin{align}\label{xPEW}
\norm{\Chi_{x} P_{\bom}({E})}_{2}&\le \W_{\bom,x}({E})\norm{T^{-1}_{x}P_{\bom}({E})}_{2}= \W_{\bom,x}({E}) \sqrt{\mu_{\bom,x}({{E}})},
\\ \label{xPEWL}
\norm{\Chi_{x,L} P_{\bom}({E})}_{2}&\le \W_{\bom,x,L}({E})\norm{T^{-1}_{x}P_{\bom}({E})}_{2}= \W_{\bom,x,L}({E}) \sqrt{\mu_{\bom,x}({{E}})}.
\end{align}

If $H_{\bom}$ has pure point spectrum in an interval $I$, it follows from \eq{geneigexp} and  \eq{Ppurepoint} that for all  bounded Borel functions $f$   we have
\beq\label{PIexpas}
f(H_{\bom})P_{\bom}(I) = \int_{I}f({E}) {P}_{\bom,x}(E)  \di \mu_{\bom,x}(E)\quad \text{for all}\quad x \in \R^{d},
\eeq
where
\begin{align} \label{defPEx01}
{P}_{\bom,x}(E):=\begin{cases} 
\pa{\mu_{\bom,x}(E)}^{-1}P_{\bom}(E) &
 \text{if $P_{\bom}(E)\not=0$}\\0 & \text{otherwise}\end{cases}.
\end{align}

\section{From the multiscale analysis to localization}\label{secMSAloc}

We will now assume that the conclusions of the multiscale analysis  (i.e., of Proposition~\ref{propA}) hold for all energies in   a bounded open interval $ \I $, and  prove a theorem that encapsulates localization in the interval $\I$.  All forms of localization will be derived from this theorem.

We fix $ {\nu} > \frac d 2$, which  will be generally omitted from the notation.

 \begin{theorem}\label{thmmainev} Let $H_{\bom}$ be a generalized Anderson Hamiltonian on  $\L^{2}(\R^{d})$.  
Consider  a bounded open interval $ \I \subset \R$,
 $m>0$,  $p>0 $, and   $ \vs \in]0,1[ $, and  assume there is a scale $\cL$ such that all scales $L\ge\cL$ are $(E,m,\varsigma, p)$-good for all energies $E \in \I$.  Set
 \beq\label{defhatm}
{M}={M}(m,p) := \tfrac m {30^{ \widehat{n}+2 }}\quad \text{where}  \quad\widehat{n}  =\widehat{n}(p):= \min\set{n \in \N; \;\; 2^{\frac 1 n} - 1 <p} . 
  \eeq 
Fix  $\widetilde{p}\in ]0,p[$, and pick  $\vartheta =\frac  \beta 2$, where   $\beta={\rho}^{n_{1}}$ with  $\rho>0$ and $n_{1}\in \N$ such that
\beq \label{prho2n1}
 \pa{1 +  p}^{-1} < {\rho}<1 \quad    \text{and} \quad   (n_{1}+1)\beta < p -\widetilde{p} ,
\eeq 
and set, at scale $L$,
 \beq
\I_{L}:=\set{{E} \in \I; \quad \dist \pa{{E}, \R\setminus \I} >  \e^{- {M} L^{\vartheta} }}.\label{defIL}
\eeq 
Then, given a sufficiently large  scale $L $,  for each $x_{0}\in \R^{d}$   there exists an event $\cU_{L,x_{0}}$ with the following properties:
\begin{enumerate}
\item We have
\beq\label{cUdesired}
 \cU_{L,x_{0}}\in \cF_{\Lambda_{L_{+}}(x_{0})}\quad \text{and}\quad  \P\set{\cU_{L,x_{0}} }\ge  1 - L^{- \widetilde{p}  d}.
\eeq

   \item  If $\bom \in\cU_{L,x_{0}}$ and  ${E} \in \I_{L}$,
whenever
\beq \label{distEIred-}
W_{\bom,x_{0}}({E}) >   \e^{-  M L^{\vartheta} },
 \eeq 
 we conclude that 
\beq \label{EDI9L9}
W_{\bom,x_{0},L}({E}) \le \e^{- M L} .
\eeq

\item If  $\bom \in\cU_{L,x_{0}}$,   we have 
\beq \label{EDI9L99}
W_{\bom,x_{0}}({E})W_{\bom,x_{0},L}({E}) \le \e^{-  \frac 1 2 {M}  L^{ \vartheta}}\qtx{for   all} {E} \in  \I_{L} .
\eeq
\end{enumerate}
\end{theorem}

\begin{remark}  If  $p\in\rb{ \frac 1 3, \frac 3 8}$, as in Theorem~\ref{thmMSA}, we have
$\widehat{n} =3$. 
\end{remark}

   The proof of this theorem will require several propositions.   The scale $\cL$ will always be assumed to be sufficiently large; in particular  we assume $m \ge  \cL^{- \frac \vs 2}$.  We consider only scales  $L\ge \cL$. We use the following notation:  $A\up{\I}=A\cap \I$ for $A \subset \R$.
   
   We assume the hypotheses of Theorem~\ref{thmmainev} in the remainder of this section.
   
\subsection{The first spectral reduction}

\begin{proposition}\label{propenergyred1} Given $b \ge 1$, there exists a constant  $K_{d,p,b}\ge 1 $ with the following property: Fix $K\ge K_{d,p,b}$. Then,
given a sufficiently large  scale $L $,  for each $x_{0}\in \R^{d}$   there is an event $\cQ_{L,x_{0}}$, with
\beq\label{cQdesired}
\cQ_{L,x_{0}}\in \cF_{\Lambda_{L}(x_{0})} \quad\text{and}\quad  \P\set{\cQ_{L,x_{0}} }\ge  1 - L^{-2bd} ,
\eeq
such that for  $\bom \in\cQ_{L,x_{0}}$, given  ${E} \in \I$ such that
\beq \label{distEI}
 W_{\bom,x_{0}}({E}) >   \e^{- \widehat{m} \sqrt{\frac L K} }\quad \text{and}  \quad \dist \pa{{E}, \R\setminus \I} > \e^{-\widehat{m} \sqrt{\frac L K}},
 \eeq 
 where $ \widehat{m}= \widehat{m}(m,p):= 30{M}$ with  ${M}$ given in \eq{defhatm},
it follows that
\beq\label{Ldist}
\dist \pa{{E}, \sigma\up{\I}(H_{\bom,\Lambda_{L}(x_{0})})} \le  \e^{-\widehat{m}L } .
\eeq
\end{proposition}

The proof of this proposition will rely on several lemmas.

\subsubsection{A site percolation model}
Given a box $\Lambda_{L^{\pr}}(x_{0})$ and a scale $\ell \ll L^{\pr}$, we set $L^{\pr\pr}=L^{\pr}+\ell$,   let $\alpha =\alpha_{L^{\pr\pr},\ell}$  be as in \eq{alphaL}, and consider the graph
\beq
\G=\G_{x_{0},L^{\pr},\ell}:= x_{0}+  \alpha\ell  \Z^{d} \quad \text{with edges}\quad\set{ \set{r,r^{\pr}} \subset \G; \; \norm{r-r^{\pr}}=\alpha \ell }.
\eeq
Note that for  $r,r^{\pr}\in \G$ we have 
\beq
\norm{r-r^{\pr}}=\alpha \ell \quad \Longleftrightarrow \quad r\ne r^{\pr} \;\text{and} \; \Lambda_{\ell}(r)\cap \Lambda_{\ell}(r^{\pr})\ne \emptyset.
\eeq
The external boundary of $\Gamma \subset \G$ is defined as
\beq
\partial^{+}\Gamma:= \set{r \in \G \setminus \Gamma; \quad \set{ r, r^{\pr}} \; \text{is an edge for some} \; r^{\pr}\in \Gamma }.
\eeq
We have   $\# \pa{ \partial^{+} \set{r}}=3^{d}-1$ for all  $r \in \G$, i.e., any site is connected by edges to $3^{d}-1$ other sites.  We  call  $y_{0},y_{1},\ldots,y_{k}\in \G$ a path if $\{y_{j-1},y_{j}\}$ is an edge of $\G$ for $j=1,2,\ldots,k$; it is a self-avoinding path if  the $y_{0},y_{1},\ldots,y_{k}$ are distinct.

Given an energy $E\in \I$, we consider the following  site percolation model on the graph $\G$:
every site $r \in \G_{\Lambda_{L^{\pr\pr}}(x_{0})}^{(\ell)}=\G \cap \Lambda_{L^{\pr\pr}}(x_{0})$ (cf.\eq{bbG}) is \emph{bad} with probability one;  a site $r\in \G \setminus \G_{\Lambda_{L^{\pr\pr}}(x_{0})}^{(\ell)}$ is \emph{good} if the box $\Lambda_{\ell}(r)$ is $(\bom,E,m,\varsigma, p)$-good and \emph{bad} otherwise.  We let $\mathbb{A}_{E}={\mathbb{A}}_{E}(\bom)={\mathbb{A}}_{E,x_{0},L^{\pr},\ell}(\bom)$ denote the  cluster of bad sites containing $ \G_{\Lambda_{L^{\pr\pr}}(x_{0})}^{(\ell)}$ (i.e., the connected component of the subgraph of bad sites containing  $\G_{\Lambda_{L^{\pr\pr}}(x_{0})}^{(\ell)}$).

 We now take scales $\ell ,{\widetilde{L}}$ with  $\ell \ll L^{\pr}$ and $100^{ d}\ell \le {\widetilde{L}}$.  Given  an energy $E \in \R$,  we consider the event
 \beq\label{YLM}
\cY\up{E}_{x_{0},L^{\pr},\ell,{\widetilde{L}}}:=
\begin{cases}
\set{{\mathbb{A}}_{E} \subset \Lambda_{L^{\pr}+{\widetilde{L}}-3 \ell}(x_{0})}& \text{if} \quad E \in \I\\
\Omega & \text{if} \quad E \notin \I
\end{cases}.
\eeq
Note that   $\cY\up{E}_{x_{0},L^{\pr},\ell,{\widetilde{L}}} \in \cF_{{\Lambda_{L^{\pr}+{\widetilde{L}}, L^{\pr}}(x_{0})}}$ for all $E \in \R$, and  it follows from \eq{goodmeas} that $\cY\up{E}_{x_{0},L^{\pr},\ell,{\widetilde{L}}}$ is jointly measurable in $\pa{E, {\bom}_{{\Lambda_{L^{\pr}+{\widetilde{L}}, L^{\pr}}(x_{0})}}}$.

\begin{lemma} \label{slemYLM}
For all  $E\in \I$ we have
\beq \label{probYLM}
\P\set{\cY\up{E}_{x_{0},L^{\pr},\ell,{\widetilde{L}}}}\ge 1 -  \pa{\tfrac {4L^{\pr}}{\ell}}^{d-1} 3^{\br{ \frac {\widetilde{L}} {2 \ell}}d}\ell^{-c_{d}p\frac {\widetilde{L}} { \ell}}.
\eeq
In particular, if $L^{\pr}={\widetilde{L}}=\frac L 2 $ and $\ell= \sqrt{L}$,  we get
\beq\label{probYL2}
\P\set{\cY\up{E}_{x_{0},\frac L 2,\sqrt{L},\frac L 2}}\ge 1 - L^{-{c}_{d,p}  \sqrt{L}}.
\eeq
\end{lemma}

\begin{proof}
Fix $E \in \I$, and suppose ${\mathbb{A}}_{E} \not\subset \Lambda_{L^{\pr}+{\widetilde{L}}-3 \ell}(x_{0})$. Then there exists a self-avoiding path $y_{0},y_{1},\ldots,y_{k}$ in $\G$, such that  $\norm{y_{0}-x_{0}}= \frac {L^{\pr}}2$,   $y_{1},y_{2},\ldots,y_{k} \notin\G_{\Lambda_{L^{\pr\pr}}(x_{0})}^{(\ell)}$,  $\norm{y_{k}-x_{0}} \ge \frac {L^{\pr}+{\widetilde{L}}-3\ell}2$, and all $y_{0},y_{1},\ldots,y_{k}$ are bad sites.  It follows that $ \frac {L^{\pr}+{\widetilde{L}}-3\ell}2 \le  \frac {L^{\pr}}2 + k\alpha \ell$, so
$k \ge \frac {{\widetilde{L}} -3\ell}{2\alpha \ell}\ge \br{ \frac {\widetilde{L}} {2 \ell}}$. We thus conclude that if ${\mathbb{A}}_{E} \not\subset \Lambda_{L^{\pr}+{\widetilde{L}}-3\ell}(x_{0})$ we can find a self-avoiding path $y_{0},y_{1},\ldots,y_{\br{ \frac {\widetilde{L}} {2 \ell}}}$ of bad sites with $\norm{y_{0}-x_{0}}= \frac {L^{\pr}}2$ and    $y_{1},y_{2},\ldots,y_{\br{ \frac {\widetilde{L}} {2 \ell}}} \notin\G_{\Lambda_{L^{\pr\pr}}(x_{0})}^{(\ell)}$.  The number of such self-avoiding paths is bounded by 
 $\pa{\frac {4L^{\pr}}{\ell}}^{d-1} 3^{\br{ \frac {\widetilde{L}} {2 \ell}}d}$.  Since sites $y,y^{\pr} \notin\G_{\Lambda_{L^{\pr\pr}}(x_{0})}^{(\ell)}$ are independent unless $\norm{y-y^{\pr}}\le \alpha \ell$, such a  self-avoiding path   must contain at least  $\br{3^{-d}\br{ \frac {\widetilde{L}} {2 \ell}}}\ge c_{d}^{\pr}\frac {\widetilde{L}} { \ell}$
independent sites, and hence its probability of having only bad sites is $\le \ell^{-c^{\pr}_{d}pd\frac {\widetilde{L}} { \ell}}$.  Thus
\beq
\P\set{{\mathbb{A}}_{E} \not\subset \Lambda_{L^{\pr}+{\widetilde{L}}-3\ell}(x_{0})}\le \pa{\tfrac {4L^{\pr}}{\ell}}^{d-1} 3^{\br{ \frac {\widetilde{L}} {2 \ell}}d}\ell^{-c^{\pr}_{d}pd\frac {\widetilde{L}} { \ell}}.
\eeq
\end{proof}

Given $\Gamma \subset \G$ and $0\le\eps_{1}<\eps_{2}$,  we set
\begin{align}
\widehat{\Gamma}&:=\bigcup_{x \in \Gamma} \Lambda_{\ell}(x),\\
\partial\up{\eps_{1},\eps_{2}}\widehat{\Gamma}&:= \set{x \in \R^{d}; \; \eps_{1} <\dist \pa{x, \widehat{\Gamma}}<\eps_{2}}.
\end{align}
Note that $\widehat{\Gamma}$ is a connected subset of $\R^{d}$ if $\Gamma$ is a connected subset of $\G$. 

\begin{lemma} \label{sublemphi}
 Let $E\in \I$ and $\bom \in \cY\up{E}_{x_{0},L^{\pr},\ell,{\widetilde{L}}}$. Then ($\mathbb{A}_{E}={\mathbb{A}}_{E}(\bom)$):
  \begin{enumerate}
  
  \item   For all  $r \in \partial^{+}{\mathbb{A}}_{E}$ we have that $\Lambda_{\ell}(r)\subset
   \Lambda_{L^{\pr}+{\widetilde{L}}}(x_{0})$ and  the box $\Lambda_{\ell}(r)$  is  $(\bom,E,m,\varsigma, p)$-good.
  
  \item  There exists a function   ${\phi}={\phi}_{\bom,E} \in C^2_{\mathrm{c}}(\R^{d})$, with  $0\le {\phi} \le 1$,  such that
\begin{gather}\label{defphi5}
{\phi}\equiv 1 \quad \text{on} \quad   \widehat{{\mathbb{A}}_{E}} ,\\
{\phi}\equiv 0 \quad \text{on} \quad  \R^{d}\setminus  \Lambda_{L^{\pr}+{\widetilde{L}}-\ell}(x_{0}) ,\\
\supp \nabla \phi \subset  \partial\up{4,8}\widehat{{\mathbb{A}}_{E}} ,\\
\abs{\nabla {\phi}}, \abs{\Delta {\phi}}\le C_d, \quad \text{the constant $C_{d}$ depending only on $d$}, \label{nablaphi5}
\end{gather}
and  for all $x \in \R^{d}$ with $\La_{\frac 1 2}(x)\cap \supp \nabla \phi \ne \emptyset$ there exists $r(x) \in \partial^{+}{\mathbb{A}}_{E}$
 such that   $\Lambda_{\frac \ell 5}(x) \subset \Lambda_{\ell}(r(x))$. 
   
\end{enumerate}
\end{lemma}

\begin{proof}Since $\bom \in \cY\up{E}_{x_{0},L^{\pr},\ell,{\widetilde{L}}}$, we have  $ \widehat{{\mathbb{A}}_{E}} \subset  \Lambda_{L^{\pr}+{\widetilde{L}}-2\ell}(x_{0})$.  (i) follows from the definition of  $ \widehat{{\mathbb{A}}_{E}}$.  To prove (ii), let $\psi$ be the characteristic function of the set $ \set{x \in \R^{d}; \; \dist \pa{x,  \widehat{{\mathbb{A}}_{E}}}\le 6}$.  Pick a a nonnegative function $\eta \in C^{2}(\R^{d})$ ,   with compact support in $\Lambda_{1}(0)$,  $\int_{\R^{d}} \eta(x)\,  \di x=1$, and $\abs{\nabla {\eta}}, \abs{\Delta {\eta}}\le C^{\pr}_d$. Then $\phi= \eta \ast \psi$ has all the desired properties.

Let $x \in \R^{d}$ with $\La_{\frac 1 2}(x)\cap \supp \nabla \phi \ne \emptyset$. Then, in view of    \eq{ell5cover},  there exists $r(x) \in \partial^{+}{\mathbb{A}}_{E}$  with $\Lambda_{\frac \ell 5}(x) \subset \Lambda_{\ell}(r(x))$.
Since  $x \in  \partial\up{3,9}\widehat{{\mathbb{A}}_{E}}\cap \Lambda_{\ell}(r)$ for some $r \in \G$ implies    $r \in \partial^{+}{\mathbb{A}}_{E}$, we conclude that $r(x) \in \partial^{+}{\mathbb{A}}_{E}$.
\end{proof}

\subsubsection{The  energy trap}

\begin{lemma} \label{lem1trap} 
Given a sufficiently large  scale $L $,  for each $x_{0}\in \R^{d}$   there exists an event $\cT_{L,x_{0}}$, with
\beq\label{cTdesired}
\cT_{L,x_{0}}\in \cF_{{\Lambda_{L, \frac L 2}(x_{0})}} \quad\text{and}\quad  \P\set{\cT_{L,x_{0}}} \ge  1 - L^{-{c}_{d,p,m,\abs{\I}} {\sqrt{L}} } ,
\eeq
such that for $\bom \in \cT_{L,x_{0}}$ we have 
\beq\label{sqrtdistsp}
W_{\bom,x_{0}}({E})\dist \pa{{E}, \sigma(H_{\bom,\Lambda_{L}(x_{0})})} \le  \e^{- \frac m {15} \sqrt{L} } \quad \text{for all}\quad {E} \in \I .
\eeq
In particular,  we conclude that, for   $\bom \in \cT_{L,x_{0}}$ and ${E} \in \I$,
\beq\label{sqrtdistsp2}
  W_{\bom,x_{0}}({E}) >   \e^{- \frac m {30} \sqrt{L} }  \quad \Longrightarrow \quad\dist \pa{{E}, \sigma(H_{\bom,\Lambda_{L}(x_{0})})} \le  \e^{- \frac m {30} \sqrt{L} } .\eeq
\end{lemma}

\begin{proof} Fix a scale $L$ and  $x_0 \in \R^d$.
 Since $\I$ is a bounded interval, we can find $\set{ E_{j}}_{j=1,2,\ldots,J}\subset \I$ such that
 \beq
 \I \subset \bigcup_{j=1}^{J} \br{E_{j}- \e^{-2m\sqrt{L}}, E_{j}+ \e^{-2m\sqrt{L}}} \quad \text{and} \quad J\le \e^{2m\sqrt{L}}\abs{\I}.
 \eeq
We   set
 \beq
 \cT_{L,x_{0}}= \bigcap_{j=1}^{J}\cY\up{E_{j}}_{x_{0},\frac L 2,\sqrt{L},\frac L 2}.
  \eeq
 The estimate \eq{cTdesired} follows immediately from  \eq{probYL2}. 
 
  Let $\bom \in\cT_{L,x_{0}}$ and ${E} \in \I$ with $\Theta_{\bom}({E})\not= \emptyset$.  Pick $j \in \set{1,2,\ldots,J}$ such that  we have 
  ${E}\in \br{E_{j}- \e^{-2m\sqrt{L}}, E_{j}+ \e^{-2m\sqrt{L}}}$, write  $\mathbb{A}_{E_j}={\mathbb{A}}_{E_j}(\bom)$,  and let $\phi={\phi}_{\bom,E_{j}}$ be the function given in Lemma~\ref{sublemphi}.  Let $\psi\in \Theta_{\bom}({E})$,  a generalized eigenfunction.  Then   $\phi \psi \in  \mathcal{D}(H_{\bom,\Lambda})$ and we have \eq{localeig}, where $\La=\Lambda_{L}(x_{0})$.  It follows that, for $\cL$ sufficiently large, 
  \begin{align}\label{phipsiest}
& \norm{ \pa{H_{\bom,\Lambda}-{E}}\phi \psi}^{2}= \norm{  W_\Lambda(\phi)\psi}^{2}= \sum_{\substack{x \in x_{0}+ \frac 12 \Z^{d}\\  \La_{\frac 1 2}(x)\cap \supp \nabla \phi \ne \emptyset }} \norm{ \Chi_{\La_{\frac 1 2}(x)} W_\Lambda(\phi)\psi}^{2}\\  \notag
& \qquad \quad \le C_{d, \I,V_{\mathrm{per}}}\!\!  \sum_{\substack{x \in x_{0}+ \frac 12 \Z^{d}\\  \La_{\frac 1 2}(x)\cap \supp \nabla \phi \ne \emptyset }} \norm{ \Chi_{x} \psi}^{2}\le C_{d, \I,V_{\mathrm{per}}}^{2} \e^{-\frac {2m^{\pr}} {11} \sqrt{L} }\! \!\sum_{\substack{x \in x_{0}+ \frac 12 \Z^{d}\\  \La_{\frac 1 2}(x)\cap \supp \nabla \phi \ne \emptyset }} \!\! \norm{\psi}^{2}_{\La_{\sqrt{L}}(r(x))}
\\ \notag & \qquad \quad  \le 
{C}^{\pr}_{d, \I,V_{\mathrm{per}}} L^{d}  \e^{-\frac {2m^{\pr}} {11}\sqrt{L}} \norm{\psi}_{\Lambda_{L, \frac L 2}(x_{0})}^{2} \le  \e^{-\frac {2m} {15}\sqrt{L}} \norm{T_{x_{0}}^{-1}\psi}^{2} , \end{align}
where we used \eq{Wphi} and \eq{nablaphi5}, 
 applied the interior estimate given in \eq{interior} as in the derivation of  \eq{usinginterior},  used Lemma~\ref{sublemphi} with $\ell=\sqrt{L}$  ($r(x) \in \partial^{+}{\mathbb{A}}_{E_j}$ is given in the lemma),   applied  Lemma~\ref{lemSLI}\ref{sublemEDII2}, using   \eq{EDI22} with $m^{\pr}$ as in \eq{mpr} taking $\ell=\sqrt{L}$, and then used  \eq{mpr}  to write the final estimate in terms of $m$. Since it follows from \eq{defphi5} that  $\norm{\phi \psi}\ge \norm{\psi}_{\Lambda_{ \frac L 2}(x_{0})}\ge \norm{\Chi_{x_{0}}\psi}$,  we conclude that
\beq \label{phipsiest2}
\dist \pa{{E}, \sigma(H_{\bom,\Lambda_{L}(x_{0})})} \le \frac{\norm{ \pa{H_{\bom,\Lambda}-{E}}\phi \psi}} {\norm{\phi \psi}}\le \e^{-\frac {m} {15}\sqrt{L}}  \ \frac {\norm{T_{x_{0}}^{-1}\psi}}{ \norm{\Chi_{x_{0}}\psi}}.
\eeq
The desired  \eq{sqrtdistsp} now follows using  \eq{defGWx},    and  it yields \eq{sqrtdistsp2} .
\end{proof} 
 
\subsubsection{The energy bootstrap}

We fix $b \ge 1$, let  $\widehat{n}=\widehat{n}(p)$ be as  in \eq{defhatm}, and set
\beq    
 \eta =\eta(p):=2^{\frac 1 {\widehat{n}}}-1 < p,  \qtx{so} \eta \in ]0,1] \qtx{and} (1 +\eta)^{ \widehat{n}}=2,
 \label{etaprop1} 
  \eeq

We now fix a scale $L$, let $\ell_{0}=\sqrt{L}$, and set  $\ell_{k}=\ell_{k-1}^{1+\eta}$ for $k=1,2,\ldots,\widehat{n}$, so   $\ell_{\widehat{n}}=L$ by \eq{etaprop1}. We take $J\in \N$, to be determined later, and let $L_{0}=L$, $L_{k}= L_{k-1} + 2 J \ell_{k}$ for $k=1,2,\ldots,\widehat{n}$. We have 
\beq\label{Lhatn}
L_{\widehat{n}}=L + 2J \sum_{k=1}^{\widehat{n}}\ell_{k}\le \pa{1+ 2J \widehat{n}} L.
\eeq

Given $x_{0}\in \R^{d}$ and  $E \in \I$, we consider the events $\widetilde{\cY}\up{E}_{x_{0},L_{k-1},\ell_{k},2J\ell_{k}}$, $k=1,2,\ldots,\widehat{n}$, defined similarly to the event in \eq{YLM}, but with a modified site percolation model:  a site $r$ is now either \emph{pgood} or \emph{pbad} according to whether the corresponding box $\La_{\ell_{k}}$ is $(\bom,E,m,\varsigma,\eta)$-pgood or not (see Definition~\ref{defpgood}), and the set $\widetilde{\mathbb{A}}_{E}(\bom)$, defined similarly to $ {\mathbb{A}}_{E}(\bom)$, is now a cluster of bad sites. Requiring $2J\ge 100^{d}$, Lemma~\ref{slemYLM} still applies, with  ${\widehat{p}}={\widehat{p}}(p):= \tfrac {p-\eta}{2(1 + \eta)}$   substituted for $p$ in view of Lemma~\ref{lemprobpgood}, yielding for all $k=1,2,\ldots,\widehat{n}$ the estimate
\begin{align}\notag
\P\set{\widetilde{\cY}\up{E}_{x_{0},L_{k-1},\ell_{k},2J\ell_{k}}}& \ge 1 -  \pa{\tfrac {4L_{k-1}}{\ell_{k}}}^{d-1} 3^{J d}\ell_{k}^{- 2c_{d} {\widehat{p}} J}\ge 1 -  \pa{\tfrac {4L_{\widehat{n}}}{\ell_{0}}}^{d-1} 3^{J d}\ell_{0}^{-2 c_{d} {\widehat{p}} J}\\
&
\ge 1-  \pa{ 4\pa{1+ 2J \widehat{n}}L^{\frac 12}}^{d-1} 3^{J d}L^{-  {c_{d} {\widehat{p}}}   J}\ge 1 - L^{-6b d},\label{probYLMx}
\end{align}
provided $J \ge  C_{d,p,b}$.  We fix $J_{{d,p,b}} :=\br{ \max \set{ C_{d,p,b},\frac { 100^{d}}2}}+1$ so if $J \ge J_{{d,p,b}}$ the estimate   \eq{probYLMx} holds for all $k=1,2,\ldots,\widehat{n}$ .

For each $k=0,1,\ldots,\widehat{n}-1$  the finite volume operator $H_{\La_{L_{k}}(x_{0}),\bom}$, which depends only on $\bom_{\La_{L_{k}}(x_{0})}$, is a nonnegative self-adjoint operator with discrete spectrum.   We let
$\set{E_{j}\up{k} (\bom_{\La_{L_{k}}(x_{0})})}_{j\in \N}$ be the enumeration of these eigenvalues given by the min-max principle,
 as in \eq{mima}.  Each $E_{j}\up{k}=E_{j}\up{k} (\bom_{\La_{L_{k}}(x_{0})})$ is a continuous function of $\bom_{\La_{L_{k}}(x_{0})}$.   We define events
\beq \label{Zk}
\cZ_{k}=\cZ_{\La_{L_{k}}(x_{0})}:= \bigcap_{j\in \N} \widetilde{\cY}\up{E_{j}\up{k-1}}_{x_{0},L_{k-1},\ell_{k},2J\ell_{k}} \in \cF_{\La_{L_{k}}(x_{0})} \quad \text{for} \quad k=1,2,\ldots,\widehat{n}.
\eeq
Note that  $\cZ_{k} \in \cF_{\La_{L_{k}}(x_{0})}$ since the event $\widetilde{\cY}\up{E}_{x_{0},L_{k-1},\ell_{k},2J\ell_{k}}$  is jointly measurable in $\pa{E, {\bom}_{{\Lambda_{L_{k}, L_{k-1}}(x_{0})}}}$ and each $E_{j}\up{k-1}$ is a measurable function of
$\bom_{\La_{L_{k-1}}(x_{0})}$.  Since general estimates yield (cf. \cite[Eq.~(A.7)]{GK5})
\beq\label{numbereig}
\tr \set{\Chi_{\I}\pa{H_{\La_{L}(x_{0}),\bom}}}\le  C_{d,\Vper,\sup \I} L^d  \quad \text{for all}\quad L \ge 10,
\eeq
it follows from \eq{YLM} and  \eq{probYLMx} that
\beq \label{probZk}
\P\set{\cZ_{k}}\ge  1 - L^{-4bd} \quad \text{for} \quad k=1,2,\ldots,\widehat{n}.
\eeq

\begin{lemma}\label{lembootstrap} Given a sufficiently large  scale $L $,  for each $x_{0}\in \R^{d}$   there exists an event  $ \cZ_{L,x_{0} }$, with
\beq\label{probZ}
 \cZ_{L,x_{0} }\in \cF_{\La_{L_{\widehat{n}}}(x_{0})} \quad \text{and} \quad \P\set{\cZ_{L,x_{0}}}\ge 1 -\widehat{n} L^{-4bd}\ge 1 - L^{-3bd},
\eeq
 such that for all $\bom \in \cZ_{L,x_{0} }$,  if ${E} \in \I$ satisfies\begin{gather}\label{distL1}
 \dist \pa{{E}, \sigma^{(\I)}(H_{\bom,\Lambda_{L}(x_{0})})} \le  \e^{-  \tfrac m {30} \sqrt{L} },\\
   \dist \pa{{E}, \R\setminus \I}> \e^{-  {\widehat{m} } \sqrt{L} }  ,    \label{distRI}\\
  W_{\bom,x_{0}}({E}) >   \e^{- \widehat{m} \sqrt{L} },\label{distLWE}
\end{gather} 
where $\widehat{m}:= \tfrac m {30^{ \widehat{n}+1 }}=30{M}$ (see  \eq{defhatm}), it follows that
\beq \label{distL2}
\dist \pa{{E}, \sigma\up{\I}(H_{\bom,\Lambda_{L_{\widehat{n}}}(x_{0})})} \le  \e^{-\widehat{m} {L_{\widehat{n}}} } .
\eeq
\end{lemma}

\begin{proof} Given $L$ and $x_{0}$, we set
\beq\label{cZ}
\cZ_{L,x_{0} }:= \bigcap_{k=1}^{\widehat{n}}\cZ_{k} ,
\eeq
so   \eq{probZ}   follows immediately from  \eq{probZk}.

Let $\widetilde{m}= \tfrac m {30}$. Given $\bom \in \cZ_{L,x_{0} }$ and ${E} \in \I$  satisfying \eq{distL1} and $ W_{\bom,x_{0}}({E}) >0$, we   pick $E^{\pr} \in \sigma\up{\I}(H_{\bom,\Lambda_{L}(x_{0})})$ such that  $\abs{{E} - E^{\pr}  } \le  \e^{- \widetilde{m} \sqrt{L} }$ and  $\psi\in \Theta_{\bom}({E})$.   We have   $\bom \in \widetilde{\cY}\up{E^{\pr} }_{x_{0},L,\ell_{1},2J\ell_{1}}$, so we let $\phi={\phi}_{\bom,E^{\pr} }$ be the function given in Lemma~\ref{sublemphi}.  Note that Lemma~\ref{sublemphi} applies as stated for the modified site percolation model, the only modification being that a box $\Lambda_{\ell_{1}}(r)$ with  $r \in \partial^{+} \widetilde{\mathbb{A}}_{E^{\pr} }(\bom)$ is now
     $(\bom,E^{\pr} ,m,\varsigma,\eta)$-pgood, and hence, using Lemma~\ref{lempggodtogood}, $\Lambda_{\ell_{1}}(r)$   is $(\bom,{E},\widetilde{\widetilde{m}},\varsigma)$-good, where 
     \beq \label{massM15}
\widetilde{\widetilde{m}}= \widetilde{m}\pa{1- C_{d,p,m}  \ell_{1}^{-\frac {\min\set{\vs,\eta}}{1 + \eta}} } .
\eeq
Proceeding as in \eq{phipsiest} and \eq{phipsiest2}, we get  ($L$ large)
\beq \label{phipsiest24}
\frac{\norm{ \pa{H_{\bom,\Lambda_{L_{1}}(x_{0})}-{E}}\phi \psi}} {\norm{\phi \psi}}\le  \e^{-   \frac {\widetilde{\widetilde{m}}} {13} \ell_{1} } \ \frac {\norm{T_{x_{0}}^{-1}\psi}}{ \norm{\Chi_{x_{0}}\psi}}\le  \e^{- \frac {\widetilde{m}} {15} \ell_{1} }  \ \frac {\norm{T_{x_{0}}^{-1}\psi}}{ \norm{\Chi_{x_{0}}\psi}} ,
\eeq
the  generalized eigenfunction $\psi$ being arbitrary, so we conclude  that
\beq \label{distL23}
\dist \pa{{E}, \sigma(H_{\bom,\Lambda_{L_{1}}(x_{0})})} \le  \e^{-\frac {\widetilde{m} }{15} {\ell_{1}} }\pa{W_{\bom,x_{0}}({E})}^{-1}.
\eeq
Since it follows from \eq{distLWE} that 
\beq
 W_{\bom,x_{0}}({E}) >    \e^{-\frac {\widetilde{m} }{30} {\ell_{1}} },
\eeq
we get, using also \eq{distRI}, that
\beq\label{sqrtdistsp234}
\dist \pa{{E}, \sigma\up{\I}(H_{\bom,\Lambda_{L_{1}}(x_{0})})} \le  \e^{-\frac {\widetilde{m} }{30} {\ell_{1}} }.
\eeq

Repeating the argument $\widehat{n}-1$ times we get \eq{distL2}.
\end{proof}

\subsubsection{Completing the proof of Proposition~\ref{propenergyred1}}

\begin{proof}[Proof of Proposition~\ref{propenergyred1}]
Given a scale $L$, let $\breve{L}$ be the unique scale such that $\breve{L}_{\widehat{n}}=L$
(see \eq{Lhatn}). We take $J\ge J_{d,p,b}$, so   $K=1 +2J\widehat{n}\ge K_{d,p,b}:= 1 +2J_{d,p,b}\widehat{n}$, and hence $\breve{L}\ge \frac L K $.  Recalling Lemmas~\ref{lem1trap} and \ref{lembootstrap}, we let
\beq
\cQ_{L,x_{0}}= \cT_{\breve{L},x_{0}} \cap  \cZ_{\breve{L},x_{0}} \in   \cF_{{\Lambda_{\breve{L}, \frac {\breve{L}} 2}(x_{0})}}\cap \cF_{\La_{\breve{L}_{\widehat{n}}}(x_{0})}\subset \cF_{\La_{L}(x_{0})},
\eeq
so
\beq
 \P\set{\cQ_{L,x_{0}}} \ge  1 - {\breve{L}}^{-{c}_{d,p,m,\abs{\I}} {\sqrt{{\breve{L}}}} }- {\breve{L}}^{-3d} \ge 1 - L^{-2bd}.
\eeq
Let $\bom \in\cQ_{L,x_{0}}$ and ${E} \in \I$ satisfying \eq{distEI}. It follows that
\beq
 W_{\bom,x_{0}}({E}) >   \e^{- \frac m {30} \sqrt{\breve{L}} }\quad\text{and}\quad  \dist \pa{{E}, \R\setminus \I}> \e^{-  \frac m {30} \sqrt{\breve{L}} }   
\eeq
so we conclude from  Lemma~\ref{lem1trap} that
\beq\label{sqrtdistsp244}
\dist \pa{{E}, \sigma\up{\I}(H_{\bom,{\Lambda_{\breve{L}}(x_{0})})}} \le  \e^{- \frac m {30} \sqrt{\breve{L}} }.
\eeq
Since \eq{sqrtdistsp244} is just \eq{distL1} at scale $\breve{L}$, and \eq{distEI} implies \eq{distRI} and \eq{distLWE} at scale $\breve{L}$, 
 Lemma~\ref{lembootstrap} now yields   \eq{distL2}  for the scale $\breve{L}$, which is the desired \eq{Ldist}.
\end{proof}

\subsection{The second spectral reduction}
If $p\le 1$ we need a second spectral reduction.

Given a scale $L$, we set $L_{n}=L^{{\rho}^{n}}$ for   $n=0,1,\ldots,n_{1}$  (note $L_{0}=L $, $L_{n_{1}}=L^{{\beta}}$), where   $\rho,n_{1}, \beta$ are as in Theorem~\ref{thmmainev}.

\begin{definition}
The  \emph{reduced spectrum} of
 the operator  $H_{\bom}$ in the box $\La_{L}(x_{0})$, in the energy interval $\I$, is given by\begin{align}\label{redspI}
& \sigma\up{\I,\mathrm{red}}\pa{H_ {\bom,\La_{L}(x_{0})}} := \\
&\qquad \qquad  \set{E \in \sigma\up{\I}\pa{H_{\bom, \La_{L}(x_{0})}}; \;  \dist \pa{E, \sigma\up{\I}(H_{\bom,\Lambda_{L_{n}}(x_{0})})} \le 2 \e^{-\widehat{m}L_{n} }, n=1,\ldots,n_{1}},
\notag
\end{align}
where  $\widehat{m}$ is given in \eq{defhatm}. 
\end{definition}

Note that the set $\set{(E,\bom); \ E \in  \sigma\up{\I,\mathrm{red}}\pa{H_{\bom,\La_{L}(x_{0})}}}$
 is jointly measurable in $\pa{E, {\bom}_{\Lambda_{L}(x_{0})} }$.

\begin{proposition}\label{prop2red} Let $b\ge 1$  and fix  $K \ge K_{d,p,b}$, where  $ K_{d,p,b}$ is  the constant of Proposition~\ref{propenergyred1}.  
Given a sufficiently large  scale $L $,  for each $x_{0}\in \R^{d}$   there exists an event $\cX_{L,x_{0}}$, with
\beq\label{cXdesired}
\cX_{L,x_{0}}\in \cF_{\Lambda_{L}(x_{0})} \quad\text{and}\quad  \P\set{\cX_{L,x_{0}} }\ge  1 - L^{- b {\beta}d} ,
\eeq
such that for all $\bom \in\cX_{L,x_{0}}$: 
\begin{enumerate}
\item
  If ${E} \in \I$ satifies
\beq \label{distEIred}
W_{\bom,x_{0}}({E}) >   \e^{- \widehat{m} \sqrt{\frac {L^{{\beta}}} K} }\quad \text{and}  \quad \dist \pa{{E}, \R\setminus \I} >  \e^{- \widehat{m} \sqrt{\frac {L^{{\beta}}} K} },
 \eeq 
it follows that
\beq\label{Ldistred}
\dist \pa{{E},  \sigma\up{\I,\mathrm{red}}\pa{H_ {\bom,\La_{L}(x_{0})}}} \le  \e^{-\widehat{m}L }.
\eeq
\item We have
\beq\label{sigmaredest}
\#  \sigma\up{\I,\mathrm{red}}\pa{H_{\bom,\La_{L}(x_{0})}} \le  C_{d,\Vper,\I,p,\rho,n_{1}}  L^{(n_{1}+1 )\beta d}.
\eeq
\end{enumerate}
\end{proposition}

The proof will use several lemmas.

\begin{lemma}\label{lemwidetildQ} Given a sufficiently large scale $L$ and $x_{0}\in \R^{d}$, consider the event
\beq\label{widetildeQ}
\widetilde{ \cQ}_{L,x_{0}}:=\bigcap_{n=0}^{n_{1}}   \cQ_{L_{n},x_{0}}\in \cF_{\La_{L}(x_{0})},
\eeq
where  $ \cQ_{L^{\pr},x_{0}}$ is the event given in Proposition~\ref {propenergyred1} at scale $L^{\pr}$.  Then
\beq \label{probwtQL}
\P\set{\widetilde{ \cQ}_{L,x_{0}} }>  1  - (n_{1}+1) L^{-2b{\beta}d}
\eeq 
Moreover, if $\bom \in \widetilde{ \cQ}_{L,x_{0}}$, we have \eq{Ldistred} for any    ${E} \in \I$  satisfying \eq{distEIred}. 
\end{lemma}

\begin{proof} The estimate \eq{probwtQL} follows immediately from \eq{widetildeQ}, \eq{cQdesired}, and \eq{prho2n1}. The second part of the lemma is an immediate consequence of Proposition~\ref {propenergyred1}. 
\end{proof}

Given  scales $L^{\pr}< L$ with $L^{\rho}< \frac {L - L^{\pr}} 7$ and $x_{0}\in \R$, we consider  the annulus $\Lambda_{L,L^{\pr}}=\Lambda_{L,L^{\pr}}(x_{0}) $. We  let
 $\cR_{n}=\{\Lambda_{L_{n}}(r)\}_{r \in R_{n}}$ denote the standard  $L_{n}$-covering  of  the annulus $\Lambda_{L,L^{\pr}}$ for $n=1,2,\ldots,n_{1}$ (see Section~\ref{seccovannulus}).   Given  $K_{2}\in \N$ (to be chosen later), we set
 \beq
 \cS_{\Lambda_{L,L^{\pr}}}:=\set{\bigcup_{r \in R^{\pr}_{n_{1}}} \Lambda_{3L_{n_{1}}}(r); \;\;
R^{\pr}_{n_{1}} \subset  R_{n_{1}} \ \text{with} \ \# R^{\pr}_{n_{1}} \le K_{2} }.
 \eeq
 Similarly to Definition~\ref{defnotsobad},
the annulus $ \Lambda_{L,L^{\pr}}$   is said to be $(\bom,E,K_2)$-notsobad   if   there exists a $(\bom,L,L^{\pr},E)$-singular set  $\Theta\in  \cS_{\Lambda_{L,L^{\pr}}}$: 
 for all $ x \in  \Lambda_{L,L^{\pr}}\setminus {\Theta}$ there is a
$(\bom,E,m,\varsigma)$-good  box $\Lambda_{L_{n}}(r)\in \cR_{n}$, for some $n\in \{1,\ldots,n_{1}\}$, with $ \Lambda_{\frac { L_{n}} 5}(x )  \cap \Lambda_{L,L^{\pr}}\subset \Lambda_{L_{n}}(r)$.    An event $\cN$ is  $(\Lambda_{L,L^{\pr}},E,K_2)$-notsobad if
$\cN \in \cF_{\Lambda_{L,L^{\pr}}}$ and 
the annulus $\Lambda_{L,L^{\pr}}$   is  $(\bom,E,K_2)$-notsobad for all  $\bom \in \cN$.    We  have the analogue of Lemma~\ref{lemast}: If $K_{2}\ge \widehat{K}_{2}= \widehat{K}_{2}(d,p,b)$, and $L\ge \widehat{L}= \widehat{L}(d,p,b, K_2)$, then for all $E\in \I$ there exists  a $( \Lambda_{L,L^{\pr}},E,K_2)$-notsobad
 event $\cN_{ \Lambda_{L,L^{\pr}}}\up{E}$
 with
\beq \label{notsobad23}
\P\{\cN_{ \Lambda_{L,L^{\pr}}}\up{E}\} > 1- L^{-5bd}.
\eeq
(The proof of  Lemma~\ref{lemast} applies since ${\rho}> \pa{1+p}^{-1}$.)  We fix $K_{2}= [\widehat{K}_{2}] +1$, so  \eq{notsobad23} holds for $L$ large,  and   set $\cN_{ \Lambda_{L,L^{\pr}}}\up{E}=\Omega$ if $E\notin \I$.
The event
$\cN_{ \Lambda_{L,L^{\pr}}}\up{E}$ is jointly measurable in $(E,\bom_{\Lambda_{L,L^{\pr}}})$, so 
\begin{align}
\cN_{ \Lambda_{L,L^{\pr}}}= \bigcap_{E \in \sigma \pa{H_{\La_{L^{\pr}},\bom}}} {\cN}\up{E}_{ \Lambda_{L,L^{\pr}}}\in \cF_{\La_{L}},
\end{align}
and it follows from \eq{notsobad23} and \eq{numbereig} that
\beq \label{notsobad234}
\P\set{\cN_{ \Lambda_{L,L^{\pr}}}} > 1-  C_{d,\Vper,\sup \I} L^{-4bd}.
\eeq

Given a box  $\La_{L}(x_{0})$,   we define
 ``multi-spectrum'' of the operator  $H_{\bom}$, in the energy interval $\I$, by
\beq  \label{Sigmak}
 \bSig_{H_{\bom}, L, x_{0}}\up{k}:= \prod_{n=k}^{n_{1}}\sigma\up{\I}\pa{H_{\bom,\La_{L_{n}}(x_{0})}} \qtx{for}    k=0,1,\dots, n_{1}.
\eeq
A ``multi-eigenvalue'' $\bE\up{k}=\set{E_{n}}_{n=k}^{n_{1}} \in  \bSig_{H_{\bom}, L, x_{0}}\up{k}$ will be called ``linked'' if  
\beq
\abs{E_{n}-E_{n^{\pr}}} \le 4 \e^{-  {\widehat{m}}  L_{\max\set{n,n^{\pr}}}} \quad \text{for all}\quad  n,n^{\pr}\in\set{k,k+1,\ldots,n_{1}}.
\eeq
The ``reduced multi-spectrum'' is then defined as
\beq\label{reducedsp}
\bSig_{H_{\bom}, L, x_{0}}\up{k, \mathrm{red}}:=\set{\bE\up{k} \in  \bSig_{H_{\bom}, L, x_{0}}\up{k}; \ \bE\up{k} \ \text{is linked}},  \quad  k=0,1,\dots, n_{1}.
\eeq

\begin{lemma} \label{lemNL} Given a (sufficiently large) scale $L$ and $x_{0 }\in \R^{d}$, consider the event
\begin{align}
\cN_{ L,x_{0}}:= {\cN}_{ \Lambda_{L_{0},L_{1}}(x_{0})}\cap\set{ \bigcap_{n=1} ^{n_{1}-1}  \set{{\cN}_{ \Lambda_{L_{n},L_{n+1}}(x_{0})}\cap {\cN}_{ \Lambda_{2L_{n},L_{n+1}}(x_{0})} }}.
\end{align}
Then $\cN_{ L,x_{0}} \in \cF_{\La_{L}(x_{0})}$ and
\beq \label{notsobad255}
\P\set{\cN_{ L,x_{0}}} > 1- C_{d,\Vper,\sup \I}\, n_{1} { L_{n_{1}-1}}^{-4bd}
\ge 1 -   C_{d,\Vper,\sup \I}\, n_{1} L^{-4b\frac{ {\beta}}{\rho}d}.
\eeq
Moreover, for all
$\bom \in \cN_{ L,x_{0}}$ we have
\beq\label{sigmared}
\# \sigma\up{\I,\mathrm{red}}\pa{H_{\bom,\La_{L}(x_{0})}}\le \# \bSig_{H_{\bom}, L, x_{0}}\up{0,\mathrm{red}}\le  C_{d,\Vper,\I,p,\rho,n_{1}}  L^{(n_{1}+1){\beta} d}.
\eeq
\end{lemma}

\begin{proof} We have $\cN_{ L,x_{0}} \in \cF_{\La_{L}(x_{0})}$ by construction. Since $\beta={\rho}^{n_{1}}$, the estimate \eq{notsobad255} follows immediately from \eq{notsobad234}.

The first inequality in \eq{sigmared} is obvious, we only need to estimate $\#  \bSig_{H_{\bom}, L,x_0}\up{0,\mathrm{red}}  $  for $\bom \in\cN_{ L,x_{0}}$.  We will write  $\Lambda_L= \Lambda_L(x_0)$, $\Lambda_{L_{n}}=\Lambda_{L_{n}}(x_0)$.

It follows from  \eq{numbereig} that
\beq\label{rho2red}
\# \bSig_{H_{\bom}, L,x_0}\up{n_{1},\mathrm{red}}= \# \sigma\up{\I}\pa{H_{\bom,\La_{L_{n_{1}}}}}
\le   C_{d,\Vper,\sup \I}\pa{ L_{n_{1}}}^d = C_{d,\Vper,\sup \I} L^{{\beta}d}.
\eeq
Let $k\in\set{1,2,\ldots,n_{1}}$.   We set $  \widetilde{L}_{k-1}= {L}_{k-1}$ and 
$  \widetilde{L}_{n}= \widehat{L}_{n}=2{L}_{n}$ for $n=k,k+1,\ldots, n_{1}-1$, and let   ${ \Lambda_{\widetilde{L}_{n},L_{n+1}}}= { \Lambda_{{L}_{n},L_{n+1}}(x_0)} $.
We take  $\bE\up{k}=\set{E_{n}}_{n=k}^{n_{1}} \in  \bSig_{H_{\bom}, L,x_0}\up{k,\mathrm{red}}$.   Since  $\bom \in\cN_{ L,x_{0}}$, we have
$\bom \in{ \bigcap_{n=k-1} ^{n_{1}-1}  {\cN}\up{E_n}_{ \Lambda_{\widetilde{L}_{n},L_{n+1}}} }$,
so   let $\Theta_{\bom,E_{n}}\in   \cS_{\Lambda_{\widetilde{L}_{n-1}, L_{n}} }$ be the corresponding  $(\bom,\widetilde{L}_{n-1},L_{n},E_{n})$-singular set for $n=k,k+1,\ldots, n_{1}$,  and set
\beq
\Theta_{\bom,\bE\up{k}}= \Lambda_{2 L_{n_{1}}}\cup \set{\cup_{n=k}^{n_{1}}\Theta_{\bom,E_{n}}}.
\eeq
We have 
\beq\label{absTheta}
\abs{\Theta_{\bom,\bE\up{k}}}\le \pa{2 L_{n_{1}}}^{d} +K_{2}\sum_
{n=k}^{n_{1}} \pa{3\widetilde{L}_{n-1}^{{\beta}}}^{d}\le 6^{d}(n_{1}-k+2)K_{2}L_{k-1}^{{\beta}d}.
\eeq

Given $k=1,2,\ldots, n_1$ and $\bE\up{k}\in   \bSig_{H_{\bom}, L,x_0}\up{k,\mathrm{red}}$, we set 
\beq
 \Sigma_{H_{\bom}, L,x_0}\up{k-1}(\bE\up{k})= \set{E \in \sigma\up{\I}\pa{H_{\bom,\La_{L_{k-1}}}};\;  (E, \bE\up{k}) \in   \bSig_{H_{\bom}, L,x_0}\up{k-1,\mathrm{red}}},
\eeq
and note that
\beq\label{Sigmak(k-1)}
\#  \bSig_{H_{\bom}, L,x_0}\up{k-1,\mathrm{red}}\le\pa {\max_{\bE\up{k}\in   \bSig_{H_{\bom}, L,x_0}\up{k,\mathrm{red}}} \#  \Sigma_{H_{\bom}, L,x_0}\up{k-1}(\bE\up{k})} \pa{ \#  \bSig_{H_{\bom}, L,x_0}\up{k,\mathrm{red}}}.
\eeq

We now fix $\bE\up{k}\in   \bSig_{H_{\bom}, L,x_0}\up{k,\mathrm{red}}$. 
Given $E \in \Sigma_{H_{\bom}, L,x_0}\up{k-1}(\bE\up{k})$, let $\psi_{E}$ be a normalized  eigenfunction of $H_{\bom,\La_{L_{k-1}}}$ corresponding to the eigenvalue $E$.
If $x \in \La_{L_{k-1}}
\setminus  \Theta_{\bom,\bE\up{k}}$, there exists $n \in \set{k,k+1,\ldots, n_{1}}$, $j\in \set{1,2,\ldots, n_{1}}$, and a $(\bom,E_{n},m,\varsigma)$-good  box
 $\Lambda_{\ell_{n,j}}\subset  \La_{L_{k-1}}$, where $\ell_{n,j}=\pa{\widetilde{L}_{n-1}}_{j} \ge L^{{\rho}^{n+j-1}}$, such that
$ \Lambda_{\frac { \ell_{n,j}} 5}(x )  \cap \La_{L_{k-1}}\subset \Lambda_{ \ell_{n,j}}$.  (This is ensured by our choice of   the $  \widetilde{L}_{n}$.)
Since $\abs{E-E_{n}}\le 4\e^{- {\widehat{m}}  L_{n}}$, it follows from Lemma~\ref{lemjgoodm1} that  the box $\Lambda_{\ell_{n,j}}$ is
$(\bom,E,\frac {\widehat{m}} 2,\varsigma)$-jgood,  and hence we get, proceeding as in
\eq{EDI22}, that
\beq\label{EDI2255}
\norm{\Chi_{x}\psi_{E}}\le \e^{- \frac{\widehat{m}}{25}\ell_{n,j}}    \le
\e^{- \frac{\widehat{m}}{25} L^{{\rho}^{2 n_{1}-1}} }, 
\eeq
so we conclude that
\beq
\norm{\Chi_{\Theta_{\bom,\bE\up{k}}}\psi_{E}}^{2}\ge 1 -  \e^{- \frac{2\widehat{m}}{25} L^{{\rho}^{2 n_{1}-1}} } L^{{\rho}^{k-1}d}\ge \tfrac 1 2  .
\eeq
Thus,
\begin{align}\notag
\#  \Sigma_{H_{\bom}, L,x_0}\up{k-1}(\bE\up{k})&\le 2 \tr \set{ \Chi_{\I}\pa{H_{\bom,\La_{L_{k-1}}}}   \Chi_{\Theta_{\bom,\bE\up{k}} }  }\le  2C_{d,\Vper,\I} \abs{\Theta_{\bom,\bE\up{k}}}\\
& \le C_{d,\Vper,\I,p,\rho,n_{1}} L_{k-1}^{{\beta}d}, \label{EDI225536}
\end{align}
where we used \cite[Lemma~A.4]{GK5} (as in \eq{N2bound})  and \eq{absTheta}.  

 In view of   \eq{rho2red} and  \eq{Sigmak(k-1)},  and recalling $\rho< 1$,  we get
\beq
\# \bSig_{H_{\bom}, L, x_{0}}\up{0,\mathrm{red}} \le \pa{C^\pr_{d,\Vper,\I,p,\rho,n_{1}}  L^{{\beta}d}}^{n_{1}+1}\le C^{\pr\pr}_{d,\Vper,\I,p,\rho,n_{1}}  L^{(n_{1}+1){\beta} d}.  \eeq
\end{proof}

We are now ready to prove Proposition~\ref{prop2red}.

\begin{proof}[Proof of Proposition~\ref{prop2red}]  Setting
\beq
\cX_{L,x_{0}}:=\widetilde{ \cQ}_{L,x_{0}} \cap \cN_{L,x_{0}},
\eeq
Proposition~\ref{prop2red} is an immediate consequence of Lemmas~\ref{lemwidetildQ} and \ref{lemNL}
\end{proof}

\subsection{Annuli of good boxes}

We are now ready to prove  Theorem~\ref{thmmainev}.

\begin{proposition}\label{propmainev}
Given a sufficiently large  scale $L $,  for each $x_{0}\in \R^{d}$   there exists an event $\cU_{L,x_{0}}$
 as in \eq{cUdesired},
such that  for all $\bom \in\cU_{L,x_{0}}$, if  ${E} \in \I_{L}$ satisfies \eq{distEIred-}, then every box $\La_{\frac L {100}}$ in  the standard  $\frac L {100}$-covering  of  the  annulus $  \Lambda_{L_{+},L_{-}}(x_{0})$ is  $(\bom,{E}, 70 \widehat{m} ,\varsigma)$-jgood.
 \end{proposition}

\begin{proof}
Given  $E \in \I$, we let  $\cM_{L,x_{0}}\up{E}$ be the event that
all the boxes in  the standard  $\frac L {100}$-covering  of  the  annulus $ \Lambda_{L_{+},L_{-}}= \Lambda_{L_{+},L_{-}}(x_{0})$ are $(\bom,E,m,\varsigma)$-good,  and set  $\cM_{L,x_{0}}\up{E}=\Omega$ if  $E \notin \I$.  The event is jointly measurable in $(E,\bom_{\Lambda_{L_{+},L_{-}}})$, and, using \eq{number22},
\beq \label{L100cov}
\P\set{\cM_{L,x_{0}}\up{E}} > 1- \pa{2002}^{d} \pa{100}^{pd} L^{-pd}\quad \text{for}\quad E \in \I.
\eeq
Setting
\beq\label{Mlx0}
\cM_{L,x_{0}}= \bigcap_{E \in \sigma\up{\I,\mathrm{red}}\pa{H_{\bom,\La_{L_{-}}(x_{0})}} } \cM_{L,x_{0}}\up{E} \in \cF_{\La_{L_{+}}(x_{0})},
\eeq
 it follows from \eq{sigmaredest} and \eq{L100cov} that
\beq\begin{split} \label{L100covprob}
\P\set{\cM_{L,x_{0}} }&> 1- C_{d,\Vper,\I,p,\rho,n_{1}}  \pa{2002}^{d} \pa{100}^{pd} L^{-pd}   L^{(n_{1}+1 )\beta d}\\
& \ge 1 -C^{\pr}_{d,\Vper,\I,p,\rho,n_{1}}  L^{-(p-(n_{1}+1 )\beta ) d}.
\end{split}\eeq

 We now require that $K$, fixed in  Proposition~\ref{propenergyred1} subject only to the condition $K \ge K_{d,p,b}$, is sufficiently large to ensure that,  given a scale $L$, if ${E} \in \I_{L}$ satisfies \eq{distEIred-}, then ${E}$ satisfies  \eq{distEIred} at scale $L_{-}$: 
  \beq
  \e^{-  M L^{\frac \beta 2} }
  \ge   \e^{- \widehat{m} \sqrt{\frac {L_{-}^{{\beta}}} K} }, \qtx{i.e.,}  K \ge 900 \pa{\tfrac {499}{500}}^{\beta}.
 \eeq
 
 We  introduce the event
 \beq \label{event2red}
 \cU_{L,x_{0}}= \cX_{L_{-},x_{0}} \cap  \cM_{L,x_{0}}  \in \cF_{\La_{L_{+}}(x_{0})},
 \eeq
where $ \cX_{L_{-},x_{0}}$ is the event given in Proposition~\ref {prop2red} with  $b =1 + \frac 1 \beta ( p-(n_{1}+1 )\beta )$ . It follows from \eq{cXdesired},     \eq{L100covprob} and \eq{prho2n1}  that 
 \beq \label{prob2red}
\P\set{\cU_{L,x_{0}} }>  1  -  L_{-}^{-b {\beta}d} -C^{\pr}_{d,\Vper,\I,p,\rho,n_{1}}  L^{-(p-(n_{1}+1 )\beta ) d} \ge 1 -  L^{-\widetilde{p}  d}.
\eeq 

Fix $\bom \in\cU_{L,x_{0}}$, and  let ${E} \in \I_{L}$ satisfy \eq{distEIred-}, so it follows
 that \eq{distEIred} holds at  scale $L_{-}$.  Proposition~\ref{prop2red} then gives \eq{Ldistred}  at scale $L_{-}$:
\beq\label{Ldistred2}
\dist \pa{{E},  \sigma\up{\I,\mathrm{red}}\pa{H_ {\bom,\La_{L_{-}}(x_{0})}}} \le e^{-\widehat{m} L_{-}}= \e^{-\frac {499\widehat{m} }{5}\frac L {100}}.
\eeq
Thus, given a box $\La_{\frac {L} {100}}$   in the 
standard  $\frac {L} {100}$-covering  of  the  annulus $  \La_{L_{+},L_{-}}(x_{0})$,  it follows from  \eq{Mlx0} that the box  $\La_{\frac {L} {100}}$ is $(\bom,E,m,\varsigma)$-good for all energies  $E \in \sigma\up{\I,\mathrm{red}}\pa{H_{\bom,\La_{L_{-}}(x_{0})}}$.  We conclude from
 \eq{Ldistred2} and Lemma~\ref{lemjgoodm1} that the box  $\La_{\frac {L} {100}}$ is  $(\bom,{E}, 70 \widehat{m} ,\varsigma)$-jgood. 
\end{proof}

\begin{proof}[Proof of Theorem~\ref{thmmainev}]
The theorem follows from Proposition~\ref{propmainev}, with $\cU_{L,x_{0}}$ the event given in Proposition~\ref{propmainev}.

We fix  $\bom \in\cU_{L,x_{0}}$ and
${E} \in \I_{L}$.   Recall $\vartheta= \frac  \beta 2$.

If  \eq{distEIred-} holds, Proposition~\ref{propmainev} guarantees that  every box $\La_{\frac L {100}}$ in  the standard  $\frac L {100}$-covering  of  the  annulus $  \Lambda_{L_{+},L_{-}}(x_{0})$ is  $(\bom,{E}, 70 \widehat{m} ,\varsigma)$-jgood, so it follows from  Lemma~\ref{lemgoodWL} that
\beq \label{WxLEgood754}
W_{\bom,x,L}({E})\le    \e^{-\frac {70 \widehat{m}} {2000} L} \le  \e^{-\frac {\widehat{m}} {30} L}=\e^{-{{M}}  L},
\eeq
proving \eq{EDI9L9}.

To prove \eq{EDI9L99}, note that  either  ${E}$ satisfies  \eq{distEIred-}, so we  have  \eq{EDI9L9}, and hence, recalling \eq{boundGW},
\beq \label{EDI9L999}
W_{\bom,x_{0}}({E}) W_{\bom,x_{0},L}({E}) \le 2^{\frac \nu 2} \e^{-  {{M}} L},
\eeq
or we have
\beq \label{distEIred-2}
W_{\bom,x_{0}}({E})\le   \e^{- {{M}}  L^{\vartheta} },
\eeq
so using  \eq{boundGWL} we get
\beq \label{EDI9L98}
W_{\bom,x_{0}}({E})W_{\bom,x_{0},L}({E}) \le  
2^{\frac \nu 2} L^{\nu} \e^{- { M} L^{\vartheta}  }\le   \e^{- \frac 1 2 { {M}}  L^{\vartheta}  }.
\eeq
The desired   \eq{EDI9L99} follows.
  \end{proof}

\begin{remark}\label{remp>1} If $p>1$, the proof of Theorem~\ref{thmmainev} is much simpler; it does not require the second energy reduction of Proposition~\ref{prop2red}.
 The event $\cM_{L,x_{0}}$ in \eq{Mlx0} is replaced by
\beq\label{Mlx01}
\widetilde{\cM}_{L,x_{0}}= \bigcap_{E \in \sigma\up{\I}\pa{H_{\La_{L}(x_{0}),\bom}} } \cM_{L,x_{0}}\up{E} \in \cF_{\La_{L_{+}}(x_{0})},
\eeq
so we have
\beq \label{L100covprob1}
\P\set{\widetilde{\cM}_{L,x_{0}} }> 1-   \pa{2002}^{d} \pa{100}^{pd} L^{-pd} C_{d,\Vper,\I} L^{  d}\ge 1 -C^{\pr}_{d,\Vper,\I}  L^{- (p-1)  d}.
\eeq 
The event $ \cU_{L,x_{0}}$ in \eq{event2red} is replaced by
 \beq \label{event2red1}
\widetilde{\cU}_{L,x_{0}}= \cQ_{L,x_{0}} \cap  \widetilde{\cM}_{L,x_{0}}  \in  \cF_{\La_{L_{+}}(x_{0})},
 \eeq
where $ \cQ_{L,x_{0}}$ is the event given in Proposition~\ref {propenergyred1}. It follows from \eq{cQdesired} and    \eq{L100covprob1}   that 
 \beq \label{prob2red1}
\P\set{\widetilde{\cU}_{L,x_{0}} }>  1  -  L^{-2 bd}  -C^{\pr}_{d,\Vper,\I}  L^{- (p-1)  d}
>  1   -C^{\pr\pr}_{d,\Vper,\I}  L^{- (p-1)  d},
\eeq 
choosing $b=1 + \frac{p-1}2$.
The proof of  Theorem~\ref{thmmainev}  then proceeds as before, with $\vartheta=1$ in  \eq{distEIred-} and  \eq{EDI9L99}.
\end{remark}

\begin{remark}\label{remp>3} If $p >3$, we can prove a modified version of  Theorem~\ref{thmmainev}, that does not require either Proposition~\ref {propenergyred1} or Proposition~\ref{prop2red}; it suffices to use Lemma~\ref{lem1trap}.  The   conditions ${E} \in \I_{L}$ and   \eq{distEIred-}  are  replaced by 
\beq\label{distEIred-rev} 
 {E} \in \I \qtx{and} W_{\bom,x_{0}}({E}) >   \e^{- \frac m {30} \sqrt{L} } .
\eeq
 We replace the event
$\cM_{L,x_{0}}\up{E}$ by $\widehat{\cM}_{L,x_{0}}\up{E}$, the event that all the boxes in  the standard  $\sqrt{L}$-covering  of  the  annulus $ \Lambda_{L_{+},L_{-}}(x_{0})$ are $(\bom,E,m,\varsigma)$-good,  and set  $\widehat{\cM}_{L,x_{0}}\up{E}=\Omega$ if  $E \notin \I$. We have
\beq \label{L100cov2}
\P\set{\widehat{\cM}_{L,x_{0}}\up{E}} > 1-20^{d} L^{\frac 1 2 d}L^{-\frac p 2 d}= 1-20^{d} L^{-\frac {p-1} 2 d}\quad \text{for}\quad E \in \I.
\eeq
We set
\beq\label{Mlx02}
\widehat{\cM}_{L,x_{0}}= \bigcap_{E \in \sigma\up{\I}\pa{H_{\La_{L}(x_{0}),\bom}} } \widehat{\cM}_{L,x_{0}}\up{E} \in \cF_{\La_{L_{+}}(x_{0})},
\eeq
so we have
\beq \label{L100covprob2}
\P\set{\widehat{\cM}_{L,x_{0}} }> 1-20^{d} L^{-\frac {p-1} 2 d}C_{d,\Vper,\I} L^{  d}\ge 1 - C_{d,\Vper,\I}^{\pr\pr}L^{-\frac {p-3} 2 d}.
\eeq 
The event $ \cU_{L,x_{0}}$ in \eq{event2red} is replaced by
 \beq \label{event2red2}
\widehat{\cU}_{L,x_{0}}= \cT_{L,x_{0}} \cap \widehat{\cM}_{L,x_{0}}  \in\cF_{\La_{L_{+}}(x_{0})},
 \eeq
 where $ \cT_{L,x_{0}}$ is the event in Lemma~\ref{lem1trap}. It follows from \eq{cTdesired} and    \eq{L100covprob2}   that 
 \beq \label{prob2red2}
\P\set{\widehat{\cU}_{L,x_{0}} }>  1  -   L^{-{c}_{d,p,m,\abs{\I}} {\sqrt{L}} }  -C^{\pr\pr}_{d,\Vper,\I} L^{-\frac {p-3} 2 d}\ge 1- L^{-\frac {p-3} 3 d} .
\eeq 
The proof of  Proposition~\ref{propmainev} then proceeds as before, except that we use  Lemma~\ref{lem1trap} and boxes of side $\sqrt{L}$ instead of $\frac L {100}$. We conclude, using Lemma~\ref{lemjgoodm1}, that if ${E}$ satisfies \eq{distEIred-rev}, then  all the boxes in  the standard  $\sqrt{L}$-covering  of  the  annulus $ \Lambda_{L_{+},L_{-}}(x_{0})$ are $(\bom,{E},\frac m {60},\varsigma)$-jgood.  Applying  Lemma~\ref{lemgoodWL}, modified for boxes of side $\sqrt{L}$ instead of $\frac L {100}$, we obtain (cf. \eq{EDI9L9})
\beq
W_{\bom,x_{0},L}({E})\le  \e^{- \frac m {1000} \sqrt{L} }. 
\eeq
It follows that we have (cf. \eq{EDI9L99})
\beq
W_{\bom,x_{0}}({E})W_{\bom,x_{0},L}({E})\le  \e^{- \frac m {1000} \sqrt{L} }. \qtx{for all} {E} \in \I.
\eeq
 This simpler result implies  pure point spectrum with sub-exponential decay of eigenfunctions, as well as dynamical localization.
\end{remark}

\section{Localization}\label{secloc}

In this section we derive all the usual forms of localization from Theorem~\ref{thmmainev}.
We will assume only the conclusions of this theorem.  More precisely, we will assume only the existence of the events $\cU_{L,0}$ satisfying the conclusions of Theorem~\ref{thmmainev} for  some fixed $\widetilde{p}, \vartheta, M$.  In particular, we do not assume the conclusions of the multiscale analysis, which were the hypotheses for Theorem~\ref{thmmainev}.
We fix $ {\nu} > \frac d 2$, which  will be generally omitted from the notation.

\subsection{Anderson localization and finite multiplicity of eigenvalues} \label{secMSAas1}

A simple Borel-Cantelli Lemma  argument based on Theorem~\ref{thmmainev}  yields Anderson localization and  finite multiplicity of eigenvalues. We only need the events of Theorem~\ref{thmmainev} at $x_{0}=0$.

\begin{theorem}\label{thmAndloc} Let $H_{\bom}$ be a generalized Anderson Hamiltonian on  $\L^{2}(\R^{d})$. 
 Let $ \I \subset \R$ be a bounded open interval, for which 
 there is a scale $\cL_{1}$ such that  for all $L\ge \cL_{1}$   there exists an event $\cU_{L,0}$ as in Theorem~\ref{thmmainev}.
Then the following holds with probability one:
  \begin{enumerate}
 \item  {$H_{\bom}$}  has  pure point spectrum  in the interval  {$\I$}.

\item   If   {$\psi$}  is an {eigenfunction} of  {$H_{\bom}$} 
with {eigenvalue}  {$E \in \I$}, then $\psi$ is exponentially localized  with rate of decay $ M$,   more precisely,
\begin{equation}\label{expdecay222}
 \|\Chi_x \psi\| \le C_{\bom,E}\norm{T^{-1} \psi}\, e^{- M\norm{x}} \qquad \text{for all}\quad  x \in \R^{d}.
\end{equation}

\item  For all  {$E \in \I$} we have
 \begin{equation}\label{expdecay2229}
 \|\Chi_x P_{\bom}(E)\|_{2} \le C^{\pr}_{\bom,E}\, e^{- M\norm{x}} \qquad \text{for all}\quad  x \in \R^{d}.
\end{equation}

\item  The {eigenvalues} of  {$H_{\bom}$} in  {$\I$} have {finite multiplicity}: 
\beq\label{finmult}
\tr P_{\bom}(E) < \infty  \quad\text{for all} \quad E\in  \I.
\eeq
\end{enumerate}
\end{theorem}

\begin{proof} It suffices to prove the theorem in every closed interval $I \subset \I$.  
We fix $I$, and  pick a scale $L_{0}\ge \cL_{1}$ such that $I \subset \I_{L_{0}}$ (see \eq{defIL}).  We introduce scales $L_{k+1}=2L_{k}$ for $k=1,2,\ldots.$, and  set $ \cU_{k} =\cU_{L_{k},0} $. It follows from the Borel-Cantelli Lemma, using \eq{cUdesired}, that 
\beq\label{cUinfty}
\P \set{\cU_{\infty}}=1, \quad \text{where}\quad \cU_{\infty}=\liminf_{k\to \infty}\cU_{k}.
\eeq

Fix $\bom \in \cU_{\infty}$; there exists $k_{\bom}\in \N$ such that $\bom\in \cU_{k}$  for all $k \ge k_{\bom}$.  If ${E}\in I $ is a generalized eigenvalue  of $H_{\bom}$, i.e., $\Theta_{\bom}({E})\not= \emptyset$, and hence $W_{\bom,0}({E}) >0$,  we set
\beq 
k_{\bom,{E}}= \min \set{k \in\N;\;  k \ge k_{\bom} \; \text{and  \eq{distEIred-} holds for ${E}$ and $L_{k}$ (with $x_{0}=0$)}}< \infty.  \label{kbomcE}
\eeq
Given $\psi \in \Theta_{\bom}({E})$, it follows from  \eq{EDI9L9} that  \beq
  \|\Chi_{0,L_{k}}\psi\|\le  \norm{T^{-1} \psi}\e^{- M L_{k}} \qtx{for all} k \ge k_{\bom,{E}}.
 \eeq
 If $x\in \R^{d}$ with   $\norm{x } \ge  L_{k_{\bom,{E}}} $, we can always find $k \ge k_{\bom,{E}}$ such that $x \in \bar{\La}_{L_{k+1},L_{k}}(0)$, so 
 \begin{equation}\label{expdecay22299}
 \|\Chi_x \psi\| \le   \|\Chi_{0,L_{k}}\psi\|\le  \norm{T^{-1} \psi}\e^{- M L_{k}}\le  \norm{T^{-1} \psi}\e^{- M \norm{x}}.
 \end{equation}
It follows that that $\psi \in \H= \L^{2}(\R^{d})$  and satisfies   \eq{expdecay222}. It now follows from \eq{geneigPgeneig} that  \eq{allgeneig} holds with $B=I$.  We conclude that that  $H_{\bom}$ has pure point spectrum in $I$, and if $\psi$ is an eigenfunction of $H_{\bom}$ with eigenvalue  ${E}\in I$ it has the exponential decay given in \eq{expdecay222}.

The estimate \eq{expdecay2229} is an immediate consequence of \eq{expdecay222}, and implies \eq{finmult}.
\end{proof}

\subsection{Eigenfunctions correlations  and  dynamical localization}\label{secMSAas2}

Another Borel-Cantelli Lemma  argument based on Theorem~\ref{thmmainev} yields eigenfunctions correlations.  In particular, we obtain pure point spectrum, finite multiplicity of eigenvalues,
SUDEC (summable uniform decay of eigenfunction correlations; see \cite{GKsudec}) and   SULE (semi-uniformly localized eigenfunctions; see \cite{DRJLS,GKboot,GKsudec}), and dynamical localization.   We will need the events of Theorem~\ref{thmmainev} for all $x\in \Z^{d}$.
  We do not assume or use Theorem~\ref{thmAndloc}.

\begin{theorem}\label{thmSUDEC} Let $H_{\bom}$ be a generalized Anderson Hamiltonian on  $\L^{2}(\R^{d})$.  
 Let $ \I \subset \R$ be a bounded open interval, for which 
 there is a scale $\cL_{2}$ such that  for all $L\ge \cL_{2}$ and  $x\in \Z^{d}$   there exists an event $\cU_{L,x}$ as in Theorem~\ref{thmmainev}. Let $\eps >0$ and fix an open interval  $I \subset \bar{I} \subset \I$. The following holds with probability one:
 \begin{enumerate}
\item  For all   ${E} \in I$ we have
  \beq\label{WSUDEC}
W_{\bom,x}({E})W_{\bom,y}({E})\le  C_{\bom,I,\eps}\e^{\norm{x}^{{(1 + \eps)\frac {\vartheta}{\widetilde{p}}}  } }
 \e^{-\frac 1 3 M\norm{x-y}^{{\vartheta}}} \quad\text{for all}\quad x, y \in \R^{d}.
\eeq

\item  {$H_{\bom}$}  has  pure point spectrum  in the interval  {$I$}. Moreover,  the {eigenvalues} of  {$H_{\bom}$} in  {$I$} have {finite multiplicity}.

\item (SUDEC)  For all $E\in I$ and 
 $\phi,\psi \in  \Ran P_{\bom}(E)$ we have
\beq\label{SUDECas}
\norm{\Chi_{x} \phi}\norm{\Chi_{y} \psi}\le C_{\bom,I,\eps}^{\pr}\norm{T^{-1} \phi}\norm{T^{-1} \psi} \e^{\norm{x}^{{(1 + \eps)\frac {\vartheta}{\widetilde{p}}}  } } \e^{-\frac 1 4 M \norm{x-y}^{{\vartheta}}} \qtx{for all}  x, y \in \R^{d}.
\eeq
In addition,  for all $E\in I$ we have
\beq \label{SUDECasP}
 \|\Chi_x P_{\bom}(E)\|_{2} \norm{\Chi_{y} P_{\bom}(E)}_{2}\le C_{\bom,I,\eps}^{\pr}\, \mu_{\bom}({{E}})\, \e^{\norm{x}^{{(1 + \eps)\frac {\vartheta}{\widetilde{p}}}  } } \e^{- \frac 1 4 M \norm{x-y}^{{\vartheta}}} \qtx{for all}  x, y \in \R^{d}.
 \eeq

\item  (SULE)   For all $E\in I$ there exist a center of localization $y_{\bom,E}\in \R^{d}$
for all eigenfunctions with eigenvalue $E$, i.e.,  for all $\phi \in  \Ran P_{\bom}(E)$ we have 
\beq\label{SULEas}
\norm{\Chi_{x} \phi}\le C_{\bom,I,\eps}^{\pr\pr}\norm{T^{-1} \phi}\e^{\norm{y_{\bom,E}}^{{(1 + \eps)\frac {\vartheta}{\widetilde{p}}}  }} \e^{- \frac 1 4 M \norm{x-y_{\bom,E}}^{{\vartheta}}} \qtx{for all}  x  \in \R^{d}.
\eeq
In addition,  for all $E\in I$ we have
\beq\label{SULEasP}
\norm{\Chi_{x}  P_{\bom}(E)}_{2}\le C_{\bom,I,\eps}^{\pr\pr}\sqrt{\mu_{\bom}({{E}})}\, \e^{\norm{y_{\bom,E}}^{{(1 + \eps)\frac {\vartheta}{\widetilde{p}}}  }} \e^{- \frac 1 4 M \norm{x-y_{\bom,E}}^{{\vartheta}}} \qtx{for all}  x  \in \R^{d}.
\eeq

\item  We have
\beq \label{NSULE}
N_{\bom,I}(L):= \sum_{E \in I; \, \norm{y_{\bom,E}}\le L  }\tr P_{\bom}(E) \le  C_{\bom,I,\eps}\,   L^{(1 +\eps)\frac {d}{\widetilde{p}}} \quad \text{for all}\quad L \ge 1.
\eeq

\end{enumerate}
  
\end{theorem}

\begin{proof} Fix  $\eps >0$. Given  $k \in \N$, we set $L_{k}=2^{k}$,
 and consider the event  
\beq
\cJ_{k}:=  \bigcap_{x \in \Z^{d}; \, \norm{x}^{1 + \eps} \le \tau  L_{k}^{ \widetilde{p} } } \cU_{L_{k},x} \, ,
\eeq
where $\cU_{L,x},
{M},\widetilde{p}, \beta$ are as in  Theorem~\ref{thmmainev}, and $\tau>0$ is a constant to be chosen later. It follows from \eq{cUdesired} that
\beq
\P\set{\cJ_{k}}\ge 1 - C_{d,M,\eps,\tau} L_{k}^{-\frac {\eps}{1+\eps}\widetilde{p} d}.
\eeq
 
 Applying the Borel-Cantelli Lemma we conclude that 
\beq\label{cUinfty2}
\P \set{\cJ_{\infty}}=1, \quad \text{where}\quad \cJ_{\infty}=\liminf_{k\to \infty}\cJ_{k}.
\eeq
Thus, for $\bom \in \cJ_{\infty}$
 there exists $k_{1}(\bom)\in \N$ such that $\bom \in  \cU_{L_{k},x}$ 
for all $k \ge  k_{1}(\bom)$ and $x\in \Z^{d}$ with $\norm{x}^{1 + \eps} \le \tau  L_{k}^{ \widetilde{p} }$. 

We now fix $\bom \in \cJ_{\infty}$ and an open interval
$I \subset \bar{I} \subset \I$. We set 
\beq 
k_{1}(\bom,I)=\min \set{k\in \N;\  k\ge  k_{1}(\bom), k \ge 2, \, I \subset \I_{L_{k}}},
\eeq
where $\I_{L}$ is defined in \eq{defIL}.
 Given $x \in \Z^{d}$, we define $k_{2}(x) \in \N$, $k_{2}(x) \ge 2$,  by
\beq
\tau L_{k_{2}(x)-1}^{\widetilde{p}}    < \norm{x}^{1 + \eps} \le\tau  L_{k_{2}(x)}^{\widetilde{p}},  
\eeq
when possible, and set  $k_{2}(x) =1$ otherwise.
We let $k_{3}(\bom,I,x)=\max \set{ k_{1}(\bom,I),k_{2}(x)}$; note that $k_{3}(\bom,I,x)\ge 2$.  It follows from \eq{EDI9L99}, using \eq{Wconv},  that 
 for  all   ${E}\in  I$  and $ y\in \R^{d}\setminus  {\La}_{L_{k_{3}(\bom,I,x)}}(x)$ we have 
\beq \label{preSUDEC}
W_{\bom,x}({E})W_{\bom,y}({E})\le  2^{\nu}\norm{x-y}^{\nu}
 \e^{- \frac 1 2 M \norm{x-y}^{{\vartheta}}}.
 \eeq
If $y\in  {\La}_{L_{k_{3}(\bom,I,x)}}(x)$, we have
\begin{align} \label{preSUDEC2}
W_{\bom,x}({E})W_{\bom,y}({E})& = W_{\bom,x}({E})W_{\bom,y}({E}) \e^{ \frac 1 2 M \norm{x-y}^{{\vartheta}}} \e^{-  \frac 1 2 M \norm{x-y}^{{\vartheta}}} \\ \notag 
&
\le 2^\nu \e^{ \frac 1 2 M \pa{\frac 1 2  {L_{k_{3}(\bom,I,x)}} }^{{\vartheta}}}  \e^{- \frac 1 2 M \norm{x-y}^{{\vartheta}}}   \le 2^\nu \e^{\frac 1 2 M  L_{k_{3}(\bom,I,x)-1}^{{\vartheta}}}  \e^{- \frac 1 2 M \norm{x-y}^{{\vartheta}}}  
\\ &
\le
\begin{cases} 
 2^\nu \mathrm{e}^{\norm{x}^{{(1 + \eps)\frac {\vartheta}{\widetilde{p}}}  } }
\e^{- \frac 1 2 M \norm{x-y}^{{\vartheta}}} & \quad  \text{if} \quad k_{3}(\bom,I,x)=k_2(x)\\
 2^\nu  \e^{\frac 1 2 M L_{k_{1}(\bom,I)-1}^{{\vartheta}}} 
\e^{- \frac 1 2 M \norm{x-y}^{{\vartheta}}} & \quad  \text{if} \quad k_{3}(\bom,I,x)= k_{1}(\bom,I)
\end{cases}   ,\notag
\end{align}
where we used  \eqref{boundGW} and made an appropriate chice of the constant $\tau$. 
The estimate \eq{WSUDEC} follows from \eq{preSUDEC} and \eq{preSUDEC2}.

It follows from \eq{WSUDEC} that for all  ${E}\in  I$ and  and all $\phi,\psi \in \widetilde{\Theta}_{\bom}({E})$ we have, for all $ x, y \in \R^{d}$,
\begin{align} \label{SUDECas2}
\norm{\Chi_{x} \phi}\norm{\Chi_{y} \psi}&\le C_{\bom,I,\eps}\norm{T_{x}^{-1} \phi}\norm{T_{y}^{-1} \psi} \e^{\norm{x}^{{(1 + \eps)\frac {\vartheta}{\widetilde{p}}}  }} \e^{-\frac 1 3 M  \norm{x-y}^{{\vartheta}}} \\ 
&\le 2^{\nu}C_{\bom,I,\eps}\scal{x}^{\nu}\scal{y}^{\nu}\norm{T^{-1} \phi}\norm{T^{-1} \psi} \e^{{(1 + \eps)\frac {\vartheta}{\widetilde{p}}}  } \e^{- \frac 1 3 M \norm{x-y}^{{\vartheta}}}.\notag
\end{align}
Thus  $\widetilde{\Theta}_{\bom}({E}) \subset \H$  for all  ${E}\in  I$.  It now follows from \eq{geneigPgeneig} that  \eq{allgeneig} holds with $B=I$, 
 and hence  $H_{\bom}$ has pure point spectrum in $I$.  The estimate \eq{SUDECas}  follows from \eq{SUDECas2}. The estimate \eq{SUDECasP} is an immediate consequence of \eq{SUDECas}, and implies  $\tr P_{\bom}(E) < \infty$ for all $E\in  I$.

Given $E\in I$ with  $P_{\bom}(E)\not= 0$, we pick $\psi \in  \Ran P_{\bom}(E)$, $\psi\not = 0$, and pick $y_{\bom,E}\in \Z^{d}$ (not unique) such that
\beq
\norm{\Chi_{y_{\bom,E}} \psi}= \max_{y\in \Z^{d}} \norm{\Chi_{y} \psi}.
\eeq
It follows that  (see \cite[Eq.~(4.22)]{GKsudec})
\begin{equation}\label{Tmaxphi}
W_{\bom,y_{\bom,E}}(E)\ge \frac {\norm{\Chi_{y_{\bom,E}} \psi}}  {\norm{T_{y_{\bom,E}}^{-1} \psi}} \ge C_{d }> 0 .
 \end{equation}
If $P_{\bom}(E)= 0$ we take $y_{\bom,E}\in \Z^{d}=0$.   Then for all $E \in I$ and  all  $\psi \in  \Ran P_{\bom}(E)$,  \eq{SULEas} and  \eq{SULEasP} follow from \eq{WSUDEC} (taking $y=y_{\bom,E}$) 
and \eq{Tmaxphi}.

To prove \eq{NSULE}, note that if follows from  \eq{SULEasP} that  for all $E \in I$ and $R \ge 1$ we have
\begin{align}
&\norm{\Chi_{\R^{d} \setminus \La_{2R}(y_{\bom,E})}  P_{\bom}(E)}_{2}^{2}\le \sum_{x \in  \Z^{d }\setminus  \La_{2R -1}(y_{\bom,E})} \norm{\Chi_{x}  P_{\bom}(E)}_{2}^{2}\\ 
& \qquad \qquad\qquad   \le  C_{\bom,I,\eps}\,\mu_{\bom}(I)\,\e^{2\norm{y_{\bom,E}}^{{(1 + \eps)\frac {\vartheta}{\widetilde{p}}}  }} \e^{-  \frac 1 2{M} \pa{\frac R 2 }^{{\vartheta}}}= C_{\bom,I,\eps}^{\pr}\e^{2\norm{y_{\bom,E}}^{{(1 + \eps)\frac {\vartheta}{\widetilde{p}}}  }} \e^{-  \frac 1 2{M} \pa{\frac R 2 }^{{\vartheta}}}.\notag
\end{align}
There is a constant $D_{\bom,I,\eps}\ge 1$ such that for all $L \ge 1$ 
\beq
R \ge D_{\bom,I,\eps} L^{\frac {1 +\eps}{\widetilde{p}}}\quad \Longrightarrow \quad C_{\bom,I,\eps}^{\pr}\e^{2L^{{(1 + \eps)\frac {\vartheta}{\widetilde{p}}}  }} \e^{-  \frac 1 2{M} \pa{\frac R 2 }^{{\vartheta}}} \le \tfrac 1 2 .
\eeq
Thus, given $L \ge 1$, letting  $R_{L}:= D_{\bom,I,\eps} L^{\frac {1 +\eps}{\widetilde{p}}}$,
we have
\beq
\norm{\Chi_{\La_{2R_{L}}(y_{\bom,E})}  P_{\bom}(E)}_{2}^{2} > \tfrac 1 2 \quad \text{whenever} \quad  \norm{y_{\bom,E}} \le L.
\eeq

It follows, using also \eq{trestmeas9988}, that
\begin{align}\label{NSULE2}
N_{\bom,I}(L)&\le  2 \sum_{\substack{E \in I\\ \norm{y_{\bom,E}}\le L}  }\norm{\Chi_{\La_{2R_{L}}(y_{\bom,E})}  P_{\bom}(E)}_{2}^{2}\\
&   \le   2  \sum_{\substack{E \in I\\ \norm{y_{\bom,E}}\le L}  }\norm{\Chi_{\La_{2(L+R_{L})}(0)}  P_{\bom}(E)}_{2}^{2}\le 2 \norm{\Chi_{\La_{2(L+R_{L})}(0)}  P_{\bom}(I)}_{2}^{2}  \notag \\ \notag
& \ \le C_{I} (L + R_{L})^{d}\le C_{\bom,I,\eps}\,  L^{\frac {(1 +\eps)d}{\widetilde{p}}}.
\end{align}
\end{proof}

 We can now prove dynamical  localization with probability one.

 \begin{corollary}\label{cordynloc}  Let $H_{\bom}$ be a generalized Anderson Hamiltonian satisfying the hypotheses of Theorem~\ref{thmSUDEC} in a bounded open interval $\I$.   Let $\eps >0$ and fix an open interval  $I \subset \bar{I} \subset \I$.  The following holds with probability one:
 
  \begin{enumerate}
  
  \item For all  $E \in I$ we have
  \beq\label{decaynorm-1}
\norm{\Chi_{y}P_{\bom}(E) \Chi_{x}}_{1} \le  C_{\bom,I,\eps}^{\pr}\, {\mu_{\bom}({{E}}) }\,\e^{\norm{x}^{{(1 + \eps)\frac {\vartheta}{\widetilde{p}}}  }} \e^{- \frac  1 4 M \norm{x-y}^{{\vartheta}}} \qtx{for all}   x, y \in \R^{d}. 
  \eeq

\item  We have
\beq\label{decaykernelas}
\sup_{f \in \cB_{b,1}}    \norm{\Chi_{y} f(H_{\bom})P_{\bom}(I)
\Chi_{x} }_{1} \le  C_{\bom,I,\eps}^{\pr\pr}  \e^{\norm{x}^{{(1 + \eps)\frac {\vartheta}{\widetilde{p}}}  }} \e^{-  \frac {1} 4 M  \norm{x-y}^{{\vartheta}}}\qtx{for all}   x, y \in \R^{d}.
\eeq

\item For all $b >0$ and $x_{0}\in \R^{d}$  we have 
\beq \label{dynlocas}
\sup_{f \in \cB_{b,1}}    \norm{\scal{X-x_{0}}^{bd} f(H_{\bom})P_{\bom}(I)
\Chi_{x_{0}} }_{1}       \le   C_{\bom,I,\eps,b} \, \e^{\norm{x_{0}}^{{(1 + \eps)\frac {\vartheta}{\widetilde{p}}}  }},
\eeq
and, in particular, 
\beq \label{dynloc2as}
\sup_{t \in \R}    \norm{\scal{X-x_{0}}^{bd} \e^{-itH_{\bom}} P_{\bom}(I)
\Chi_{x_{0}} }_{1} \le C_{\bom,I,\eps,b} \, \e^{\norm{x_{0}}^{{(1 + \eps)\frac {\vartheta}{\widetilde{p}}}  }}.
\eeq

\item For all  $E \in I$ we have
  \beq\label{decayFermip}
\norm{\Chi_{y}P_{\bom}\up{E} \Chi_{x}}_{1} \le  C_{\bom,I,\I,\eps}\,\e^{\norm{x}^{{(1 + \eps)\frac {\vartheta}{\widetilde{p}}}  }} \e^{- \frac {1} 4 M \norm{x-y}^{{\vartheta}}}\qtx{for all}   x, y \in \R^{d}.
  \eeq
\end{enumerate}

\end{corollary}

\begin{proof}  
Since
\begin{align}
\norm{\Chi_{x}P_{\bom}({E}) \Chi_{y} }_{1} &\le\norm{\Chi_{x}P_{\bom}({E})  }_{2}\norm{\Chi_{y}P_{\bom}({E}) }_{2},
\end{align}
\eq{decaynorm-1} follows immediately from  \eq{SUDECasP}.

Given $f \in \cB_{b,1}$, it follows from  \eq{PIexpas} and \eq{decaynorm-1} that
\begin{align}
 \norm{\Chi_{y} f(H_{\bom})P_{\bom}(I)
\Chi_{x} }_{1}&  \le \int_{I}\abs{f(E)} \norm{\Chi_{y}  {P}_{\bom,0}(E) \Chi_{x} }_{1}  \di \mu_{\bom}(E)
\\
&  \le C_{\bom,I,\eps}^{\pr} \, \mu_{\bom}(I) \e^{\norm{x}^{{(1 + \eps)\frac {\vartheta}{\widetilde{p}}}  }} \e^{-  \frac {1} 4 M  \norm{x-y}^{{\vartheta}}},  \notag
\end{align}
which is \eq{decaykernelas}.

Given  $b >0$ and $x_{0}\in \R^{d}$, \eq{dynlocas} and \eq{dynloc2as}  follow from \eq{decaykernelas}.

To prove \eq{decayFermip}, we proceed as in \cite[Proof of Theorem~3]{GKsudec}.
We write $I = ]\alpha_{1},\alpha_{2}[$, let $\delta =\frac 12  \mathrm{dist}( I, \R \setminus \I)>0$, and consider the open interval $I_{1}=]\alpha_{1} - \frac \delta 2,\alpha_{2} +\frac \delta 2[\subset \overline{I_{1}}\subset \I$.  We set $\zeta=  \frac 1 2( 1 +\frac {\beta} 2) \in  ]\frac {\beta} 2 ,1[$ and $\zeta^{\prime}= \frac 1 2( 1 +\zeta)  \in ]\zeta,1[$. We pick a
 ${\rm L}^1$-Gevrey function $g$  of class $\frac 1 {\zeta^\prime}$ on 
$ ]- 1,\infty[$, such that  $0\le g \le 1$,  $g\equiv1$ on $]-\infty,{\alpha_{1}}-\frac \delta 2]$
and $g\equiv 0$ on $[{\alpha_{2}}+\frac \delta 2,\infty[$. (See
\cite[Definition 1.1]{BGK}; such a function always exists.) 
 For all $E \in I$ we have
\beq\begin{split}\label{P=gf}
P_{\bom}^{(E)}&=  g^{2}(H_{\bom}) + f_E(H_{\bom}),\qtx{where} \\
f_E(t)& : = \Chi_{]-\infty,E]}(t) - g^{2}(t) =f_E(H_{\bom}) P_{\bom}(I_1)\in   \cB_{b,1}.
\end{split}\eeq
 Since we proved   \eqref{decaykernelas}, we have
\beq\label{decaykernelas35}
    \norm{\Chi_{y} f_{E}(H_{\bom})
\Chi_{x} }_{1} \le  C_{\bom,I,\eps}^{\pr\pr}  \e^{\norm{x}^{{(1 + \eps)\frac {\vartheta}{\widetilde{p}}}  }} \e^{-  \frac {1} 4 M  \norm{x-y}^{{\vartheta}}} \qtx{for all}  x, y \in \R^{d}. 
\eeq

The function $g$ was chosen so we can use
 \cite[Theorem 1.4]{BGK}, obtaining
 \begin{equation}\label{BGKest1}
\left\|\Chi_{x}g (H_{\bom}) \Chi_{y}\right\| \le C_{g}
\,\mathrm{e}^{-C^{\pr}_{g}\norm{x-y}^{\zeta
}}\qtx{for all}  x, y \in \R^{d}.
\end{equation}
We also have, using  \eq{trestmeas9988}, that
\begin{align}\label{estg1}
 \left\|\Chi_{x}g (H_{\bom}) \Chi_{y}\right\|_1& \le
 \left\|\Chi_{x}\sqrt{g }(H_{\bom})\right\|_2
 \left\|\Chi_{y}\sqrt{g} (H_{\bom})\right\|_2    \le \left\|\Chi_{x}P_{\bom}^{({\alpha_{2}}+\frac \delta 2)}\right\|_2
 \left\|\Chi_{y}P_{\bom}^{({\alpha_{2}}+\frac \delta 2)}\right\|_2 \\ &\le   C_{d,\norm{V^{-}_{\mathrm{per}}},{\alpha_{2}}+\frac \delta 2}^{2}. \notag
\end{align}
It follows that 
\begin{align}
 \left\|\Chi_{x}g (H_{\bom}) \Chi_{y}\right\|_2^{2} \le  \left\|\Chi_{x}g (H_{\bom}) \Chi_{y}\right\|  \left\|\Chi_{x}g (H_{\bom}) \Chi_{y}\right\|_1 \le C_{d,\norm{V^{-}_{\mathrm{per}}},{\alpha_{2}}+\frac \delta 2,g}  
\,\mathrm{e}^{-C^{\pr}_{g}\norm{x-y}^{\zeta}}.
\end{align}
Thus, given $x,y \in \R^{d}$ we get
\begin{align}\label{BGKest}
\left\|\Chi_{x}g^{2} (H_{\bom}) \Chi_{y}\right\|_1 &\le \sum_{z \in \Z^{d}} \norm{\Chi_{x}{g (H_{\bom})} \Chi_{z}}_{2}\norm{\Chi_{z} {g (H_{\bom})} \Chi_{y}}_{2}\\ \notag
& \le C_{1}\sum_{z \in \Z^{d}}  \mathrm{e}^{-\frac 1 2 C^{\pr}_{g}\norm{x-z}^{\zeta}} \mathrm{e}^{-\frac 1 2 C^{\pr}_{g}\norm{z-y}^{\zeta}}\le C_{2}\, {e}^{- C_3 \norm{x-y}^{\zeta}},
\end{align}
where $C_{1}=C_{d,\norm{V^{-}_{\mathrm{per}}},{\alpha_{2}}+\frac \delta 2,g}$ and $C_{2}, C_{3}$ depend only on   $d,\norm{V^{-}_{\mathrm{per}}},I,\I,\zeta$.

 Since   $\zeta>  \frac \beta 2$, the estimate \eqref{decayFermip}  now follows from \eq{P=gf}, \eq{decaykernelas35}, and \eqref{BGKest}.
    \end{proof}

\subsection{Localization in expectation}\label{secMSAexp}

  We will now derive eigenfunctions correlations estimates in expectation  from Theorem~\ref{thmmainev},  and use them to get dynamical localization in expectation, as well as pure point spectrum, finite multiplicity of eigenvalues, etc,  as in \cite{GKsudec}. We do not assume or use the results of Subections~\ref{secMSAas1} and \ref{secMSAas2}.

We recall that we pick $\nu > \frac d 2$, and that $\W_{\bom,x}({E})$ and $\W_{\bom,x,L}({E})$, defined in \eq{defGWx2} and \eq{defGWxL2}, are measurable functions of  $(\bom,{E})$ for each $x\in \R^d$, and satisfy \eq{WW<W}.  

 \begin{theorem}\label{thmsgenSUDEC}  Let $H_{\bom}$ be a generalized Anderson Hamiltonian on  $\L^{2}(\R^{d})$.  
 Let $ \I \subset \R$ be a bounded open interval, for which 
 there is a scale $\cL_{3}$ such that  for all $L\ge \cL_{3}$ and  $x\in \R^{d}$   there exists an event $\cU_{L,x}$ as in Theorem~\ref{thmmainev}.
Then the following holds for all  open intervals
 $I \subset \bar{I} \subset \I$:
\begin{enumerate}

\item
  For all $x,y \in \R^{d}$ we have
  \beq \label{genSUDEC2}
\E \set{\norm{\W_{\bom,x}({E})\W_{\bom,y}({E})}_{\L^\infty(I,\di \mu_{\bom}({E}))}} \le  C \scal{x-y}^{-{\widetilde{p}}d},
\eeq
with a constant $C=C_{d,\widetilde{p},\vartheta,M, \nu, \cL_{3}}$. 

\item For all $x_{0}\in \R^{d}$, $L \ge 1$,      and  $s \in ]0,\frac {\widetilde{p}  d}{\nu}[$ we have
\beq \label{genSUDEC}
\E \set{\norm{\W_{\bom,x_{0}}({E})\W_{\bom,x_{0},L}({E})}^{s}_{\L^\infty(I,\di \mu_{\bom}({E}))} }\le  C L^{-(\widetilde{p}d -s {{\nu}} )},
\eeq
 with a constant $C=C_{d,\widetilde{p},\vartheta,M, \nu, \cL_{3},s}$.
 
 \item For all $x_{0}\in \R^{d}$, $s \in ]0,\frac {\widetilde{p}  d}{\nu}[$ and $r\in [0,{\widetilde{p}}d -s {{\nu}} [$   we have, for $\P$-a.e.\  $\bom$,
\beq \label{sgenSUDEC99}
\norm{\W_{\bom,x_{0}}({E})
\W_{\bom,x_{0},2^{k-1}}({E})}_{\L^\infty(I,\di \mu_{\bom}({E}))}\le C_{\bom,I,s,r}   2^{-k \frac r s}\quad \text{for}\quad k=0,1,2,\ldots .
\eeq
As a consequence, $H_{\bom}$ has pure point spectrum in the  interval  $I$.

  \end{enumerate} 
 \end{theorem}

 \begin{remark} (ii) and (iii)  hold for any $s \in ]0,2{\widetilde{p}}[$, since in this case we can choose ${\nu}>\frac d 2 $ such that ${\widetilde{p}}d -s{{\nu}}>0$. 
  \end{remark}
  
    We  set
 \begin{align}
{\Chi_{x}\up{k}}={\Chi_{x,2^{k-1}}}\quad \text{and}\quad \W_{\bom,x}\up{k}({E})=\W_{\bom,x,2^{k-1}}({E})\quad \text{for}\quad k\in \N.
\end{align}
We also set 
$\Chi_{x}\up{0}= \Chi_{x}$ and $ \W_{\bom,x}\up{0}({E})=\W_{\bom,x}({E})$ for convenience. Note that
\beq
1 \le \sum_{k=0}^{\infty} \Chi_{x}\up{k}.
\eeq

 \begin{proof}[Proof of Theorem~\ref{thmsgenSUDEC}] We take $L$ sufficiently large to insure that $I \subset \I_{L}$ and we can apply Theorem~\ref{thmmainev}. We will  prove \eq{genSUDEC};  the correlation estimate \eq{genSUDEC2}  is proved in a similar way.
 In this case,
 applying \eq{WW<W}, \eq{EDI9L99}, \eq{boundGW2}, \eq{boundGWL2}, and \eq{cUdesired}, we have
 \begin{align} \notag
&\E \set{\norm{\W_{\bom,x_{0}}({E})\W_{\bom,x_{0},L}({E})}^{s}_{\L^\infty(I,\di \mu_{\bom}({E}))} }\le   \e^{- \frac s 2 { M} L^{{\vartheta}}} \P\set{\Omega } +
2^{s\nu}L^{s {\nu} } \P\set{\Omega\setminus \cU_{L,x_{0}} }\\ \label{genSUDEC3}
&\qquad \qquad \qquad \qquad \qquad  \le   \e^{- \frac s 2 { M} L^{{\vartheta}}} + 2^{s\nu} L^{s {\nu} }  L^{- {\widetilde{p}}d}\le (1+ 2^{{\widetilde{p}} d}) L^{-({\widetilde{p}}d -s {{\nu}} )}.
\end{align}
Using the bounds \eq{boundGW2} and \eq{boundGWL2} we get \eq{genSUDEC} for all $L\ge 1$.

Given  $r\in [0,\widetilde{p}d -s {{\nu}} [$, it follows from \eq{genSUDEC} that 
 \begin{equation}\label{sgenSUDEC89}
\E \set{\norm{ \sum_{k=0}^{\infty}  2^{ k r}
\left[\W_{\bom,x_{0}}({E})
\W_{\bom,x_{0}}\up{k}({E})\right]^s}_{\L^\infty(I,\di \mu_{\bom}({E}))}}
\le C_{d,{\nu},p,s,r} < \infty,
\end{equation}
 and \eq{sgenSUDEC99} is an immediate consequence of \eq{sgenSUDEC89} using the Borel-Cantelli Lemma.  Given $\bom$ for which \eq{sgenSUDEC99} holds  and  $\phi \in \H_{+}$, it follows, using \eq{defGWx2} and \eq{defGWxL2},
 that for
 $\mu_{\bom}$-a.e.\  ${E} \in I$ we have
   \begin{align}\label{SUDEC1}
\norm{ \Chi_{x_{0}} \Pb_{\bom}({E})\phi }\norm{ \Pb_{\bom}({E})\phi }& \le  \sum_{k=0}^{\infty} \set{\norm{ \Chi_{x_{0}} \Pb_{\bom}({E})\phi }\norm{\Chi_{x_{0},k} \Pb_{\bom}({E})\phi }}\\
&  \le  C_{\bom,I,s,r}\pa{1- 2^{- \frac r s}}^{-1} \norm{T_{x_{0}}^{-1}\Pb_{\bom}({E})\phi }^{2}< \infty .  \notag
\end{align}

If $\Pb_{\bom}({E})\phi \not =0$, we have $\norm{ \Chi_{x_{0}} \Pb_{\bom}({E})\phi }\not=0$ for some $x_{0}\in \R^{d}$, and hence  $\norm{ \Pb_{\bom}({E})\phi }<\infty$ by \eq{SUDEC1}, so we conclude that
$ \Pb_{\bom}({E})\phi  \in \H= \L^{2}(\R^{d})$. Thus we have \eq{allgeneig} with $B=I$, 
and we conclude that that  $H_{\bom}$ has pure point spectrum in $I$. \end{proof}

   Since  $H_{\bom}$, as in Theorem~\ref{thmsgenSUDEC},  has pure point spectrum in the interval $\I$ with probability one,   we might as well  work with 
eigenfunctions, not generalized eigenfunctions. We use the notation given in  \eq{defPEx01}.

 \begin{corollary}\label{corlstdynloc}   Let $H_{\bom}$ be a generalized Anderson Hamiltonian satisfying the hypotheses of Theorem~\ref{thmsgenSUDEC} in a bounded open interval $\I$.   Let
 $I \subset \bar{I} \subset \I$ be an open interva and  $s \in ]0,\frac {{\widetilde{p}} d}{\nu}[$.  Then  
  
  \begin{enumerate}
  
  \item For all $x_{0}\in \R^{d}$ and $L\ge 1$ we have
\begin{align} \label{decaynorm1}
 \E\set{\sup_{{E} \in I}\norm{ \Chi_{x_{0},L} P_{\bom}({E}) \Chi_{x_{0}} }_{1}^{s}}  \le C_{1} \E\set{\sup_{{E} \in I}\norm{ \Chi_{x_{0},L} P_{\bom,x_{0}}({E}) \Chi_{x_{0}} }_{1}^{s}} \le  C_{2} L^{-({\widetilde{p}}d -s{{\nu}} )},
\end{align} 
 with $C_{1}=C_{1,d,{\nu},\norm{V^{-}_{\mathrm{per}}},I,s}$ and  $C_{2}=C_{2,d,\norm{V^{-}_{\mathrm{per}}},\widetilde{p},\vartheta,M, \nu, \cL_{3},I,s}$. 
 \item We have 
\begin{align}\label{decayWP}
\E \left\{\sup_{{E} \in I}\pa{ \norm{\Chi_{x_{0}} P_{\bom}({E})}_{2}^{2}
\left(\tr P_{\bom}({E})\right) }
^{\frac s 2}\right\} 
< \infty,
\end{align}
 and hence 
 for $\P$-a.e.\  $\bom$ the eigenvalues of  $H_{\bom}$ in $I$
are of finite multiplicity.

\end{enumerate}

\end{corollary}

\begin{proof}  
Recalling \eq{xPEW} and \eq{xPEWL}, we have  
\begin{align}
\norm{\Chi_{x_{0},L}P_{\bom}({E}) \Chi_{x_{0}} }_{1}&\le   \norm{\Chi_{x_{0}}P_{\bom}({E})  }_{2}\norm{\Chi_{x_{0},L}P_{\bom}({E}) }_{2}\\  & \le \mu_{\bom,x_{0}}({E}) \W_{\bom,x_{0}}({E}) \W_{\bom,x_{0},L}({E}) ,\notag
\end{align}
and  \eq{decaynorm1} follows from   \eq{genSUDEC} and \eq
{trestmeas99}. 

In addition, we have
\begin{align}
\pa{\norm{\Chi_{x_{0}} P_{\bom}({E})}_{2}^{2}
\left(\tr P_{\bom}({E})\right) }^{\frac s 2}& \le  \sum_{k=0}^{\infty} \set{\norm{\Chi_{x_{0}} P_{\bom}({E})}_{2}\norm{\Chi_{x_{0}}\up{k} P_{\bom}({E})}_{2}}^{s}\\ \notag
&  \le \set{ \sum_{k=0}^{\infty}\set{ \W_{\bom,x_{0}}({E}) \W_{\bom,x_{0}}\up{k}({E})}^{s}} \set{\mu_{\bom,x_{0}}({I})}^{s},
\end{align}
so \eq{decayWP} follows from   \eq{genSUDEC}, and  \eq{trestmeas99}.

Since for $\P$-a.e.\  $\bom$ the operator $H_{\bom}$ has pure point spectrum in the  interval  $I$,   it follows from \eq{decayWP} that for $\P$-a.e.\  $\bom$ we have
\beq
\norm{\Chi_{x_{0}} P_{\bom}({E})}_{2}^{2}
\left(\tr P_{\bom}({E})\right) <\infty \quad  \text{for all} \quad {E} \in I,
\eeq
and hence, since $\Chi_{x_{0}} P_{\bom}({E})\not=0$ for some $x_{0}\in \R^{d}$ if $P_{\bom}({E})\not=0$, we have
$\tr P_{\bom}({E})<\infty $ for all ${E} \in I$.
\end{proof}

We can now prove dynamical  localization in expectation. 

\begin{corollary} \label{coroldynlocexp}  Let $H_{\bom}$ be a generalized Anderson Hamiltonian satisfying the hypotheses of Theorem~\ref{thmsgenSUDEC} in a bounded open interval $\I$. 
The following holds for all $x_{0}\in \R^{d}$ and    open intervals
 $I \subset \bar{I} \subset \I$:
\begin{enumerate}

\item  For all  $L \ge 1$      and  $s \in ]0,\frac {{\widetilde{p}}  d}{\nu}[$  we have
\begin{gather}\label{decaykernel}
\E\set{\sup_{f \in \cB_{b,1}}    \norm{\Chi_{x_{0},L}\ f(H_{\bom})P_{\bom}(I)
\Chi_{x_{0}} }_{1}^{s}      } \le  C L^{-({\widetilde{p}}d -s {{\nu}} )},\\
\label{decayFermipexp}
\E\set{\sup_{E \in I}\norm{\Chi_{x_{0},L}P_{\bom}\up{E} \Chi_{x_{0}}}_{1}^{s}} \le   C L^{-({\widetilde{p}} d -s{{\nu}} )}, 
\end{gather}

with $C=C_{d,\norm{V^{-}_{\mathrm{per}}},\widetilde{p},\vartheta,M, \nu, \cL_{3},I,s}$. 

\item Given $b >0$, for all 
 $s \in \left]0,\frac {{\widetilde{p}}}{ b +\frac 1 2  }\right[$  we have 
\begin{gather} \label{HSdynloc}
\E\set{\sup_{f \in \cB_{b,1}}    \norm{\scal{X-x_{0}}^{bd} f(H_{\bom})P_{\bom}(I)
\Chi_{x_{0}} }_{1}^{s}      } \le  C< \infty,\\
 \label{HSdynloc2}
\E\set{\sup_{t \in \R}    \norm{\scal{X-x_{0}}^{bd} \e^{-itH_{\bom}} P_{\bom}(I)
\Chi_{x_{0}} }_{1}^{s}      } \le  C <\infty,\\
\E\set{\sup_{E \in I}    \norm{\scal{X-x_{0}}^{bd}P_{\bom}\up{E}
\Chi_{x_{0}} }_{1}^{s}      } \le  C< \infty, \label{decayFermipexp5}
\end{gather}
with $C=C_{d,\norm{V^{-}_{\mathrm{per}}},\widetilde{p},\vartheta,M, \nu, \cL_{3},I,b,s}$.

\end{enumerate}
\end{corollary}

\begin{proof}  Given $f \in \cB_{b,1}$, it follows from  \eq{PIexpas} that
\begin{align}
 \norm{\Chi_{x_{0},L}f(H_{\bom})P_{\bom}(I)
\Chi_{x_{0}} }_{1}&  \le \int_{I}\abs{f({E})} \norm{\Chi_{x_{0},L}  {P}_{\bom,x_{0}}({E}) \Chi_{x_{0}} }_{1}  \di \mu_{\bom,x_{0}}({E})
\\
&  \le \sup_{E \in I} \norm{\Chi_{x_{0},L}P_{\bom,x_{0}}({E}) \Chi_{x_{0}}}_{1} \mu_{\bom,x_{0}}(I),  \notag
\end{align}
and hence \eq{decaykernel} is an immediate consequence of \eq{decaynorm1}.

The estimate \eq{decayFermipexp} is proven similarly to  \eq{decayFermip}.  We introduce 
 the decomposition $P_{\bom}^{(E)}=  g^{2}(H_{\bom}) + f_E(H_{\bom})$ 
  as in  \eq{P=gf}, and \eq{decayFermipexp} follows from   \eq{P=gf}, \eq{decaykernel}, and \eqref{BGKest}.

Given $b>0$ and  $s \in \left]0,\frac {{\widetilde{p}}}{ b +\frac 1 2  }\right[$, we pick $\nu > \frac d 2$ such that
 $s \in \left]0,\frac {{\widetilde{p}}}{ b +\frac \nu d  }\right[$.
Since
\begin{align}
 \norm{\scal{X-x_{0}}^{bd} f(H_{\bom})P_{\bom}(I) \Chi_{x_{0}} }_{1}\le C_{d,b}
 \sum_{k=0}^{\infty} 2^{kbd}  \norm{\Chi_{x_{0}}\up{k} f(H_{\bom})P_{\bom}(I)
\Chi_{x_{0}} }_{1},
\end{align}
 the estimate \eq{HSdynloc} follows from \eq{decaykernel};
  \eq{HSdynloc2} is a special case of \eq{HSdynloc}. Similarly,  \eq{decayFermipexp5} follows from \eq{decayFermipexp}.
 \end{proof}

\section{Log-H\" older continuity of the integrated density of states}\label{seclogH}
We will now assume that the conclusions of the multiscale analysis  (i.e., of Proposition~\ref{propA}) hold for all energies in   a bounded open interval $ \I $, and  prove log-H\" older continuity of the integrated density of states.

Given a generalized Anderson Hamiltonian $H_{\bom}$ and $x_0 \in \R^d$, we set
\beq
N_{x_0}(E) = \E \tr \set{\Chi_{x_0}P_{\bom}^{(E)} \Chi_{x_0}} \qtx{for} E \in \R.
\eeq

\begin{theorem}\label{logHolder} Let $H_{\bom}$ be a generalized Anderson Hamiltonian on  $\L^{2}(\R^{d})$.  
Consider  a bounded open interval $ \I \subset \R$,
 $m>0$,  $p>0 $, and   $ \vs \in]0,1[ $, and  assume there is a scale $\cL$ such that all scales $L\ge\cL$ are $(E,m,\varsigma, p)$-good for all energies $E \in \I$. Then, for all $0<\widehat{p}<p$,  closed interval $I\subset \I$ with length $\abs{I}\le \frac 1 2$, and $x_0 \in \R^d$, 
 we have
  \beq\label{Nxest}
\abs{N_{x_0}(E_{2})-N_{x_0}(E_{1})} \le \frac {C_{\widehat{p},I}} {\abs {\log \abs{E_{2}-E_{1}} }^{\widehat{p}d}} \qtx{for all} E_{1},E_{2} \in I.
\eeq
\end{theorem}

The proof of this theorem will use the 
  Helffer-Sj\"ostrand formula  (see  \cite[Section~2.2]{Da} and  \cite[Appendix B]{HuS} for details).   Given  $ g\in C^\infty(\R)$, $n\in \N$, and $a>0$, we define a quasi-analytic extension of $g$ of order $n$ by
  \begin{equation}\label{quasianalytic}
\tilde{g}_{n,a}(z) : = \left\{ \sum_{r=0}^n  \frac 1 {r!} g^{(r)}(u) (iv)^r \right\}
\xi\pa{ \frac {av} {\scal{u}}} ,
\end{equation}
where $z= u + i v$, $\scal{u}= (1 + \abs{u}^{2})^{\frac 1 2}$, and  $\xi \in C^\infty(\R)$  such that $0\le\xi \le 1$,  $\xi(u)=1$ if $\abs{u}\le 1$,   $\xi(u)=0$ if $\abs{u}\ge 2$. (We choose and fix $\xi$.) We set  $\di \tilde{g}_{n,a}(z) := \frac 1 {2\pi}\partial_{\bar{z}}\tilde{g}_{n,a}(z)
\,\mathrm{d} u\, \mathrm{d} v $, with $\partial_{\bar{z}}= \partial_u + i
\partial_v$, and $|\di \tilde{g}_{n,a}(z)| := \frac 1 {2\pi}
|\partial_{\,\overline{z}}\tilde{g}_{n,a}(z)| \,\mathrm{d} u\, \mathrm{d} v$.  Proceeding as in the derivation of  \cite[Eq.~(B.8)]{HuS}, we get, for all $n\in \N$, $a>0$, and $s \in [0,n]$,
\begin{equation}\label{HShigherorder}
  \int_{\R^{2}} \! |\di \tilde{g}_n(z)| \abs{\Im  z}^{-(s+1)}  \le C_{n,s}
  \hnorm{g}_{n,s,a}\le C_{n,s} \max\set{a^{s+1}, a^{s-n}} \hnorm{g}_{n}, 
\end{equation}
with
\beq\label{hnormg}
\hnorm{g}_{n,s,a}:=  \sum_{r=0}^{n+1}a^{-(r-s-1)} \int_{\mathbb{R}}\!\mathrm{d}u \scal{u}^{r-s -1}
  |g^{(r)}(u)|; \; \hnorm{g}_{n}=\hnorm{g}_{n,0,1}.
\eeq
  In particular, if  $ \hnorm{g}_n < \infty$,  then for any self-adjoint operator
$K$ and $a>0$ we have
\begin{equation}\label{HS}
 g (K) = \int_{\R^{2}} \!\di \tilde{g}_{n,a}(z) \, (K-z)^{-1} ,
\end{equation}
where the integral converges absolutely in operator norm.

\begin{remark} In the usual  Helffer-Sj\"ostrand formula there is no   parameter  $a$ in the definition of the quasi-analytic extension, i.e.,   $a=1$ in  \eq{quasianalytic} (e.g., \cite{Da,HuS}).  The proof of  Theorem~\ref{logHolder}   requires  the insertion of the  parameter  $a$ in \eq{quasianalytic}, which is then chosen  according to  the scale $L$--we will  need  $a\approx \e^{L^{1-\vs}}$.
\end{remark}

\begin{proof}[Proof of Theorem~\ref{logHolder}] Let $\eta \in ]0,p[$  and $I\subset \I$ be a closed interval with length $\abs{I}\le \frac 1 2$.
Without loss of generality we assume $\eta > \frac \vs {1-\vs}$.
We consider scales  $L\ge\cL$  such that $\dist (I,\R\setminus \I)>\frac 1 2\e^{- L^{1- \vs}}$.  Let $I_L\subset I $ be a closed interval of length $\abs{I_L}=\e^{- L^{1- \vs}}$, so it can be written as $I_L= \br{E - \frac 1 2\e^{- L^{1- \vs}}, E + \frac 1 2\e^{- L^{1- \vs}}}$ with $E \in I$. Set
$\widetilde{I}_L=  \br{E - \e^{- L^{1- \vs}}, E + \e^{- L^{1- \vs}}}\subset \I$.  We fix  $h_L \in C^\infty (\R)$,  $0\le h_L\le 1$, such that 
\beq \label{hL}
 \supp h_L \subset \widetilde{I}_L,\;   h_L \Chi_{I_L}=\Chi_{I_L},\; \text{and}\;  
\abs{h_L\up{j}}\le  C_d\, \e^{ j L^{1- \vs}}\; \text{for}\; j=1,2,\ldots, d+2,
\eeq
with  $C_d$  a constant independent of $L$. 

Given  $x_0\in \R^d$, 
we let $\cY_{L}=\cY_{L,x_0}$ be the event that the box $\La_L=\Lambda_L(x_0)$
is  $(\bom,E,m,\varsigma,\eta)$-pgood (cf. Definition~\ref{defpgood}).  Since all large scales $L\ge\cL$ are $(E,m,\varsigma, p)$-good by hypothesis,  we have, using \eq{probxgood} and \eq{trestmeas}, that
\begin{align}\label{PtohL}
 \E \tr \set{\Chi_{x_0}P_{\bom}(I_L)\Chi_{x_0}}&\le  \E \tr \set{\Chi_{x_0} h_L (H_{\bom})\Chi_{x_0}}\\ \notag
 &  \le  \E\set{ \tr \set{\Chi_{x_0} h_L (H_{\bom})\Chi_{x_0}; \, \cY_{L}} }+  C_{\I} L^{-\frac {p-\eta}{1+\eta}d}, 
 \end{align}
with a constant  $C_{\I}=C_{d,{\nu},\norm{V^{-}_{\mathrm{per}}},\sup \I}$.

If $\bom \in\cY_{L}$, $\Lambda_L$ is $(\bom,E,{M_{1}},\varsigma)$-good by Lemma~\ref{lempggodtogood}  (with $M_1$ given in  \eq{massM1}), and hence   $h_L (H_{\bom,\La_L})=0$.  Thus, 
\begin{align}\label{preHS}
\tr \set{\Chi_{x_0} h_L (H_{\bom})\Chi_{x_0}}= \tr \set{\Chi_{x_0} h_L (H_{\bom})\Chi_{x_0} - \Chi_{x_0} h_L (H_{\bom,\La_L})\Chi_{x_0}} \qtx{for} \bom \in \cY_L.
\end{align}

The  right-hand-side of \eq{preHS} may now be estimated by the 
  Helffer-Sj\"ostrand formula.  We apply the  Helffer-Sj\"ostrand formula to $h_L(H_{\bom})$ and  $h_L (H_{\bom,\La_L})$, with $a\ge 1$ in \eq{quasianalytic} to be chosen later depending on $L$.  We take $\phi_0 \in C_c^\infty(\R)$, such that $0\le \phi_0 \le1$,  $\phi_0=1$ on  $\Lambda_{\frac L 2 }(x_0)$, and $\supp \phi_0 \subset \Lambda_{\frac L 2  + 10}(x_0)$. We have, with $n\in \N$ to be chosen later (we omit $n$ and $a$ from the notation),
\begin{align}\label{TLambda}
T_{\bom}^{L}& =T_{\bom}^{L,x_0}:=\Chi_{x_0} h_L (H_{\bom})\Chi_{x_0} - \Chi_{x_0} h_L (H_{\bom,\La_L})\Chi_{x_0}\\ \notag
&= \int_{\R^{2}} \!\di\widetilde{h_L}(z) \set{ \Chi_{x_0} R_{\bom} (z) \Chi_{x_0} -  \Chi_{x_0} R_{\bom,\La_L}(z)  \Chi_{x_0} } \\ \notag
& =\int_{\R^{2}} \!\di\widetilde{h_L}(z) \set{ \Chi_{x_0} R_{\bom} (z) \phi_0\Chi_{x_0} -  \Chi_{x_0} \phi_0 R_{\bom,\La_L}(z)  \Chi_{x_0} }  \\ 
& =   \int_{\R^{2}} \!\di\widetilde{h_L}(z)  \set{ \Chi_{x_0} R_{\bom} (z)W(\phi_0)R_{\bom,\La_L}(z)  \Chi_{x_0} }  , \notag
\end{align}
where we used the geometric resolvent identity as in \eq{geometricresolvent}.

We now pick functions $\phi_i \in C_{c}^\infty(\R)$, $i=1,2,\ldots,2 k -1$, where $k\in \N$ will be chosen later, such that
$0\le \phi_i\le 1$, $\phi_{i}=1$ on $\supp{\nabla \phi_{i-1}}$, and  $\supp \phi_i \subset \La_{\frac L 2  + 50,\frac L 2  - 50 }(x_0)$.
  Using the resolvent identity $2k-1$ times,  noticing $\phi_i \Chi_{x_0}=0$ for $i=1,2,\ldots,2 k -1$, and writing  $\Chi_{\nabla \phi}= \Chi_{\supp \nabla \phi}$,  we get
\begin{align}\label{expres}
&{\Chi_{x_0}}R_{\bom}(z)W(\phi_0)
 =
{\Chi_{x_0}} R_{\bom}(z)W(\phi_{2k-1})R_{\bom}(z)W(\phi_{2k-2})\ldots  R_{\bom}(z)W(\phi_1)R_{\bom}(z)W(\phi_0)\\  
\notag 
&\qquad \qquad  = \set{{\Chi_{x_0}} R_{\bom}(z)} \set{W(\phi_{2k-1})R_{\bom}(z)W(\phi_{2k-2})} \set{\Chi_{\nabla \phi_{2k-2}} R_{\bom}} \\ \notag
&\qquad  \qquad  \qquad \qquad  \qquad  \times \set{W(\phi_{2k-3})R_{\bom}(z)W(\phi_{2k-4})}
\ldots \set{\Chi_{\nabla \phi_{2}} R_{\bom}(z)} \set{W(\phi_{1})R_{\bom}(z)W(\phi_{0})}.
\end{align}

Given $\phi\in C_{c}^\infty(\R)$, it follows from  \eq{Wphi} that for all $\bom \in \Omega$
\beq \label{Wphibound}
\norm{\pa{H_{\bom} +1}^{-\frac 12 } W(\phi)}=\norm{W(\phi)\pa{H_{\bom} +1}^{-\frac 12 } } \le C_\phi:=C_1 \pa{\norm{\Delta \phi}_\infty + \norm{\nabla \phi}_\infty},
\eeq
where $C_1=C_{d,\norm{V^{-}_{\mathrm{per}}}}$.  Moreover,  for all $x \in \R^d $ we have 
\beq \label{normpd}
\norm {\Chi_x {\pa{H_{\bom} +1}^{-1}}}_{k_d}\le C_2 < \infty \quad \text{with} \quad k_d = [\tfrac d 2] + 1,
\eeq
the constant $C_2=C_{d,\norm{V_{\mathrm{per}}},U_+}$ being independent of $x$  (cf.\  \cite[Eqs. (130)-(136)]{KKS}).  We have
\beq \label{R2compact}
\norm{\pa{H_{\bom}+1}R_{\bom}(z)}\le 1 + \frac {1 +\abs{z}}{\abs{\Im z}}\le  2 + \frac {1 +\abs{\Re z}}{\abs{\Im z}} ,
\eeq
Using \eq{Wphibound}, \eq{R2compact} and \eq{normpd}, we have
\beq  \label{Wphibound2}
 \norm{W(\phi_{i})R_{\bom}(z)W(\phi_{i-1})}\le C_{\phi_i}  C_{\phi_{i-1}} \pa{2 + \frac {1 +\abs{\Re z}}{\abs{\Im z}}}
\eeq
and,  for  all measurable sets $\Xi \subset \Lambda_L$, we get 
\beq \label{normpd2}
\norm{\Chi_{\Xi} R_{\bom}(z)}_{k_d}\le  C_2 \pa{2 + \frac {1 +\abs{\Re z}}{\abs{\Im z}}}
 L^d.
\eeq

We now take $k=k_d$ as in \eq{normpd}, and note that we can choose the functions
 $\phi_i \in C_{c}^\infty(\R)$, $i=0,1,\ldots,2 k_d-1$ so that   all  $C_{\phi_i}\le C_3=C_{d,\norm{V^{-}_{\mathrm{per}}}}$, a constant independent of $\Lambda_L$,  From \eq{expres}, \eq{Wphibound2} and \eq{normpd2}, we get, for all $\bom \in \Omega$,
\begin{align}\label{breakest}
\norm{{\Chi_{x_0}}R_{\bom}(z)W(\phi_0)R_{\bom,\La_L}(z){\Chi_{x_0}} }_1 \
 \le  C_4 L^{d k_d} \pa{2 + \frac {1 +\abs{\Re z}}{\abs{\Im z}}}
^{2k_d}\norm{\chi_{\nabla \phi_{0}}R_{\bom,\La_L}(z){\Chi_{x_0}}}, 
\end{align}
with a constant $C_4=  \pa{C_2 C_3^2}^{kd}=C_{d,\norm{V_{\mathrm{per}}},U_+}$.

We can now estimate  $\tr T_{\bom}^{L}$. First, note that, with $\I_+ =\sup \I < \infty$, 
\beq\label{supphL1}
\supp \widetilde{h_L} \subset \set{z=u + iv; \; u \in \widetilde{I}_L, \, \abs{v} \le \tfrac 2 a \scal{u}}\subset  \widetilde{I}_L + i \br{- \tfrac 2 a\scal{\I_+},  \tfrac 2 a\scal{\I_+}},
\eeq
and hence (recall $a \ge 1$)
\begin{equation} \label{supphL}
 \pa{2 + \frac {1 +\abs{\Re z}}{\abs{\Im z}} }\le  \pa{ \frac {\tfrac 4 a\scal{\I_+}+1 +\I_+}{\abs{\Im z}} }\le \frac { C_{\I_+}} {\abs{\Im z}}  \qtx{for all} z \in\supp \widetilde{h_L},
\end{equation}
with $C_{\I_+}=5
(1 + \I_+)$.  Combining \eq{TLambda}, \eq{breakest},  \eq{supphL}, and \eq{HShigherorder}, and using the fact that  $ \hnorm{g}_{n}$ in \eq{hnormg} is monotone increasing in $n$, we get
\begin{align}\label{trTtL}
\abs{\tr T_{\bom}^{L}}&\le  \int_{\R^{2}} \abs{\di\widetilde{h_L}(z)}  \norm{ \Chi_{x_0} R_{\bom} (z)W(\phi_0)R_{\bom,\La_L}(z)  \Chi_{x_0} }_1 \\ \notag
& \le  C_4  C_{\I_+}^{2k_d} L^{d k_d}  \int_{\R^{2}} \abs{\di\widetilde{h_L}(z)} {\abs{\Im z}^{-2k_d}} \norm{\chi_{\nabla \phi_{0}}R_{\bom,\La_L}(z){\Chi_{x_0}}}\\  \notag
&\le  C_4  C_{\I_+}^{2k_d} L^{d k_d} a^{2k_d}\hnorm{h_L}_{2k_d-1}\set{\max_{z \in \supp \widetilde{h_L}} \norm{\chi_{\nabla \phi_{0}}R_{\bom,\La_L}(z){\Chi_{x_0}}}}\\ \notag
& \le  C_4 C_d C_{\I_+}^{d+2} L^{\frac d 2 (d+2)} a^{d+2}\hnorm{h_L}_{d+1}\set{\max_{z \in \supp \widetilde{h_L}} \norm{\chi_{\nabla \phi_{0}}R_{\bom,\La_L}(z){\Chi_{x_0}}}}.
\end{align}
In view of \eq{hL} and \eq{hnormg}, we have
\beq
\hnorm{h_L}_{d+1}\le C_{d,\I_+} \abs{\widetilde{I}_L} \e^{ (d+2)L^{1- \vs}}= 2C_{d,\I_+}\e^{ (d+1)L^{1- \vs}} \qtx{for all} \bom \in \Omega.
\eeq

We now ready to estimate the quantity in  \eq{preHS}.   We choose $a=  { 2 \scal{\I_+}}  \e^{ L^{1- \vs}}$, so it follows from \eq{supphL1} that
\beq
\supp_{z \in \supp \widetilde{h_L}} \abs{z-E}\le \e^{ -L^{1- \vs}}+ \tfrac 2a \scal{\I_+}\le 2\e^{ -L^{1- \vs}} .
\eeq
Since $\eta > \frac \vs {1-\vs}$, we may  take $L$ large enough to ensure $ 2\e^{ -L^{1- \vs}} <  \e^{-m L^{\frac 1{1 + \eta}}}$, so  Lemma~\ref{lempggodtogood} guarantees that, for large $L$, for all $\bom \in \cY_L$ the box $\Lambda_{L}$ is   $(\bom,z,\frac m 2,\varsigma)$-good  for all $z \in \supp \widetilde{h_L}$. Thus, for large $L$,
 \beq\label{pgoodimplication}
 \max_{z \in \supp \widetilde{h_L}} \norm{\chi_{\nabla \phi_{0}}R_{\bom,\La_L}(z){\Chi_{x_0}}}\le \pa{\tfrac L 2 +11}^d\e^{-\frac m 2  \frac {L} 4 }\le \e^{-\frac m {10}  L }.
 \eeq
It follows from  \eq{preHS}, \eq{TLambda},\eq {trTtL}, and \eq{pgoodimplication} that for all  $\bom \in \cY_L$ we have, again taking $L$ large,
\begin{align} \label{YLest}
& \tr \set{\Chi_{x_0} h_L (H_{\bom})\Chi_{x_0}} \\  & \qquad \qquad  \le
   C_4 C_d C_{\I_+}^{d+2}  L^{\frac d 2 (d+2)}\pa{2C_{d,\I_+}\e^{ (d+1)L^{1- \vs}}}\pa{ { 2 \scal{\I_+}\e^{ L^{1- \vs}}}}^{d+2} \e^{-\frac m {10}  L } \le \e^{-\frac m {20}  L }  . \notag
\end{align}

 Combining \eq{PtohL} and \eq{YLest},  we get, for large $L$,
 \beq
 \E \tr \set{\Chi_{x_0}P_{\bom}(I_L)\Chi_{x_0}} \le \e^{-\frac m {20}  L } +  C_{\I}  L^{-\frac {p-\eta}{1+\eta}d}\le 2 C_{\I} L^{-\frac {p-\eta}{1+\eta}d}.
 \eeq
In particular, for all intervals $J \subset I$ with sufficiently small length $\abs{J}$, we have
\beq
 \E \tr \set{\Chi_{x_0}P_{\bom}(J)\Chi_{x_0}}  \le  2 C_{\I}  \abs{\log \abs{J}}^{-\frac {p-\eta}{(1+\eta)(1-\vs)}d}.
 \eeq 
The estimate \eq{Nxest} follows.
\end{proof}

\begin{remark}
The proof of Theorem~\ref{logHolder} uses  the pgood boxes  of Definition~\ref{defpgood} because  we need Lemma~\ref{lempggodtogood}.   It does not suffice to use good boxes.
\end{remark}

\appendix

\section{A quantitative unique continuation principle for Schr\"odinger operators}\label{appendixQUP}

In this appendix we rewrite Bourgain and Kenig's quantitative unique continuation principle
for Schr\"odinger operators \cite{BK} in a convenient form for our purposes.
We also give an application of this quantitative unique continuation principle to periodic Schr\"odinger operators, giving an alternative proof to Combes, Hislop and Klopp's  lower bound estimate for spectral projections  \cite{CHK1}.

We use the norm 
$\abs{x}:=\pa{\sum_{j=1}^{d} \abs{x_{j}}^{2}}^{\frac 12 }$ for $x=\pa{x_1,x_2,\ldots, x_d} \in \R^d$; all distances in $\R^{d}$ will be measured with respect to this norm. Given $x\in \R^{d}$ and  $\delta>0$, we set 
$B(x,\delta):= \set{y \in \R^{d}; \, \abs{y-x}<\delta}$  and   $B(x,\delta)^\ast:= B((x,\delta)\setminus \set{x}$.   Given subsets $A$ and $B$ of $\R^{d}$, and  a  function $\vphi$ on the  set $B$,  we set $\vphi_{A}:=\vphi \Chi_{A\cap B}$.  In particular, given $x\in   \R^{d}$ and $\delta >0$ we write $\vphi_{x,\delta}: =\vphi_{B(x,\frac \delta 2)}$.

We also set
\beq
 C_1= \e^{\int_0^1 \frac {1 - \e^{-t}} t \di t} ; \quad \text{note} \quad 2< \e^{\frac 3 4}< C_1 < \e<3.  \label{C1}
 \eeq

\subsection{The  quantitative unique continuation principle}

The following theorem is our version of  \cite[Lemma~3.10]{BK}.

\begin{theorem}\label{thmucp} Let  $G$ be an  open subset  of $\R^d$. Let 
$\psi \in\mathrm{H}^2(G)$ and ${\zeta} \in \L^2(G)$  be  real-valued functions satisfying
\beq \label{eq}
-\Delta {\psi} +V{\psi}={\zeta}  \quad \text{a.e.\  on $G$},
\eeq
where $V$ is a real measurable function on G with $\norm{V}_{\infty} \le K <\infty$.   
 Fix $\delta, D_{0}, D
 $   such that   $0<  \frac\delta 4 \le D_{0}\le 
 D$.    There exists a constant $m= m(d,\delta,D_0) >0$ such that, given
  a measurable set   ${\Theta} \subset G$ with $ \diam {\Theta} \le D$,  and  $x \in G$ such that
  \begin{gather} 
  R:= \dist \pa{x, {\Theta}} \ge D
  \label{xR}
 \quad  \text{and} \quad B(x, 4C_1 R+ 2D_0)\subset G,
  \end{gather}
where $C_1$ is the constant in \eq{C1}, we have
\begin{align} \label{UCPbound}
  (1 +K) \norm{{\psi_{x,\delta}}}^2_2   +
\norm{{\zeta_{G}}}_2^2 \ge
  R^{-m\pa{1 + K^{\frac 2 3} + \log  \pa{ {\norm{\psi_{G}}_{2}}{ \norm{\psi_{\Theta}}_2^{-1}}}}R^{\frac 4 3}}\norm{\psi_{{\Theta}}}_2 ^2.
   \end{align}
\end{theorem}

If the open set $G$ is bounded, the second condition in \eq{xR}  restricts  the application of  Theorem~\ref{thmucp} to  sites $x\in G$ sufficiently far away from the boundary of $G$. When $G$ is a box $\Lambda$, and \eq{eq} holds on $\Lambda$ with  either Dirichlet or periodic boundary condition, Theorem~\ref{thmucp} can be extended to sites $x\in \La$ near the boundary of $\La$
as in the following corollary.

\begin{corollary}\label{corQUCPD}  Consider the Schr\"odinger operator $H_\Lambda:= -\Delta_\Lambda +V $ on $\mathrm{L}^2(\Lambda)$, where  $\Lambda= \Lambda_L(x_0)= x_0 + ] -\frac L 2, \frac L 2[^d$, the open box of side $L>0$ centered at $x_0 \in \R^d$,  $\Delta_\Lambda$ is the Laplacian with either Dirichlet or periodic boundary condition on $\Lambda$, and   $V$ a is bounded potential on $\Lambda$ with  $\|V\|_\infty \le K<\infty$. Let    $\psi \in \D(\Delta_\Lambda)$. 

\begin{enumerate}
\item \label{corQUCPDi}
 Fix $\delta,  D
 $ such that   $ 0< \frac\delta 4 \le 
 D$, There exists a constant $\widetilde{m}=\tilde{m}(d,\delta,D) >0$ such that, given
 a measurable set   ${\Theta} \subset \Lambda$ with $ \diam {\Theta} \le D$, and 
 $x \in \Lambda$  such that
  \beq  \label{xR1}
B(x,\tfrac \delta 2)
 \subset \Lambda \qtx{and}  R:= \dist \pa{x, {\Theta}} \ge D,
  \eeq
we have
  \begin{align} \label{UCPbound1}
  (1 +K) \norm{\psi_{x,\delta}}^2_2   +\pa{29 \sqrt{d}}^{d} 
\norm{\pa{H_{\Lambda} \psi}_{\Lambda}}
_2^2 \ge
  R^{-\widetilde{m}\pa{1 + K^{\frac 2 3} + \log \pa{\norm{\psi_{\Lambda}}_{2} { \norm{\psi_{{\Theta}}}_2^{-1}}}}R^{\frac 4 3}}\norm{\psi_{{\Theta}}}_2 ^2.
   \end{align}

   \item \label{corQUCPDii} Let $L \ge 2$ and $0<\delta \le L$.  Then there exists a constant $\widehat{m}=\widehat{m}(d,\delta) >0$ such that for all  $x \in \Lambda$ with $B(x,\tfrac \delta 2) \subset \Lambda$ we have 
     \begin{align} \label{UCPbound12}
  (1 +K)  \norm{\psi_{x,\delta}}^2_2  +\pa{41}^{d} 
\norm{\pa{H_{\Lambda} \psi}_{\Lambda}}_{2}^{2} \ge
  L^{-\widehat{m} \pa{1 + K^{\frac 2 3} }L^{\frac 4 3}}\norm{\psi_{\Lambda} }_2 ^2.
   \end{align}
   \end{enumerate}
\end{corollary}

We will prove Theorem~\ref{thmucp}  from Bourgain and Kenig's  Carleman-type inequality estimate \cite[Lemma~3.15]{BK}, which we state in the next lemma.

\begin{lemma}\label{lemCTE}  Consider the function  $w(x)=\vphi(\abs{x})$ on $\R^d$, where
\beq
 \vphi(s):= \e^{-\int_0^s \frac {1 - \e^{-t}} t \, \di t} s \qquad \text{for} \quad s \in [0,\infty[,
 \eeq
 is a strictly increasing continuous function on $[0,\infty[$, $C^\infty$ on $]0,\infty[$.  In particular, we have 
\begin{gather}
\tfrac 1  {C_1} \abs{x} \le w(x) \le \abs{x}   \quad\text{for all }\quad  x \in B(0,1),
 \end{gather}
 where $C_1$ is the constant in \eq{C1}.
Then there are positive  finite constants $C_2$ and $C_3$, depending only on $d$, such that for all 
 $\alpha\ge  C_2$  and all real valued functions $f\in C_{\mathrm{c}}^\infty(B(0,1)^\ast)$ we have
\begin{equation} \label{carl}
\alpha^3 \int_{\R^d} w^{-1-2\alpha} f^2 \, \di x  \le C_3 \int_{\R^d} w^{2-2\alpha} (\Delta f)^2  \, \di x.
\end{equation}
\end{lemma}

We refer to \cite{BK} for the proof.  We shall use Lemma~\ref{lemCTE} with a function $f$ that is not necessarily smooth, but    $f \in \mathrm{H}^2_{\mathrm{loc}}$.  However in our case $f$ is compactly supported away from zero, and thus we can use the following extension  of Lemma~\ref{lemCTE}.

\begin{lemma}\label{lemCarl2} Let  $f\in \mathrm{H}^2(B(0,1))$, real valued with $\supp f \subset B(0,1)^\ast$.
Then \eq{carl} holds for all 
 $\alpha\ge  C_2$.
\end{lemma}

\begin{proof}
The proof follows from Lemma~\ref{lemCTE} from an approximation argument. Let $f$ be as in the lemma, and pick $h\in C_{\mathrm{c}}^\infty(\R)$ with $\int h(t) \, \di t = 1$, and set $h_\eta(t):=\eta^{-d}h(\frac t \eta)$.
Note that for $\eta$ small enough we have  $f_\eta:=f\ast h_\eta \in C_{\mathrm{c}}^\infty(B(0,1)^\ast)$. Thus, for such $\eta$'s, Lemma~\ref{lemCTE} applies to $f_\eta$. Then, as $\eta$ goes to zero, $f\eta$ converges to $f$ in $\mathrm{L}^2(\R^d)$ and $\Delta f_\eta=(\Delta f) \ast h_\eta$ to $\Delta f$ in $\mathrm{L}^2(\R^d)$. Since $w^{-1}$ is bounded above and below on $B(0,R)\setminus B(0,\delta)$ for any $\delta >0$, the lemma follows.
\end{proof}

We now rewrite these lemmas as follows.

\begin{lemma}  Given  $\vrho >0$, there exists a function $w_\vrho (x)=\vphi_\vrho(\abs{x})$ on $\R^d$, where $\vphi_\vrho$ is a  strictly  increasing continuous real-valued  function  on $[0,\infty[$, $C^\infty$ on $]0,\infty[$, such that
\beq\label{wvrho}
\tfrac 1 {C_1 \vrho} \abs{x} \le w_\vrho(x) \le \tfrac 1 { \vrho}\abs{x}   \quad\text{for all  $x \in B(0,\vrho)$},
\eeq
and for all   $\alpha \ge C_2$ and  all real valued functions  $f\in \mathrm{H}^2( B(0,\vrho))$ with $\supp f \subset B(0,\vrho)^\ast$ we have
\begin{equation} \label{carlvrho}
\alpha^3 \int_{\R^d}w_\vrho^{-1-2\alpha} f^2  \, \di x \le C_3 \vrho^4 \int_{\R^d}  w_\vrho^{2-2\alpha} (\Delta f)^2  \, \di x,
\end{equation}
where $C_1, C_2, C_3$ are  the constants of Lemma~\ref{lemCTE}.
\end{lemma}

\begin{proof}  \eq{carlvrho} follows from \eq{carl} by a change of variables, with
$w_\vrho(x)= w\pa{\tfrac 1 { \vrho} x}$.
\end{proof}

We are ready to prove Theorem~\ref{thmucp}  and Corollary~\ref{corQUCPD}.

\begin{proof}[Proof of Theorem~\ref{thmucp}] Without loss of generality we assume
\beq
\norm{\psi_{{\Theta}}}_2 =1.\label{01}
\eeq

Let $x_0 \in G$  satisfy \eq{xR} with $R:= \dist \pa{x_0, {\Theta}}$, and set $A := 4 C_{1} >4$.     
 For convenience we may assume $x_0=0$, in which case $  {\Theta}\subset B(0, A R)$,  and  take $G= B(0, A R +2{D_0} )$.  

Let us consider a function $\eta \in C^\infty_{\mathrm{c}}(\R^d)$ given by $\eta(x)= \xi(\abs{x})$, where $\xi $  is an even  $C^\infty$ function  on $\R$ such that
\begin{align}
0& \le \xi(s) \le 1  \quad \text{for all $s \in \R$},\\
\xi(s) &= 0 \quad \text{if either $\abs{s} \le \tfrac \delta 8$ or  $\abs{s} \ge AR +{D_0}$},\\
\xi(s) &= 1 \quad \text{if $\tfrac \delta 4 \le \abs{s} \le  AR$},\\
\abs{\xi^{(j)}(s)} & \le C_4\quad \text{for all $s \in \R$, $j=1,2$},
\end{align}
where $C_4=C_4(d,\delta,{D_0})$ is a finite constant (independent of $A$ and $R$).  Note that
$\abs{\nabla \eta} \le C_4 \sqrt{d}$ and $\abs{\Delta \eta} \le C_4 d$.

We now apply \eq{carlvrho} to the function $\eta \psi$ with $\vrho= 2 AR$. Given  $\alpha \ge C_2>1$ (without loss of generality we take  $C_2 >1$),   we get
\begin{align} \label{longest}
&\frac {\alpha^3}{3 C_3 \vrho^4} \int_{\R^d} w_\vrho^{-1-2\alpha} \eta^2 {\psi}^2  \, \di x \le \tfrac 1 3  \int_{\R^d}  w_\vrho^{2-2\alpha} (\Delta (\eta  {\psi}))^2  \, \di x
\\ & \quad   \le \int_{\R^d}  w_\vrho^{2-2\alpha}\eta^2 (\Delta  {\psi})^2  \, \di x  + 4 \int_{\supp \nabla \eta} w_\vrho^{2-2\alpha} \abs{\nabla \eta}^2 \abs{\nabla  {\psi}}^2  \, \di x +   \int_{\supp \nabla \eta}  w_\vrho^{2-2\alpha} (\Delta \eta)^2  {\psi}^2  \, \di x, \notag\end{align}
where $\supp \nabla \eta \subset \set{\frac \delta 8 \le \abs{x} \le \frac \delta 4}\cup \set{AR \le \abs{x} \le AR +{D_0}}$.

It follows from \eq{eq}, recalling  $\norm{V}_{\infty} \le K$, and using also the fact that $w_\vrho\le 1 $ on $\supp \eta$,   that 
\beq \begin{split}\label{longest0}
 \int_{\R^d} w_\vrho^{2-2\alpha}\eta^2 (\Delta  {\psi})^2 \, \di x & = \int_{\R^d}  w_\vrho^{2-2\alpha}\eta^2 (V {\psi} - {\zeta})^2  \, \di x\\
& \le  2  K^2 \int_{\R^d} w_\vrho^{-1-2\alpha}\eta^2  {\psi}^2  \, \di x +2 \int_{\R^d}  w_\vrho^{2-2\alpha}\eta^2 {\zeta}^2 \, \di x.
\end{split} \eeq
We take 
\beq
\alpha = \alpha_{0}\rho^{\frac 4 3}\quad \text{where} \quad \alpha_{0}\ge \max \set{\pa{18  C_3 K^2}^{\frac 1 3}, C_2 \pa{8 C_1{D_0}}^{-\frac 4 3}}, \label{alpha0}
\eeq
so we have
\beq \label{alpha01}
 \frac {\alpha^3}{3 C_3 \vrho^4} = \frac {\alpha_0^3}{3 C_3 } \ge  6 {K^2} .
\eeq
Using \eq{wvrho} and \eq{01}, and recalling that  $\diam {\Theta}\le D\le R$, we have
\beq\label{longest2}
 \int_{\R^d}w_\vrho^{-1-2\alpha} \eta^2  {\psi}^2  \, \di x \ge  \pa{\frac { \vrho}{R + \diam {\Theta}}}^{1 + 2 \alpha}\norm{ {\psi}_{{\Theta}}}^2_2  \ge { A}^{1 + 2 \alpha}.
 \eeq
Combining \eq{longest}, \eq{longest0}, \eq{alpha01}, and \eq{longest2},
we  conclude that
\begin{align}
 &  \frac {2\alpha_0^3}{9 C_3 }{A}^{1 + 2 \alpha}\le  \\ & \qquad  4 \int_{\supp \nabla \eta}  w_\vrho^{2-2\alpha} \abs{\nabla \eta}^2 \abs{\nabla  {\psi}}^2  \, \di x +   \int_{\supp \nabla \eta} w_\vrho^{2-2\alpha} (\Delta \eta)^2  {\psi}^2  \, \di x + 2  \int_{\supp  \eta}  w_\vrho^{2-2\alpha}\eta^2 {\zeta}^2 \, \di x. \notag
\end{align}

We have
 \begin{align}\notag
&  \int_{ \set{AR \le \abs{x} \le AR +{D_0}}}   w_\vrho^{2-2\alpha}\pa{ 4 \abs{\nabla \eta}^2 \abs{\nabla  {\psi}}^2 +  (\Delta \eta)^2  {\psi}^2 } \, \di x \\\notag
&\qquad  \qquad \qquad   \le  C_4^2 d^2 \pa{\frac {C_1 \vrho}{AR}}^{2\alpha -2} \int_{ \set{AR \le \abs{x} \le AR +{D_0}}}  \pa{ 4  \abs{\nabla  {\psi}}^2 +    {\psi}^2 } \, \di x\\
&\qquad \qquad  \qquad  \le  C_5  \pa{\frac {C_1 \vrho}{AR}}^{2\alpha -2} \int_{ \set{AR -{D_0} \le \abs{x} \le AR +2{D_0}}}  \pa{{\zeta}^2 + (1 +K)  {\psi}^2 } \, \di x\\ \notag
&\qquad \qquad \qquad  \le  C_5 \pa{\frac {C_1 \vrho}{AR}}^{2\alpha -2}\pa{\norm{{\zeta_{G}}}_2^2 +  (1 +K)\norm{ {\psi}_{G}}_{2}^{2}}\\ 
&\qquad \qquad  \qquad  = C_5 \pa{2C_1 }^{2\alpha -2}\pa{\norm{{\zeta_{G}}}_2^2 +  (1 +K)\norm{ {\psi}_{G}}_{2}^{2}} , \notag
\end{align}
where we used an interior estimate (	e.g., \cite[Lemma~A.2]{GK5}) and $C_5=C_5(d,\delta,{D_0})$ is a constant.

Similarly,
 \begin{align}\notag
&\int_{\set{\frac \delta 8 \le \abs{x} \le \frac \delta 4}}  w_\vrho^{2-2\alpha}\pa{ 4 \abs{\nabla \eta}^2 \abs{\nabla  {\psi}}^2 +  (\Delta \eta)^2  {\psi}^2 } \, \di x\\ \notag
&\qquad  \qquad  \qquad \qquad  \le  C_4^2 d^2 \pa{8 \delta^{-1} C_1 \vrho}^{2\alpha -2} \int_{\set{\frac \delta 8 \le \abs{x} \le \frac \delta 4}}  \pa{ 4  \abs{\nabla  {\psi}}^2 +    {\psi}^2 } \, \di x\\ \notag
&\qquad \qquad \qquad  \qquad  \le  C_{6}  \pa{8  \delta^{-1} C_1 \vrho}^{2\alpha -2} \int_{ \set{ \abs{x} \le \frac \delta 2}}    \pa{{\zeta}^2 + (1 +K)  {\psi}^2 } \, \di x \\
&\qquad  \qquad \qquad  \qquad  \le    C_{6}   \pa{8  \delta^{-1} C_1 \vrho}^{2\alpha -2}\pa{\norm{{\zeta_{G}}}_2^2 +  (1 +K)\norm{ {\psi}_{0,\delta}}^2_2}\\ \notag
&\qquad  \qquad \qquad  \qquad  =   C_6 \pa{16  \delta^{-1} C_1A R}^{2\alpha -2}\pa{\norm{{\zeta_{G}}}_2^2 +  (1 +K)\norm{ {\psi}_{0,\delta}}^2_2}\\\notag
&\qquad   \qquad \qquad  \qquad  = C_6   \pa{64  \delta^{-1}   C_1^{2} R}^{2\alpha -2}\pa{\norm{{\zeta_{G}}}_2^2 +  (1 +K)\norm{ {\psi}_{0,\delta}}^2_2}, \notag
\end{align}
where  $C_6=C_6(d,\delta,{D_0})$ is a constant.
  
In addition, 
\beq
2 \int_{\supp  \eta} w_\vrho^{2-2\alpha}\eta^2 {\zeta}^2 \, \di x \le  2\pa{8  \delta^{-1}  C_1 \vrho}^{2\alpha -2} \norm{{\zeta_{G}}}_2^2= 2 \pa{64  \delta^{-1} C_1^{2} R}^{2\alpha -2}\norm{{\zeta_{G}}}_2^2.
\eeq

Thus, if
\beq
C_5  (1 +K) \norm{ {\psi}_{G}}_{2}^{2} \pa{2C_1 }^{2\alpha -2} \le \tfrac 1 2  \frac {2\alpha_0^3}{9 C_3 }{A}^{1 + 2 \alpha}= { \frac {\alpha_0^3}{9 C_3 }}\pa{4 C_{1}}^{1 + 2 \alpha},
\eeq
or, equivalently,
\beq  \label{keycond}
\alpha_0^{3} 4^{\alpha}\ge \tfrac 9 {16}  C_{5}C_{1}^{-2} (1 +K)\norm{ {\psi}_{G}}_{2}^{2} ,
\eeq
we conclude that
\begin{align} \label{almostbound}
& { \frac {\alpha_0^3}{9 C_3 }}\pa{4 C_{1}}^{1 + 2 \alpha} \\
& \quad  \le   C_6  \pa{64  \delta^{-1}   C_1^{2} R}^{2\alpha-2 }  (1 +K)\norm{ {\psi}_{0,\delta}}^2_2   +\pa{ \pa{C_6 +2}\pa{64  \delta^{-1}   C_1^{2} R}^{2\alpha-2 }  +  C_5 \pa{2C_1 }^{2\alpha -2} 
}\norm{{\zeta}_{G}}_2^2 \notag \\
& \quad \le  C_7  \pa{ \beta_1  \delta^{-1}   C_1^{2} R}^{2\alpha-2 }\pa{  (1 +K)\norm{ {\psi}_{0,\delta}}^2_2  +  \norm{{\zeta}_{G}}_2^2  },\notag
 \end{align}
 where we used  $R \ge D\ge D_0$,   set $C_7= \max \set{C_5,C_6 +2 }$, and took
 \beq
 \beta_1=   \max \set{64, 2 \delta \pa{C_1{D_0}}^{-1}}.
 \eeq
It follows that 
\begin{align} \label{almostbound3}
& C_8   { \frac {\alpha_0^3}{ C_3 }}
  \pa{\beta R}^{ - 2\alpha } \le    (1 +K)  \norm{ {\psi}_{0,\delta}}^2_2   +
\norm{{\zeta}_{G}}_2^2 ,
 \end{align}
with a constants  $C_8= C_8(d,\delta,{D_0},C_1) >0$ and 
\beq \label{beta}
\beta= \tfrac 1 4 \beta_1 \delta^{-1}   C_1 = \max \set{ 16 \delta^{-1}   C_1,\pa{2{D_0}}^{-1}}.
\eeq

 Since $R \ge D$ and we require \eq{alpha0}, to satisfy \eq{keycond}  it suffices to also require
 \beq
  4^{\alpha_0 \pa{4C_1{D_0}}^{\frac 4 3}}\ge \tfrac 9 {16} C_2^{-3} \pa{8 C_1{D_0}}^{ 4 } C_{5}C_{1}^{-2} (1 +K)\norm{ {\psi}_{G}}_{2}^{2},
 \eeq
 that is,
 \begin{align}
\alpha_0\ge  \pa{4C_1{D_0}}^{-\frac 4 3} \pa{\log 4}^{-1} \log \pa{   C_{9}  (1 +K)\norm{ {\psi}_{G}}_{2}^{2}},
\end{align}
 where $C_9=C_9(d,\delta,{D_0})$.

Thus we can satisfy \eq{alpha0} and \eq{keycond} by taking
\beq
\alpha= \alpha_1 R^{\frac 4 3},\quad \text{with} \quad
\alpha_1= C_{10}\pa{1 + K^{\frac 2 3} + \log  \norm{ {\psi}_{G}}_{2}},\label{alpha1}
\eeq
 for some appropriate constant  $C_{10}=C_{10}(d,\delta,{D_0})$.

 It now  follows from \eq{almostbound3}, \eq{beta} and  \eq{alpha1} that we can find  a constant $m=m(d,\delta,{D_0})>0$ such that
 \beq
R^{-m\pa{1 + K^{\frac 2 3} + \log  \norm{ {\psi}_{G}}_{2}} R^{\frac 4 3}} \le    (1 +K)  \norm{ {\psi}_{0,\delta}}^2_2   +
\norm{{\zeta}_{G}}_2^2
 \eeq
 for all $R \ge D$.
\end{proof}

\begin{proof}[Proof of Corollary~\ref{corQUCPD}]  Without loss of generality we take $x_0=0$, i.e., $\La=\La_{L}(0)$.   We will prove the corollary  for the case of Dirichlet boundary condition, the modifications for the (easier) case of periodic boundary condition will be obvious.

Let  $\Delta_\Lambda$  be the Dirichlet Laplacian  on $\Lambda$, and let    $V$ be a bounded potential on $\La$  with  $\|V\|_\infty \le K<\infty$. 
Given $\vphi \in \mathrm{L}^2(\Lambda)$, we extend it to a function $\widetilde{\vphi}\in \mathrm{L}^2_{\mathrm{loc}}(\R^d)$ by setting $\widetilde{\vphi}=\vphi$ on $\Lambda$ and $\widetilde{\vphi}=0 $ on $\partial \Lambda$, and requiring that 
for all $x \in \R^d$ and $j \in \set{1,2 \dots,d}$ we have
\beq \label{widetildephi}
\widetilde{\vphi}(x)=- \widetilde{\vphi}(x + (L  -2 \widehat{x_j}) \e_j),
\eeq
where $\set{\e_j}_{j =1,2\ldots,d}$ is the canonical orthonormal basis in $\R^d$, and for each  $t\in  \R$ we define $\hat{t}\in  ] -\frac L 2, \frac L 2]$ by $t =
kL + \hat{t}$ with $k \in \Z$.   Note that if $\Lambda^\pr= \Lambda_{L^\pr}(0)=  ] -\frac {L^\pr} 2, \frac {L^\pr}  2[^d$,  we have
 \beq \label{normtilde0}
 \| \widetilde{\vphi} _{ \Lambda^\pr}  \|_2^{2} = \pa{2n+1}^{d}  \| {\vphi}_{ \Lambda}  \|_2^{2} \qtx{if} L^\pr= \pa{2n +1}L  \qtx{for some} n\in \N.
 \eeq
We also  extend the potential $V$ to a potential $\widehat{V}$ on  $\R^d$ by  by setting $\widehat{V}=V$ on $\Lambda$ and $V=0$ on $\partial \Lambda$, and requiring that 
for all $x \in \R^d$ and $j \in \set{1,2 \dots,d}$ we have
\beq
\widehat{V}(x)=\widehat{V}(x + (L  -2 \widehat{x_j}) \e_j).
\eeq
In particular,  $\|\widehat{V}\|_\infty =\norm{V}_\infty \le K$.

 Using the fact that for all    eigenfunctions $\phi$ of $\Delta_\Lambda$ (given explicitly  in \cite[Eq. (113) in Chapter XIII]{RS4})  we have  $\widetilde{\phi}\in C^\infty(\R^d)$, we conclude that  $\psi \in \D(\Delta_\Lambda)$  implies   $\widetilde{\psi} \in \mathrm{H}^2_{\mathrm{loc}}(\R^d)$, satisfying
\beq\label{eqwidehat}
-\Delta \widetilde{\psi} + \widehat{V}  \widetilde{\psi} =  \widetilde{H_\Lambda \psi}\quad \text{a.e.\ in} \quad  
\R^d.
\eeq

Now let  $\delta, D,\Theta$ be as in Corollary~\ref{corQUCPD}\ref{corQUCPDi}, and set $D_{0}=D$. In view of \eq{xR1} we may assume $D\le R\le \sqrt{d}L$ without loss of generality.  We take $\Lambda_{1}= \Lambda_{L_1}(0)$, with
\beq\label{defL1}
L_1= \pa{2[[(4C_1 +2)\sqrt{d}]]+1}L \le 29 \sqrt{d}L ,
\eeq
where $[[t]]$ denotes the smallest integer bigger or equal to $t$,  and we used \eq{C1}. Fix
 $x\in \Lambda$ satisfying \eq{xR1},  it follows that  $x$ satisfies \eq{xR} with
 $G=\Lambda_{1}$.  We now apply Theorem~\ref{thmucp} with $G=\Lambda_{1}$. Given   $\psi \in \D(\Delta_\Lambda)$,   $ \widetilde{\psi}$ satisfies \eq{eqwidehat} on $\Lambda_{1}$, and hence  \eq{UCPbound} yields
 \begin{align} \label{UCPbound3}
  (1 +K) \norm{{ \widetilde{\psi}_{x,\delta}}}^2_2   +
\norm{{\pa{ \widetilde{H_\Lambda \psi}}_{\La_{1}}}}_2^2 \ge
  R^{-{m}\pa{1 + K^{\frac 2 3} + \log  \pa{ {\norm{ \widetilde{\psi}_{\La_{1}}}_{2}}{ \norm{ \widetilde{\psi}_{\Theta}}_2^{-1}}}}R^{\frac 4 3}}\norm{ \widetilde{\psi}_{{\Theta}}}_2 ^2,
   \end{align}
with a constant  ${m}= {m}(d,\delta,D) >0$. Taking into account \eq{widetildephi}, \eq{normtilde0}, and \eq{defL1}, we get \eq{UCPbound1}

To prove  Corollary~\ref{corQUCPD}\ref{corQUCPDii}, let $L \ge 2$,  $0<\delta \le L$, and $x \in \Lambda$ with $B(x,\tfrac \delta 2) \subset \Lambda$.  We take $\Lambda_{2}= \Lambda_{L_2}(0)$, with
\beq\label{defL2}
L_2= \pa{2[[(6C_1 +3)]]+1}L \le 41L,
\eeq
where we used  \eq{C1}. We let  $\Theta_x= \Lambda + 2L \e^{(x)}\subset \Lambda_2$, where $\e^{(x)} \in \set{\pm \e_{j}}_{j=1,2,\ldots,d}$ is chosen such that   $R:= \dist \pa{x, {\Theta_x}}  \in \br{L, \frac 3 2 L}$.  It follows that  $x$ satisfies \eq{xR} with
 $G=\Lambda_{2}$, so we apply Theorem~\ref{thmucp} with $G=\Lambda_{2}$,  $D_{0}=\frac \delta 2$, $D=L$, and $\Theta=\Theta_x$. Given   $\psi \in \D(\Delta_\Lambda)$,   $ \widetilde{\psi}$ satisfies \eq{eqwidehat} on $\Lambda_{2}$,  we have $\|{\widetilde{\psi}_{{\Theta}}}\|_2=\norm{\psi_{\Lambda}}_{2}$, and hence  \eq{UCPbound} yields
 \begin{align} \label{UCPbound4}
  (1 +K) \norm{{ \widetilde{\psi}_{x,\delta}}}^2_2   +
\norm{{\pa{ \widetilde{H_\Lambda \psi}}_{\La_{2}}}}_2^2 \ge
  \pa{\tfrac 3 2 L}^{-{m^\pr}\pa{1 + K^{\frac 2 3} } \pa{\tfrac 3 2 L}^{\frac 4 3}}\norm{\psi_{\Lambda}}_{2},
   \end{align}
with a constant  ${m^\pr}= {m^\pr}(d,\delta) >0$.  Using \eq{widetildephi}, \eq{normtilde0}, and \eq{defL2}, we get \eq{UCPbound12}
 \end{proof}

\subsection{Application to Schr\"odinger operators with periodic potentials}

 Consider  the Schr\"odinger operator $H =   -\Delta + V$
on 
$\mathrm{L}^2(\mathbb{R}^d)$, 
where $\Delta$ is the $d$-dimensional Laplacian operator and  $V$ is a bounded periodic potential with period $q >0 $, i.e., periodic with respect to the group $q \Z^d$.  Without loss of generality  we assume $ \inf \sigma(H)=0$, i.e., $0 \in  \sigma(H) \subset [0,\infty[$.

Given $\delta \in ]0, q]$, we set ${b_{\delta}}= \Chi_{B(0,\frac \delta 2)} $, and consider the $q$-periodic bounded operator  $W_{\delta}$ on $\mathrm{L}^2(\mathbb{R}^d)$ given by multiplication by the function
\beq\label{multW}
 W_{\delta}(x) = \sum_{m \in q\Z^d}  {b_{\delta}}(x-m).
      \eeq

 We also consider  the corresponding finite volume operators.     Given $L \in q\N$, we set $H_L =   -\Delta_L + V$ on $\mathrm{L}^2(\Lambda_L,\di x)$, where $\Lambda_{L}=\Lambda_{L}(0)$.
$\Delta_L$ is the Laplacian with periodic boundary condition on $\Lambda_L$, which we  identify with the torus $\R^d / L\Z^d$ in the usual way.  We will also write $H_{\infty}=H$.

Combes, Hislop and Klopp \cite[Section~4]{CHK1} proved that for every compact interval $I$ there exists a constant  $C_{I,\delta} =C_{d,V,I,\delta}>0$, such that  for all $L\in q\N \cup \{\infty\}$ we have 
\beq
\Chi_{I}(H_{L}) W_{\delta} \Chi_{I}(H_{L}) \ge C_{I,\delta}  \Chi_{I}(H_{L}).
\eeq
Their proof relies on the unique continuation principle for Schr\" odinger operators, and for this reason does not provide much information on the constant $C_{I,\delta}>0$.  We will show that the quantitative unique continuation principle can be used to prove a modified form of their result with control of the constant.  

\begin{theorem} Let  $H =   -\Delta + V$ be a periodic  Schr\"odinger operator on 
$\mathrm{L}^2(\mathbb{R}^d)$ as above, with period $q \ge 2$, and let $W_{\delta}$ be as in \eq{multW}.  Given $E_{0}>0$, set $K_{0}= E_{0} + \norm{V}_{\infty}$.  There exists a constant $\widehat{m}=\widehat{m}(d,\delta) >0$, such that, defining $\gamma >0$  by
\beq\label{defgamma}
\gamma^{2} =  \tfrac 1 2 \pa{41}^{-d}  q^{-\widehat{m} \pa{1 + K_{0}^{\frac 2 3} }q^{\frac 4 3}},
\eeq
 for any closed  interval $I \subset [0,E_{0}]$ with $\abs{I} \le 2\gamma$ and any scale $L\in q\N \cup \{\infty\}$ we have
\beq \label{floqf230}
\Chi_{I}(H_{L}) W_{\delta} \Chi_{I}(H_{L}) \ge  \pa{41}^{d}\gamma
^{2}(1 + K_{0})^{-1}  \Chi_{I}(H_{L}).
\eeq
\end{theorem}

\begin{proof}

 We will need to review Floquet Theory (see \cite[Section XIII.6]{RS4}). We let $Q={\Lambda}_{q}(0)$ be the basic period cell, $\widetilde{Q}={
{\Lambda}}_{\frac{2\pi }{q}}(0)$ the dual basic cell.  
 We define the Floquet
transform 
\begin{equation}
{\cF}\colon \mathrm{L}^2(\R^d, \di x) \to
 \int_{\widetilde{Q}}^{\oplus }\mathrm{L}^2(Q, \di x) \,dk\cong L^{2}\left(\widetilde{Q}, \di k
;\mathrm{L}^2(Q, \di x) \right)
\end{equation}
by 
\begin{equation}
({\cF}\psi )(k,x)=\left( \frac{q}{2\pi }\right) ^{\frac{d}{2}}\sum_{m\in
q \mathbb{Z}^{d}}{\rm e}^{-ik\cdot m}\psi (x-m),\;\;\ x\in
Q,\ k\in \widetilde{Q},
\end{equation}
if $\psi $ has compact support; it extends by continuity to a unitary
operator.

The $q$-periodic operator $H$ is decomposable in this direct integral
representation, more precisely, 
\begin{equation} \label{floq}
{\cF}H{\cF}^{*}=\int_{\widetilde{Q}}^{\oplus }{H}
_{Q}(k)\,\di k,
\end{equation}
where for each $k\in \mathbb{R}^{d}$ we set $H_Q(k) = -{\Delta}_Q(k)  + V$, where  ${\Delta}_Q(k)$ is the 
Laplacian on $Q$ with $k$-quasi-periodic boundary condition, i.e., defined on functions of the form $\psi(x)={\rm e}^{-ik\cdot x} \vphi(x)$ with $\vphi$ a periodic function on $Q$.  Note that $H_Q(0)=H_q$.  Moreover, If  
 $p \in \frac{2 \pi}{q} \mathbb{Z}^d$,
 then for all $k \in \mathbb{R}^d$ we have 
${H}_{Q}(k + p) =\mathrm{e}^{-i p \cdot x}{H}_{Q}(k)\mathrm{e}^{i p \cdot x}$.

If $\L\in q\N$, similar considerations apply to
the operator $H_L$, which is $q$-periodic on the torus 
$\Lambda_L \cong \R^d / L\Z^d $. The Floquet transform 
\begin{equation}
{\cF}_{L }\colon \mathrm{L}^2(\Lambda_L, \di x)
 \to \bigoplus_{k\in \frac{2\pi }{
L }\mathbb{Z}^{d}\cap \widetilde{Q}}\mathrm{L}^2(Q, \di x)
\end{equation}
is a unitary operator now defined by 
\begin{equation}
({\cF}_{L }\psi )(k,x)=\left( \frac{q}{L }\right) ^{\frac{d}{2}
}\sum_{m\in q\mathbb{Z}^{d}\cap {\Lambda}_{L }}
{\rm e}^{-ik\cdot m}\psi (x-m),
\end{equation}
where $x\in Q,\ k\in \frac{2\pi }{L }\mathbb{Z}^{d}\cap 
\widetilde{Q},\psi \in  \mathrm{L}^2(\Lambda_L, \di x)$, and  $\psi (x-m)$ is properly
interpreted in the torus ${\Lambda }_{L }$. We also
have 
\begin{equation}
{\cF}_{L }{H}_{L}{\cF}_{L
}^{*}=\bigoplus_{k\in \frac{2\pi }{L }\mathbb{Z}^{d}\cap 
\widetilde{Q}}{H}_{Q}(k) .
\end{equation}
It follows that for any bounded Borel function $f$ we have
\begin{equation} \label{floqf}
{\cF}f(H){\cF}^{*}=\int_{\widetilde{Q}}^{\oplus }f({H}
_{Q}(k))\,\di k, \quad {\cF}_{L }f({H}_{L}){\cF}_{L
}^{*}=\bigoplus_{k\in \frac{2\pi }{L }\mathbb{Z}^{d}\cap 
\widetilde{Q}}f({H}_{Q}(k)).
\end{equation}

Let us fix $\delta \in ]0, q]$ and  $E_0>0$. We set $K_{0}= \norm{V}_{\infty}+E_0$, so $\norm{V -E}_\infty \le K_{0}$ for all $E \in I_0$. Given $k \in \widetilde{Q}$, we consider the Schr\" odinger operator ${H}_{Q}(k)$ on $\L^2(Q)$, and   proceed  similarly to  the proof of  Corollary~\ref{corQUCPD}\ref{corQUCPDii}. Since we have   $k$-quasi-periodic boundary condition,    we extend a function  $\vphi \in \mathrm{L}^2(Q)$ to a function $\widetilde{\vphi}\in \mathrm{L}^2_{\mathrm{loc}}(\R^d)$ by requiring
$\widetilde{\vphi}={\vphi}$ on $Q$ and
$\widetilde{\vphi}( x+ m)={\rm e}^{-ik\cdot m} \widetilde{\vphi}(x)$ for all $x \in \R^d$ and $m \in q\Z^d$.  If $\psi \in \D({\Delta}_Q(k))$, then   $\widetilde{\psi} \in \mathrm{H}^2_{\mathrm{loc}}(\R^d)$ and we have
\beq\label{eqwidehat34}
-\Delta \widetilde{\psi} + {V}  \widetilde{\psi} =  \widetilde{{H}_{Q}(k) \psi}\quad \text{a.e.\ in} \quad  
\R^d.
\eeq
We apply Theorem~\ref{thmucp} with   $G= \Lambda_{L_2}(0)$, where  $L_2$ is given in \eq{defL2} (recall  $L=q$). 
Proceeding as in the derivation of \eq{UCPbound4} and \eq{UCPbound12}, using $q\ge 2$, we get
 \begin{align} \label{UCPboundper}
  (1 +K_{0}) \norm{\pa{{b_{\delta}}\psi }_{Q}}^2_2   + (41 )^{d}
\norm{\pa{\pa{{H}_{Q}(k) -E}\psi}_{Q}}_{2}^2 \ge
 q^{-\widehat{m} \pa{1 + K_{0}^{\frac 2 3} }q^{\frac 4 3}}\norm{\psi_{Q} }_2 ^2
   \end{align}
for all $E \in [0,E_0]$,  with a constant
   $\widehat{m}=\widehat{m}(d,\delta)>0$.

   We now take  $I=[E-\eps,E +\eps]\subset [0,E_0]$.  If  $\psi= {\Chi}_I({H}_{Q}(k)) \psi$,  we have
      \beq
   \norm{\pa{\pa{{H}_{Q}(k) -E}\psi}_{Q}}_{2}\le \eps \norm{\psi_{Q}}_{2},
  \eeq 
 and it  follows from \eq{UCPboundper} that
 \begin{align} \label{UCPboundper22}
   (1 +K_{0}) \norm{\pa{{b_{\delta}}\psi }_{Q}}^2_2  +
\eps^2  ( 41)^{d}\norm{\psi_{Q}}_{2}^2 \ge
 q^{-\widehat{m} \pa{1 + K_{0}^{\frac 2 3} }q^{\frac 4 3}}\norm{\psi_{Q} }_2 ^2.
   \end{align} 
  Thus, if $\eps \le \gamma$, where $\gamma$
   is given in \eq{defgamma},
    we get
  \begin{align} \label{UCPboundper33}
   (1 +K_{0}) \norm{\pa{{b_{\delta}}\psi}_{Q}}^2_2   \ge
\tfrac 1 2   q^{-\widehat{m} \pa{1 + K_{0}^{\frac 2 3} }q^{\frac 4 3}}\norm{\psi_{Q} }_2 ^2=
 \pa{41}^{d}\gamma^{2}\norm{\psi_{Q} }_2 ^{2},
   \end{align} 
   that is,
   \beq \label{unifk}
   {\Chi}_I({H}_{Q}(k)){b_{\delta}}{\Chi}_I({H}_{Q}(k)) \ge  \pa{41}^{d}\gamma
^{2}(1 + K_{0})^{-1}  {\Chi}_I({H}_{Q}(k)).
   \eeq

   Given an  interval $I$, we have 
  \begin{align} \label{floqf2}
{\cF}\set{{\Chi}_I(H)W_{\delta}{\Chi}_I(H)}{\cF}^{*} &=\int_{\widetilde{Q}}^{\oplus }\set{{\Chi}_I({H}_{Q}(k)){b_{\delta}}{\Chi}_I({H}_{Q}(k))}\,\di k,  \\
\intertext{and, for $L\in q\N$,}
 {\cF}_{L }\set{{\Chi}_I(H_L)W_{\delta}{\Chi}_I(H_L)}{\cF}_{L
}^{*}&=\bigoplus_{k\in \frac{2\pi }{L }\mathbb{Z}^{d}\cap 
\widetilde{Q}}\set{{\Chi}_I({H}_{Q}(k)){b_{\delta}}{\Chi}_I({H}_{Q}(k))}.\label{floqf2L}
\end{align}
Thus for   $I=[E-\eps,E +\eps]\subset [0,E_0]$, with  $\eps\le \gamma$,  it follows from   \eq{unifk}, \eq{floqf2}, and \eq{floqf2L}, that  for all  $L\in q\N \cup\set{\infty}$ we have 
 \begin{align} \label{floqf23}
 {{\Chi}_I(H_{L})W_{\delta}{\Chi}_I(H_{L})}\ge   \pa{41}^{d}\gamma
^{2}(1 + K_{0})^{-1}  {\Chi}_I({H_{L}}),   
\end{align}
so we proved \eq{floqf230}.
\end{proof}

\begin{remark}
Note that \eq{floqf230} holds for $I=[0,E_1]$ where 
\beq 
E_1^{2} = 2 \pa{41}^{-d} q^{-\widehat{m}  \pa{1 + \pa{\norm{V}_{\infty} + E_1}^{\frac 2 3} }q^{\frac 4 3}}.
\eeq
Note that this equation has a solution $E_1 >0$.
\end{remark}

\bigskip
\footnotesize
\noindent\textit{Acknowledgments.}
Fran\c cois  Germinet was partially supported by the ANR BLAN 0261.
Abel Klein was  supported in part by the NSF under grant DMS-1001509.

\end{document}